\documentclass[final,5p,times,twocolumn]{elsarticle}

\usepackage{amssymb}
\usepackage{amsmath}

\usepackage{amsmath,amsfonts}
\usepackage{algorithmic}
\usepackage{algorithm}
\usepackage{array}
\usepackage[caption=false,font=normalsize,labelfont=sf,textfont=sf]{subfig}
\usepackage{textcomp}
\usepackage{stfloats}
\usepackage{verbatim}
\usepackage{graphicx}

\usepackage{orcidlink}
\usepackage{amsmath,amssymb,amsfonts}
\usepackage{algorithmic}
\usepackage{graphicx}
\usepackage{textcomp}
\usepackage{xcolor}
\usepackage[switch]{lineno}
\usepackage{algorithmic}
\usepackage{algorithm}

 \usepackage[english]{babel}
 \usepackage{pstricks}
 \usepackage{verbatim}
 \usepackage{graphicx}
 \usepackage{psfrag}
 \usepackage{fancybox}
 \usepackage{amsmath}
 \usepackage{amsfonts}
 \usepackage{amssymb}
 \usepackage{listings}
 \usepackage{enumerate}

\newcommand\notsotiny{\@setfontsize\notsotiny\@vipt\@viipt}
 \makeatother
 
\csname PSTplotLoaded\endcsname
\let\PSTplotLoaded 
\ifx\PSTricksLoaded \else\input pstricks.tex\fi
\ifx\PSTXKeyLoaded \else \input pst-xkey.tex\fi
\ifx\PSTtoolsloaded \else\input pst-tools.tex\fi
\ifx\PSTtoolsloaded \else\input pstricks-add.tex\fi
\ifx\PSTFPloaded \else   \input pst-fp.tex  \fi
\ifx\MultidoLoaded \else \input multido.tex \fi
\def\fileversion{1.92}
\def\filedate{2019/05/16}
\edef\TheAtCode{\the\catcode`\@}
\catcode`\@=11
\pst@addfams{pst-plot}
\def\pst@linetype{2}%
\SpecialCoor

\newdimen\pstRadUnit
\newdimen\pstRadUnitInv
\begingroup
\catcode`\{=13
\catcode`\}=13
\catcode`\(=13
\catcode`\)=13
\catcode`\,=13
\catcode`\!=1
\catcode`\*=2
\catcode`\ =13
\catcode`\_=13
\catcode`\^^M=13
\gdef\pst@datadelimiters!% Begin def
\catcode`\{=13%
\catcode`\}=13%
\catcode`\(=13%
\catcode`\)=13%
\catcode`\,=13%
\catcode`\ =13%
\catcode`\^^M=13%
\def,##1!%
\ifcat\noexpand,\noexpand##1%
\expandafter##1%
\else\space%
D\space##1%
\fi*%
\let(,\let),\let{,\let},\let ,\let^^M,\let_\@empty*% End def
\endgroup
\define@key[psset]{pst-plot}{ignoreLines}[0]{\def\psk@ignoreLines{#1}}
\psset[pst-plot]{ignoreLines=0}
\newcount\pst@linecnt
\begingroup
\catcode`\,=13
\catcode`\_=13
\gdef\savedata@#1[#2]{%
  \xdef\pst@tempg{#2_}%
  \endgroup
  \let#1\pst@tempg
  \global\let\pst@tempg\relax
  \ignorespaces}
\gdef\readdata@{%
  \read1 to \pst@tempA
  \ifnum\pst@linecnt=\psk@nStep
    \global\pst@linecnt=0
    \expandafter\readdata@@\pst@tempA_\@nil
  \fi
  \global\advance\pst@linecnt by 1
  \ifeof1\else\expandafter\readdata@\fi}
\gdef\pst@@readfile#1#2\@nil{\addto@pscode{,#1#2}}%
\gdef\readdata@@#1#2\@nil{\xdef\pst@tempg{\pst@tempg,#1#2}}%
\endgroup
\def\readdata{\@ifnextchar[{\readdata@i}{\readdata@i[]}}
\def\readdata@i[#1]#2#3{%
  \openin1=#3
  \begingroup
  \ifx#1\relax#1\else\psset{#1}\fi
  \def\pst@tempg{}%
  \ifeof1
    \@pstrickserr{Data file `#3' not found.}\@ehpa
  \else
    \pst@datadelimiters
    \catcode`\[=1
    \catcode`\]=2
    \pst@cnta=0
    \loop \ifnum\the\pst@cnta<\psk@ignoreLines
      \advance\pst@cnta by 1\relax
      \read1 to \pst@tempA
    \repeat
    \psDEBUG[pst-plot]{>>> ignored \the\pst@cnta\space data lines}%
    \global\pst@linecnt=\psk@nStep
    \readdata@
  \fi
  \endgroup
  \global\let#2\pst@tempg%
  \global\let\pst@tempg\relax%
\ignorespaces}
\def\pst@readfile#1{{\let\readdata@@\pst@@readfile\readdata\pst@tempg{#1}}}
\def\pst@altreadfile#1{%
  \openin1=#1
  \ifeof1
    \@pstrickserr{Data file `#1' not found.}\@ehpa
  \else
    \catcode`\{=10
    \catcode`\}=10
    \catcode`\(=10
    \catcode`\)=10
    \catcode`\,=10
    \catcode`\^^M=10
    \catcode`\[=1
    \catcode`\]=2
    \pst@@altreadfile
  \fi}
\def\pst@@altreadfile{%
  \read1 to \pst@tempg
  \expandafter\pst@@@altreadfile\pst@tempg\@empty\@nil
  \ifeof1\else\expandafter\pst@@@altreadfile\fi}
\def\pst@@@altreadfile#1#2\@nil{\addto@pscode{#1#2}}%
\def\savedata#1{\begingroup\pst@datadelimiters\savedata@{#1}}
\newread\RCD@file
\def\psreadDataColumn{\@ifnextchar[\psreadDataColumn@i{\psreadDataColumn@i[]}}
\def\psreadDataColumn@i[#1]{%
  \psset{#1}%
  \psreadDataColumn@ii
}
\def\psreadDataColumn@ii#1#2#3#4{%    #1:column #2:delimiter #3:\result #4:data file
  \immediate\openin\RCD@file=#4\relax
  \global\let#3=\@empty
  \pst@cnta=0
  \loop \ifnum\the\pst@cnta<\psk@ignoreLines
      \advance\pst@cnta by 1\relax
      \read\RCD@file to \@tempa
  \repeat
  \loop
    \read\RCD@file to \@tempa
    \ifeof\RCD@file\else
      \edef\@tempa{\@tempa#2}%
      \def\reserved@b{}%
      \@tempswafalse
      \@tempcnta=#1\relax
    \expandafter\@tfor\expandafter\reserved@a
      \expandafter:\expandafter=\@tempa\do{% loop over every token
      \if\reserved@a#2\relax%                delimiter?
        \advance\@tempcnta \m@ne
        \ifnum \@tempcnta=\z@
          \expandafter\g@addto@macro\expandafter#3\expandafter{\reserved@b\space}%
          \@tempswatrue
        \fi
        \def\reserved@b{}% ???
      \else
        \edef\reserved@b{\reserved@b\reserved@a}
      \fi
      \if@tempswa\@break@tfor\fi
    }%
  \repeat
  \immediate\closein\RCD@file
}

\def\beginplot@line{\begin@OpenObj}
\def\endplot@line{\psline@ii}
\def\beginplot@polygon{\begin@ClosedObj}
\def\endplot@polygon{\pspolygon@ii}
\def\beginplot@curve{\begin@OpenObj}
\def\endplot@curve{\pscurve@ii}
\def\beginplot@ecurve{\begin@OpenObj}
\def\endplot@ecurve{\psecurve@ii}
\def\beginplot@ccurve{\begin@ClosedObj}
\def\endplot@ccurve{\psccurve@ii}
\def\beginplot@dots{\begin@SpecialObj}
\def\endplot@dots{\psdots@ii}
\define@key[psset]{pst-plot}{Hue}[180]{\pst@getint{#1}\pst@HueAngle}
\psset[pst-plot]{Hue=180}
\def\beginplot@colordots{\begin@SpecialObj}
\def\endplot@colordots{%
  \addto@pscode{%
    \psk@dotsize
    \@nameuse{psds@\psk@dotstyle}
    newpath
    /MaxValue 0 def
    /m n 2 mul def
    n { 
      dup MaxValue gt { dup /MaxValue ED } if
      m 2 roll
    } repeat
    n { dup MaxValue div  % y y
      \pst@number\psyunit div abs % to orig y value
      \pst@HueAngle\space 360 div exch dup sethsbcolor % 180 Y Y hsb color
      transform floor .5 add exch floor
      .5 add exch itransform Dot stroke } repeat }%
  \end@SpecialObj%
}
\def\beginplot@bubble{\begin@SpecialObj}
\def\endplot@bubble{%
  \addto@pscode{%
    newpath
    n { dup % x y y
      \pst@number\psyunit div abs % to orig y value
      transform floor .5 add exch floor
      .5 add exch itransform 
      0 360 arc \psk@fill 
      stroke } repeat }%
  \end@SpecialObj%
}
\def\beginplot@bezier{\begin@OpenObj}
\def\endplot@bezier{\psbezier@ii}
\def\beginplot@cbezier{\begin@ClosedObj}
\def\endplot@cbezier{\pscbezier@ii}
\def\beginplot@cspline{\begin@OpenObj}%% Christoph Bersch
\def\endplot@cspline{\pscspline@ii}
\let\beginplot@LineToYAxis\beginplot@line  % all from pst-plot added 2007-06-26 (hv)
\def\endplot@LineToYAxis{\psLineToYAxis@ii}
\let\beginqp@LineToYAxis\beginqp@line
\let\doqp@LineToYAxis\doqp@line
\let\endqp@LineToYAxis\endqp@line
\let\testqp@LineToYAxis\testqp@line
\let\beginplot@LineToXAxis\beginplot@line
\def\endplot@LineToXAxis{\psLineToXAxis@ii}
\let\beginqp@LineToXAxis\beginqp@line
\let\doqp@LineToXAxis\doqp@line
\let\endqp@LineToXAxis\endqp@line
\let\testqp@LineToXAxis\testqp@line
\newif\ifPst@interrupt \Pst@interruptfalse
\define@key[psset]{pst-plot}{barwidth}[0.25cm]{\pst@getlength{#1}\Add@barwidth}
\psset[pst-plot]{barwidth=0.25cm}
\define@key[psset]{pst-plot}{interrupt}[]{\expandafter\pst@interrupt#1,,,\@nil}
\def\pst@interrupt#1,#2,#3,#4\@nil{%
  \ifx\relax#1\relax \Pst@interruptfalse
  \else
    \Pst@interrupttrue
    \def\pst@interrupt@YMax{#1 }%
    \def\pst@interrupt@YMaxSep{#2 }%
    \def\pst@interrupt@YMaxDiff{#3 }%
  \fi
}
\def\psbar@ii{\addto@pscode{false \tx@NArray \psbar@iii}}

\def\psbar@iii{%
  \ifPst@interrupt
    /YMax \pst@interrupt@YMax \strip@pt\psyunit\space mul def
    /YMaxSep \pst@interrupt@YMaxSep \strip@pt\psyunit\space mul def
    /YMaxDiff \pst@interrupt@YMaxDiff \strip@pt\psyunit\space mul def
    /Tilde { % on stack DX
      /Op ED % add or sub
      /DX ED
      currentpoint 2 copy
      /Y ED /X ED   % x y  
      X DX add Y YMaxSep 2 div Op   
      X DX dup add add Y           
      curveto
      currentpoint 2 copy pop /X ED 
      X DX add Y YMaxSep 2 div neg Op  
      X DX dup add add Y    
      curveto      
    } def  
    newpath
    n { 
      /Yval exch def /Xval exch def 
      Xval \number\Add@barwidth 0.5 mul sub 0 moveto 
      Yval YMax le {  
        0 Yval rlineto \number\Add@barwidth 0 rlineto 
        0 Yval neg rlineto \number\Add@barwidth neg 0 rlineto
      }{
        0 YMax rlineto 
        \number\Add@barwidth 4 div 
        { add } Tilde
        0 YMax neg rlineto 
        \number\Add@barwidth neg 0 rlineto
        closepath
        Xval \number\Add@barwidth 0.5 mul sub YMax YMaxSep add moveto 
        0 Yval YMax sub YMaxSep sub YMaxDiff sub rlineto 
        \number\Add@barwidth 0 rlineto 
        0 Yval YMax YMaxSep add sub YMaxDiff sub neg rlineto 
        \number\Add@barwidth 4 div neg 
        { sub } Tilde
      } ifelse
    } repeat
  \else
    newpath
    n { 
      /Yval exch def /Xval exch def 
      Xval \number\Add@barwidth 0.5 mul sub 0 moveto 
      0 Yval rlineto \number\Add@barwidth 0 rlineto 
      0 Yval neg rlineto \number\Add@barwidth neg 0 rlineto
    } repeat
  \fi
}%
\def\beginplot@bar{\begin@SpecialObj}
\def\endplot@bar{%
  \psbar@ii\psk@fillstyle\ifpsshadow\pst@closedshadow\fi%
  \pst@stroke
  \end@SpecialObj}
\def\psybar@ii{\addto@pscode{false \tx@NArray \psybar@iii}}
\def\psybar@iii{%
  newpath
  n { 
    /Yval exch def /Xval exch def 
    0 Yval \number\Add@barwidth 0.5 mul sub moveto 
    Xval 0 rlineto 0 \number\Add@barwidth rlineto 
    Xval neg 0 rlineto 0 \number\Add@barwidth neg rlineto
  } repeat
}%
\def\beginplot@ybar{\begin@SpecialObj}
\def\endplot@ybar{%
  \psybar@ii\psk@fillstyle\ifpsshadow\pst@closedshadow\fi%
  \pst@stroke
  \end@SpecialObj}
\def\psLSM@ii{\addto@pscode{ false \tx@NArray \psLSM@iii }}
\def\psLSM@iii{%
  /xiSquare 0 def				% xi*xi
  /xi 0 def					% xi
  /fi 0 def					% f(xi)
  /xifi 0 def					% xi*f(xi)
  exch dup dup /xEnd ED /xStart ED exch
  n { 						% number of data pairs
    /Yval ED /Xval ED 				% save x y values
    /xi xi Xval add def				% sum xi
    /xiSquare xiSquare Xval dup mul add def	% sum xi*xi
    /xifi xifi Xval Yval mul add def		% sum xi*yi, same as xi*f(xi)
    /fi fi Yval add def				% sum yi, same as f(xi)
    Xval xStart lt { /xStart Xval def } if	% find the lowest xi
    Xval xEnd gt { /xEnd Xval def } if		% find the largest xi
  } repeat
  /u xiSquare fi mul xi xifi mul sub n xiSquare mul xi dup mul sub div def
  /v n xifi mul xi fi mul sub n xiSquare mul xi dup mul sub div def
  \Pst@Debug\space 0 gt { 			% print the equation
    /NimbusSanL-Regu findfont 12 scalefont setfont	
    0 -50 moveto (y=) show 			% print y=
    v \pst@number\psyunit \pst@number\psxunit div div 20 string cvs show ( x+) show		% m*x+
    u \pst@number\psyunit div 20 string cvs show } if
  newpath
  (\psk@xStart) length 0 gt 			% special start value?
    { \psk@xStart\space \pst@number\psxunit mul }
    { xStart } ifelse 
  dup v mul u add 				% xStart f(xStart)  
  moveto		 			% goto first point x1 y(x1)
  (\psk@xEnd) length 0 gt 			% special end value?
    { \psk@xEnd\space \pst@number\psxunit mul }
    { xEnd } ifelse 
  dup v mul u add 				% xEnd f(xEnd)	
  lineto					% line to second point x2 y(x2)
}%
\def\beginplot@LSM{\begin@SpecialObj}
\def\endplot@LSM{%
  \psLSM@ii\psk@fillstyle\ifpsshadow\pst@closedshadow\fi%
  \pst@stroke
  \end@SpecialObj}
\define@key[psset]{pst-plot}{IQLfactor}{\pst@checknum{#1}\pst@IQLfactor}
\psset[pst-plot]{IQLfactor=1.5}
\define@key[psset]{pst-plot}{postAction}[]{\def\psk@postAction{%
  \ifx\relax#1\relax\else\pst@number\psyunit div #1 \pst@number\psyunit mul \fi }}
\psset[pst-plot]{postAction=}
\define@key[psset]{pst-plot}{mediancolor}[black]{\pst@getcolor{#1}\median@linecolor}
\psset[pst-plot]{mediancolor=black}
\define@boolkey[psset]{pst-plot}[Pst@]{markMedian}[true]{}
\psset[pst-plot]{markMedian=false}

\def\psBoxplot@ii{%
  \addto@pscode{
    /Barwidth \number\Add@barwidth 2 div def  
    /Endwidth Barwidth \psk@arrowlength\space mul def  
   NArray bubblesort
   /NArray ED 				% save sorted array
   [ NArray { yUnit mul } forall ] /NArray ED % multiply with y unit
   NArray 0 get /MinVal ED		% save minimum
   NArray m 1 sub get /MaxVal ED	% maximum
   m 2 div cvi /M ED 			% the middle
   NArray length 2 mod 0 eq {		% even numbers of entries
     M 1 sub NArray exch get 		% even number of values
     NArray M get          		% and the upper one
     add 2 div /Median ED  		% the median
   }{
     NArray M get /Median ED  		% odd numbers of values
   } ifelse
   m 4 mod 0 eq {	  		% get the lower Quartil even/odd
     m 4 div cvi dup 1 sub NArray exch get
     exch NArray exch get
     add 2 div floor /LowerQuartil ED
   }{ 
     NArray M 2 div cvi get /LowerQuartil ED 
   } ifelse				% end even/odd 
   m 0.75 mul dup dup cvi sub 0 eq {	% get the upper Quartil
     cvi dup 1 sub NArray exch get exch NArray exch get
     add 2 div floor /UpperQuartil ED
   }{					% upper quartil
     NArray m 0.75 mul floor cvi get /UpperQuartil ED
   } ifelse 
   /IQL UpperQuartil LowerQuartil sub \pst@IQLfactor\space mul def
   0 1 m 1 sub { % Index on stack
     dup /Index ED
     NArray exch get LowerQuartil sub abs IQL sub 0 gt { 
       \psk@dotsize
       \@nameuse{psds@\psk@dotstyle}
       0 NArray Index get \psk@postAction
       Dot
       NArray Index LowerQuartil UpperQuartil LowerQuartil sub \pst@IQLfactor\space mul sub 
       dup /MinVal ED put % replace with 1.5 IQL
       NArray Index 1 add get /MinVal ED 
    } { exit } ifelse
   } for
   m 1 sub -1 0 {	% Index on stack
     dup /Index ED
     NArray exch get UpperQuartil sub abs IQL sub 0 gt { 
       \psk@dotsize
       \@nameuse{psds@\psk@dotstyle}
       0 NArray Index get \psk@postAction\space
       Dot
       NArray Index UpperQuartil LowerQuartil sub \pst@IQLfactor\space mul UpperQuartil add 
       dup /MaxVal ED put % replace with 1.5 IQL
       NArray Index 1 sub get /MaxVal ED 
     }{ exit } ifelse
   } for
   Endwidth neg MaxVal \psk@postAction moveto			% we are on top / lower whisker
   Endwidth dup add 0 rlineto 
   0 MaxVal \psk@postAction moveto 
   0 UpperQuartil \psk@postAction lineto			% upper quartil
   MinVal \psk@postAction MaxVal \psk@postAction lt {
     0 LowerQuartil \psk@postAction moveto			% line to lower whisker
     0 MinVal \psk@postAction lineto 
     Endwidth neg MinVal \psk@postAction moveto 
     Endwidth dup add 0 rlineto 
   } if
   gsave
   \pst@number\pslinewidth SLW
   \pst@usecolor\pslinecolor
   \tx@setStrokeTransparency 
   \@nameuse{psls@\pslinestyle}
   stroke
   grestore
   newpath
   Barwidth neg LowerQuartil \psk@postAction moveto	% lower quartil
   Barwidth neg UpperQuartil \psk@postAction lineto
   Barwidth dup add 0 rlineto
   Barwidth LowerQuartil \psk@postAction lineto
   closepath
   \pst@usecolor\psfillcolor
   gsave \pst@usecolor\psfillcolor \tx@setTransparency fill grestore
   \@nameuse{psls@solid}
   \ifPst@markMedian
     \pst@number\pslabelsep neg Median moveto currentpoint 
     /YMedian ED /XMedian ED 
      Barwidth neg Median \psk@postAction lineto  % median
   \else
      Barwidth neg Median \psk@postAction moveto  % median
   \fi
   Barwidth dup add 0 rlineto 
   \pst@number\pslinewidth SLW
   \pst@usecolor\median@linecolor
   \tx@setStrokeTransparency
   stroke
  }
}% 
\def\beginplot@Boxplot{\init@pscode}
\def\endplot@Boxplot{%
  \psBoxplot@ii\psk@fillstyle\ifpsshadow\pst@closedshadow\fi%
  \pst@stroke
  \end@SpecialObj}
\def\psBoxplot{\def\pst@par{}\pst@object{psBoxplot}}
\def\psBoxplot@i#1{%
  \leavevmode
  \pst@killglue
  \begingroup
  \addbefore@par{barwidth=40pt,arrowlength=0.75}%
  \addto@par{plotstyle=Boxplot}%
  \use@par
  \@nameuse{beginplot@\psplotstyle}%
  \addto@pscode{
    /D {} def
    [ #1 ] /NArray ED 
    NArray aload length /m ED
    /xUnit \pst@number\psxunit def
    /yUnit \pst@number\psyunit def
  }%
  \@nameuse{endplot@\psplotstyle}%
  \ignorespaces%
}
\define@key[psset]{pst-plot}{plotstyle}[line]{%
  \@ifundefined{beginplot@#1}%
    {\@pstrickserr{Plot style `#1' not defined}\@eha}%
    {\def\psplotstyle{#1}}}
\psset[pst-plot]{plotstyle=line}
\define@key[psset]{pst-plot}{plotpoints}[50]{%
  \pst@cntg=#1\relax
  \ifnum\pst@cntg<2
    \@pstrickserr{plotpoints parameter must be at least 2}\@ehpa
  \else
    \advance\pst@cntg-1
    \edef\psk@plotpoints{\the\pst@cntg\space}%
  \fi}
\psset[pst-plot]{plotpoints=50}
\define@key[psset]{pst-plot}{yMaxValue}[1.e30]{\def\psk@yMaxValue{#1 }\def\psk@yMinValue{#1 neg }}
\psset{yMaxValue=1.e30}
\define@key[psset]{pst-plot}{yMinValue}[-1.e30]{\def\psk@yMinValue{#1 }}
\psset{yMinValue=-1.e30}
\def\beginqp@line{\pst@oplineto}
\def\doqp@line{ 
  dup
  \psk@yMaxValue \pst@number\psyunit mul gt 
    { moveto }
    { dup \psk@yMinValue \pst@number\psyunit mul lt 
      { moveto }
      { L } ifelse 
    } ifelse
}
\def\endqp@line{%
  \ifPst@variableLW \addto@pscode{ \pst@flattenpath }\fi%
  \end@OpenObj}%

\def\testqp@line{%
  \ifdim\pslinearc>\z@\else
    \ifshowpoints\else
      \ifx\psk@arrowA\@empty
        \ifx\psk@arrowB\@empty
          \@psttrue
        \fi
      \fi
    \fi
  \fi}
\def\beginqp@polygon{moveto }
\def\doqp@polygon{ 
      dup
      \psk@yMaxValue \pst@number\psyunit mul gt 
      { moveto }{ 
          dup
          \psk@yMinValue \pst@number\psyunit mul lt 
          { moveto }{ L } ifelse 
      } ifelse
}
\def\endqp@polygon{%
  \addto@pscode{closepath}%
  \end@ClosedObj}
\def\testqp@polygon{%
  \ifdim\pslinearc>\z@\else
    \ifshowpoints\else
      \@psttrue
    \fi
  \fi}
\def\beginqp@dots{%
  \psk@dotsize
  \@nameuse{psds@\psk@dotstyle}
  Dot }
\def\doqp@dots{Dot }
\def\endqp@dots{\end@SpecialObj}
\def\testqp@dots{\@psttrue}
\def\beginqp@bezier{/n 0 def \pst@oplineto}
\def\doqp@bezier{/n n 1 add def n 3 mod 0 eq { % we need 3 points   
    dup \psk@yMaxValue\space \pst@number\psyunit mul gt 
    { moveto pop pop pop pop}
    { dup \psk@yMinValue\space \pst@number\psyunit mul lt 
      { moveto pop pop pop pop}{ curveto } ifelse 
    } ifelse 
  } if
}
\def\endqp@bezier{%
  \addto@pscode{n 3 mod { pop pop } repeat}
  \end@OpenObj}%
\def\testqp@bezier{%
  \ifshowpoints\else
    \ifx\psk@arrowA\@empty
      \ifx\psk@arrowB\@empty
        \@psttrue
      \fi
    \fi
  \fi}
\def\beginqp@cbezier{/n 0 def moveto }
\def\doqp@cbezier{\doqp@bezier}
\def\endqp@cbezier{%
  \addto@pscode{n 3 mod { pop pop } repeat closepath}
  \end@ClosedObj}%
\def\testqp@cbezier{\ifshowpoints\else\@psttrue\fi}
\def\tx@LineToYAxis{LineToYAxis }
\def\psLineToYAxis@ii{%
\addto@pscode{\pst@cp \psline@iii \psk@Ox\space \pst@number\psxunit mul \tx@LineToYAxis}%
\end@OpenObj}
\def\tx@LineToXAxis{LineToXAxis }
\def\psLineToXAxis@ii{%
\addto@pscode{\pst@cp \psline@iii \psk@Oy\space \pst@number\psyunit mul \tx@LineToXAxis}%
\end@OpenObj}
\define@key[psset]{pst-plot}{PSfont}[NimbusRomNo9L-Regu]{\def\psk@PSfont{/#1 }}
\define@key[psset]{pst-plot}{valuewidth}[10]{\pst@getint{#1}\psk@valuewidth }
\define@key[psset]{pst-plot}{fontscale}[10]{\pst@checknum{#1}\psk@fontscale }
\define@key[psset]{pst-plot}{decimals}[-1]{\pst@getint{#1}\psk@decimals }
\psset[pst-plot]{PSfont=NimbusRomNo9L-Regu,fontscale=10,valuewidth=10,decimals=-1}
\newdimen\psxlabelsep
\newdimen\psylabelsep
\define@key[psset]{pst-plot}{xlabelsep}[5pt]{\pssetlength\psxlabelsep{#1}}
\define@key[psset]{pst-plot}{ylabelsep}[5pt]{\pssetlength\psylabelsep{#1}}
\psset[pst-plot]{xlabelsep=5pt,ylabelsep=5pt}

\newif\ifPst@valuesStar\Pst@valuesStarfalse
\newif\ifPst@xvalues\Pst@xvaluesfalse
\def\psvalues@ii{\addto@pscode{ false \tx@NArray \psvalues@iii }}
\def\psvalues@iii{
  \psk@PSfont findfont \psk@fontscale scalefont setfont 
  newpath 
  n { /yO ED /xO ED
      gsave
      \ifPst@xvalues
        xO \pst@number\psxunit div
      \else
        yO \pst@number\psyunit div
      \fi
      \psk@decimals 0 eq { cvi } if
      \psk@decimals 0 gt { 10 \psk@decimals exp dup 3 1 roll mul cvi exch div } if
      \psk@valuewidth string cvs /Str ED
      \ifPst@valuesStar
      Str stringwidth pop /yS \psk@fontscale def /xS ED 
      gsave newpath 
        xO \ifPst@xvalues \pst@number\pslabelsep add \fi 
        yO \ifPst@xvalues \psk@fontscale 4 div sub \else \pst@number\pslabelsep add \fi 
        moveto \ifx\psk@rot\@empty\else\psk@rot rotate \fi
        xS 0 rlineto 0 yS rlineto xS neg 0 rlineto 0 yS neg rlineto 
        closepath  1 setgray fill stroke 
      grestore 
      \fi
      xO \ifPst@xvalues \pst@number\pslabelsep add \fi
      yO \ifPst@xvalues \psk@fontscale 4 div sub \else \pst@number\pslabelsep add \fi 
      moveto \ifx\psk@rot\@empty\else\psk@rot rotate \fi 
      Str show 
      grestore } repeat 
}%
\def\beginplot@values{\Pst@valuesStarfalse\begin@SpecialObj}
\expandafter\def\csname beginplot@values*\endcsname{\Pst@valuesStartrue\begin@SpecialObj}
\def\beginplot@xvalues{\Pst@valuesStarfalse\begin@SpecialObj}
\expandafter\def\csname beginplot@xvalues*\endcsname{\Pst@valuesStartrue\begin@SpecialObj}
\def\endplot@values{%
  \Pst@xvaluesfalse%  
  \psvalues@ii%
  \pst@stroke
  \end@SpecialObj}
\@namedef{endplot@values*}{\endplot@values}
\def\endplot@xvalues{%
  \Pst@xvaluestrue%  
  \psvalues@ii%
  \pst@stroke
  \end@SpecialObj}
\@namedef{endplot@xvalues*}{\endplot@xvalues}
\def\psdataplot{\def\pst@par{}\pst@object{dataplot}}
\def\dataplot{\def\pst@par{}\pst@object{dataplot}}
\def\dataplot@i#1{%
  \pst@killglue
  \begingroup
    \use@par
    \@pstfalse
    \@nameuse{testqp@\psplotstyle}%
    \if@pst
      \dataplot@ii{\addto@pscode{#1}}%
    \else
      \listplot@ii{\addto@pscode{#1}}%
    \fi
  \endgroup
  \ignorespaces}
\def\dataplot@ii#1{%
  \@nameuse{beginplot@\psplotstyle}%
    \addto@pscode{%
      /Dx { \pst@number\psxunit mul /D { Dy } def } def
      /Dy { \pst@number\psyunit mul Do /D { Dx } def } def
      /D { /D { Dx } def } def
      /Do {
        \@nameuse{beginqp@\psplotstyle}%
        /Do { \@nameuse{doqp@\psplotstyle}} def
      } def}%
    #1%			% this is \pst@readfile{#1} for fileplot
    \addto@pscode{ D }%
  \@nameuse{endqp@\psplotstyle}}
\def\psfileplot{\def\pst@par{}\pst@object{fileplot}}
\def\fileplot{\def\pst@par{}\pst@object{fileplot}}
\def\fileplot@i#1{%
  \pst@killglue%
  \begingroup%
    \use@par%
    \@pstfalse%
    \@nameuse{testqp@\psplotstyle}%
    \if@pst\dataplot@ii{\pst@readfile{#1}}\else\listplot@ii{\pst@altreadfile{#1}}\fi%
  \endgroup%
  \ignorespaces}
\def\pslistplot{\pst@object{listplot}}
\def\listplot{\pst@object{listplot}}
\def\listplot@i#1{\listplot@ii{\addto@pscode{#1}}}
\def\listplot@ii#1{%
  \@nameuse{beginplot@\psplotstyle}%
  \addto@pscode{/D {} def mark}%
  #1%
  \addto@pscode{
    \tx@PreparePoints
    \pst@number\psxunit
    \pst@number\psyunit
    \tx@ScalePoints
  }%
  \@nameuse{endplot@\psplotstyle}%
}
\define@boolkey[psset]{pst-plot}[Psk@]{xyValues}[true]{}
\define@boolkey[psset]{pst-plot}[Pst@]{ChangeOrder}[true]{}
\psset[pst-plot]{xyValues,ChangeOrder=false}
\pst@def{PreparePoints}<{%
  counttomark /m exch def
  /maxYValues \psk@plotNoMax\space def
  /YValuePos \psk@plotNo\space def
  /XValuePos \psk@plotNoX\space def
  /maxYvalue \ifx\empty\psk@plotYMax false \else true \fi def
  maxYvalue { /YMaxValue \psk@plotYMax\space def } if
  \ifPsk@xyValues\else % we have only y values
    /mm m def
    /M m 1 add def
    m { mm exch M 2 roll /M M 1 add def /mm mm 1 sub def } repeat
    /m m dup add def
  \fi
  \ifPst@ChangeOrder
    /m0 m def
    m maxYValues 1 add div 1 sub cvi {
      m0 maxYValues 1 add roll /m0 m0 maxYValues 1 add sub def
    } repeat
  \fi
  /n m maxYValues 1 add div cvi def
  XValuePos 1 gt {% multiple x values?            x y y xNo y
    n {
      maxYValues 1 add XValuePos neg roll    % y x y y xNo
      dup /XValue ED
      maxYValues 1 add XValuePos 1 sub  roll % y y xNo y x
      pop XValue                             % y y xNo y xNo
      maxYValues 1 add 1 roll                % xNo y y xNo y 
      m maxYValues 1 add  roll               % next values
    } repeat
  } if % no multiple data files
  maxYValues 1 gt {% multiple data files?           x y y yNo y 
    n {
      maxYValues YValuePos 1 sub neg roll % x yNo y y y ...
      maxYValues 1 sub { pop } repeat     % x yNo
      /m m maxYValues 1 sub sub def
      m 2 roll
    } repeat
  } if % no multiple data files
  /xMax -99999 def /yMax -99999 def
  /xP 0 def /yP 0 def
  m copy
  n {
    /y exch def /x exch def
    xMax x lt { /xMax x def } if
    yMax y lt {/yMax y def } if
    xP x gt { /xP x def } if
    yP y gt { /yP y def } if
  } repeat
  \psk@xStep\space 0 gt \psk@yStep\space 0 gt or (\psk@xStart) length 0 gt or
  (\psk@yStart) length 0 gt or (\psk@xEnd) length 0 gt or (\psk@yEnd) length 0 gt or {
    (\psk@xStart) length 0 gt {\psk@xStart\space }{ xP } ifelse /xStart exch def
    (\psk@yStart) length 0 gt {\psk@yStart\space }{ yP } ifelse /yStart exch def
    (\psk@xEnd) length 0 gt { \psk@xEnd\space }{ xMax } ifelse /xEnd exch def
    (\psk@yEnd) length 0 gt { \psk@yEnd\space }{ yMax } ifelse /yEnd exch def
    n {
      m -2 roll
      2 copy /yVal exch def /xVal exch def
      xVal xP ge
      yVal yP ge and
      xVal xEnd le and
      yVal yEnd le and
      xVal xStart ge and
      yVal yStart ge and {
        /xP xP \psk@xStep\space add def
        /yP yP \psk@yStep\space add def
      }{%
        pop pop
        /m m 2 sub def
      } ifelse
    } repeat
  }{%
    /ncount 1 def
    (\psk@nEnd) length 0 gt { \psk@nEnd\space }{ m } ifelse 
    /nEnd exch def
    n {
      m -2 roll
      \psk@nStep\space 1 gt { ncount \psk@nStart\space sub \psk@nStep\space mod 0 eq }{ true } ifelse
        ncount nEnd le and
        ncount \psk@nStart\space ge and not {
          pop pop
          /m m 2 sub def
        } if
        /ncount ncount 1 add def
      } repeat
  } ifelse
  maxYvalue { % delete values gt YMaxValue
    /m0 m def
    n {
      dup YMaxValue gt { pop pop /m m 2 sub def } if
      m -2 roll 
    } repeat
  } if
}>
\define@boolkey[psset]{pst-plot}[Pst@]{polarplot}[true]{}
\define@boolkey[psset]{pst-plot}[Pst@]{VarStep}[true]{}
\define@key[psset]{pst-plot}{PlotDerivative}{\def\psk@PlotDerivative{#1}}%
\define@key[psset]{pst-plot}{VarStepEpsilon}{\def\psk@VarStepEpsilon{#1}}%
\define@key[psset]{pst-plot}{method}{\def\psk@method{#1}}%     	   adams - rk4
\def\@rkiv{rk4}%		Runge-Kutta 4  method
\def\@varrkiv{varrkiv}%		Runge-Kutta 4 with an adaptive step method
\def\@adams{adams}%		Adams method
\def\@default{default}%		Adams method
\psset[pst-plot]{VarStep=false,PlotDerivative=none,VarStepEpsilon=default,polarplot=false,method={}}
\def\psplotinit#1{\xdef\psplot@init{#1 }}
\def\psplot@init{}
\def\psplot{\def\pst@par{}\pst@object{psplot}}
\def\psplot@i#1#2{\@ifnextchar[{\psplot@x{#1}{#2}}{\psplot@x{#1}{#2}[]}}
\def\psplot@x#1#2[#3]#4{%
  \pst@killglue
  \begingroup
    \use@par
    \@nameuse{beginplot@\psplotstyle}%
    \ifPst@polarplot
      \addto@pscode{
        \psplot@init
        #3 
        /x #1 def
        /x1 #2 def
        /dx x1 x sub \psk@plotpoints div def
        /F@pstplot \ifPst@algebraic (#4)
                    \ifx\psk@PlotDerivative\@none\else
                      \psk@PlotDerivative\space { (x) tx@Derive begin Derive end } repeat
                    \fi\space
                    tx@AlgToPs begin AlgToPs end cvx
                 \else { #4 } \fi  def
        \ifPst@VarStep
          /StillZero 0 def /LastNonZeroStep dx def
          /F2@pstplot tx@Derive begin (#4) (x) Derive (x) Derive end
                     \ifx\psk@PlotDerivative\@none\else
                       \psk@PlotDerivative\space { (x) tx@Derive begin Derive end } repeat
                     \fi\space
                    tx@AlgToPs begin AlgToPs end cvx def
          /epsilon12 \ifx\psk@VarStepEpsilon\@default tx@Derive begin F2@pstplot end dx 3 exp abs mul abs
                    \else\psk@VarStepEpsilon\space 12 mul \fi def
          /ComputeStep {
            dup 1e-4 lt
            { pop StillZero 2 ge { LastNonZeroStep 2 mul } { LastNonZeroStep } ifelse /StillZero StillZero 1 add def }
            { epsilon12 exch div 1 3 div exp /StillZero 0 def }
            ifelse } bind def
        \fi
        /xy {% Adapted from \parametricplot@i
          F@pstplot x \ifPst@algebraic RadtoDeg \fi PtoC
          \pst@number\psyunit mul exch
          \pst@number\psxunit mul exch
        } def}%
    \else% polarplot
    \addto@pscode{
      \psplot@init
      #3 
      /x #1 def
      /x1 #2 def
      /dx x1 x sub \psk@plotpoints div def
      /F@pstplot \ifPst@algebraic (#4)
                    \ifx\psk@PlotDerivative\@none\else
                      \psk@PlotDerivative\space { (x) tx@Derive begin Derive end } repeat
                    \fi\space
                    tx@AlgToPs begin AlgToPs end cvx
                 \else { #4 } \fi  def
      \ifPst@VarStep
         /StillZero 0 def /LastNonZeroStep dx def
         /F2@pstplot tx@Derive begin (#4) (x) Derive (x) Derive end
                     \ifx\psk@PlotDerivative\@none\else
                       \psk@PlotDerivative\space { (x) tx@Derive begin Derive end } repeat
                     \fi\space
                    tx@AlgToPs begin AlgToPs end cvx def
         /epsilon12 \ifx\psk@VarStepEpsilon\@default tx@Derive begin F2@pstplot end dx 3 exp abs mul abs
                    \else\psk@VarStepEpsilon\space 12 mul \fi def
         /ComputeStep {
           dup 1e-4 lt
           { pop StillZero 2 ge { LastNonZeroStep 2 mul } { LastNonZeroStep } ifelse /StillZero StillZero 1 add def }
           { epsilon12 exch div 1 3 div exp /StillZero 0 def }
           ifelse } bind def
      \fi
      /xy { x \pst@number\psxunit mul F@pstplot \pst@number\psyunit mul
      } def}%
    \fi
    \gdef\psplot@init{}%
    \ifx\pslinestyle\psls@@symbol
      \psplot@iii
    \else
      \@pstfalse
      \@nameuse{testqp@\psplotstyle}%
      \if@pst\psplot@ii\else\psplot@iii\fi
    \fi
  \endgroup
  \ignorespaces}
\def\psplot@ii{%
  \ifPst@VarStep%
    \addto@pscode{%
      mark xy \@nameuse{beginqp@\psplotstyle}
      { F2@pstplot abs ComputeStep
        x 2 copy add dup x1 gt {pop x1} if /x exch def F2@pstplot abs ComputeStep
        /x 3 -1 roll def 2 copy gt { exch } if pop
        /x x 3 -1 roll add dup x1 gt {pop x1} if def
        xy \@nameuse{doqp@\psplotstyle}
        x x1 eq { exit } if} loop}%
  \else
    \pst@killglue%
    \addto@pscode{
      /ps@Exit false def
      xy \@nameuse{beginqp@\psplotstyle}
      \ifx\psk@method\@varrkiv\else\psk@plotpoints 1 sub \fi {
        /x x dx add \ifx\psk@method\@varrkiv  dup x1 gt { pop x1 } if \fi def
        xy \@nameuse{doqp@\psplotstyle}
        \ifx\psk@method\@varrkiv  x x1 eq { exit } if \fi
      } 
      ps@Exit { exit } if
      \ifx\psk@method\@varrkiv loop \else repeat \fi
      ps@Exit not {
        /x x1 def
        xy \@nameuse{doqp@\psplotstyle}
      } if }%
  \fi%
  \@nameuse{endqp@\psplotstyle}}
\def\psplot@iii{%
  \ifPst@VarStep%
    \addto@pscode{
      /n 2 def
      mark
      { xy n 2 roll F2@pstplot abs
        ComputeStep x 2 copy add dup x1 gt {pop x1} if
        /x exch def F2@pstplot abs ComputeStep
        /x 3 -1 roll def 2 copy gt { exch } if pop
        /x x 3 -1 roll dup /LastNonZeroStep exch def add dup x1 gt {pop x1} if def /n n 2 add def
        x x1 eq { exit } if } loop
      xy 
      n 2 roll}%
  \else\pst@killglue%
    \addto@pscode{
      mark
      /n 2 def
      \ifx\psk@method\@varrkiv\else\psk@plotpoints\fi {
        xy
        n 2 roll
        /n n 2 add def
        /x x dx add \ifx\psk@method\@varrkiv  dup x1 gt { pop x1 } if \fi def
        \ifx\psk@method\@varrkiv  x x1 eq { exit } if \fi
      } \ifx\psk@method\@varrkiv loop\else repeat \fi \space
      /x x1 def
      xy 
      2 copy \tx@UserCoor 2 array astore /FinalState ED
      n 2 roll}%
  \fi%
  \@nameuse{endplot@\psplotstyle}}
\def\psparametricplot{\pst@object{parametricplot}}% 	hv 2008-11-22
\def\parametricplot{\pst@object{parametricplot}}
\def\parametricplot@i#1#2{\@ifnextchar[{\parametricplot@x{#1}{#2}}{\parametricplot@x{#1}{#2}[]}}
\def\parametricplot@x#1#2[#3]{\@ifnextchar[{\parametricplot@xi{#1}{#2}[#3]}{\parametricplot@xi{#1}{#2}[#3][]}}
\def\parametricplot@xi#1#2[#3][#4]#5{%
  \pst@killglue%
  \begingroup%
    \use@par%
    \@nameuse{beginplot@\psplotstyle}%
    \addto@pscode{%
      #3 %prefix PS code
      \psplot@init
      /t #1 def
      /t1 #2 def
      /dt t1 t sub \psk@plotpoints div def
      /F@pstplot \ifPst@algebraic (#5)
                    \ifx\psk@PlotDerivative\@none\else
                      \psk@PlotDerivative\space { (t) tx@Derive begin Derive end } repeat
                    \fi\space
                    tx@AlgToPs begin AlgToPs end cvx
                 \else { #5 } \fi  def
      \ifPst@VarStep
         /StillZero 0 def /LastNonZeroStep dt def
         /F2@pstplot tx@Derive begin (#5) (t) Derive (t) Derive end
                     \ifx\psk@PlotDerivative\@none\else
                       \psk@PlotDerivative\space { (t) tx@Derive begin Derive end } repeat
                     \fi\space
                    tx@AlgToPs begin AlgToPs end cvx def
         /epsilon12 \ifx\psk@VarStepEpsilon\@default
                       tx@Derive begin F2@pstplot end Pyth
                       dt 3 exp abs mul
                    \else\psk@VarStepEpsilon\space 12 mul \fi def
         /ComputeStep {
           dup 1e-4 lt
           { pop StillZero 2 ge { LastNonZeroStep 2 mul } { LastNonZeroStep } ifelse /StillZero StillZero 1 add def }
           { epsilon12 exch div 1 3 div exp /StillZero 0 def }
           ifelse } bind def
      \fi
      /xy {
        \ifPst@algebraic F@pstplot \else #5 \fi
        \pst@number\psyunit mul exch
        \pst@number\psxunit mul exch
      } def
      }%
    \gdef\psplot@init{}%
    \@pstfalse
    \@nameuse{testqp@\psplotstyle}%
    \if@pst\parametricplot@ii{#4}\else\parametricplot@iii{#4}\fi
  \endgroup%
  \ignorespaces}
\def\parametricplot@ii#1{% para is the post code
  \ifPst@VarStep%
    \addto@pscode{%
      mark xy \@nameuse{beginqp@\psplotstyle}
      { F2@pstplot Pyth ComputeStep
        t 2 copy add dup t1 gt {pop t1} if /t exch def F2@pstplot Pyth ComputeStep
        /t 3 -1 roll def 2 copy gt { exch } if pop
        /t t 3 -1 roll add dup t1 gt {pop t1} if def
        xy \@nameuse{doqp@\psplotstyle}
        t t1 eq { exit } if } loop}%
  \else\pst@killglue%
    \addto@pscode{%
      /ps@Exit false def
      xy \@nameuse{beginqp@\psplotstyle}
      \psk@plotpoints 1 sub {
        /t t dt add def
        xy \@nameuse{doqp@\psplotstyle}
        ps@Exit { exit } if 
      } repeat
      ps@Exit not {
        /t t1 def
        xy \@nameuse{doqp@\psplotstyle}
      } if 
    }%
  \fi%
  \addto@pscode{ #1 }%
  \@nameuse{endqp@\psplotstyle}}
\def\parametricplot@iii#1{%
  \ifPst@VarStep%
    \addto@pscode{%
      /n 2 def
      mark
      { xy n 2 roll F2@pstplot Pyth
        ComputeStep t 2 copy add dup t1 gt {pop t1} if
        /t exch def F2@pstplot Pyth ComputeStep
        /t 3 -1 roll def 2 copy gt { exch } if pop
        /t t 3 -1 roll dup /LastNonZeroStep exch def add dup t1 gt {pop t1} if def /n n 2 add def
        t t1 eq { exit } if } loop
      xy 
      2 copy \tx@UserCoor 2 array astore /FinalState ED
      n 2 roll}%
  \else\pst@killglue%
    \addto@pscode{
      mark
      /n 2 def
      \psk@plotpoints {
        xy
        n 2 roll
        /n n 2 add def
        /t t dt add def
      } repeat
      /t t1 def
      xy
      n 2 roll}%
  \fi%
  \addto@pscode{ #1 }%
  \@nameuse{endplot@\psplotstyle}}
\newdimen\psk@subticksize\psk@subticksize=\z@
\newdimen\pst@xticksizeA
\newdimen\pst@xticksizeB
\newdimen\pst@xticksizeC
\newdimen\pst@yticksizeA
\newdimen\pst@yticksizeB
\newdimen\pst@yticksizeC
\define@key[psset]{pst-plot}{ticks}[all]{\pst@expandafter\psset@@ticks{#1}\@nil\psk@ticks}
\def\psset@@ticks#1#2\@nil#3{%
  \ifx#1a\let#3\z@\else%				0=a)ll
    \ifx#1x\let#3\@ne\else%				1=x
      \ifx#1y\let#3\tw@\else%				2=y
        \ifx#1n\let#3\thr@@\else%			3=n)one
          \@pstrickserr{Bad argument: `#1#2'}\@ehpa
  \fi\fi\fi\fi}
\psset[pst-plot]{ticks=all}
\define@key[psset]{pst-plot}{labels}[all]{\pst@expandafter\psset@@ticks{#1}\@nil\psk@labels}% same as ticks
\psset[pst-plot]{labels=all}
\define@key[psset]{pst-plot}{Ox}[0]{\def\psk@Ox{#1}}
\psset[pst-plot]{Ox=0}
\define@key[psset]{pst-plot}{Dx}[1]{\def\psk@Dx{#1}}
\psset[pst-plot]{Dx=1}
\define@key[psset]{pst-plot}{dx}[0]{%
  \pssetxlength\pst@dimg{#1}%
  \edef\psk@dx{\number\pst@dimg}}
\psset[pst-plot]{dx=0}
\define@key[psset]{pst-plot}{Oy}[0]{\def\psk@Oy{#1}}
\psset[pst-plot]{Oy=0}
\define@key[psset]{pst-plot}{Dy}[1]{\def\psk@Dy{#1}}
\psset[pst-plot]{Dy=1}
\define@key[psset]{pst-plot}{dy}[0]{%
  \pssetylength\pst@dimg{#1}%
  \edef\psk@dy{\number\pst@dimg}}
\psset[pst-plot]{dy=0}
\define@boolkey[psset]{pst-plot}[]{showXorigin}[true]{}
\define@boolkey[psset]{pst-plot}[]{showYorigin}[true]{}
\define@boolkey[psset]{pst-plot}[]{showorigin}[true]{%
  \ifshoworigin
    \showXorigintrue\showYorigintrue
  \else
    \showXoriginfalse\showYoriginfalse
  \fi
}
\psset[pst-plot]{showorigin=true}
\long\def\psrotatebox#1#2{%
  \leavevmode
  \Grot@setangle{#1}%
  \setbox\z@\hbox{{#2}}%
  \Grot@x\z@
  \Grot@y\z@
  \Grot@box}
\def\Grot@setangle#1{\edef\Grot@angle{#1}}
\def\Grot@Px#1#2#3{%
        #1\Grot@cos#2%
        \advance#1-\Grot@sin#3}
\def\Grot@Py#1#2#3{%
        #1\Grot@sin#2%
        \advance#1\Grot@cos#3}
\def\Grot@box{%
  \begingroup
  \CalculateSin\Grot@angle
  \CalculateCos\Grot@angle
  \edef\Grot@sin{\UseSin\Grot@angle}%
  \edef\Grot@cos{\UseCos\Grot@angle}%
  \Grot@r\wd\z@  \advance\Grot@r-\Grot@x
  \Grot@l\z@     \advance\Grot@l-\Grot@x
  \Grot@h\ht\z@  \advance\Grot@h-\Grot@y
  \Grot@d-\dp\z@ \advance\Grot@d-\Grot@y
  \ifdim\Grot@sin\p@>\z@
    \ifdim\Grot@cos\p@>\z@
      \Grot@Py\Grot@height \Grot@r\Grot@h%B
      \Grot@Px\Grot@right  \Grot@r\Grot@d%E
      \Grot@Px\Grot@left   \Grot@l\Grot@h%C
      \Grot@Py\Grot@depth  \Grot@l\Grot@d%D
    \else
      \Grot@Py\Grot@height \Grot@r\Grot@d%E
      \Grot@Px\Grot@right  \Grot@l\Grot@d%D
      \Grot@Px\Grot@left   \Grot@r\Grot@h%B
      \Grot@Py\Grot@depth  \Grot@l\Grot@h%C
    \fi
  \else
    \ifdim\Grot@cos\p@<\z@
      \Grot@Py\Grot@height \Grot@l\Grot@d%D
      \Grot@Px\Grot@right  \Grot@l\Grot@h%C
      \Grot@Px\Grot@left   \Grot@r\Grot@d%E
      \Grot@Py\Grot@depth  \Grot@r\Grot@h%B
    \else
      \Grot@Py\Grot@height \Grot@l\Grot@h%C
      \Grot@Px\Grot@right  \Grot@r\Grot@h%B
      \Grot@Px\Grot@left   \Grot@l\Grot@d%D
      \Grot@Py\Grot@depth  \Grot@r\Grot@d%E
    \fi
  \fi
  \advance\Grot@height\Grot@y
  \advance\Grot@depth\Grot@y
  \Grot@Px\dimen@  \Grot@x\Grot@y
  \Grot@Py\dimen@ii \Grot@x\Grot@y
  \dimen@-\dimen@     \advance\dimen@-\Grot@left
  \dimen@ii-\dimen@ii \advance\dimen@ii\Grot@y
  \setbox\z@\hbox{%
    \kern\dimen@
    \raise\dimen@ii\hbox{\Grot@start\box\z@\Grot@end}}%
  \ht\z@\Grot@height
  \dp\z@-\Grot@depth
  \advance\Grot@right-\Grot@left\wd\z@\Grot@right
  \leavevmode\box\z@
  \endgroup}
\define@key[psset]{pst-plot}{labelFontSize}[{}]{\def\psk@xlabelFontSize{#1}\def\psk@ylabelFontSize{#1}}%
\define@key[psset]{pst-plot}{xlabelFontSize}[{}]{\def\psk@xlabelFontSize{#1}}%
\define@key[psset]{pst-plot}{ylabelFontSize}[{}]{\def\psk@ylabelFontSize{#1}}%
\define@boolkey[psset]{pst-plot}[Pst@]{mathLabel}[true]{%
  \ifPst@mathLabel%
    \Pst@xmathLabeltrue \Pst@ymathLabeltrue
    \def\pshlabel##1{$\psk@xlabelFontSize##1$}%
    \def\psvlabel##1{$\psk@ylabelFontSize##1$}%
  \else%
    \Pst@xmathLabelfalse \Pst@ymathLabelfalse
    \def\pshlabel##1{\psk@xlabelFontSize##1}%
    \def\psvlabel##1{\psk@ylabelFontSize##1}%
  \fi}
\define@boolkey[psset]{pst-plot}[Pst@]{xmathLabel}[true]{%
  \ifPst@xmathLabel%
    \def\pshlabel##1{$\psk@xlabelFontSize##1$}\else\def\pshlabel##1{\psk@xlabelFontSize##1}\fi}
\define@boolkey[psset]{pst-plot}[Pst@]{ymathLabel}[true]{%
  \ifPst@ymathLabel%
    \def\psvlabel##1{$\psk@ylabelFontSize##1$}\else\def\psvlabel##1{\psk@ylabelFontSize##1}\fi}
\psset[pst-plot]{labelFontSize={},mathLabel}
\define@boolkey[psset]{pst-plot}[Pst@]{xAxis}[true]{}
\define@boolkey[psset]{pst-plot}[Pst@]{yAxis}[true]{}
\define@boolkey[psset]{pst-plot}[Pst@]{xyAxes}[true]{%
    \@nameuse{Pst@xAxis#1}\@nameuse{Pst@yAxis#1}}%
\psset[pst-plot]{xAxis,yAxis}%
\define@key[psset]{pst-plot}{xlabelPos}[b]{\pst@expandafter\psset@@xlabelPos#1\@nil}
\define@key[psset]{pst-plot}{ylabelPos}[l]{\pst@expandafter\psset@@ylabelPos#1\@nil}
\def\psset@@xlabelPos#1#2\@nil{%
  \ifx#1t\relax
    \let\psk@xlabelPos\tw@%		2=top
    \pst@xticksizeC=\pst@xticksizeB
  \else
    \ifx#1a\relax
      \let\psk@xlabelPos\@ne %	 	1=axis
      \pst@xticksizeC=\z@
    \else
      \def\psk@xlabelPos{\z@}%		0=bottom	
      \pst@xticksizeC=\pst@xticksizeA
  \fi\fi
}
\def\psset@@ylabelPos#1#2\@nil{%
  \ifx#1r\relax
    \def\psk@ylabelPos{\tw@}%		2=right
    \pst@yticksizeC=\pst@yticksizeB
  \else
    \ifx#1a\relax
      \def\psk@ylabelPos{\@ne}% 	1=axis
      \pst@yticksizeC=\z@
    \else 
      \def\psk@ylabelPos{\z@}%		0=left	
      \pst@yticksizeC=\pst@yticksizeA
  \fi\fi
}
\psset[pst-plot]{xlabelPos=b, ylabelPos=l}%
\define@key[psset]{pst-plot}{xyDecimals}[{}]{\def\psk@xDecimals{#1}\def\psk@yDecimals{#1}}
\define@key[psset]{pst-plot}{xDecimals}[{}]{\def\psk@xDecimals{#1}}
\define@key[psset]{pst-plot}{yDecimals}[{}]{\def\psk@yDecimals{#1}}
\psset[pst-plot]{xyDecimals={}}%
\define@key[psset]{pst-plot}{xlogBase}[{}]{\def\psk@xlogBase{#1}}
\define@key[psset]{pst-plot}{ylogBase}[{}]{\def\psk@ylogBase{#1}}
\define@key[psset]{pst-plot}{xylogBase}[{}]{\def\psk@xlogBase{#1}\def\psk@ylogBase{#1}}%
\psset[pst-plot]{xylogBase={}}%
\define@key[psset]{pst-plot}{trigLabelBase}[0]{\pst@getint{#1}{\psk@xtrigLabelBase}\let\psk@ytrigLabelBase\psk@xtrigLabelBase}
\define@key[psset]{pst-plot}{xtrigLabelBase}[0]{%
  \pst@getint{#1}{\psk@xtrigLabelBase}\psset{xtrigLabels}}
\define@key[psset]{pst-plot}{ytrigLabelBase}[0]{%
  \pst@getint{#1}{\psk@ytrigLabelBase}\psset{ytrigLabels}}
\psset[pst-plot]{trigLabelBase=0}
\define@key[psset]{pst-plot}{fractionLabelBase}[0]{\pst@getint{#1}{\psk@xfractionLabelBase}\let\psk@yfractionLabelBase\psk@xfractionLabelBase}
\define@key[psset]{pst-plot}{xfractionLabelBase}[0]{%
  \pst@getint{#1}{\psk@xfractionLabelBase}\psset{xfractionLabels}}
\define@key[psset]{pst-plot}{yfractionLabelBase}[0]{%
  \pst@getint{#1}{\psk@yfractionLabelBase}\psset{yfractionLabels}}
\psset[pst-plot]{fractionLabelBase=0}
\def\setDefaulthLabels{%
  \ifPst@xmathLabel\def\pshlabel##1{$\psk@xlabelFontSize##1$}\else\def\pshlabel##1{\psk@xlabelFontSize##1}\fi
  \def\pst@@@hlabel##1{%
      \edef\@xyDecimals{\psk@xDecimals}%
      \ifnum\psk@labels<\tw@\relax% labels=all|x
        \ifx\psk@xlogBase\@empty
          \pshlabel{\psk@xlabelFontSize\expandafter\@LabelComma##1..\@nil\psk@xlabelFactor}%
        \else
          \ifPst@xmathLabel
            \pshlabel{\psk@xlabelFontSize\psk@xlogBase^{\expandafter\@stripDecimals##1..\@nil}}%
          \else
            \pshlabel{\psk@xlabelFontSize\psk@xlogBase\textsuperscript{\expandafter\@stripDecimals##1..\@nil}}%
          \fi
        \fi
      \fi
    }%
    \ifPst@xmathLabel\def\pshlabel##1{$\psk@xlabelFontSize##1$}\else\def\pshlabel##1{\psk@xlabelFontSize##1}\fi
}
\def\setTrighLabels{%
    \def\pst@@@hlabel##1{\pshlabel{##1}}%
    \def\pshlabel##1{%
      \ifnum\psk@xtrigLabelBase<2
        \def\de@nominator{\@ne}\else\def\de@nominator{\psk@xtrigLabelBase}\fi
      \def\pst@tempA{##1}% 
      \pst@abs{\pst@tempA}\pst@cntm 
      \pst@mod{\pst@cntm}{\de@nominator}\pst@cntp % cntb=##1 modulo trigLabelBase
      \ifnum\@ne>\pst@cntp                  % 1 > modulo -> then we have pi/x
        \pst@cnto=\pst@cntm \divide\pst@cnto by \de@nominator  
	\ifPst@xmathLabel
          $\psk@xlabelFontSize
  	  \ifnum\pst@tempA<0 -\fi
          \ifnum\pst@cnto=\@ne                % #1 = trigLabelBase
            \pi                 	      % print pi
          \else
            \ifnum\pst@cnto=\z@ 0\else
            \the\pst@cnto\pi 	              % print \pst@cnto/\de@nominator pi
          \fi\fi$%   
	\else%
          \psk@xlabelFontSize
  	  \ifnum\pst@tempA<0 -\fi
          \ifnum\pst@cnto=\@ne%                % #1 = trigLabelBase
            $\pi$%                             % print pi
          \else%
            \the\pst@cnto$\pi$%                % print \pst@cnto/\de@nominator pi
          \fi%
	\fi%
      \else%
	\ifPst@xmathLabel%
          $\psk@xlabelFontSize%
          \ifnum\pst@cntp=\@ne%                % < 1 pi?
            \if\pst@cntm=\@ne%
              \frac{\pi}{\de@nominator}%   % pi/x
            \else\ifnum\pst@tempA=-1 \frac{-\pi}{\de@nominator}%
              \else \ifnum\pst@tempA=1 \frac{\pi}{\de@nominator}%
                \else\frac{\pst@tempA\pi}{\de@nominator}% (x pi)/y
            \fi\fi\fi%
          \else%
            \ifnum\pst@tempA=1 \frac{\pi}{\de@nominator}%
            \else\ifnum\pst@tempA=\de@nominator \pi%
              \else\frac{\pst@tempA\pi}{\de@nominator}% 
          \fi\fi\fi$%
	\else%
          \psk@xlabelFontSize%
          \ifnum\pst@cntp=\@ne%                % < 1 pi?
            \if\pst@cntm=\@ne%
              $\frac{\pi}{\de@nominator}$%   % pi/x
            \else\ifnum\pst@tempA=-1 $\frac{-\pi}{\de@nominator}$%
              \else \ifnum\pst@tempA=1 $\frac{\pi}{\de@nominator}$%
                \else$\frac{\pst@tempA\pi}{\de@nominator}$% (x pi)/y
            \fi\fi\fi
          \else
            \ifnum\pst@tempA=1 $\frac{\pi}{\de@nominator}$%
            \else\ifnum\pst@tempA=\de@nominator $\pi$%
              \else$\frac{\pst@tempA\pi}{\de@nominator}$% 
          \fi\fi\fi
	\fi
      \fi
    }%
}
\def\setDefaultvLabels{%
  \ifPst@ymathLabel\def\psvlabel##1{$\psk@ylabelFontSize##1$}\else\def\psvlabel##1{\psk@ylabelFontSize##1}\fi
    \def\pst@@@vlabel##1{%
      \edef\@xyDecimals{\psk@yDecimals}%
      \ifodd\psk@labels % labelss=all||y (0,2)
      \else%
        \ifx\psk@ylogBase\@empty
          \psvlabel{\expandafter\@LabelComma##1..\@nil\psk@ylabelFactor}%
        \else%
          \ifPst@ymathLabel%
            \psvlabel{\psk@ylogBase^{\expandafter\@stripDecimals##1..\@nil }}%
	  \else
            \psvlabel{\psk@ylogBase\textsuperscript{\expandafter\@stripDecimals##1..\@nil }}%
          \fi%
        \fi%
      \fi%
    }%
}%
\def\setTrigvLabels{%
  \def\pst@@@vlabel##1{\psvlabel{##1}}%
    \def\psvlabel##1{%
      \ifnum\psk@ytrigLabelBase<2 \def\de@nominator{\@ne}\else\def\de@nominator{\psk@ytrigLabelBase}\fi
      \def\pst@tempA{##1} 
      \pst@abs{\pst@tempA}\pst@cntm 
      \pst@mod{\pst@cntm}{\de@nominator}\pst@cntp % cntb=##1 modulo trigLabelBase
      \ifnum\@ne>\pst@cntp                  % 1 > modulo -> then we have pi/x
        \pst@cnto=\pst@cntm \divide\pst@cnto by \de@nominator  
	\ifPst@ymathLabel%
          $\psk@ylabelFontSize
  	  \ifnum\pst@tempA<0 -\fi
          \ifnum\pst@cnto=\@ne                % #1 = trigLabelBase
            \pi                 	      % print pi
          \else
            \the\pst@cnto\pi 	              % print \pst@cnto/\de@nominator pi
          \fi$%   
	\else%
          \psk@ylabelFontSize%
  	  \ifnum\pst@tempA<0 -\fi
          \ifnum\pst@cnto=\@ne%                % #1 = trigLabelBase
            $\pi$%                             % print pi
          \else
            \the\pst@cnto$\pi$%                % print \pst@cnto/\de@nominator pi
          \fi
	\fi
      \else
	\ifPst@ymathLabel%
          $\psk@ylabelFontSize
          \ifnum\pst@cntp=\@ne%                % < 1 pi?    $
            \if\pst@cntm=\@ne%
              \frac{\pi}{\de@nominator}%   % pi/x
            \else\ifnum\pst@tempA=-1 \frac{-\pi}{\de@nominator}%
              \else \ifnum\pst@tempA=1 \frac{\pi}{\de@nominator}%
                \else\frac{\pst@tempA\pi}{\de@nominator}% (x pi)/y
            \fi\fi\fi%
          \else%
            \ifnum\pst@tempA=1 \frac{\pi}{\de@nominator}%
            \else\ifnum\pst@tempA=\de@nominator \pi%
              \else\frac{\pst@tempA\pi}{\de@nominator}% 
          \fi\fi\fi$%
	\else
          \psk@ylabelFontSize
          \ifnum\pst@cntp=\@ne%                % < 1 pi?
            \if\pst@cntm=\@ne
              $\frac{\pi}{\de@nominator}$%   % pi/x
            \else\ifnum\pst@tempA=-1 $\frac{-\pi}{\de@nominator}$%
              \else \ifnum\pst@tempA=1 $\frac{\pi}{\de@nominator}$%
                \else$\frac{\pst@tempA\pi}{\de@nominator}$% (x pi)/y
            \fi\fi\fi
          \else
            \ifnum\pst@tempA=1 $\frac{\pi}{\de@nominator}$%
            \else\ifnum\pst@tempA=\de@nominator $\pi$%
              \else$\frac{\pst@tempA\pi}{\de@nominator}$% 
          \fi\fi\fi
	\fi
      \fi
    }%
}%$
\def\setFractionvLabels{%
  \def\pst@@@vlabel##1{\psvlabel{##1}}%
  \def\psvlabel##1{%
      \ifnum\psk@yfractionLabelBase<2 \def\de@nominator{\@ne}\else\def\de@nominator{\psk@yfractionLabelBase}\fi
      \def\pst@tempA{##1}% 
      \pst@abs{\pst@tempA}\pst@cntm 
      \pst@mod{\pst@cntm}{\de@nominator}\pst@cntp % cntb=##1 modulo trigLabelBase
      \ifnum\@ne>\pst@cntp                  % 1 > modulo -> then we have pi/x
        \pst@cnto=\pst@cntm \divide\pst@cnto by \de@nominator  
	\ifPst@ymathLabel$\psk@ylabelFontSize\ifnum\pst@tempA<0 -\fi\the\pst@cnto\psk@ylabelFactor$%
	\else             \psk@ylabelFontSize\ifnum\pst@tempA<0 -\fi\the\pst@cnto\psk@ylabelFactor
	\fi
      \else
	\ifPst@ymathLabel
          $\psk@ylabelFontSize
          \ifnum\pst@cntp=\@ne                % < 1?    $
            \if\pst@cntm=\@ne
              \frac{1}{\de@nominator}\psk@ylabelFactor%   % 1/x
            \else\ifnum\pst@tempA=-1 \frac{-1}{\de@nominator}\psk@ylabelFactor%
              \else \ifnum\pst@tempA=1 \frac{1}{\de@nominator}\psk@ylabelFactor%
                \else\frac{\pst@tempA}{\de@nominator}\psk@ylabelFactor% x/y
            \fi\fi\fi
          \else
            \ifnum\pst@tempA=1 \frac{1}{\de@nominator}\psk@ylabelFactor%
            \else\ifnum\pst@tempA=\de@nominator 1\psk@xlabelFactor \else\frac{\pst@tempA}{\de@nominator}\psk@ylabelFactor%
          \fi\fi\fi$
	\else
          \psk@ylabelFontSize
          \ifnum\pst@cntp=\@ne%                % < 1?
            \if\pst@cntm=\@ne
              $\frac{1}{\de@nominator}\psk@ylabelFactor$%   % 1/x
            \else\ifnum\pst@tempA=-1 $\frac{-1}{\de@nominator}\psk@ylabelFactor$%
              \else \ifnum\pst@tempA=1 $\frac{1}{\de@nominator}\psk@ylabelFactor$%
                \else$\frac{\pst@tempA}{\de@nominator}\psk@ylabelFactor$% x/y
            \fi\fi\fi%
          \else%
            \ifnum\pst@tempA=1 $\frac{1}{\de@nominator}\psk@ylabelFactor$%
            \else\ifnum\pst@tempA=\de@nominator 1\psk@ylabelFactor
              \else$\frac{\pst@tempA}{\de@nominator}\psk@ylabelFactor$%   %$
          \fi\fi\fi
	\fi
      \fi
    }%
}%$
\def\setFractionhLabels{%
  \def\pst@@@hlabel##1{\pshlabel{##1}}%
  \def\pshlabel##1{%
      \ifnum\psk@xfractionLabelBase<2 \def\de@nominator{\@ne}\else\def\de@nominator{\psk@xfractionLabelBase}\fi
      \def\pst@tempA{##1}% 
      \pst@abs{\pst@tempA}\pst@cntm 
      \pst@mod{\pst@cntm}{\de@nominator}\pst@cntp% cntb=##1 modulo trigLabelBase
      \ifnum\@ne>\pst@cntp                  % 1 > modulo -> then we have 1/x
        \pst@cnto=\pst@cntm \divide\pst@cnto by \de@nominator  
	\ifPst@xmathLabel$\psk@xlabelFontSize\ifnum\pst@tempA<0 -\fi\the\pst@cnto\psk@xlabelFactor$%
	\else             \psk@xlabelFontSize\ifnum\pst@tempA<0 -\fi\the\pst@cnto\psk@xlabelFactor
	\fi
      \else
	\ifPst@xmathLabel
          $\psk@xlabelFontSize% $
          \ifnum\pst@cntp=\@ne
            \if\pst@cntm=\@ne \frac{1}{\de@nominator}\psk@xlabelFactor%   % 1/x
            \else\ifnum\pst@tempA=-1 \frac{-1}{\de@nominator}\psk@xlabelFactor%
              \else\ifnum\pst@tempA=1 \frac{1}{\de@nominator}\psk@xlabelFactor%
                \else\frac{\pst@tempA}{\de@nominator}\psk@xlabelFactor% x/y
            \fi\fi\fi%
          \else%
            \ifnum\pst@tempA=1 \frac{1}{\de@nominator}\psk@xlabelFactor%
            \else\ifnum\pst@tempA=\de@nominator 1\psk@xlabelFactor\else\frac{\pst@tempA}{\de@nominator}\psk@xlabelFactor%
          \fi\fi\fi$
	\else
          \psk@xlabelFontSize
          \ifnum\pst@cntp=\@ne
            \if\pst@cntm=\@ne $\frac{1}{\de@nominator}\psk@xlabelFactor$%            % 1/x
            \else\ifnum\pst@tempA=-1 $\frac{-1}{\de@nominator}\psk@xlabelFactor$%
              \else \ifnum\pst@tempA=1 $\frac{1}{\de@nominator}\psk@xlabelFactor$%
                \else$\frac{\pst@tempA}{\de@nominator}\psk@xlabelFactor$% x/y
            \fi\fi\fi
          \else
            \ifnum\pst@tempA=1 $\frac{1}{\de@nominator}\psk@xlabelFactor$%
            \else\ifnum\pst@tempA=\de@nominator 1\psk@xlabelFactor%
              \else$\frac{\pst@tempA}{\de@nominator}\psk@xlabelFactor$%   %$
          \fi\fi\fi
	\fi
      \fi
    }%
}%$
\define@boolkey[psset]{pst-plot}[Pst@]{xtrigLabels}[true]{%
  \ifPst@xtrigLabels \setTrighLabels \else \setDefaulthLabels\fi}
\define@boolkey[psset]{pst-plot}[Pst@]{ytrigLabels}[true]{%
  \ifPst@ytrigLabels \setTrigvLabels \else \setDefaultvLabels \fi}
\define@boolkey[psset]{pst-plot}[Pst@]{trigLabels}[true]{%
  \ifPst@trigLabels\psset[pst-plot]{xtrigLabels,ytrigLabels=false}
  \else            \psset[pst-plot]{xtrigLabels=false,ytrigLabels=false}%
  \fi}
\psset[pst-plot]{trigLabels=false}
\define@boolkey[psset]{pst-plot}[Pst@]{xfractionLabels}[true]{%
  \ifPst@xfractionLabels \setFractionhLabels \else \setDefaulthLabels \fi}
\define@boolkey[psset]{pst-plot}[Pst@]{yfractionLabels}[true]{%
  \ifPst@yfractionLabels \setFractionvLabels \else \setDefaultvLabels \fi}
\define@boolkey[psset]{pst-plot}[Pst@]{fractionLabels}[true]{%
  \ifPst@fractionLabels \setFractionhLabels\setFractionvLabels\Pst@xfractionLabelstrue\Pst@yfractionLabelstrue\fi}
\psset[pst-plot]{fractionLabels=false}
\def\psk@logLines{3}
\define@key[psset]{pst-plot}{logLines}[none]{\pst@expandafter\psset@@logLines#1\@nil\psk@logLines}
\def\psset@@logLines#1#2\@nil#3{%
  \ifx#1a\relax
    \let#3\z@
    \Pst@maxxTickstrue\Pst@maxyTickstrue
    \set@xticksize{0 4pt}\set@yticksize{0 4pt}%
    \def\psk@xsubticksize{1}\def\psk@ysubticksize{1}%
  \else
    \ifx#1x\relax
      \let#3\@ne
      \Pst@maxxTickstrue\Pst@maxyTicksfalse
      \set@xticksize{0 4pt}\def\psk@xsubticksize{1}%
    \else
      \ifx#1y\relax
        \let#3\tw@
	\Pst@maxyTickstrue\Pst@maxxTicksfalse
	\set@yticksize{0 4pt}\def\psk@ysubticksize{1}%
      \else
        \ifx#1n\let#3\thr@@\else
          \@pstrickserr{Bad argument: `#1#2'}\@ehpa
  \fi\fi\fi\fi}
\psset[pst-plot]{logLines=none}%
\define@key[psset]{pst-plot}{ylabelFactor}[\relax]{\def\psk@ylabelFactor{#1}}
\define@key[psset]{pst-plot}{xlabelFactor}[\relax]{\def\psk@xlabelFactor{#1}}
\define@boolkey[psset]{pst-plot}[Pst@]{showOriginTick}[true]{}%
\psset[pst-plot]{xlabelFactor=\relax,ylabelFactor=\relax,showOriginTick}%

\def\psxTick{\pst@object{psxTick}}% idea by Martin Chicoine
\def\psxTick@i{\@ifnextchar({\psxTick@ii{0}}\psxTick@ii}
\def\psxTick@ii#1(#2)#3{{%
  \pst@killglue
  \addbefore@par{arrows=-,linewidth=\psk@xtickwidth\pslinewidth}
  \ifPst@xtrigLabels\addto@par{xtrigLabels=false}\fi 
  \use@par
  \edef\temp@coor{(!#2 \pst@number\pst@xticksizeB \pst@number\psyunit div)(!#2 \pst@number\pst@xticksizeA \pst@number\psyunit div)}%
  \expandafter\psline\temp@coor
  \rput[t]{#1}(! \psk@origin 
                 #2 \pst@number\psxlabelsep \pst@number\pst@xticksizeB add
                 \pst@number\psyunit div neg ){\pshlabel{#3\vphantom{1}}}%
  }\ignorespaces}
\def\psyTick{\pst@object{psyTick}}% idea by Martin Chicoine
\def\psyTick@i{\@ifnextchar({\psyTick@ii{0}}\psyTick@ii}
\def\psyTick@ii#1(#2)#3{{%
  \pst@killglue
  \addbefore@par{arrows=-,linewidth=\psk@ytickwidth\pslinewidth}
  \ifPst@ytrigLabels \setDefaultvLabels \fi
  \use@par
  \edef\temp@coor{(!\pst@number\pst@yticksizeB \pst@number\psxunit div #2)(!\pst@number\pst@yticksizeA \pst@number\psxunit div #2)}%
  \expandafter\psline\temp@coor
    \rput[r]{#1}(!\psk@origin
                  \pst@number\pst@yticksizeB \pst@number\psylabelsep add
                  \pst@number\psxunit div neg #2){\psvlabel{#3}}}\ignorespaces}
\define@boolkey[psset]{pst-plot}[Pst@]{markPoint}[true]{}%
\psset[pst-plot]{markPoint}
\def\psCoordinates{\pst@object{psCoordinates}}
\def\psCoordinates@i(#1){%
  \pst@killglue%
  \begingroup
  \addbefore@par{showpoints=false,markPoint}
  \use@par
  \psline(#1|0,0)(#1)% single lines to allow arrows
  \psline(#1)(0,0|#1)%
  \ifPst@markPoint\psdot(#1)\fi%
  \endgroup
  \ignorespaces
}
\def\stripDecimals#1{\expandafter\@stripDecimals#1..\@nil}
\def\@stripDecimals#1.#2.#3\@nil{%
  \def\pst@dummy{#1}%
  \ifx\pst@dummy\@empty\the\@zero\else#1\fi% the integer part
}
\newcount\@digitcounter\@digitcounter=0\relax
\def\@inc@digitcounter{\global\advance\@digitcounter by 1\relax}
\def\@get@digitcounter{\the\@digitcounter\relax}
\def\@Reset@digitcounter{\global\@digitcounter=0\relax}
\def\@zeroFill{%
  \ifnum \@xyDecimals>\@get@digitcounter
    \bgroup
      0\@inc@digitcounter\@zeroFill
    \egroup
  \fi
}
\def\@process@digits#1#2;{%
  \ifx *#1\@zeroFill\else#1\@inc@digitcounter 
  \ifnum\@xyDecimals>\@get@digitcounter\expandafter\@process@digits#2;\fi\fi%
}
\def\@writeDecimals#1{%
  \ifx\@xyDecimals\@empty% take value as is
    \def\@tempa{#1}% write only if not empty
    \ifx\@tempa\@empty% write nothing
    \else\ifmmode\expandafter\mathord\expandafter{\psk@decimalSeparator}\else\psk@decimalSeparator\fi#1\fi%
  \else% write only \xy@decimals
    \ifnum\@xyDecimals>\@zero
      \ifmmode\expandafter\mathord\expandafter{\psk@decimalSeparator}\else\psk@decimalSeparator\fi%
        \@Reset@digitcounter
        \expandafter\@process@digits#1*;%
      \fi%
  \fi%
}
\def\@LabelComma#1.#2.#3\@nil{%
  \def\pst@tempA{#1}%
  \ifx\pst@tempA\@empty\the\@zero\else#1\fi% the integer part
  \def\pst@tempA{#2}%
  \ifx\pst@tempA\@empty\@writeDecimals{}\else\@writeDecimals{#2}\fi
}
\def\set@xticksize#1{%
  \pst@expandafter\pst@getydimdim{#1} {} {}\@nil% y-unit!! 
  \ifdim\pst@dimm>\pst@dimn% 		%	first > second value
    \pst@xticksizeA=\the\pst@dimn%
    \pst@xticksizeB=\the\pst@dimm%
  \else%
    \pst@xticksizeA=\the\pst@dimm%
    \pst@xticksizeB=\the\pst@dimn%	first > second value
  \fi%
  \edef\psk@xticksize{\pst@number\pst@xticksizeA \pst@number\pst@xticksizeB}%
  \ifnum\psk@xlabelPos<\z@\relax% top
    \pst@xticksizeC=\pst@dimn
  \else
    \pst@xticksizeC=\pst@dimm%	bottom	
  \fi
}
\def\set@yticksize#1{%
  \pst@expandafter\pst@getxdimdim{#1} {} {}\@nil% x-unit!
  \ifdim\pst@dimm>\pst@dimn\relax%   		%	first > second value
    \pst@yticksizeA=\the\pst@dimn%
    \pst@yticksizeB=\the\pst@dimm%
  \else%
    \pst@yticksizeA=\the\pst@dimm%
    \pst@yticksizeB=\the\pst@dimn%	first > second value
  \fi%
  \edef\psk@yticksize{\pst@number\pst@yticksizeA \pst@number\pst@yticksizeB}%
  \ifnum\psk@ylabelPos<\z@	% right	
    \pst@yticksizeC=\pst@dimn%
  \else%
      \pst@yticksizeC=\pst@dimo%  left
  \fi%
}
\newif\ifPst@maxxTicks
\newif\ifPst@maxyTicks
\define@key[psset]{pst-plot}{ticksize}[-4pt 4pt]{%
  \def\pst@tempA{max}%
  \def\pst@tempB{#1}%
  \ifx\pst@tempA\pst@tempB%
    \Pst@maxxTickstrue\Pst@maxyTickstrue%
    \set@xticksize{0 4pt}\set@yticksize{0 4pt}%
  \else%
    \Pst@maxxTicksfalse\Pst@maxyTicksfalse%
    \set@xticksize{#1}\set@yticksize{#1}%
  \fi}
\define@key[psset]{pst-plot}{xticksize}{%
  \def\pst@tempA{max}%
  \def\pst@tempB{#1}%
  \ifx\pst@tempA\pst@tempB
    \Pst@maxxTickstrue\set@xticksize{0 4pt}%
  \else\set@xticksize{#1}\Pst@maxxTicksfalse\fi}
\define@key[psset]{pst-plot}{yticksize}{%
  \def\pst@tempA{max}%
  \def\pst@tempB{#1}%
  \ifx\pst@tempA\pst@tempB%
    \Pst@maxyTickstrue\set@yticksize{0 4pt}%
  \else\set@yticksize{#1}\Pst@maxyTicksfalse\fi}%
\psset[pst-plot]{ticksize=-4pt 4pt}
\define@key[psset]{pst-plot}{tickstyle}[full]{\pst@expandafter\psset@@tickstyle{#1}\@nil}
\def\psset@@tickstyle#1#2\@nil{%
  \ifx#1f\let\psk@tickstyle\z@\else			% 0=f)ull
    \ifx#1t\let\psk@tickstyle\@ne			% 1=t)op
      \edef\psk@xticksize{0 \pst@number\pst@xticksizeB}%
      \edef\psk@yticksize{0 \pst@number\pst@yticksizeB}%
    \else\ifx#1b\let\psk@tickstyle\m@ne			% -1=b)ottom
      \edef\psk@xticksize{\pst@number\pst@xticksizeA 0}%
      \edef\psk@yticksize{\pst@number\pst@yticksizeA 0}%
      \else\ifx#1i\let\psk@tickstyle\tw@%		% 2=i)nner (for frame)
        \else\@pstrickserr{Bad tick style: `#1#2'}\@ehpa
  \fi\fi\fi\fi}
\psset[pst-plot]{tickstyle=full}
\define@key[psset]{pst-plot}{subticks}[1]{\def\psk@xsubticks{#1}\def\psk@ysubticks{#1}}
\define@key[psset]{pst-plot}{xsubticks}[1]{\def\psk@xsubticks{#1}}
\define@key[psset]{pst-plot}{ysubticks}[1]{\def\psk@ysubticks{#1}}
\define@key[psset]{pst-plot}{subticksize}[0.75]{\def\psk@xsubticksize{#1}\def\psk@ysubticksize{#1}}
\define@key[psset]{pst-plot}{xsubticksize}[0.75]{\def\psk@xsubticksize{#1}}
\define@key[psset]{pst-plot}{ysubticksize}[0.75]{\def\psk@ysubticksize{#1}}
\define@key[psset]{pst-plot}{tickwidth}[0.5\pslinewidth]{%
  \pst@getlength{#1}\psk@xtickwidth%
  \pst@getlength{#1}\psk@ytickwidth}
\define@key[psset]{pst-plot}{xtickwidth}[0.5\pslinewidth]{\pst@getlength{#1}\psk@xtickwidth}
\define@key[psset]{pst-plot}{ytickwidth}[0.5\pslinewidth]{\pst@getlength{#1}\psk@ytickwidth}
\define@key[psset]{pst-plot}{subtickwidth}[0.25\pslinewidth]{%
  \pst@getlength{#1}\psk@xsubtickwidth%
  \pst@getlength{#1}\psk@ysubtickwidth}
\define@key[psset]{pst-plot}{xsubtickwidth}[0.25\pslinewidth]{\pst@getlength{#1}\psk@xsubtickwidth}
\define@key[psset]{pst-plot}{ysubtickwidth}[0.25\pslinewidth]{\pst@getlength{#1}\psk@ysubtickwidth}
\define@key[psset]{pst-plot}{labelOffset}[0pt]{%
  \pst@getlength{#1}\psk@xlabelOffset%
  \pst@getlength{#1}\psk@ylabelOffset}
\define@key[psset]{pst-plot}{xlabelOffset}[0pt]{\pst@getlength{#1}\psk@xlabelOffset}
\define@key[psset]{pst-plot}{ylabelOffset}[0pt]{\pst@getlength{#1}\psk@ylabelOffset}
\define@key[psset]{pst-plot}{frameOffset}[0pt]{\pst@getlength{#1}\psk@frameOffset}
\define@key[psset]{pst-plot}{tickcolor}[black]{%
    \pst@getcolor{#1}\psk@xtickcolor%
    \pst@getcolor{#1}\psk@ytickcolor}
\define@key[psset]{pst-plot}{xtickcolor}[black]{\pst@getcolor{#1}\psk@xtickcolor}
\define@key[psset]{pst-plot}{ytickcolor}[black]{\pst@getcolor{#1}\psk@ytickcolor}
\define@key[psset]{pst-plot}{subtickcolor}[gray]{%
  \pst@getcolor{#1}\psk@xsubtickcolor%
  \pst@getcolor{#1}\psk@ysubtickcolor}
\define@key[psset]{pst-plot}{xsubtickcolor}[gray]{\pst@getcolor{#1}\psk@xsubtickcolor}
\define@key[psset]{pst-plot}{ysubtickcolor}[gray]{\pst@getcolor{#1}\psk@ysubtickcolor}
\define@key[psset]{pst-plot}{xticklinestyle}[solid]{%
  \@ifundefined{psls@#1}%
    {\@pstrickserr{Line style `#1' not defined}\@eha}%
    {\def\psxticklinestyle{#1}}}
\define@key[psset]{pst-plot}{xsubticklinestyle}[solid]{%
  \@ifundefined{psls@#1}%
    {\@pstrickserr{Line style `#1' not defined}\@eha}%
    {\def\psxsubticklinestyle{#1}}}
\define@key[psset]{pst-plot}{yticklinestyle}[solid]{%
  \@ifundefined{psls@#1}%
    {\@pstrickserr{Line style `#1' not defined}\@eha}%
    {\def\psyticklinestyle{#1}}}
\define@key[psset]{pst-plot}{ysubticklinestyle}[solid]{%
  \@ifundefined{psls@#1}%
    {\@pstrickserr{Line style `#1' not defined}\@eha}%
    {\def\psysubticklinestyle{#1}}}
\define@key[psset]{pst-plot}{ticklinestyle}[solid]{%
  \@ifundefined{psls@#1}%
    {\@pstrickserr{Line style `#1' not defined}\@eha}%
    {\def\psxticklinestyle{#1}\def\psyticklinestyle{#1}}}
\define@key[psset]{pst-plot}{subticklinestyle}[solid]{%
  \@ifundefined{psls@#1}%
    {\@pstrickserr{Line style `#1' not defined}\@eha}%
    {\def\psxsubticklinestyle{#1}\def\psysubticklinestyle{#1}}}
\psset[pst-plot]{subticksize=0.75,subticks=1,tickcolor=black,ticklinestyle=solid,
  subticklinestyle=solid,subtickcolor=gray,tickwidth=0.5\pslinewidth,
  subtickwidth=0.25\pslinewidth,labelOffset=0pt,frameOffset=0pt}
\define@key[psset]{pst-plot}{nStep}[1]{\def\psk@nStep{#1}}
\define@key[psset]{pst-plot}{nStart}[0]{\def\psk@nStart{#1}}
\define@key[psset]{pst-plot}{nEnd}[{}]{\def\psk@nEnd{#1}}
\define@key[psset]{pst-plot}{xStep}[0]{\def\psk@xStep{#1}}
\define@key[psset]{pst-plot}{yStep}[0]{\def\psk@yStep{#1}}
\define@key[psset]{pst-plot}{xStart}[{}]{\def\psk@xStart{#1}}
\define@key[psset]{pst-plot}{xEnd}[{}]{\def\psk@xEnd{#1}}
\define@key[psset]{pst-plot}{yStart}[{}]{\def\psk@yStart{#1}}
\define@key[psset]{pst-plot}{yEnd}[{}]{\def\psk@yEnd{#1}}
\define@key[psset]{pst-plot}{plotNoX}[1]{\def\psk@plotNoX{#1}}
\define@key[psset]{pst-plot}{plotNo}[1]{\def\psk@plotNo{#1}}
\define@key[psset]{pst-plot}{plotNoMax}[1]{\def\psk@plotNoMax{#1}}
\define@key[psset]{pst-plot}{plotYMax}[1]{\def\psk@plotYMax{#1}}
\psset[pst-plot]{nStep=1, nStart=0, nEnd={},%
  xStep=0, yStep=0, xStart={}, xEnd={},  yStart={}, yEnd={}, 
  plotNo=1,plotNoMax=1,plotNoX=1,plotYMax={}}%
\def\pstScalePoints(#1,#2)#3#4{%
  \def\pstXScale{#1 }%
  \def\pstYScale{#2 }%
  \def\pstXPSScale{#3 }%
  \def\pstYPSScale{#4 }%
  \pst@def{ScalePoints}<%
    /yVal ED /xVal ED
    /yPSOp { #4 yVal mul #2 mul } def
    /xPSOp { #3 xVal mul #1 mul } def
    counttomark dup dup cvi eq not { exch pop } if
    /m exch def /n m 2 div cvi def
    n {
      \ifPst@polarplot exch cvi 360 mod PtoC \fi  % x cvi 360 mod PtoC
      yPSOp m 1 roll xPSOp m 1 roll 
      /m m 2 sub
      def } repeat>%
}
\pstScalePoints(1,1){}{}% the default -> no special operators
\def\psxs@none{\let\psk@arrowA\@empty\let\psk@arrowB\@empty\psxs@axes}
\def\psxs@axes{{%
  \ifPst@xAxis\psxs@@axes\pst@dima\pst@dimb\pst@dimc\pst@dimd{}{x}\fi%
  \ifPst@yAxis\psxs@@axes\pst@dima\pst@dimb\pst@dimc\pst@dimd{exch}{y}\fi%
}}
\newif\ifSpecialLabelsDone
\def\psaxes{\pst@object{psaxes}}
\def\psaxes@i{%
  \let\pst@par@save\pst@par
  \pst@getarrows\psaxes@ii}
\def\psaxes@ii(#1){\@ifnextchar({\psaxes@iii(#1)}{\psaxes@iv(0,0)(0,0)(#1)}}
\def\psaxes@iii(#1)(#2){\@ifnextchar({\psaxes@iv(#1)(#2)}{\psaxes@iv(#1)(#1)(#2)}}
\def\psaxes@iv(#1)(#2)(#3){\@ifnextchar[{\psaxes@v(#1)(#2)(#3)}{\psaxes@vii(#1)(#2)(#3)}}%
\def\psaxes@v(#1)(#2)(#3)[#4]{\@ifnextchar[{\psaxes@vi(#1)(#2)(#3)[#4]}{\psaxes@vi(#1)(#2)(#3)[#4][]}}%
\def\psaxes@vi(#1)(#2)(#3)[#4,#5][#6,#7]{%
  \psaxes@vii(#1)(#2)(#3)%
  \let\pst@par\pst@par@save
  \begingroup
  \SpecialCoor
  \use@par
  \ifshowgrid\psgrid[style=gridstyleA]\fi
  \uput{\psxlabelsep}[#5](#3|#1){#4}\uput{\psylabelsep}[#7](#1|#3){#6}%
  \endgroup
  \ignorespaces
}
\def\psaxes@vii(#1,#2)(#3,#4)(#5,#6){%
  \pst@killglue
  \begingroup
  \ifdim\pst@dimc<\z@\relax 
    \ifdim\pst@dimd<\z@\relax % axes show to left and down
      \addbefore@par{xlabelPos=t,ylabelPos=r}%
  \fi\fi
  \use@par%	now the same with an optional unit=... in par
  \pssetxlength\pst@dimc{#5}% ur-x
  \pssetylength\pst@dimd{#6}% ur-y
    \pssetxlength\pst@dimg{#1}% o-x
    \pssetylength\pst@dimh{#2}% o-y
    \pssetxlength\pst@dima{#3}% ll-x
    \pssetylength\pst@dimb{#4}% ll-y
    \pst@dima=\dimexpr\pst@dima-\pst@dimg\relax
    \pst@dimb=\dimexpr\pst@dimb-\pst@dimh\relax
    \pst@dimc=\dimexpr\pst@dimc-\pst@dimg\relax
    \pst@dimd=\dimexpr\pst@dimd-\pst@dimh\relax
   \setbox\pst@hbox=\hbox\bgroup
    \ifshowgrid\psgrid[style=gridstyleA]\fi
    \@nameuse{psxs@\psk@axesstyle}%  \psxs@axes or \psxs@frame or \psxs@polar
    \ifPst@xAxis
      \SpecialLabelsDonefalse
      \begingroup
      \ifnum\psk@dx=\z@
        \pst@dimg=\psk@Dx\psxunit
        \ifdim\pst@dimg<\p@ 
          \pst@cnta=\psk@Dx
          \edef\psk@Dx{\the\numexpr-1*\pst@cnta}%
        \fi% v.1.21
        \edef\psk@dx{\number\pst@dimg}%
      \fi
      \pst@hlabels{\pst@dimc}{\psk@arrowB}{#3}{#5}% Right
      \ifPst@yAxis\showoriginfalse\fi
      \pst@hlabels{\pst@dima}{\psk@arrowA}{#3}{#5}% Left
      \endgroup
    \fi
    \ifPst@yAxis
      \SpecialLabelsDonefalse
      \begingroup
      \ifdim\pst@dima=\z@ \else\ifPst@xtrigLabels\showoriginfalse\fi\fi
      \ifnum\psk@dy=\z@
        \pst@dimg=\psk@Dy\psyunit
        \ifdim\pst@dimg<\p@ 
          \pst@cnta=\psk@Dy
          \edef\psk@Dy{\the\numexpr-1*\pst@cnta}%
        \fi% v.1.21
        \edef\psk@dy{\number\pst@dimg}%
      \fi
      \pst@vlabels{\pst@dimb}{\psk@arrowA}{#4}{#6}%
      \ifPst@xAxis\ifdim\pst@dima<\z@ \showoriginfalse\fi\fi % no 0 when x- axis is crossing
      \pst@vlabels{\pst@dimd}{\psk@arrowB}{#4}{#6}%
      \endgroup
    \fi
  \egroup%
  \pssetxlength\pst@dimg{#1}%
  \pssetylength\pst@dimh{#2}%
  \leavevmode
  \psput@cartesian\pst@hbox
  \endgroup
  \ignorespaces
}
\newif\ifis@yAxis%
\def\psxs@@axes#1#2#3#4#5#6{% llx,lly,urx,ury,exch,x|y,arrowA,arrowB
  \pst@killglue
  \begin@SpecialObj
    \ifx#6x\relax%				% x-axis?
      \is@yAxisfalse
      \ifnum\psk@dx=\z@
        \pst@dimg=\psk@Dx\psxunit
        \def\psk@dx{\number\pst@dimg}%
      \fi
    \else
      \is@yAxistrue
      \ifnum\psk@dy=\z@
        \pst@dimg=\psk@Dy\psyunit
        \def\psk@dy{\number\pst@dimg}%
      \fi
    \fi
    \let\pst@linetype\pst@arrowtype
    \def\pst@axes{axes}%
    \pst@addarrowdef
    \addto@pscode{
      /showOrigin \ifPst@showOriginTick true \else false \fi def 	% ticks for 0/0 ?
      \ifis@yAxis 0 \pst@number#4 \else \pst@number#3 0 \fi
      \ifis@yAxis 0 \pst@number#2 \else \pst@number#1 0 \fi
      ArrowA
      CP 4 2 roll
      ArrowB 
      2 copy
      /yEnd exch def /xEnd exch def
      \ifx\psk@axesstyle\@none   
        pop pop % axesstyle = none (only ticks) or frame (already drawn)
      \else
        gsave                              		% save current state
        L                                  		% the line with arrows 
        \@nameuse{psls@\pslinestyle}                 	% linestyle for the axes
        stroke                                       	% draw the main line
        grestore
      \fi
      /yStart exch def
      /xStart exch def
      \number\psk@ticks\space dup 2 mod 0 eq \ifis@yAxis true \else false \fi and 
      exch 2 lt \ifis@yAxis false \else true \fi and or {
      /viceversa 
        \ifis@yAxis\pst@number#2 \pst@number#4 \else\pst@number#1 \pst@number#3 \fi
         gt { true }{ false } ifelse def           % other way round
      /epsilon 0.01 def                            % rounding errors
      /minTickline \ifis@yAxis \pst@number#1 \else \pst@number#2 \fi def
      /maxTickline \ifis@yAxis \pst@number#3 \else \pst@number#4 \fi def
      /dT \ifis@yAxis \psk@dy \else \psk@dx \fi\space abs  % added abs 2006-07-07
        65536 div viceversa { neg } if def                 % div to get pt instead of sp
      /DT \ifis@yAxis \psk@Dy \else \psk@Dx \fi\space abs viceversa { neg } if def  
      /subTNo \ifis@yAxis\psk@ysubticks\else\psk@xsubticks\fi \space def
      subTNo 0 gt { /dsubT dT subTNo div def}{ /dsubT 0 def } ifelse  % deltaSubTick
      \ifis@yAxis \psk@yticksize \else \psk@xticksize \fi
      /tickend exch def /tickstart exch def
      /Twidth \ifis@yAxis \psk@ytickwidth \else \psk@xtickwidth \fi\space def
      /subTwidth \ifis@yAxis \psk@ysubtickwidth \else \psk@xsubtickwidth \fi\space def
      /STsize \ifis@yAxis \psk@ysubticksize \else \psk@xsubticksize \fi\space def
      /TColor {
        \ifis@yAxis\pst@usecolor\psk@ytickcolor
        \else\pst@usecolor\psk@xtickcolor\fi\space } def
      /subTColor {
        \ifis@yAxis\pst@usecolor\psk@ysubtickcolor
        \else\pst@usecolor\psk@xsubtickcolor\fi\space } def
      /MinValue { \ifis@yAxis yStart \else xStart \fi
        \ifx\psk@arrowA\@empty\else 
          \psk@arrowsize\space CLW mul add \psk@arrowlength\space mul 
           viceversa { sub epsilon add }{ add epsilon sub } ifelse \fi } def
      /MaxValue { \ifis@yAxis yEnd \else xEnd \fi 
        \ifx\psk@arrowB\@empty\else
          \psk@arrowsize\space CLW mul add \psk@arrowlength\space mul 
           viceversa { add epsilon sub }{ sub epsilon add } ifelse \fi } def
      /logLines {
        \ifnum\psk@logLines=\z@ true \else         % all axes
          \ifnum\psk@logLines<\tw@                 % x axis
            \ifis@yAxis false \else true \fi       % do we have x or y axis
          \else
            \ifnum\psk@logLines<\thr@@             % y axis
              \ifis@yAxis true \else false \fi     % do we have x or y axis
            \else 
              false                                % no one
            \fi
          \fi
        \fi
      } def
      /LSstroke {                                  % set linestyle and stroke
        \ifis@yAxis \@nameuse{psls@\psyticklinestyle}
        \else       \@nameuse{psls@\psxticklinestyle}\fi 
        stroke} def
      /subLSstroke {                               % set sublinestyle and stroke
        \ifis@yAxis \@nameuse{psls@\psysubticklinestyle}
        \else       \@nameuse{psls@\psxsubticklinestyle}\fi 
        stroke} def
      0 dT MaxValue 1 add {                        % the positive part of the axes, step unit is pt
        /cntTick exch def                          % the index
        logLines {                                 % log lines?
          gsave
          1 1 DT {
           1 sub /OffSet exch def
          -10 subTNo 1 add div dup 10 add exch dup -0.1 mul 1 add {                   % do not write a line for 10 and 1, trace lines between 10 and 1 by steps of 10/subTno
            /dx exch def                           % save index
            /x dx log OffSet add \ifis@yAxis\pst@number\psyunit\else\pst@number\psxunit\fi\space mul cntTick add def       %
            x abs MaxValue abs le {                % out of range?
	      \ifis@yAxis
	        \ifPst@maxyTicks true \else false \fi
	      \else
	        \ifPst@maxxTicks true \else false \fi
	      \fi
                { x minTickline #5 moveto
                  x maxTickline #5 lineto }
                { x tickstart STsize mul #5 moveto
                  x tickend STsize mul #5 lineto } ifelse
            } if
          } for } for
          subTwidth SLW subTColor                  % set line width and subtick color
          subLSstroke
          grestore                                 % restore main tick status
          stroke
          /dsubT 0 def                             % no other subticks
        } if 					   % end logLines
        dsubT abs 0 gt {                           % du we have subticks?
          gsave                                    % save graphic state
          /cntsubTick cntTick dsubT add def
          subTNo 1 sub {
            cntsubTick abs MaxValue abs le {       % out of range?
    	    \ifis@yAxis
              \ifPst@maxyTicks true \else false \fi
    	    \else
              \ifPst@maxxTicks true \else false \fi
    	    \fi
              { cntsubTick minTickline STsize mul #5 moveto
                cntsubTick maxTickline STsize mul #5 lineto }
              { cntsubTick tickstart STsize mul #5 moveto
                cntsubTick tickend STsize mul #5 lineto } ifelse
            }{ exit }  ifelse
            /cntsubTick cntsubTick dsubT add def
          } repeat 
          subTwidth SLW subTColor               % set line width and subtick color
          subLSstroke
          grestore                              % restore tick status
        } if
        showOrigin {
          gsave
          \ifis@yAxis
            \ifPst@maxyTicks true \else false \fi
          \else
            \ifPst@maxxTicks true \else false \fi
          \fi
            { cntTick minTickline #5 moveto
              cntTick maxTickline #5 lineto }
            { cntTick tickstart #5 moveto        % line begin main Tick
              cntTick tickend #5 lineto } ifelse % lineto tick end
          Twidth SLW TColor                      % set line width and tick color
          LSstroke
          grestore
        }{ /showOrigin true def } ifelse         % only for the very first tick valid
      } for
      /showOrigin \ifPst@showOriginTick true \else false \fi def % ticks for 0/0 ?
      /dT dT neg def                               % the other side of the axis
      /dsubT dsubT neg def
      0 dT MinValue epsilon viceversa { add }{ sub } ifelse {
        /cntTick exch def
        logLines {                                 % log lines?
          gsave
          1 1 DT cvi {
            1 sub /OffSet exch def
          -10 subTNo 1 add div dup 10 add exch dup -0.1 mul 1 add {                   % do not write a line for 10 and 1, trace lines between 10 and 1 by steps of 10/subTno
            /dx exch def                           % save index
            /x dx log OffSet add \ifis@yAxis\pst@number\psyunit\else\pst@number\psxunit\fi\space mul cntTick add def
            x abs MinValue abs le {                % out of range?
	      \ifis@yAxis
	        \ifPst@maxyTicks true \else false \fi
	      \else
	        \ifPst@maxxTicks true \else false \fi
	      \fi
                { x minTickline #5 moveto
                  x maxTickline #5 lineto }
                { x tickstart STsize mul #5 moveto
                  x tickend STsize mul #5 lineto } ifelse
            } if
          } for } for
          /dsubT 0 def 
          subTwidth SLW subTColor                  % set line width and subtick color
          subLSstroke
          grestore
        }                                          % end loglines
        dsubT abs 0 gt {                           % do we have subticks?
          gsave                                    % save main state
          /cntsubTick cntTick dsubT add def
          subTNo 1 sub {
            cntsubTick abs MinValue abs le {       % out of range?
              cntsubTick tickstart STsize mul #5 moveto
              cntsubTick tickend STsize mul #5 lineto
            }{ exit } ifelse
            /cntsubTick cntsubTick dsubT add def
          } repeat % for
          subTwidth SLW subTColor                  % set line width and subtick color
          subLSstroke
          grestore                                 % restore main state
        } if
        showOrigin {
          gsave
          cntTick tickstart #5 moveto         	% line begin main Tick
          cntTick tickend #5 lineto    	       	% lineto tick end
          Twidth SLW TColor                         % set line width and tick color
          LSstroke
          grestore
        }{ /showOrigin true def } ifelse         % only for the very first tick valid
      } for
    } if
   }%	end of \pscode
  \end@SpecialObj%
  \ifx\psk@axesstyle\@none\else
    \ifPst@yAxis\psline[linecolor=\pslinecolor](0,#2)(0,#4)\fi
    \ifPst@xAxis\psline[linecolor=\pslinecolor](#1,0)(#3,0)\fi
  \fi
  \ignorespaces
}%
\def\psxs@frame{%
  \psset{axesstyle=none}%
  \begin@SpecialObj%
    \addto@pscode{					% the frame
      \pst@number\pst@dima \psk@frameOffset sub \pst@number\pst@dimb \psk@frameOffset sub moveto 	% lower left
      \pst@number\pst@dimc \psk@frameOffset add \pst@number\pst@dimb \psk@frameOffset sub L	% upper left
      \pst@number\pst@dimc \psk@frameOffset add \pst@number\pst@dimd \psk@frameOffset add L 	% upper right
      \pst@number\pst@dima \psk@frameOffset sub \pst@number\pst@dimd \psk@frameOffset add L 	% lower right
      closepath 
      }%
    \pst@stroke%
    \psk@fillstyle%
  \end@SpecialObj%
  \let\psk@arrowA\@empty%
  \let\psk@arrowB\@empty%
  \pst@xticksizeC=\z@\pst@yticksizeC=\z@  
  \ifPst@xAxis\psxs@@axes\pst@dima\pst@dimb\pst@dimc\pst@dimd{}{x}\fi%		x axis
  \ifPst@yAxis\psxs@@axes\pst@dima\pst@dimb\pst@dimc\pst@dimd{ exch }{y}\fi%	y axis
  \ifnum\psk@tickstyle=\tw@	% llx,lly,urx,ury,exch,x|y,arrowA,arrowB	
    \psDEBUG[psxs@frame]{psk@tickstyle=2 (inner)}%
    \psDEBUG[psxs@frame]{pst@dima=\pst@number\pst@dima}%
    \psDEBUG[psxs@frame]{pst@dimb=\pst@number\pst@dimb}%
    \psDEBUG[psxs@frame]{pst@dimc=\pst@number\pst@dimc}%
    \psDEBUG[psxs@frame]{pst@dimd=\pst@number\pst@dimd}%
    \ifPst@xAxis\psxs@@axes\pst@dima\pst@dimb\pst@dimc\pst@dimd{ neg \pst@number\pst@dimd add }{x}\fi%	% upper x axis
    \ifPst@yAxis\psxs@@axes\pst@dima\pst@dimb\pst@dimc\pst@dimd{ neg \pst@number\pst@dimc add exch }{y}\fi%  right y axis
  \fi%
}
\def\psxs@polar{% (rx,ry) % all other values are ignored
  \pst@killglue
  \begingroup
  \edef\pst@dimC{\strip@pt\pst@dimc}% 			RadiusX
  \pstFPDiv\pstR@dius{\pst@dimC}{\strip@pt\psxunit}%	in cm and as integer
  \edef\pst@dimD{\strip@pt\pst@dimd}% 			RadiusX
  \pstFPDiv\psk@EndAngle{\pst@dimD}{\strip@pt\psyunit}%	in cm and as integer
  \ifnum\psk@EndAngle=0 \def\psk@EndAngle{360}\fi
  \use@keep@par
  \pstFPDiv\pstN@lpha{\psk@EndAngle}{\psk@Dy}% 			No. of (int) main lines
  \pstFPdiv\pstd@lpha{\psk@Dy}{\psk@ysubticks}% 	sub dAlpha
  \pstFPdiv\pstdR@dius{1}{\psk@xsubticks}%		sub dRadius
  \pst@cntm=\psk@xsubticks\advance\pst@cntm by \m@ne
  \multido{\iA=\psk@Dx+\psk@Dx,\rB=\pstdR@dius+\psk@Dx,\iB=0+1}{\pstR@dius}{%
    \multido{\rA=\rB+\pstdR@dius}{\the\pst@cntm}{%
      \psarc[linestyle=\psxsubticklinestyle,
         linecolor=\psk@xsubtickcolor,linewidth=\psk@xsubtickwidth pt](0,0){\rA}{0}{\psk@EndAngle}}    
    \psarc[linestyle=\psxticklinestyle,linecolor=\psk@xtickcolor,
		linewidth=\psk@xtickwidth pt](0,0){\iA}{0}{\psk@EndAngle}%
    \ifnum\psk@labels<2\relax% is all or x (0,1)
      \uput[-45](\iB,0){\pshlabel{\iB}}\uput[45](0,\iB){\pshlabel{\iB}}%
    \fi%
  }%
  \pst@cntm=\psk@ysubticks\advance\pst@cntm by \m@ne
  \multido{\iA=\psk@Oy+\psk@Dy,\rB=\pstd@lpha+\psk@Dy}{\pstN@lpha}{%
    \multido{\rA=\rB+\pstd@lpha}{\the\pst@cntm}{\psline[linestyle=\psysubticklinestyle,
      linecolor=\psk@ysubtickcolor,linewidth=\psk@ysubtickwidth pt](\pstR@dius;\rA)} 
    \psline[linestyle=\psyticklinestyle,
      linecolor=\psk@ytickcolor,linewidth=\psk@ytickwidth pt](\pstR@dius;\iA)%
    \ifodd\psk@labels\else% is all or y (0,3)
      \uput[\iA](\pstR@dius;\iA){\psvlabel{\iA\psk@ylabelFactor}}%
    \fi%
  }%
  \ifnum\psk@EndAngle<360 \psline[linestyle=\psyticklinestyle,
      linecolor=\psk@ytickcolor,linewidth=\psk@ytickwidth pt](\pstR@dius;0)\fi
  \endgroup\ignorespaces%
  \Pst@xAxisfalse\Pst@yAxisfalse%
}
\def\@polar{polar}
\define@key[psset]{pst-plot}{axesstyle}[axes]{%
  \@ifundefined{psxs@#1}%
    {\@pstrickserr{Axes style `#1' not defined}\@eha}%
    {\def\psk@axesstyle{#1}%
     \ifx\psk@axesstyle\@polar\psset{Dy=30}\fi}}
\psset[pst-plot]{axesstyle=axes}
\define@key[psset]{pst-plot}{xLabels}[]{\def\psk@xLabels{#1}}
\define@key[psset]{pst-plot}{xLabelsRot}[0]{\pst@getangle{#1}\pst@xLabelsRot}
\psset[pst-plot]{xLabels=,xLabelsRot=0}
\define@key[psset]{pst-plot}{yLabels}[]{\def\psk@yLabels{#1}}
\define@key[psset]{pst-plot}{yLabelsRot}[0]{\pst@getangle{#1}\pst@yLabelsRot}
\psset[pst-plot]{yLabels=,yLabelsRot=0}
\def\pst@hlabels#1#2#3#4{%
  \ifSpecialLabelsDone
  \else
    \kern\psk@xlabelOffset pt            % set the x offset?
    \ifx\empty\psk@xLabels
      \ifdim#1=\z@
      \else                   % start from 0 ?
        \ifx#2\empty
        \else
          \advance#1\ifdim#1>\z@-\fi7\pslinewidth
        \fi
        \pst@cnta=#1\relax                % Distance (in sp) to end.
        \divide\pst@cnta\psk@dx\relax     % Number of ticks/labels
        \ifnum\pst@cnta=\z@
        \else
          \pst@dimb=\psk@dx sp            % Space between ticks.
          \ifnum\psk@labels<\tw@ \ifPst@xAxis\pst@@hlabels\fi\fi
          \showoriginfalse
        \fi
      \fi
   \else
     \ifnum\psk@xlabelPos=\tw@ \def\pst@tempC{90}\else\def\pst@tempC{-90}\fi
       \pstFPsub\pst@pmtempa{#4}{#3}%
       \pstFPDiv\pst@pmtempb{\pst@pmtempa}{\psk@Dx}%
       \pstFPadd\pst@pmtempc{\pst@pmtempb}{-1}%
       \pstFPadd\pst@pmtempd{\pst@pmtempb}{1}%
       \ifdim\pst@pmtempb pt < \z@ 
         \def\pst@pmtempe{\pst@int{\pst@pmtempc}}%
       \else
         \def\pst@pmtempe{\pst@int{\pst@pmtempd}}%
       \fi
       \multido{\nA=0+1,\rA=#3+\psk@Dx}{\pst@pmtempe}{%
         \ifdim \nA pt < \z@ \def\nB{-\nA} \else \def\nB{\nA} \fi
         \uput{\psxlabelsep}[\pst@tempC]{\pst@xLabelsRot}(\rA,0){%
              \strut\expandafter\pshlabel\expandafter{\psPutXLabel{\nB}}}}%
       \SpecialLabelsDonetrue
    \fi
  \fi
}
\def\pst@@hlabels{%
  \setbox\z@=\vbox{%			save all in a box
    \ifcase\psk@xlabelPos
      \vskip-\pst@xticksizeA\vskip\psxlabelsep\or% 1
      \vskip-1ex\vskip-\pslabelsep\or% 2
      \vskip-\pst@xticksizeB\vskip-\psxlabelsep\vskip-1ex% 3
    \fi
    \ifnum\pst@cnta<\z@ \pst@dimb=-\pst@dimb\fi
    \hbox to \z@{%
      \ifshoworigin\hbox to \z@{\hss\pst@@@hlabel{\psk@Ox}\hss}\fi
      \mmultido{\nA=\psk@Ox+\psk@Dx}{\pst@cnta}{%
        \hskip\pst@dimb \hbox to \z@{\hss
          \ifdim\nA pt=\z@\relax\ifshoworigin\pst@@@hlabel{0}\fi
          \else\expandafter\pst@@@hlabel{\nA}%
          \fi% prevent -0, doesn't work with \ifnum
        \hss}%
      }\hss%    1.85
    }%
  }\ht\z@\z@ \dp\z@\z@ \box\z@}% set all values to zero
\def\pst@vlabels#1#2#3#4{%
  \ifSpecialLabelsDone\else
      \ifx\empty\psk@yLabels
        \ifdim#1=\z@\else
          \ifx#2\empty\else\ifdim#1>\z@ \advance#1 by -7\pslinewidth\else\advance#1 by 7\pslinewidth\fi\fi
          \pst@cnta=#1\relax           %      % Distance (in sp) to end.
          \divide\pst@cnta\psk@dy\relax%   % Number of ticks/labels
          \ifnum\pst@cnta=\z@\else
            \pst@dima=\psk@dy sp%            % Space between ticks.
            \ifodd\number\psk@labels\else\ifPst@yAxis\pst@@vlabels\fi
          \fi
          \showoriginfalse
        \fi
      \fi
    \else
	\pstFPsub\pst@pmtempa{#4}{#3}%
	\pstFPDiv\pst@pmtempb{\pst@pmtempa}{\psk@Dy}%
	\pstFPadd\pst@pmtempc{\pst@pmtempb}{-1}%
	\pstFPadd\pst@pmtempd{\pst@pmtempb}{1}%
	\ifdim\pst@pmtempb pt < \z@ \def\pst@pmtempe{\pst@int{\pst@pmtempc}}\else\def\pst@pmtempe{\pst@int{\pst@pmtempd}}\fi
	\multido{\nA=0+1,\rA=#3+\psk@Dy}{\pst@pmtempe}{%
	  \ifdim \nA pt < \z@ \def\nB{-\nA}\else \def\nB{\nA}\fi
	  \ifnum\psk@ylabelPos=0
            \uput{\psylabelsep}[180]{\pst@yLabelsRot}(0,\rA){%
              \strut\expandafter\psvlabel\expandafter{\psPutYLabel{\nB}}}%
          \else
            \uput{\psylabelsep}[0]{\pst@yLabelsRot}(0,\rA){%
              \strut\expandafter\psvlabel\expandafter{\psPutYLabel{\nB}}}%
          \fi
        }%  
      \SpecialLabelsDonetrue
    \fi
  \fi
}
\def\pst@@vlabels{%
  \vbox to\z@{%
   \vbox to -\psk@ylabelOffset pt{}% the y label offset
    \ifnum\pst@cnta>\z@ \pst@dima=-\pst@dima\fi%  up or down label positions
    \offinterlineskip
    \ifshoworigin
      \vbox to \z@{\vss\hbox to\z@{%
        \ifcase\psk@ylabelPos
	  \hss\pst@@@vlabel{\psk@Oy}\hskip\psylabelsep\hskip-\pst@yticksizeA\or%
	  \hskip\pslabelsep\hss\pst@@@vlabel{\psk@Oy}\hss\or%		% right labels
	  \hskip\pst@yticksizeB\hskip\psylabelsep\pst@@@vlabel{\psk@Oy}%
	\fi}\vss}%
    \fi
    \mmultido{\nA=\psk@Oy+\psk@Dy}{\pst@cnta}{%
      \vbox to\pst@dima{\vss}%
      \vbox to \z@{%
        \vss\hbox to\z@{%
        \ifcase\psk@ylabelPos % and also check for -0
	  \hss\ifdim\nA pt=\z@ \ifshoworigin\pst@@@vlabel{0}\fi\else\pst@@@vlabel{\nA}\fi
	    \hskip\psylabelsep\hskip-\pst@yticksizeA\or% top = 1
	  \hss\ifdim\nA pt=\z@ \ifshoworigin\pst@@@vlabel{0}\fi\else\pst@@@vlabel{\nA}\fi
	  \ifdim\psylabelsep=\z@\hss\else\kern-\psylabelsep\fi\or% right=2
	  \hskip\pst@yticksizeB\hskip\psylabelsep
	  \ifdim\nA pt=\z@ \ifshoworigin\pst@@@vlabel{0}\fi\else\pst@@@vlabel{\nA}\fi% bottom
	\fi}\vss}%
    }\vss}%
}
\define@key[psset]{pst-plot}{xAxisLabel}[x]{\def\psk@xAxisLabel{#1}}
\define@key[psset]{pst-plot}{yAxisLabel}[y]{\def\psk@yAxisLabel{#1}}
\psset[pst-plot]{xAxisLabel=x,yAxisLabel=y}
\define@key[psset]{pst-plot}{xAxisLabelPos}[{}]{\def\psk@xAxisLabelPos{#1}}
\define@key[psset]{pst-plot}{yAxisLabelPos}[{}]{\def\psk@yAxisLabelPos{#1}}
\psset[pst-plot]{yAxisLabelPos={},xAxisLabelPos={}}
\newdimen\psk@llx
\newdimen\psk@lly
\newdimen\psk@urx
\newdimen\psk@ury
\define@key[psset]{pst-plot}{llx}[\z@]{\pssetxlength\psk@llx{#1}}
\define@key[psset]{pst-plot}{lly}[\z@]{\pssetylength\psk@lly{#1}}
\define@key[psset]{pst-plot}{urx}[\z@]{\pssetxlength\psk@urx{#1}}
\define@key[psset]{pst-plot}{ury}[\z@]{\pssetylength\psk@ury{#1}}
\psset[pst-plot]{llx=\z@, lly=\z@, urx=\z@, ury=\z@}% prevents rounding errors 
\define@boolkey[psset]{pst-plot}[Pst@]{psgrid}[true]{}
\define@key[psset]{pst-plot}{gridpara}[{}]{\def\psk@gridpara{#1}}
\define@key[psset]{pst-plot}{gridcoor}[\relax]{\def\psk@gridcoor{#1}}
\psset[pst-plot]{psgrid=false,gridpara={gridlabels=0pt,gridcolor=red!30,subgridcolor=green!30,subgridwidth=0.5\pslinewidth,
  subgriddiv=5},gridcoor=\relax}
\define@key[psset]{pst-plot}{axespos}[bottom]{\pst@expandafter\psset@@axespos{#1}\@nil}
\def\psset@@axespos#1#2\@nil{%
  \ifx#1b\let\psk@axespos\z@\else		% 0=b)bottom
    \ifx#1t\let\psk@axespos\@ne			% 1=t)op
      \else\@pstrickserr{Bad axes position: `#1#2'}\@ehpa
  \fi\fi}
\psset[pst-plot]{axespos=b}
\newdimen\pst@xunit
\newdimen\pst@yunit
\def\pslegend{\@ifnextchar[\pslegend@i{\pslegend@i[rt]}}
\def\pslegend@i[#1]{\@ifnextchar({\pslegend@ii[#1]}{\pslegend@ii[#1](\pst@number\pslabelsep,\pst@number\pslabelsep)}}
\def\pslegend@ii[#1](#2,#3)#4{%
  \gdef\pslegend@ref{#1}%
  \xdef\pslegend@sepx{#2 }%
  \xdef\pslegend@sepy{#3 }%
  \gdef\pslegend@text{#4}}
\newpsstyle{legendstyle}{fillstyle=solid,fillcolor=white,linewidth=0.5pt}
\def\pslegend@iii[#1](#2){\rput[#1](#2){\psframebox[style=legendstyle]{%
  \footnotesize\tabcolsep=2pt%
  \tabular[t]{@{}ll@{}}\pslegend@text\endtabular}}\global\let\pslegend@text\relax}
\let\pslegend@text\relax% define it as empty
\def\psgraph{\pst@object{psgraph}}
\def\psgraph@i{%
  \let\psgraph@para\pst@par
  \let\psk@save@arrowA\psk@arrowA
  \let\psk@save@arrowB\psk@arrowB
  \pst@getarrows\psgraph@ii}
\def\psgraph@ii(#1,#2){\catcode`\!=12\relax
  \@ifnextchar({\psgraph@iii(#1,#2)}{\psgraph@iv(0,0)(#1,#2)}}
\def\psgraph@iii(#1,#2)(#3,#4){\@ifnextchar({\psgraph@v(#1,#2)(#3,#4)}{\psgraph@iv(#1,#2)(#3,#4)}}
\def\psgraph@iv(#1,#2)(#3,#4)#5#6{%  no special origin defined
  \pst@killglue%
  \begingroup
  \use@keep@par
  \pstFPsub\pst@tempA{#3}{#1}%
  \pst@dimm=#5
  \pst@dimo=\pst@tempA pt
  \pstFPdiv\pst@@dx{\strip@pt\pst@dimm}{\pst@tempA}%
  \pst@xunit=\pst@@dx\p@
  \ifx!#6\let\pst@yunit=\pst@xunit\else
    \pst@dimm=#6
    \pstFPsub\pst@tempA{#4}{#2}%
    \pstFPdiv\pst@@dy{\strip@pt\pst@dimm}{\pst@tempA}%
    \pst@yunit=\pst@@dy\p@
  \fi
  \pst@dimm=#1\pst@xunit\advance\pst@dimm by \psk@llx
  \pst@dimn=#2\pst@yunit\advance\pst@dimn by \psk@lly
  \pst@dimo=#3\pst@xunit\advance\pst@dimo by \psk@urx
  \pst@dimp=#4\pst@yunit\advance\pst@dimp by \psk@ury
  \if@star\pspicture*(\pst@dimm,\pst@dimn)(\pst@dimo,\pst@dimp)\else
  \pspicture(\pst@dimm,\pst@dimn)(\pst@dimo,\pst@dimp)\fi
  \let\psxunit\pst@xunit \let\psyunit\pst@yunit
  \ifdim\pst@xunit=\pst@yunit\relax\psset{runit=\pst@xunit}\fi%
  \bgroup
    \use@par
  \ifPst@psgrid
     \expandafter\psset\expandafter{\psk@gridpara}%
      \rput[lb](0,0){\expandafter\psgrid\psk@gridcoor}  
  \fi
    \ifnum\psk@axespos=0
      \expandafter\psaxes\expandafter[\psgraph@para](#1,#2)(#3,#4)%
    \else
      \xdef\psgraph@coor{(#1,#2)(#3,#4)(#5,#6)}%
    \fi
  \egroup
  \psgraph@vi(#1,#2)(#1,#2)(#3,#4)%
}
\def\psgraph@v(#1,#2)(#3,#4)(#5,#6)#7#8{%  with special origin
  \pst@killglue%
  \let\psgraph@para\pst@par
  \begingroup%
  \use@keep@par
  \pstFPsub\pst@tempA{#5}{#3}%
  \pst@dimm=#7%
  \pst@dimo=\pst@tempA pt%
  \pstFPdiv\pst@@dx{\strip@pt\pst@dimm}\pst@tempA%
  \pst@xunit=\pst@@dx\p@%
  \ifx!#8\let\pst@yunit=\pst@xunit\else
    \pst@dimm=#8%
    \pstFPsub\pst@tempA{#6}{#4}%
    \pstFPdiv\pst@@dy{\strip@pt\pst@dimm}\pst@tempA%
    \pst@yunit=\pst@@dy\p@%
  \fi%
  \pst@dima=#3\pst@xunit \advance\pst@dima by \psk@llx%
  \pst@dimb=#4\pst@yunit \advance\pst@dimb by \psk@lly%
  \pst@dimc=#5\pst@xunit \advance\pst@dimc by \psk@urx%
  \pst@dimd=#6\pst@yunit \advance\pst@dimd by \psk@ury%
  \if@star\pspicture*(\pst@dima,\pst@dimb)(\pst@dimc,\pst@dimd)\else%
          \pspicture(\pst@dima,\pst@dimb)(\pst@dimc,\pst@dimd)\fi%
  \psset{xunit=\pst@xunit,yunit=\pst@yunit}
  \ifdim\pst@xunit=\pst@yunit \psset{runit=\pst@xunit}\fi%
  \bgroup%
    \use@par%
  \ifPst@psgrid
     \expandafter\psset\expandafter{\psk@gridpara}%
      \rput[lb](0,0){\expandafter\psgrid\psk@gridcoor}
  \fi%
    \ifnum\psk@axespos=0
      \psaxes(#1,#2)(#3,#4)(#5,#6)%
    \else
      \xdef\psgraph@coor{(#1,#2)(#3,#4)(#5,#6)}%
    \fi
  \egroup
  \psgraph@vi(#1,#2)(#3,#4)(#5,#6)%
}
\def\setxLabelC@@r#1,#2(#3,#4)(#5){%
  \pst@getcoor{#5}\pst@tempB%
  \ifx c#1 
    \pssetylength\pst@dimm{#2}%
    \rput(! #4 #3 add 2 div \pst@number\pst@dimm \pst@tempB\space exch pop add 
      \pst@number\psyunit div ){\psk@xAxisLabel}%
  \else%
    \pst@getcoor{\psk@xAxisLabelPos}\pst@tempA%
    \rput(! \pst@tempA\space \pst@tempB\space exch pop add \tx@UserCoor ){\psk@xAxisLabel}%
  \fi}
\def\setyLabelC@@r#1,#2(#3,#4)(#5){%
  \pst@getcoor{#5}\pst@tempB%
  \ifx c#2
    \pssetxlength\pst@dimm{#1}%
    \rput{90}(! \pst@number\pst@dimm \pst@tempB\space pop add \pst@number\psxunit div #4 #3 add 2 div ){\psk@yAxisLabel}%
  \else%
    \pst@getcoor{\psk@yAxisLabelPos}\pst@tempA%
    \rput{90}(! \pst@tempB\space pop \pst@tempA\space 3 1 roll add exch \tx@UserCoor ){\psk@yAxisLabel}%
  \fi}
\def\psgraph@vi(#1,#2)(#3,#4)(#5,#6){%
  \ifx\psk@xAxisLabel\@empty\else%
    \ifx\psk@xAxisLabelPos\@empty\uput[0](#5,#2){\psk@xAxisLabel}%
    \else\expandafter\setxLabelC@@r\psk@xAxisLabelPos(#3,#5)(#1,#2)\fi%
  \fi%
  \ifx\psk@yAxisLabel\@empty\else%
    \ifx\psk@yAxisLabelPos\@empty\uput[90](#1,#6){\psk@yAxisLabel}%
    \else\expandafter\setyLabelC@@r\psk@yAxisLabelPos(#4,#6)(#1,#2)\fi%
  \fi%
  \def\lt@@{lt}\def\lb@@{lb}\def\rb@@{rb}%
  \ifx\pslegend@ref\lb@@    \gdef\pslegend@coor{#3 \pslegend@sepx \pst@number\psxunit div add 
                                                   \pslegend@sepy \pst@number\psyunit div}%
  \else%
    \ifx\pslegend@ref\lt@@  \gdef\pslegend@coor{#3 \pslegend@sepx \pst@number\psxunit div add 
                                                #6 \pslegend@sepy \pst@number\psyunit div sub}%
    \else%
      \ifx\pslegend@ref\rb@@\gdef\pslegend@coor{#5 \pslegend@sepx \pst@number\psxunit div sub 
                                                   \pslegend@sepy \pst@number\psyunit div}%
      \else                 \gdef\pslegend@coor{#5 \pslegend@sepx \pst@number\psxunit div sub 
                                                #6 \pslegend@sepy \pst@number\psyunit div sub}%
      \fi%
    \fi%
  \fi%
  \xdef\psgraphLLx{#3}\xdef\psgraphLLy{#4}\xdef\psgraphURx{#5}\xdef\psgraphURy{#6}%
  \global\let\psk@arrowA\psk@save@arrowA
  \global\let\psk@arrowB\psk@save@arrowB
  \ignorespaces
}
\def\endpsgraph{%
  \ifx\relax\pslegend@text\relax \else\pslegend@iii[\pslegend@ref](!\pslegend@coor)\fi
  \expandafter\psset\expandafter{\psgraph@para}%
  \ifnum\psk@axespos>0
    \expandafter\psaxes\psgraph@coor
  \fi
  \endpspicture
  \endgroup\ignorespaces}
\@namedef{psgraph*}{\psgraph*}
\@namedef{endpsgraph*}{\endpsgraph}
\def\psPutXLabel#1{%
  \global\pst@cnto=0\relax
  \global\pst@cntp=#1\relax
  \expandafter\get@Label\psk@xLabels,,\@nil%
}
\def\psPutYLabel#1{%        
  \global\pst@cnto=0\relax
  \global\pst@cntp=#1\relax
  \expandafter\get@Label\psk@yLabels,,\@nil%
}
\def\get@Label#1,#2,#3\@nil{%
    \ifnum\the\pst@cnto<\the\pst@cntp
      \global\advance\pst@cnto by \@ne 
      \ifx\relax#3\relax\else\expandafter\get@Label#2,#3\@nil\fi%
    \else #1\fi%
}
\def\psVectorfield{\pst@object{psVectorfield}}
\def\psVectorfield@i(#1,#2)(#3,#4)#5{{%
  \addbefore@par{Dx=0.1,Dy=0.1,Ox=3,arrows=->,linewidth=0.2pt}%
  \begin@SpecialObj
  \SpecialCoor
  \pstFPsub\pst@tempA{#3}{#1}%
  \pstFPsub\pst@tempB{#4}{#2}%
  \pstFPDiv{\pst@tempC}{\pst@tempA}{\psk@Dx}%
  \pstFPDiv{\pst@tempD}{\pst@tempB}{\psk@Dy}%
  \pstVerb{ /yStrich \ifPst@algebraic (#5) tx@AlgToPs begin AlgToPs end cvx
                \else { #5 } \fi def }%
  \multido{\rX=#1+\psk@Dx}{\numexpr\pst@tempC+1}{%
    \multido{\rY=#2+\psk@Dy}{\numexpr\pst@tempD+1}{%
       \psline%
         (! /x \rX\space def 
            /y \rY\space def 
            /yTemp yStrich \psk@Dx\space \psk@Ox\space div mul def 
            \rX\space \psk@Dx\space \psk@Ox\space div sub \rY\space yTemp sub)%
         (! /x \rX\space def 
            /y \rY\space def 
            /yTemp yStrich \psk@Dx\space \psk@Ox\space div mul def 
            \rX\space \psk@Dx\space \psk@Ox\space div add \rY\space yTemp add)%
   }}%
  \end@SpecialObj
}\ignorespaces}  
\def\psFixpoint{\pst@object{psFixpoint}}
\def\psFixpoint@i#1#2#3{% #1: xStart #2: f(x) #3: number of iterations
  \pst@killglue%
  \begingroup%
  \use@par%
  \@nameuse{beginplot@\psplotstyle}%
  \addto@pscode{
    \psplot@init
      /x #1 def
      /F@pstplot \ifPst@algebraic (#2) tx@AlgToPs begin AlgToPs end cvx
                 \else { #2 } \fi  def
      /xy { x \pst@number\psxunit mul F@pstplot dup /x ED \pst@number\psyunit mul } def 
  }%
  \gdef\psplot@init{}%
  \@pstfalse%
  \@nameuse{testqp@\psplotstyle}%
  \addto@pscode{
      mark
      x \pst@number\psxunit mul 0
      /n 2 def
      #3 {
        xy 
        dup dup 
        /n n 4 add def
      } repeat 
  }%
  \@nameuse{endplot@\psplotstyle}%
  \endgroup%
  \ignorespaces}
\define@boolkey[psset]{pst-plot}[Pst@]{showDerivation}[true]{}
\psset{showDerivation}
\def\psNewton{\pst@object{psNewton}}
\def\psNewton@i#1#2{\@ifnextchar[{\psNewton@ii{#1}{#2}}{\psNewton@iii{#1}{#2}}}
\def\psNewton@ii#1#2[#3]#4{% #1:xStart #2:f(x) #3:f'(x) #4:number of iterations
  \pst@killglue%
  \begingroup%
  \addbefore@par{showDerivation}%
  \use@par%
  \@nameuse{beginplot@\psplotstyle}%
  \addto@pscode{
    \psplot@init
      /x #1 def
      /F@pstplot \ifPst@algebraic (#2) tx@AlgToPs begin AlgToPs end cvx \else { #2 } \fi  def
      /F@pstplotDerive \ifPst@algebraic (#3) tx@AlgToPs begin AlgToPs end cvx \else { #3 } \fi  def
      /newxVal { % y on stack
        F@pstplotDerive % we have m
        div neg %\pst@number\psxunit div % new x val = -y0/m
      } def
  }%
  \gdef\psplot@init{}%
  \@pstfalse%
  \@nameuse{testqp@\psplotstyle}%
  \addto@pscode{
      mark
      x 0 \tx@ScreenCoor % start point
      /n 2 def
      #4 {
        F@pstplot /yVal ED
        x yVal \tx@ScreenCoor
        /n n 2 add def
        yVal newxVal x add /x ED
        x 0 \tx@ScreenCoor 
        \ifPst@showDerivation /n n 4 add def \else moveto /n n 2 add def\fi
      } repeat 
      pstack
  }%
  \@nameuse{endplot@\psplotstyle}%
  \endgroup%
  \ignorespaces}
\def\psNewton@iii#1#2#3{% #1:xStart #2:f(x) #3:number of iterations
  \pst@killglue%
  \begingroup%
  \addbefore@par{VarStepEpsilon=0.01,showDerivation}%
  \use@par%
  \@nameuse{beginplot@\psplotstyle}%
  \addto@pscode{
    \psplot@init
      /epsilon \psk@VarStepEpsilon\space def
      /x #1 def
      /F@pstplot \ifPst@algebraic (#2) tx@AlgToPs begin AlgToPs end cvx \else { #2 } \fi  def
      /newxVal { % y on stack
        /saveX x def
        saveX epsilon add /x ED F@pstplot saveX epsilon sub /x ED F@pstplot sub epsilon dup add div % we have m
        div neg % new x val = -y0/m
        /x saveX def
      } def
  }%
  \gdef\psplot@init{}%
  \@pstfalse%
  \@nameuse{testqp@\psplotstyle}%
  \addto@pscode{
      mark
      x 0 \tx@ScreenCoor % start point
      /n 2 def
      #3 {
        F@pstplot /yVal ED
        x yVal \tx@ScreenCoor
        yVal newxVal x add /x ED
        x 0 \tx@ScreenCoor 
        \ifPst@showDerivation /n n 4 add def \else moveto /n n 2 add def\fi
      } repeat 
  }%
  \@nameuse{endplot@\psplotstyle}%
  \endgroup%
  \ignorespaces}
\def\psResetPlotValues{%
  \psset{method={}}%
}%
\catcode`\@=\TheAtCode\relax
 
\csname PSTnodesLoaded\endcsname
\let\PSTnodesLoaded 
\ifx\PSTricksLoaded \else\input pstricks.tex \fi\relax
\ifx\PSTXKeyLoaded \else \input pst-xkey \fi
\def\fileversion{1.42}
\def\filedate{2019/03/03}
\edef\TheAtCode{\the\catcode`\@}
\catcode`\@=11
\pstheader{pst-node.pro}
\pst@addfams{pst-node}
\SpecialCoor
\define@boolkey[psset]{pst-node}[Pst@]{trueAngle}[true]{}
\psset[pst-node]{trueAngle=false}
\define@boolkey[psset]{pst-node}[Pst@]{storeNodeInfo}[true]{}
\psset[pst-node]{storeNodeInfo=false}
\def\pst@nodedict{tx@NodeDict begin }
\def\pst@zapspace#1 #2{%
#1%
\ifx#2\@empty\else\expandafter\pst@zapspace\fi
#2}
\def\pst@getnode#1#2{\pst@expandafter\pst@@getnode{#1},,\@nil#2}
\def\pst@@getnode#1,#2,#3\@nil#4{%
  \ifx\@empty#3\@empty
    \edef#4{/N@\pst@zapspace#1 \@empty\space}%
  \else
    \pst@cntg=#1\relax
    \pst@cnth=#2\relax
    \edef#4{/N@M-\ifnum\psmatrixcnt=\z@ 1\else\the\psmatrixcnt\fi
    -\the\pst@cntg-\the\pst@cnth\space}%
  \fi}
\def\tx@NewNode{/NodeScale {\ifx\pstnodescale\@undefined  \else\pstnodescale \fi} def NewNode }
\def\psopenNodeFile{%
  \pst@Verb{ %globaldict begin
    (\jobname.nodes) (w) file /NodeFile exch def 
  }}
\def\pscloseNodeFile{\pstVerb{ tx@NodeDict begin NodeFile closefile end }}
\define@boolkey[psset]{pst-node}[Pst@]{showNode}[true]{\ifPst@showNode\psopenNodeFile\fi}% write coors into the logfile
\define@boolkey[psset]{pst-node}[Pst@]{markNode}[true]{}% makr the node with its name
\define@boolkey[psset]{pst-node}[Pst@]{saveNodeCoors}[true]{}
\define@key[psset]{pst-node}{NodeCoorPrefix}[]{\def\psk@NodeCoorPrefix{#1}}% if empty it is N-<Name>.x|y
\psset[pst-node]{saveNodeCoors=false,showNode=false,markNode=false,NodeCoorPrefix=}% <NodeCoorPrefix<Name>x|y>
\def\pst@newnode#1#2#3#4{%
\pst@killglue
\leavevmode
\pst@getnode{#1}\pst@thenode
\pst@Verb{
  \ifPst@saveNodeCoors
    \ifx\relax#3\relax 0 0 \else gsave \pst@dict STV CP T end #3 \tx@UserCoor grestore \fi 
    \if$\psk@NodeCoorPrefix$
      /N-#1.y exch def
      /N-#1.x exch def
    \else
      /\psk@NodeCoorPrefix#1y exch def
      /\psk@NodeCoorPrefix#1x exch def
    \fi
  \fi
  \pst@nodedict
  {#3}
  \ifx\psk@name\relax false \else \psk@name true \fi
  \pst@thenode
  #2
  {#4}
  \ifPst@showNode 
  exch dup /NodeType ED 
  exch
   NodeType 10 eq {  % pnode type
    5 copy 
    cvlit aload pop
    20 string cvs (; )   6 2 roll % InitPnode
    20 string cvs (; )   7 2 roll % type
    20 string cvs (; )   8 2 roll %/N@Name
    20 string cvs (; )   9 2 roll % true/false
    cvlit dup length 2 eq 
      { aload pop exch 
        20 string cvs (; ) 11 2 roll 
        20 string cvs (, ) 12 2 roll  % x,y
        (\string\n)                   % add newline
        13 array astore concatstringarray 
      }
      { 255 string cvs (; ) 10 2 roll 
        (\string\n)                   % add newline
        11 array astore concatstringarray 
      } ifelse 
    NodeFile exch writestring 
  } if
  NodeType 14 eq {  % dotnode
    5 copy 
    /@@temp ED 
    @@temp  % to get X Y
    4 -1 roll cvlit pop
    ( OvalNodePos ) (; )  5 2 roll
    20 string cvs (; )   6 2 roll % type
    20 string cvs (; )   7 2 roll %/N@Name
    20 string cvs (; )   8 2 roll % true/false
    Y 20 string cvs (; ) 10 2 roll
    X 20 string cvs (, ) 12 2 roll
    (\string\n)                   % add newline
    13 array astore concatstringarray 
    tx@NodeDict begin NodeFile exch writestring end
  } if
  \fi
  \tx@NewNode
  end 
}%
\global\let\psk@name\relax%
\pstree@nodehook%
\global\let\pstree@nodehook\relax}
\let\pstree@nodehook\relax
\define@boolkey[psset]{pst-node}[Pst@]{nodealign}[true]{}
\psset[pst-node]{nodealign=false}

\def\pst@nodealign{%
\pst@dimg=\ht\pst@hbox
\advance\pst@dimg by -\dp\pst@hbox
\divide\pst@dimg by \tw@
\lower\pst@dimg}
\def\tx@InitPnode{InitPnode }
\def\pnode{\@ifnextchar[{\pnode@i}{\pnode@iii}}
\def\pnode@i[#1]{\@ifnextchar({\pnode@ii[#1]}{\pnode@ii[#1](0,0)}}
\def\pnode@ii[#1](#2)#3{%
  \pst@getcoor{#1}\pst@tempA%
  \pst@getcoor{#2}\pst@tempB%
  \pst@newnode{#3}{10}{\pst@tempA \pst@tempB 3 -1 roll add 3 1 roll add exch }{\tx@InitPnode}%
  \ifPst@showNode\psdot(#3)\uput[\ifx\psk@rot\@empty0\else\psk@rot\fi]{0}(#3){#3}\fi
  \ignorespaces}
\def\pnode@iii{\@ifnextchar({\pnode@}{\pnode@(0,0)}}
\def\pnode@(#1)#2{%
  \pst@@getcoor{#1}%
  \pst@newnode{#2}{10}{\pst@coor}{\tx@InitPnode}%
  \ifPst@showNode\psdot(#2)\uput[\ifx\psk@rot\@empty0\else\psk@rot\fi]{0}(#2){#2}\fi
  \ignorespaces}
\def\pnodes{\@ifnextchar[{\pnodes@i}{\pnodes@i[0,0]}}
\def\pnodes@i[#1]{\@ifnextchar({\psnodes@ii[#1]}{\pnodes@ii}}
\def\psnodes@ii[#1](#2)#3{%
  \pnode[#1](#2){#3}%
  \@ifnextchar({\psnodes@ii[#1]}{}%
}
\def\tx@InitCnode{InitCnode }
\def\cnode{\pst@object{cnode}}
\def\cnode@i{\@ifnextchar({\cnode@ii}{\cnode@ii(0,0)}}
\def\cnode@ii(#1)#2#3{%
  \leavevmode
  \hbox{%
    \use@par
    \pst@@getcoor{#1}%
    \pssetlength\pst@dimc{#2}%
    \pst@dimg=\psk@dimen\pslinewidth
    \advance\pst@dimc-\pst@dimg
    \advance\pst@dimc.5\pslinewidth
    \ifPst@nodealign
      \kern\pst@dimc
      \vrule width\z@ height \pst@dimc depth \pst@dimc
    \fi
    \pscircle@do(#1){#2}%
    \pst@newnode{#3}{11}{\pst@coor \pst@number\pst@dimc}{\tx@InitCnode}%
    \ifPst@nodealign\kern\pst@dimc\fi%
  }%
  \ignorespaces}
\def\Cnode{\pst@object{Cnode}}
\def\Cnode@i{\@ifnextchar({\Cnode@ii}{\Cnode@ii(0,0)}}
\def\Cnode@ii(#1)#2{\cnode@ii(#1){\psk@radius}{#2}}%
\def\cnodeput{\pst@object{cnodeput}}
\def\cnodeput@i{\@ifnextchar({\cnodeput@iii}{\cnodeput@ii}}
\def\cnodeput@ii#1{%
  \addto@par{rot={#1}}%
  \@ifnextchar({\cnodeput@iii}{\cnodeput@iii(\z@,\z@)}%
}
\def\cnodeput@iii(#1)#2{%
  \pst@killglue
  \@fixedradiusfalse
  \def\pst@nodehook{\cnodeput@iv{#2}}%
  \pst@makebox{\cput@v{#1}}%
}
\def\cnodeput@iv#1{%
  \pst@newnode{#1}{11}{\pscirclebox@iv \pst@number\pslinewidth add}{\tx@InitCnode}%
  \global\let\pst@nodehook\relax
  \ignorespaces
}
\def\Cnodeput{\pst@object{Cnodeput}}
\def\Cnodeput@i{\@ifnextchar({\Cnodeput@iii}{\Cnodeput@ii}}
\def\Cnodeput@ii#1{%
  \addto@par{rot={#1}}%
  \@ifnextchar({\Cnodeput@iii}{\Cnodeput@iii(\z@,\z@)}}
\def\Cnodeput@iii(#1)#2{%
  \pst@killglue
  \@fixedradiustrue
  \def\pst@nodehook{\Cnodeput@iv{#2}}%
  \pst@makebox{\cput@v{#1}}%
}
\def\Cnodeput@iv#1{%
  \pst@newnode{#1}{11}{%
    \pst@number{\wd\pst@hbox} 2 div \pst@number\pst@dima % x y
    \pst@number\pst@dimb \pst@number\pslinewidth \psk@dimen .5 sub mul sub }% r
       {\tx@InitCnode}%
  \global\let\pst@nodehook\relax}
\def\circlenode{\pst@object{circlenode}}
\def\circlenode@i#1{\pst@makebox{\circlenode@ii{#1}}}
\def\circlenode@ii#1{%
  \begingroup
  \pst@useboxpar
  \setbox\pst@hbox=\hbox{%
    \cnodeput@iv{#1}%
    \pscirclebox@iii
    \box\pst@hbox}%
  \ifPst@nodealign \psboxseptrue \fi
  \ifpsboxsep \pscirclebox@sep \fi
  \leavevmode
  \ifPst@nodealign\pst@nodealign\fi
  \box\pst@hbox
  \endgroup}
\def\Circlenode{\pst@object{Circlenode}}
\def\Circlenode@i#1{\pst@makebox{\Circlenode@ii{#1}}}
\def\Circlenode@ii#1{%
\begingroup
  \pst@useboxpar
  \pst@dima=\ht\pst@hbox
  \advance\pst@dima by -\dp\pst@hbox
  \divide\pst@dima by \tw@
  \pssetlength\pst@dimb\psk@radius
  \setbox\pst@hbox=\hbox{%
  \Cnodeput@iv{#1}%
  \pscircle(.5\wd\pst@hbox,\pst@dima){\pst@dimb}%
  \box\pst@hbox}%
  \ifPst@nodealign \psboxseptrue \fi
  \ifpsboxsep \psCirclebox@sep \fi
  \leavevmode
  \ifPst@nodealign\pst@nodealign\fi
  \box\pst@hbox
  \endgroup}
\def\tx@GetRnodePos{GetRnodePos }
\def\tx@InitRnode{InitRnode }
\def\psnode{\pst@object{psnode}}
\def\psnode@i{\@ifnextchar(\psnode@ii{\psnode@ii(0,0)}}
\def\psnode@ii(#1)#2#3{%    #1: coordinates, #2: node name,  #3 contents
  \rput(#1){\rnode{#2}{#3}}}
\def\rnode{\@ifnextchar[{\rnode@i}{\def\pst@par{}\rnode@ii}}
\def\rnode@i[#1]{\def\pst@par{ref=#1}\rnode@ii}
\def\rnode@ii#1{\pst@makebox{\rnode@iii\rnode@iv{#1}}}
\def\rnode@iii#1#2{%
\leavevmode
\begingroup
\pst@useboxpar
#1%
\ifPst@nodealign\lower\pst@dimb\fi
\hbox{%
\pst@newnode{#2}{16}{%
\pst@number{\ht\pst@hbox}%
\pst@number{\dp\pst@hbox}%
\pst@number{\wd\pst@hbox}%
\pst@number\pst@dima%
\pst@number\pst@dimb}%
{\tx@InitRnode}%
\box\pst@hbox}%
\endgroup}
\def\rnode@iv{%
\pst@dima=\psk@xref\wd\pst@hbox
\ifx\psk@yref\relax
\pst@dimb=\z@
\else
\pst@dimb=\ht\pst@hbox
\advance\pst@dimb\dp\pst@hbox
\pst@dimb=\psk@yref\pst@dimb
\advance\pst@dimb-\dp\pst@hbox
\fi}
\define@key[psset]{pst-node}{href}[0]{\pst@checknum{#1}\psk@href}
\psset[pst-node]{href=0}
\define@key[psset]{pst-node}{vref}[0.7ex]{\def\psk@vref{#1}}
\psset[pst-node]{vref=0.7ex}
\def\Rnode{\pst@object{Rnode}}
\def\Rnode@i#1{\pst@makebox{\rnode@iii\Rnode@ii{#1}}}
\def\Rnode@ii{%
\use@par
\pst@dima=\psk@href\wd\pst@hbox
\advance\pst@dima\wd\pst@hbox
\divide\pst@dima 2
\pssetlength\pst@dimb{\psk@vref}}
\def\tx@DiaNodePos{DiaNodePos }
\def\dianode{\pst@object{dianode}}
\def\dianode@i#1{\pst@makebox{\dianode@ii{#1}}}
\def\dianode@ii#1{%
\begingroup
\pst@useboxpar
\psdiabox@iii
\setbox\pst@hbox=\hbox{%
\pst@newnode{#1}{14}{}{%
/X \pst@number\pst@dima def
/Y \pst@number\pst@dimb def
/w \pst@number\pst@dimc 2 mul def
/h \pst@number\pst@dimd 2 mul def
/NodePos { \tx@DiaNodePos } def}%
\box\pst@hbox}%
\ifPst@nodealign\psboxseptrue\fi
\ifpsboxsep\psdiabox@sep\fi
\leavevmode
\ifPst@nodealign\lower\pst@dimb\fi
\box\pst@hbox
\endgroup}
\def\tx@TriNodePos{TriNodePos }
\def\tx@InitTriNode{InitTriNode }
\def\trinode{\pst@object{trinode}}
\def\trinode@i#1{\pst@makebox{\trinode@ii{#1}}}
\def\trinode@ii#1{%
  \begingroup%
  \pst@useboxpar%
  \pstribox@iii
  \setbox\pst@hbox=\hbox{%
    \pst@newnode{#1}{14}{}{
      \pst@number\pst@dimc
      \pst@number\pst@dimd
      \ifodd\psk@trimode
        exch
        \pst@number\pst@dima
      \else
        \pst@number\pst@dimb
      \fi
      \psk@trimode
      \pst@number{\wd\pst@hbox}
      \pst@number{\ht\pst@hbox}
      \pst@number{\dp\pst@hbox}
      \tx@InitTriNode
    }%
    \box\pst@hbox%
  }%
  \ifPst@nodealign\psboxseptrue\fi
  \ifpsboxsep\pstribox@sep\fi
  \leavevmode
  \ifPst@nodealign\lower\pst@tempa\fi
  \box\pst@hbox%
  \endgroup}
\def\tx@OvalNodePos{OvalNodePos }
\def\ovalnode{\pst@object{ovalnode}}
\def\ovalnode@i#1{\pst@makebox{\ovalnode@ii{#1}}}
\def\ovalnode@ii#1{%
\begingroup
\pst@useboxpar
\psovalbox@iii
\setbox\pst@hbox=\hbox{%
\pst@newnode{#1}{14}{}{%
/X \pst@number\pst@dima def
/Y \pst@number\pst@dimb def
/w \pst@number\pst@dimc def
/h \pst@number\pst@dimd def
/NodePos { \tx@OvalNodePos } def}%
\unhbox\pst@hbox}%
\ifPst@nodealign\psboxseptrue\fi
\ifpsboxsep\psovalbox@sep\fi
\leavevmode
\ifPst@nodealign\lower\pst@dimb\fi
\box\pst@hbox
\endgroup}
\def\dotnode{\pst@object{dotnode}}
\def\dotnode@i{\@ifnextchar({\dotnode@ii}{\dotnode@ii(\z@,\z@)}}
\def\dotnode@ii(#1)#2{%
  \leavevmode
  \hbox{%
    \use@par
    \pst@@getcoor{#1}%
    \pst@getdotsize
    \pstree@nodehook
    \ifPst@nodealign
      \pst@dima=\pst@dimg
      \kern\pst@dima
      \vrule width\z@ height \pst@dimh depth \pst@dimh
    \fi
    \pst@newnode{#2}{14}{}{
      \pst@coor
      /Y exch def /X exch def
      /w \pst@number\pst@dimg def
      /h \pst@number\pst@dimh def
      /NodePos { \tx@OvalNodePos } def}%
    \psdot@ii(#1)%
    \ifPst@nodealign\kern\pst@dima\fi}%
  \ifPst@markNode\uput[\ifx\psk@rot\@empty0\else\psk@rot\fi]{0}(#2){#2}\fi
  \ignorespaces}
\def\dotnodes{\pst@object{dotnodes}}
\def\dotnodes@i{\use@par\dotnodes@ii}
\def\dotnodes@ii(#1)#2{%
  \dotnode(#1){#2}%
  \@ifnextchar(\dotnodes@ii{\def\pst@par{}}}
\define@key[psset]{pst-node}{framesize}{\pst@expandafter\psset@@framesize{#1} \@nil}
\def\psset@@framesize#1 #2\@nil{%
  \pssetlength\pst@dimg{#1}%
  \divide\pst@dimg2
  \edef\psk@framewidth{\pst@number\pst@dimg}%
  \ifx\@empty#2\@empty
    \let\psk@frameheight\psk@framewidth
  \else
    \pssetlength\pst@dimg{#2}%
    \divide\pst@dimg2
    \edef\psk@frameheight{\pst@number\pst@dimg}%
  \fi}
\psset[pst-node]{framesize=10pt}
\def\fnode{\pst@object{fnode}}
\def\fnode@i{\@ifnextchar({\fnode@ii}{\fnode@ii(\z@,\z@)}}
\def\fnode@ii(#1)#2{%
  \leavevmode
  \pst@killglue
  \hbox{%
    \use@par%
    \begin@ClosedObj%
    \ifPst@nodealign
      \kern\psk@framewidth\p@
      \vrule width\z@ height \psk@frameheight\p@ depth \psk@frameheight\p@
      \edef\pst@coor{0 0 }%
    \else\pst@@getcoor{#1}\fi
    \pst@newnode{#2}{14}{}{
      \pst@coor
      /Y exch def /X exch def
      /d \psk@dimen .5 sub CLW mul neg def
      /r \psk@framewidth d add def
      /l r neg def
      /u \psk@frameheight d add def
      /d u neg def
      /NodePos { \tx@GetRnodePos } def}%
    \addto@pscode{
      /x2 \psk@framewidth CLW \psk@dimen mul sub def
      /y2 \psk@frameheight CLW \psk@dimen mul sub def
      \pst@coor 2 copy
      y2 sub /y1 ED
      x2 sub /x1 exch def
      y2 add /y2 exch def
      x2 add /x2 exch def
      \psk@cornersize
      1 index 0 eq { pop pop \tx@Rect } { \tx@OvalFrame } ifelse}%
    \def\pst@linetype{2}%
    \showpointsfalse%
    \end@ClosedObj%
    \ifPst@nodealign\kern\psk@framewidth\p@\fi}% end of \hbox
  \ignorespaces}
\define@key[psset]{}{XnodesepA}{\pst@getlength{#1}\psk@nodesepA\def\psk@nodeseptypeA{2 }}
\define@key[psset]{}{XnodesepB}{\pst@getlength{#1}\psk@nodesepB\def\psk@nodeseptypeB{2 }}
\define@key[psset]{}{Xnodesep}{%
    \pst@getlength{#1}\psk@nodesepA
    \let\psk@nodesepB\psk@nodesepA
    \def\psk@nodeseptypeA{2 }%
    \def\psk@nodeseptypeB{2 }}
\define@key[psset]{}{YnodesepA}{\pst@getlength{#1}\psk@nodesepA\def\psk@nodeseptypeA{1 }}
\define@key[psset]{}{YnodesepB}{\pst@getlength{#1}\psk@nodesepB\def\psk@nodeseptypeB{1 }}
\define@key[psset]{}{Ynodesep}{%
    \pst@getlength{#1}\psk@nodesepA
    \let\psk@nodesepB\psk@nodesepA
    \def\psk@nodeseptypeA{1 }%
    \def\psk@nodeseptypeB{1 }}
\define@key[psset]{pst-node}{nodesepA}[0pt]{\pst@getlength{#1}\psk@nodesepA \def\psk@nodeseptypeA{0 }}
\define@key[psset]{pst-node}{nodesepB}[0pt]{\pst@getlength{#1}\psk@nodesepB \def\psk@nodeseptypeB{0 }}
\define@key[psset]{pst-node}{nodesep}[0pt]{%
  \pst@getlength{#1}\psk@nodesepA
  \let\psk@nodesepB\psk@nodesepA
  \def\psk@nodeseptypeA{0 }%
  \def\psk@nodeseptypeB{0 }}
\psset[pst-node]{nodesep=0pt}
\define@key[psset]{pst-node}{armA}[10pt]{\pst@getlength{#1}\psk@armA \def\psk@armtypeA{0 }}
\define@key[psset]{pst-node}{armB}[10pt]{\pst@getlength{#1}\psk@armB \def\psk@armtypeB{0 }}
\define@key[psset]{pst-node}{arm}[10pt]{%
  \pst@getlength{#1}\psk@armA
  \let\psk@armB\psk@armA
  \def\psk@armtypeA{0 }%
  \def\psk@armtypeB{0 }}
\psset[pst-node]{arm=10pt}
\define@key[psset]{pst-node}{XarmA}[]{\pst@getlength{#1}\psk@armA \def\psk@armtypeA{1 }}
\define@key[psset]{pst-node}{XarmB}[]{\pst@getlength{#1}\psk@armB \def\psk@armtypeB{1 }}
\define@key[psset]{pst-node}{Xarm}{%
  \pst@getlength{#1}\psk@armA
  \let\psk@armB\psk@armA
  \def\psk@armtypeA{1 }%
  \def\psk@armtypeB{1 }}
\define@key[psset]{pst-node}{YarmA}[]{\pst@getlength{#1}\psk@armA \def\psk@armtypeA{2 }}
\define@key[psset]{pst-node}{YarmB}[]{\pst@getlength{#1}\psk@armB \def\psk@armtypeB{2 }}
\define@key[psset]{pst-node}{Yarm}[]{%
  \pst@getlength{#1}\psk@armA
  \let\psk@armB\psk@armA
  \def\psk@armtypeA{2 }%
  \def\psk@armtypeB{2 }}
\define@key[psset]{pst-node}{offsetA}[0pt]{\pst@getlength{#1}\psk@offsetA}
\define@key[psset]{pst-node}{offsetB}[0pt]{\pst@getlength{#1}\psk@offsetB}
\define@key[psset]{pst-node}{offset}[0pt]{\pst@getlength{#1}\psk@offsetA\let\psk@offsetB\psk@offsetA}
\psset[pst-node]{offset=0pt}
\define@key[psset]{pst-node}{angleA}[0]{\pst@getangle{#1}\psk@angleA}
\define@key[psset]{pst-node}{angleB}[0]{\pst@getangle{#1}\psk@angleB}%
\define@key[psset]{pst-node}{angle}[0]{%
  \pst@getangle{#1}\psk@angleA
  \let\psk@angleB\psk@angleA}
\psset[pst-node]{angle=0}
\define@key[psset]{pst-node}{arcangleA}[8]{\pst@getangle{#1}\psk@arcangleA}
\define@key[psset]{pst-node}{arcangleB}[8]{\pst@getangle{#1}\psk@arcangleB}%
\define@key[psset]{pst-node}{arcangle}[8]{%
  \pst@getangle{#1}\psk@arcangleA
  \let\psk@arcangleB\psk@arcangleA}
\psset[pst-node]{arcangle=8}
\define@key[psset]{pst-node}{ncurvA}[0.67]{\pst@checknum{#1}\psk@ncurvA}
\define@key[psset]{pst-node}{ncurvB}[0.67]{\pst@checknum{#1}\psk@ncurvB}%
\define@key[psset]{pst-node}{ncurv}[0.67]{\pst@checknum{#1}\psk@ncurvA\let\psk@ncurvB\psk@ncurvA}
\psset[pst-node]{ncurv=0.67}
\def\tx@GetCenter{GetCenter }
\def\tx@XYPos{XYPos }
\def\tx@GetEdge{GetEdge }
\def\tx@AddOffset{AddOffset }
\def\tx@GetEdgeA{GetEdgeA }
\def\tx@GetEdgeB{GetEdgeB }
\def\tx@GetArmA{GetArmA }
\def\tx@GetArmB{GetArmB }
\def\check@arrow#1#2{%
  \check@@arrow#2-\@nil
  \if@pst\addto@par{arrows=#2}\def\next{#1}%
  \else\def\next{#1{#2}}\fi
  \next}
\def\check@@arrow#1-#2\@nil{%
\ifx\@nil#2\@nil\@pstfalse\else\@psttrue\fi}
\def\tx@InitNC{InitNC }
\def\nc@object#1#2#3#4#5{%
  \csname begin@#1Obj\endcsname
  \showpointsfalse
  \pst@getnode{#2}\pst@tempa
  \pst@getnode{#3}\pst@tempb
  \gdef\npos@default{#4 }%
  \addto@pscode{%
    /NCLW CLW def
    \pst@nodedict
    \psk@offsetA
    \psk@offsetB neg
    \psk@nodesepA
    \psk@nodesepB
    \psk@nodeseptypeA
    \psk@nodeseptypeB
    \pst@tempa
    \pst@tempb
    \tx@InitNC { #5 } if
    end }%
  \def\use@pscode{%
    \pst@Verb{gsave \tx@STV newpath \pst@code\space grestore}%
    \gdef\pst@code{}}%
  \csname end@#1Obj\endcsname
  \pst@shortput}
\def\npos@default{.5 }
\def\pc@object#1{%
  \@ifnextchar({\pc@@object#1}{\pst@getarrows{\pc@@object#1}}}
\def\pc@@object#1(#2)(#3){%
  \pnode(#2){@@A}\pnode(#3){@@B}%
  #1{@@A}{@@B}}
\def\tx@LPutLine{LPutLine }
\def\tx@LPutLines{LPutLines }
\def\tx@BezierMidpoint{BezierMidpoint }
\def\tx@HPosBegin{HPosBegin }
\def\tx@HPosEnd{HPosEnd }
\def\tx@HPutLine{HPutLine }
\def\tx@HPutLines{HPutLines }
\def\tx@VPosBegin{VPosBegin }
\def\tx@VPosEnd{VPosEnd }
\def\tx@VPutLine{VPutLine }
\def\tx@VPutLines{VPutLines }
\def\tx@HPutCurve{HPutCurve }
\def\tx@NCCoor{NCCoor }
\def\tx@NCLine{NCLine }
\def\ncline{\pst@object{ncline}}
\def\ncline@i{\check@arrow{\ncline@ii}}
\def\ncline@ii#1#2{\nc@object{Open}{#1}{#2}{.5}{\tx@NCLine}}
\def\pcline{\pst@object{pcline}}
\def\pcline@i{\pc@object\ncline@ii}
\def\ncLine{\pst@object{ncLine}}
\def\ncLine@i{\check@arrow{\ncLine@ii}}
\def\ncLine@ii#1#2{\nc@object{Open}{#1}{#2}{.5}%
{\tx@NCLine /LPutPos { xB yB xA yA \tx@LPutLine } def}}
\def\tx@NCLines{NCLines }
\def\nclines{\pst@object{nclines}}
\def\nclines@i{\check@arrow\nclines@ii}
\def\nclines@ii#1#2{%
\begingroup
\use@par
\def\pst@aftercoors{\nclines@iii{#1}{#2}}%
\def\pst@coors{}%
\pst@@getcoors}
\def\nclines@iii#1#2{%
\nc@object{Open}{#1}{#2}{.5}{%
tx@Dict begin \psline@iii pop end
mark \pst@coors \tx@NCLines}%
\endgroup
\ignorespaces}
\def\tx@NCCurve{NCCurve }
\def\nccurve{\pst@object{nccurve}}
\def\nccurve@i{\check@arrow{\nccurve@ii}}
\def\nccurve@ii#1#2{\nc@object{Open}{#1}{#2}{.5}{%
  /AngleA \psk@angleA\space def /AngleB \psk@angleB\space def
  \psk@ncurvB\space \psk@ncurvA\space
  \tx@NCCurve}}
\def\pccurve{\pst@object{pccurve}}
\def\pccurve@i{\pc@object\nccurve@ii}
\def\ncarc{\pst@object{ncarc}}
\def\ncarc@i{\check@arrow{\ncarc@ii}}
\def\ncarc@ii#1#2{\nc@object{Open}{#1}{#2}{.5}{%
  yB yA sub xB xA sub \tx@Atan dup
  \psk@arcangleA\space add /AngleA exch def
  \psk@arcangleB\space sub 180 add /AngleB exch def
  \psk@ncurvB\space \psk@ncurvA\space
  \tx@NCCurve}}
\def\pcarc{\pst@object{pcarc}}
\def\pcarc@i{\pc@object\ncarc@ii}
\def\tx@NCAngles{NCAngles }
\define@boolkey[psset]{pst-node}[Pst@]{pcRef}[true]{}
\psset[pst-node]{pcRef=false}
\def\ncangles{\pst@object{ncangles}}
\def\ncangles@i{\check@arrow{\ncangles@ii}}
\def\ncangles@ii#1#2{%
  \nc@object{Open}{#1}{#2}{1.5}{\ncangles@iii \tx@NCAngles}}
\def\ncangles@iii{
  tx@Dict begin \psline@iii pop end
  /AngleA \psk@angleA def
  /AngleB \psk@angleB def
  /ArmA \psk@armA \ifPst@pcRef 
    GetEdgeA yA yA1 sub dup mul xA xA1 sub dup mul add sqrt sub \fi def
  /ArmB \psk@armB def
  /ArmTypeA \psk@armtypeA def
  /ArmTypeB \psk@armtypeB def }
\def\pcangles{\pst@object{pcangles}}
\def\pcangles@i{\pc@object\ncangles@ii}
\def\tx@NCAngle{NCAngle }
\def\ncangle{\pst@object{ncangle}}
\def\ncangle@i{\check@arrow{\ncangle@ii}}
\def\ncangle@ii#1#2{%
\nc@object{Open}{#1}{#2}{1.5}{\ncangles@iii \tx@NCAngle}}
\def\pcangle{\pst@object{pcangle}}
\def\pcangle@i{\pc@object\ncangle@ii}
\def\tx@NCBar{NCBar }
\def\ncbar{\pst@object{ncbar}}
\def\ncbar@i{\check@arrow{\ncbar@ii}}
\def\ncbar@ii#1#2{\nc@object{Open}{#1}{#2}{1.5}{%
\ncangles@iii /AngleB \psk@angleA def \tx@NCBar}}
\def\pcbar{\pst@object{pcbar}}
\def\pcbar@i{\pc@object\ncbar@ii}
\define@key[psset]{pst-node}{lineAngle}[0]{%
  \ifdim#1pt=\z@\else\psset{armB=0.5}\fi
  \def\psk@lineAngle{#1}}%
\psset[pst-node]{lineAngle=0}%
\def\tx@NCDiag{NCDiag }
\def\ncdiag{\pst@object{ncdiag}}
\def\ncdiag@i{\check@arrow{\ncdiag@ii}}
\def\ncdiag@ii#1#2{%
  \nc@object{Open}{#1}{#2}{1.5}{\ncangles@iii \psk@lineAngle\space \tx@NCDiag}}
\def\pcdiag{\pst@object{pcdiag}}
\def\pcdiag@i{\pc@object\ncdiag@ii}
\def\tx@NCDiagg{NCDiagg }
\def\ncdiagg{\pst@object{ncdiagg}}
\def\ncdiagg@i{\check@arrow{\ncdiagg@ii}}
\def\ncdiagg@ii#1#2{%
  \nc@object{Open}{#1}{#2}{.5}{\ncangles@iii \psk@lineAngle\space \tx@NCDiagg}}
\def\pcdiagg{\pst@object{pcdiagg}}
\def\pcdiagg@i{\pc@object\ncdiagg@ii}
\def\tx@NCLoop{NCLoop }
\define@key[psset]{pst-node}{loopsize}{\pst@getlength{#1}\psk@loopsize}
\psset[pst-node]{loopsize=1cm}
\def\ncloop{\pst@object{ncloop}}
\def\ncloop@i{\check@arrow{\ncloop@ii}}
\def\ncloop@ii#1#2{%
\nc@object{Open}{#1}{#2}{2.5}%
{\ncangles@iii /loopsize \psk@loopsize def \tx@NCLoop}}
\def\pcloop{\pst@object{pcloop}}
\def\pcloop@i{\pc@object\ncloop@ii}
\def\tx@NCCircle{NCCircle }
\def\nccircle{\pst@object{nccircle}}
\def\nccircle@i{\check@arrow{\nccircle@ii}}
\def\nccircle@ii#1#2{%
\pssetlength\pst@dima{#2}%
\nc@object{Open}{#1}{#1}{.5}{%
/AngleA \psk@angleA def
/r \pst@number\pst@dima def
\tx@NCCircle \psarc@v end}}
\def\tx@NCBox{NCBox }
\def\ncbox{\pst@object{ncbox}}
\def\ncbox@i{\check@arrow{\ncbox@ii}}
\def\ncbox@ii#1#2{%
\def\pst@linetype{2}%
\nc@object{Closed}{#1}{#2}{.5}{%
tx@Dict begin \psline@iii pop end
\psk@boxheight \psk@boxdepth
\tx@NCBox}}
\def\pcbox{\pst@object{pcbox}}
\def\pcbox@i{\pc@object\ncbox@ii}
\def\tx@NCArcBox{NCArcBox }
\define@key[psset]{pst-node}{boxheight}[0.4cm]{\pst@getlength{#1}\psk@boxheight}
\define@key[psset]{pst-node}{boxdepth}[0.4cm]{\pst@getlength{#1}\psk@boxdepth}
\define@key[psset]{pst-node}{boxsize}[0.4cm]{%
  \pst@getlength{#1}\psk@boxheight%  
  \let\psk@boxdepth\psk@boxheight}
\psset[pst-node]{boxsize=0.4cm}
\def\ncarcbox{\pst@object{ncarcbox}}
\def\ncarcbox@i{\check@arrow{\ncarcbox@ii}}
\def\ncarcbox@ii#1#2{%
\def\pst@linetype{1}%
\nc@object{Closed}{#1}{#2}{.5}{%
\psk@arcangleA \psk@boxheight \psk@boxdepth \pst@number\pslinearc
\tx@NCArcBox}}
\def\pcarcbox{\pst@object{pcarcbox}}
\def\pcarcbox@i{\pc@object\ncarcbox@ii}
\def\tx@Tfan{Tfan }
\begingroup
\catcode`\:=13
\gdef\pst@activerot{\def:{\string:}}
\endgroup
\define@key[psset]{pst-node}{nrot}[0]{%
  \begingroup
  \pst@activerot
  \pst@expandafter{\@ifnextchar:{\psset@@nrot}{\psset@@rot}}{#1}\@nil
  \global\let\pst@tempg\psk@rot
  \endgroup
  \let\psk@nrot\pst@tempg}
\def\psset@@nrot:#1\@nil{%
  \psset@@rot#1\@nil
  \edef\psk@rot{NAngle \ifx\psk@rot\@empty\else\psk@rot add \fi}}
\psset[pst-node]{nrot=0}
\def\tx@LPutCoor{LPutCoor }
\def\tx@LPut{LPut }
\define@key[psset]{pst-node}{npos}[{}]{%
  \def\pst@tempa{#1}%
  \ifx\pst@tempa\@empty\def\psk@npos{\npos@default}\else\pst@checknum{#1}\psk@npos\fi}
\psset[pst-node]{npos=}
\def\ncput{\pst@object{ncput}}
\def\ncput@i{\pst@killglue\pst@makebox{\ncput@ii}}
\def\ncput@ii{%
  \begingroup%
  \use@par%
  \if@star\pst@starbox\fi%
  \pst@makesmall\pst@hbox%
  \pst@rotate\psk@nrot\pst@hbox%
  \ncput@iii%
  \endgroup%
  \pst@shortput}
\def\ncput@iii{%
  \leavevmode%
  \hbox{%
    \pst@Verb{
      \pst@nodedict
      /t \psk@npos def
      \tx@LPut
      end
      \tx@PutBegin}%
    \box\pst@hbox%
    \pst@Verb{\tx@PutEnd}}}
\def\naput{\pst@object{naput}}
\def\naput@i{\pst@killglue\pst@makebox{\naput@ii{NAngle 90 add}}}
\def\naput@ii#1{%
  \begingroup
  \use@par
  \if@star\pst@starbox\fi
  \def\psk@refangle{#1 }%
  \let\psk@rot\psk@nrot
  \pst@Verb{ 
    gsave  STV CP T /ps@refangle {#1 } def 
    /ps@rot { \psk@rot } def grestore }%ADDED (MJS)
  \uput@vii
  {exch pop add a \tx@PtoC h1 add exch w1 add exch }%
  {tx@Dict /NCLW known { NCLW add } if }%
  \ncput@iii
  \endgroup
  \pst@shortput}
\def\nbput{\pst@object{nbput}}
\def\nbput@i{\pst@killglue\pst@makebox{\naput@ii{NAngle 90 sub}}}
\define@key[psset]{pst-node}{tpos}[0.5]{%
  \pst@checknum{#1}\psk@tpos
  \ifdim\psk@tpos \p@<\z@
    \def\psk@tpos{.5}%
    \@pstrickserr{Bad `tpos' value: `#1'. Must be 0<tpos<1}\@ehpa
  \else
    \ifdim\psk@tpos \p@>\p@
      \def\psk@tpos{.5}%
      \@pstrickserr{Bad `tpos' value: `#1'. Must be 0<tpos<1}\@ehpa%
    \fi%
  \fi}
\psset[pst-node]{tpos=0.5}
\def\nlput{\pst@object{nlput}}
\def\nlput@i(#1)(#2)#3#4{%
  \begin@SpecialObj
  \psLDNode(#1)(#2){#3}{temp@lnput}
  \pcline[linestyle=none](#1)(temp@lnput)%
  \ncput[npos=1]{#4}%
  \end@SpecialObj}
\def\tvput{\pst@object{tvput}}
\def\tvput@i{\pst@makebox{\psput@tput{H}{1}}}
\def\tlput{\pst@object{tlput}}
\def\tlput@i{\pst@makebox{\psput@tput{H}{true}}}
\def\trput{\pst@object{trput}}
\def\trput@i{\pst@makebox{\psput@tput{H}{false}}}
\def\thput{\pst@object{thput}}
\def\thput@i{\pst@makebox{\psput@tput{V}{1}}}
\def\taput{\pst@object{taput}}
\def\taput@i{\pst@makebox{\psput@tput{V}{true}}}
\def\tbput{\pst@object{tbput}}
\def\tbput@i{\pst@makebox{\psput@tput{V}{false}}}
\def\tx@HPutAdjust{HPutAdjust }
\def\tx@VPutAdjust{VPutAdjust }
\def\psput@tput#1#2{%
  \begingroup
  \use@par
  \pst@tputmakesmall
  \leavevmode
  \hbox{%
    \pst@Verb{%
      \pst@nodedict
      /t \psk@tpos \pst@tposflip def
      tx@NodeDict /HPutPos known
        { #1PutPos }
        { CP /Y exch def /X exch def /NAngle 0 def /NCLW 0 def }
      ifelse
      /Sin NAngle sin def
      /Cos NAngle cos def
      /s \pst@number\pslabelsep NCLW add def
      /l \pst@number\pst@dima def
      /r \pst@number\pst@dimb def
      /h \pst@number\pst@dimc def
      /d \pst@number\pst@dimd def
      \ifnum1=0#2 \else
        /flag #2 def
        \csname tx@#1PutAdjust\endcsname
      \fi
      \tx@LPutCoor
      end
      \tx@PutBegin}%
    \box\pst@hbox
    \pst@Verb{\tx@PutEnd}}%
  \endgroup
  \pst@shortput}
\def\pst@tposflip{}
\def\pst@tputmakesmall{%
\pst@dima=\wd\pst@hbox
\divide\pst@dima 2
\pst@dimg=\psk@href\pst@dimg
\pst@dimb\pst@dima
\advance\pst@dima\pst@dimg % leftsize
\advance\pst@dimb-\pst@dimg % rightsize
\pst@dimd=\psk@vref\relax
\pst@dimc=\ht\pst@hbox
\advance\pst@dimc-\pst@dimd % height
\advance\pst@dimd\dp\pst@hbox % depth
\setbox\pst@hbox=\hbox to\z@{%
\kern-\pst@dima\vbox to\z@{\vss\box\pst@hbox\vskip-\pst@dimd}\hss}}
\def\MakeShortNab#1#2{%
  \def\pst@shortput@nab{%
    \def\pst@tempg{\next}%
    \ifx#1\next
      \let\pst@tempg\naput
    \else
      \ifx#2\next
        \let\pst@tempg\nbput
      \else
        \ifx\@sptoken\next
          \let\pst@tempg\pst@shortput
        \fi
      \fi
    \fi
    \pst@tempg}}
\MakeShortNab{^}{_}
\def\MakeShortTablr#1#2#3#4{%
  \def\pst@shortput@tablr{%
    \def\pst@tempg{\next}%
    \ifx#1\next
      \let\pst@tempg\taput
    \else
      \ifx#2\next
        \let\pst@tempg\tbput
      \else
        \ifx#3\next
          \let\pst@tempg\tlput
        \else
          \ifx#4\next
            \let\pst@tempg\trput
          \else
            \ifx\@sptoken\next
              \let\pst@tempg\pst@shortput
            \fi
          \fi
        \fi
      \fi
    \fi
    \pst@tempg}}
\MakeShortTablr{^}{_}{<}{>}
\def\MakeShortTab#1#2{%
  \def\pst@shortput@tab{%
    \def\pst@tempg{\next}%
    \ifx#1\next
      \def\pst@tempg{%
        \@nameuse{%
          t\ifodd\psk@treemode\ifpstreeflip b\else a\fi
          \else\ifpstreeflip r\else l\fi\fi put}}%
    \else
      \ifx#2\next
        \def\pst@tempg{%
          \@nameuse{%
            t\ifodd\psk@treemode\ifpstreeflip a\else b\fi
            \else\ifpstreeflip l\else r\fi\fi put}}%
      \else
        \ifx\@sptoken\next
          \let\pst@tempg\pst@shortput
        \fi
      \fi
    \fi
    \pst@tempg}}
\MakeShortTab{^}{_}
\define@key[psset]{pst-node}{shortput}[none]{%
  \def\pst@tempg{#1}%
  \ifx\pst@tempg\@none
    \let\pst@shortput\ignorespaces
  \else
    \@ifundefined{pst@shortput@#1}%
     {\@pstrickserr{Bad short put: `#1'}\@ehpa}%
     {\edef\pst@shortput{\noexpand\afterassignment\expandafter\noexpand
      \csname pst@shortput@#1\endcsname\noexpand\let\noexpand\next}}%
  \fi}
\psset[pst-node]{shortput=none}
\def\lput{\def\pst@par{}\pst@ifstar{\@ifnextchar[{\lput@i}{\lput@ii}}}
\def\lput@i[#1]{\addto@par{ref=#1}\lput@ii}
\def\lput@ii{\@ifnextchar({\lput@iv}{\lput@iii}}
\def\lput@iii#1{\addto@par{nrot=#1}\@ifnextchar({\lput@iv}{\ncput@i}}
\def\lput@iv(#1){\addto@par{npos=#1}\ncput@i}
\def\mput{\def\pst@par{}\pst@ifstar{\@ifnextchar[{\mput@i}{\ncput@i}}}
\def\mput@i[#1]{\addto@par{ref=#1}\ncput@i}
\def\Lput{\def\pst@par{}\pst@ifstar{\@ifnextchar[{\Lput@ii}{\Lput@i}}}
\def\Lput@i#1{\addto@par{labelsep=#1}\Lput@ii}
\def\Lput@ii[#1]{\addto@par{ref={#1}}\@ifnextchar({\Lput@iv}{\Lput@iii}}
\def\Lput@iii#1{\addto@par{nrot={#1}}\@ifnextchar({\Lput@iv}{\Lput@v}}
\def\Lput@iv(#1){\addto@par{npos=#1}\Lput@v}
\def\Lput@v{\pst@killglue\pst@makebox{\Lput@vi}}
\def\Lput@vi{%
\begingroup
\use@par
\if@star\pst@starbox\fi
\Rput@vi
\pst@makesmall\pst@hbox
\ifx\psk@rot\@empty\else\pst@rotate{ps@rot }\pst@hbox\fi% (MJS)
\ncput@iii
\endgroup
\pst@shortput}
\def\Mput{\def\pst@par{}\pst@ifstar{\@ifnextchar[{\Mput@ii}{\Mput@i}}}
\def\Mput@i#1{\addto@par{labelsep=#1}\Mput@ii}
\def\Mput@ii[#1]{\addto@par{ref={#1}}\Lput@v}
\def\aput@#1{\def\pst@par{}\pst@ifstar{\@ifnextchar[{\aput@i#1}{\aput@ii#1}}}
\def\aput@i#1[#2]{\addto@par{labelsep=#2}\aput@ii#1}
\def\aput@ii#1{\@ifnextchar({\aput@iv#1}{\aput@iii#1}}
\def\aput@iii#1#2{\addto@par{nrot=#2}\@ifnextchar({\aput@iv#1}{#1}}
\def\aput@iv#1(#2){\addto@par{npos=#2}#1}
\def\aput{\aput@\naput@i}
\def\bput{\aput@\nbput@i}
\def\Aput{\def\pst@par{}\pst@ifstar{\@ifnextchar[{\Aput@i}{\naput@i}}}
\def\Aput@i[#1]{\addto@par{labelsep=#1}\naput@i}
\def\Bput{\def\pst@par{}\pst@ifstar{\@ifnextchar[{\Bput@i}{\nbput@i}}}
\def\Bput@i[#1]{\addto@par{labelsep=#1}\nbput@i}
\def\node@coor#1;#2\@nil{%  for normal nodes (name)
  \pst@getnode{#1}\pst@tempg
  \edef\pst@coor{%
    \pst@nodedict
    tx@NodeDict \pst@tempg known
    \pslbrace \pst@tempg load \tx@GetCenter \psrbrace
    \pslbrace 0 0 \psrbrace ifelse
    end }}
\def\Node@coor[#1]#2;#3\@nil{%  for special nodes ([...]{node}node)
\begingroup
\psset{angle=0,#1}%   angle=0, to prevent problems if angle is set globally
\@ifnextchar\bgroup{\Node@@@coor}% we have an additional node [...]{node}
                   {\Node@@coor}#2\@nil% we have  [...]node
\endgroup
\let\pst@coor\pst@tempg}
\def\Node@@coor#1\@nil{%
\pst@getnode{#1}\pst@tempg
\xdef\pst@tempg{%
\pst@nodedict
tx@NodeDict \pst@tempg known
  { \psk@nodesepA \psk@angleA 
    \pst@tempg load \psk@nodeseptypeA \tx@GetEdge
    \psk@offsetA \psk@angleA \tx@AddOffset
    \pst@tempg load \tx@GetCenter
    3 -1 roll add 3 1 roll add exch }
  { CP } ifelse end }}
\def\Node@@@coor#1{%   [...]{#1}node
\pst@@getcoor{#1}%
\def\psk@angleA{%
  \pst@tempg load \tx@GetCenter \pst@coor
  3 -1 roll sub 3 1 roll sub neg \tx@Atan \psk@angleB add
  }%
\Node@@coor}
\def\nput{\pst@object{nput}}
\def\nput@i#1#2{\pst@killglue\pst@makebox{\nput@ii{#1}{#2}}}
\def\nput@ii#1#2{%
  \begingroup
  \use@par
  \if@star\pst@starbox\fi%
  \psset[pstricks]{refangle=#1}%
  \let\psk@angleA\psk@refangle
  \edef\psk@nodesepA{\pst@number\pslabelsep}%
  \def\psk@nodeseptypeA{0 }%
  \pslabelsep\z@
  \uput@vi
  \Node@@coor#2\@nil
  \let\pst@coor\pst@tempg
  \leavevmode
  \psput@special\pst@hbox
  \endgroup
  \ignorespaces}
\newcount\psrow
\newcount\pscol
\newcount\psmatrixcnt
\newskip\psrowsep
\newskip\pscolsep
\define@key[psset]{pst-node}{colsep}[1.5cm]{\pssetlength\pscolsep{#1}}
\define@key[psset]{pst-node}{rowsep}[1.5cm]{\pssetlength\psrowsep{#1}}
\psset[pst-node]{colsep=1.5cm}
\psset[pst-node]{rowsep=1.5cm}
\newif\ifpsmatrix
\ifx\mscount\@undefined\let\mscount\@multicnt\fi
\def\psmatrix{\begingroup{\ifnum0=`}\fi % Don't want to expand any &.
  \@ifnextchar[{\psmatrix@i}{\ifnum0=`{\fi}{}\psmatrix@ii}}
\def\psmatrix@i[#1]{%
  \ifnum0=`{\fi}{}%
  \psset{#1}%
  \psmatrix@ii}
\def\psmatrix@ii{%
  \KillGlue
  \edef\psm@beginmath{%
    \ifmmode$\m@th\ifinner\textstyle\else\displaystyle\fi\fi}%
  \edef\psm@endmath{\ifmmode$\fi}%
  \let\\\psm@cr
  \advance\psmatrixcnt by \@ne
  \def\psm@thenode{M-\the\psmatrixcnt-\the\psrow-\the\pscol}%
  \tabskip\z@
  \psrow=\@ne
  \pscol\z@
  \psset{shortput=tablr}%
  \leavevmode
  \vbox\bgroup\halign\bgroup&%
  \begingroup
  \global\advance\pscol by \@ne
  \csname psrowhook\romannumeral\psrow\endcsname
  \csname pscolhook\romannumeral\pscol\endcsname
  \psm@beginnode##\psm@endnode\endgroup
  \cr}
\def\endpsmatrix{%
  \crcr\egroup\unskip\egroup
  \endgroup}
\def\psm@cr{{\ifnum0=`}\fi\ps@ifnextchar[{\psm@@cr}{\psm@@@cr{}}}
\def\psm@@cr[#1]{\psm@@@cr{\vskip#1\relax}}
\def\psm@@@cr#1{%
  \ifnum0=`{\fi}{}\cr
  \noalign{%
  \global\advance\psrow 1
  \global\pscol\z@
  \vskip\psrowsep
  #1}}
\def\psm@beginnode{%
  \@ifnextchar\psm@endnode
    {\let\psm@endnode@i\relax\setbox\pst@hbox=\hbox{}}%
    {\pst@object{psm@beginnode}}}
\def\psm@beginnode@i{%
  \setbox\pst@hbox=\hbox\bgroup
  \psm@beginmath
  \begingroup
  \ignorespaces}
\def\psm@endnode@i{%
  \unskip
  \endgroup
  \psm@endmath
  \egroup
  \use@par
  \@psttrue}
\def\psm@endnode{%
  \@pstfalse
  \psm@endnode@i
  \ifnum\pscol>1\relax \pshskip\pscolsep \fi
  \psk@mnodesize
  \hfil
  \Pst@nodealigntrue
  \if@pst\csname mnode@\psk@mnode\endcsname
  \else\csname mnode@\psk@emnode\endcsname\fi
  \psk@mcol
  \psk@@mnodesize}
\def\psspan#1{\global\mscount#1\relax\pstloop\ifnum\mscount>\@ne\sp@n\repeat}
\def\pstloop#1\repeat{\gdef\pstiterate{#1\relax\expandafter\pstiterate\fi}%
  \pstiterate
  \let\pstiterate\relax}
\define@key[psset]{pst-node}{name}[\relax]{\pst@getnode{#1}\psk@name}
\let\psk@name\relax
\define@key[psset]{pst-node}{mcol}[c]{%
  \ifx r#1\relax\let\psk@mcol\relax\else
    \ifx l#1\relax\let\psk@mcol\hfill\else
    \let\psk@mcol\hfil\fi\fi}
\psset[pst-node]{mcol=c}
\define@key[psset]{pst-node}{mnodesize}[-1pt]{%
  \pssetlength\pst@dimg{#1}%
  \ifdim\pst@dimg<\z@
    \let\psk@mnodesize\relax
    \let\psk@@mnodesize\relax
  \else
    \edef\psk@mnodesize{\noexpand\hbox to\number\pst@dimg sp\noexpand\bgroup}%
    \let\psk@@mnodesize\egroup
  \fi}
\psset[pst-node]{mnodesize=-1pt}
\def\mnode@R{\rnode@iii\Rnode@ii{\psm@thenode}}
\def\mnode@r{\rnode@iii\rnode@iv{\psm@thenode}}
\def\mnode@oval{\ovalnode@ii{\psm@thenode}}
\def\mnode@tri{\trinode@ii{\psm@thenode}}
\def\mnode@dia{\dianode@ii{\psm@thenode}}
\def\mnode@C{{\Pst@nodealigntrue\cnode@ii(\z@,\z@){\psk@radius}{\psm@thenode}}}
\def\mnode@f{{\Pst@nodealigntrue\fnode@ii(\z@,\z@){\psm@thenode}}}
\def\mnode@circle{\circlenode@ii{\psm@thenode}}
\def\mnode@Circle{\Circlenode@ii{\psm@thenode}}
\def\mnode@p{\pnode(\z@,\z@){\psm@thenode}}
\def\mnode@dot{\dotnode@ii(\z@,\z@){\psm@thenode}}
\def\mnode@none{\box\pst@hbox}
\define@key[psset]{pst-node}{mnode}[R]{%
  \@ifundefined{mnode@#1}%
    {\@pstrickserr{\string\psmatrix\space node `#1' not defined.}\@ehpa}%
    {\edef\psk@mnode{#1}}}
\define@key[psset]{pst-node}{emnode}[none]{%
  \@ifundefined{mnode@#1}%
    {\@pstrickserr{\string\psmatrix\space node `#1' not defined.}\@ehpa}%
    {\edef\psk@emnode{#1}}}
\psset[pst-node]{mnode=R,emnode=none}
\def\nccoil{\pst@object{nccoil}}
\def\nccoil@i{\check@arrow{\nccoil@ii}}
\def\nccoil@ii#1#2{\nc@object{Open}{#1}{#2}{.5}{
  \tx@NCCoor
  tx@Dict begin
  4 2 roll
  \psk@coilwidth \pscoilheight
  \psk@coilarmA \psk@coilarmB
  \psk@coilaspect \psk@coilinc
  \pst@coildict \tx@Coil end
  end}%
}
\def\nczigzag{\pst@object{nczigzag}}
\def\nczigzag@i{\check@arrow{\nczigzag@ii}}
\def\nczigzag@ii#1#2{\nc@object{Open}{#1}{#2}{.5}{
  \tx@NCCoor
  tx@Dict begin
  4 2 roll
  \pscoilheight
  \psk@coilwidth
  \psk@coilarmA
  \psk@coilarmB
  \pst@coildict \tx@ZigZag end
  \psline@iii
  \tx@Line
  end}%
}
\def\psGetNodeCenter#1{ tx@NodeDict begin /N@#1 load GetCenter end % x y on stack in system coor
  \pst@number\psyunit div /#1.y exch def 	% /#1.y in user coor
  \pst@number\psxunit div /#1.x exch def }	% /#1.x in user coor
\def\psGetEdgeA#1#2{
  tx@NodeDict begin \psk@offsetA \psk@offsetB neg 
    \psk@nodesepA \psk@nodesepB 0 0 
    /N@#1 /N@#2 InitNC { NCCoor } if pop pop \tx@UserCoor end}
\def\psGetEdgeB#1#2{
  tx@NodeDict begin \psk@offsetA \psk@offsetB neg 
    \psk@nodesepA \psk@nodesepB 0 0 
    /N@#1 /N@#2 InitNC { NCCoor } if 4 2 roll pop pop \tx@UserCoor end}
\def\ncbarr{\pst@object{ncbarr}}
\def\ncbarr@i#1#2{%
  \begingroup
  \use@par%
  \psLNode(#1)(#2){0.5}{barr@tempNode}%
  \pst@dimc=\psk@angleA pt
  \pst@dimd=180pt
  \ifdim\pst@dimc=\z@\else\ifdim\pst@dimc=\pst@dimd\else\psset{angleA=0}\fi\fi
  \ncbar[arrows=-]{#1}{barr@tempNode}
  \ifdim\psk@angleA pt=\z@\relax
    \ncbar[angleA=180,angleB=180]{barr@tempNode}{#2}
  \else\ncbar[angleA=0,angleB=0]{barr@tempNode}{#2}\fi%
  \endgroup%
}
\def\psLNode(#1)(#2)#3#4{%
  \pst@getcoor{#1}\pst@tempA%
  \pst@getcoor{#2}\pst@tempB%
  \pnode(!
    \pst@tempA /YA exch \pst@number\psyunit div def
    /XA exch \pst@number\psxunit div def
    \pst@tempB /YB exch \pst@number\psyunit div def
    /XB exch \pst@number\psxunit div def
    /dx XB XA sub def
    /dy YB YA sub def
    XA dx #3\space mul add YA dy #3\space mul add){#4}}
\def\psLCNode(#1)#2(#3)#4#5{%
  \pst@getcoor{#1}\pst@tempA%
  \pst@getcoor{#3}\pst@tempB%
  \pnode(!
    \pst@tempA /YA exch \pst@number\psyunit div def
    /XA exch \pst@number\psxunit div def
    \pst@tempB /YB exch \pst@number\psyunit div def
    /XB exch \pst@number\psxunit div def
    XA #2\space mul XB #4\space mul add
    YA #2\space mul YB #4\space mul add){#5}}
\def\psLDNode(#1)(#2)#3#4{%  
  \pst@getcoor{#1}\pst@tempA%
  \pst@getcoor{#2}\pst@tempB%
  \pssetlength\pst@dimb{#3}%
  \pnode(!%
    \pst@tempA /YA exch \pst@number\psyunit div def
    /XA exch \pst@number\psxunit div def
    \pst@tempB /YB exch \pst@number\psyunit div def
    /XB exch \pst@number\psxunit div def
    /dx XB XA sub def
    /dy YB YA sub def
    /angle dy dx Atan def
    /linelength \pst@number\pst@dimb \pst@number\psunit div def
    XA linelength angle cos mul add YA linelength angle sin mul add ){#4}%
}
\def\psRelNode{\pst@object{psRelNode}}
\def\psRelNode@i(#1)(#2)#3#4{{% A - B - factor - node name
  \use@par
  \pst@getcoor{#1}\pst@tempA%
  \pst@getcoor{#2}\pst@tempB%
  \pnode(!
    \pst@tempA /YA exch \pst@number\psyunit div def
    /XA exch \pst@number\psxunit div def
    \pst@tempB /YB exch \pst@number\psyunit div def
    /XB exch \pst@number\psxunit div def
    /AlphaStrich \psk@angleA\space def
    /unit \pst@number\psyunit \pst@number\psxunit div def % yunit/xunit
    /dx XB XA sub  def
    /dy YB YA sub \ifPst@trueAngle\space unit mul \fi\space def
    /laenge dy dup mul dx dup mul add sqrt #3 mul def
    /Alpha dy dx atan def 
    /beta Alpha AlphaStrich add def
    laenge beta cos mul XA add
    laenge beta sin mul \ifPst@trueAngle\space unit div \fi\space YA add ){#4}%
}\ignorespaces}
\define@cmdkeys[psset]{pstricks-add}[PSTPSPNk@]{% Christophe Jorssen 2007
  blName,bcName,brName,
  clName,ccName,crName,
  tlName,tcName,trName}[]{}%
\psset[pstricks-add]{%
  blName=PSPbl,bcName=PSPbc,brName=PSPbr,
  clName=PSPcl,ccName=PSPcc,crName=PSPcr,
  tlName=PSPtl,tcName=PSPtc,trName=PSPtr}
\def\psDefPSPNodes{\def\pst@par{}\pst@object{psDefPSPNodes}}
\def\psDefPSPNodes@i{%
  \pst@killglue
  \begingroup
  \use@par
  \expandafter\psDefPSPNodes@ii\pic@coor}
\def\psDefPSPNodes@ii(#1)(#2)(#3){%
    \pnode(#1){PSPN@temp}\pnode([angle=45]PSPN@temp){\PSTPSPNk@blName}
    \pnode(#3){PSPN@temp}\pnode([angle=-135]PSPN@temp){\PSTPSPNk@trName}
    \pnode(\PSTPSPNk@blName|\PSTPSPNk@trName){\PSTPSPNk@tlName}
    \pnode(\PSTPSPNk@trName|\PSTPSPNk@blName){\PSTPSPNk@brName}
    \ncline[linestyle=none]{\PSTPSPNk@blName}{\PSTPSPNk@tlName}
    \ncput[npos=.5]{\pnode{\PSTPSPNk@clName}}
    \ncline[linestyle=none]{\PSTPSPNk@blName}{\PSTPSPNk@brName}
    \ncput[npos=.5]{\pnode{\PSTPSPNk@bcName}}
    \pnode(\PSTPSPNk@brName|\PSTPSPNk@clName){\PSTPSPNk@crName}
    \pnode(\PSTPSPNk@bcName|\PSTPSPNk@trName){\PSTPSPNk@tcName}
    \pnode(\PSTPSPNk@bcName|\PSTPSPNk@clName){\PSTPSPNk@ccName}
  \endgroup
  \ignorespaces}
\def\psDefBoxNodes#1#2{\rnode[tl]{#1:tl}{\rnode[Bl]{#1:Bl}{\rnode[tr]{#1:tr}{%
\rnode[bl]{#1:bl}{\rnode[Br]{#1:Br}{\rnode[br]{#1:br}{#2}}}}}}%
\pnode(!\psGetNodeCenter{#1:bl}
          \psGetNodeCenter{#1:tl} 
          #1:bl.x #1:tl.x add 2 div #1:bl.y #1:tl.y add 2 div ){#1:Cl}%
\pnode(!\psGetNodeCenter{#1:tr}
          \psGetNodeCenter{#1:br} 
          #1:tr.x #1:br.x add 2 div #1:tr.y #1:br.y add 2 div ){#1:Cr}%
\pnode(!\psGetNodeCenter{#1:Cl}
          \psGetNodeCenter{#1:Cr} 
          #1:Cl.x #1:Cr.x add 2 div #1:Cl.y #1:Cr.y add 2 div ){#1:C}%
\pnode(!\psGetNodeCenter{#1:Br}
          \psGetNodeCenter{#1:Bl} 
          #1:Br.x #1:Bl.x add 2 div #1:Br.y #1:Bl.y add 2 div ){#1:BC}%
\pnode(!\psGetNodeCenter{#1:tr}
          \psGetNodeCenter{#1:tl} 
          #1:tr.x #1:tl.x add 2 div #1:tr.y #1:tl.y add 2 div ){#1:tC}%
\pnode(!\psGetNodeCenter{#1:br}
          \psGetNodeCenter{#1:bl} 
          #1:br.x #1:bl.x add 2 div #1:br.y #1:bl.y add 2 div ){#1:bC}}%
\newcount\pst@args%used in several macros
\newcount\num@pts
\newcount\pst@argcnt% for use in parsing node expressions
\def\PST@root{}
\let\pst@next\relax
\def\my@tempA{}
\def\my@tempB{}
\def\my@tempC{}
\def\my@tempD{}
\def\my@next{}
\newif\if@paren%
\newif\if@equal%
\newif\if@colon%
\newif\ifshow
\def\plussign{+}\def\minussign{-}
\def\defaultvalue#1#2{%#1 is a command, #2 is a value, possibly a command
  \ifdefined#1\ifx#1\@empty\xdef#1{#2}\fi\else\xdef#1{#2}\fi}%
\def\testAlg#1|#2\@nil{%
\ifx\relax#2\relax%
   \let\my@next\psparnode\xdef\my@tempD{}%
\else%
   \let\my@next\algparnode\xdef\my@tempD{A}% algebraic
\fi}%
\def\trim #1{\expandafter\trim@\expandafter{#1 }#1}%
\def\trim@ #1{\trim@@ @#1 @ #1 @ @@}%
\def\trim@@ #1@ #2@ #3@@{\trim@@@\empty #2 @}%
\def\unbrace#1{#1}%
\unbrace{\def\trim@@@ #1 } #2@#3{\expandafter\def%
  \expandafter #3\expandafter {#1}}%
\def\hasparen#1(#2\@nil{%check if expression contains a (--call with \hasparen#1(\@nil
  \ifx\relax#2\relax \@parenfalse \else \@parentrue\fi}%
\def\hasequal#1=#2\@nil{%check if expression contains a =--call with \hasequal#1=\@nil
  \ifx\relax#2\relax \@equalfalse \else \@equaltrue\fi
  \hascolon#2:\@nil}%
\def\hascolon#1:#2\@nil{%check if expression contains a :--call with \hascolon#1:\@nil
\ifx\relax#2\relax \@colonfalse \else \@colontrue\fi}%
\def\equalwhat#1=#2:#3\@nil{{#2}{#3}}%
\def\parsenodexn#1(#2)#3\@nil{%
  \def\coeffA{#1}\edef\nodeA{#2}%
  \trim\coeffA%
  \ifx\nodeA\@empty\else%
    \pnode(#2){@@TMP}%
    \ifx\coeffA\@empty\def\coeffA{1}\else%
      \ifx\coeffA\plussign\def\coeffA{1}\else\ifx\coeffA\minussign\def\coeffA{-1}\fi\fi\fi% 
  \edef\cmd{\noexpand\psLCNode(@TMP\the\pst@argcnt){1}(@@TMP){\coeffA}{@TMP}}%
  \cmd%
  \advance\pst@argcnt by \@ne%
  \pnode(@TMP){@TMP\the\pst@argcnt}%
  \parsenodexn#3\@nil%
  \fi}%
\def\normalvec(#1)#2{%
  \psRelNodeVar(0,0)(#1)(0,1){#2}}%
\def\curvepnode#1#2#3{%
  \edef\my@tempA{#2}% x(t) y(t) expanded
  \expandafter\testAlg\my@tempA|\@nil\my@next {#1}{#2}{#3}}
\def\psparnode#1#2#3{%
  \pnode(!/t #1 def #2){#3}%
  \pnode(!/t #1 .001 sub def #2 
          /t #1 .001 add def 
           #2 3 -1 roll sub 3 1 roll sub neg 
           2 copy Pyth dup 3 1 roll div 3 1 roll div ){#3tang}}%unit tangent vector at t
\def\algparnode#1#2#3{%
  \pstVerb{tx@Dict begin /Func (#2) AlgParser cvx def end }
  \pnode(!/t #1 def Func){#3}
  \pnode(!/t #1 .001 sub def Func 
          /t #1 .001 add def 
          Func 3 -1 roll sub 3 1 roll sub neg 
          2 copy Pyth dup 3 1 roll div 3 1 roll div ){#3tang}%unit tangent vector at t
}%
\def\nodex#1{%
\expandafter\hasparen#1(\@nil%
\if@paren%it's an expression
  \pnode(0,0){@TMP0}%
  \pst@argcnt=0%
  \expandafter\parsenodexn#1()\@nil%
\else%
  \def\my@tempC{#1}%
  \ifx\my@tempC\@empty\pnode(0,0){@TMP}\else\pnode(#1){@TMP}\fi%
\fi}%\if@paren
\def\nodexn#1#2{%
\expandafter\hasparen#1(\@nil%%%      hv 20130917 use \expandafter
\if@paren%it's an expression
  \pnode(0,0){@TMP0}%
  \pst@argcnt=0%
  \parsenodexn#1()\@nil%
  \pnode(@TMP){#2}%
\else%
  \def\my@tempC{#1}%
  \ifx\my@tempC\@empty\pnode(0,0){#2}\else\pnode(#1){#2}\fi%
\fi}%\if@paren
\def\psxline{\pst@object{psxline}}%
\def\psxline@i{\@ifnextchar({\psxline@iii}{\psxline@ii}}%
\def\psxline@ii#1{%
\addto@par{arrows=#1}%
\psxline@iii}%
\def\psxline@iii(#1)#2#3{{%#1=basepoint, #2,#3 node expressions
\pst@killglue%
\use@par%
\nodexn{#2}{@TMP@a}%
\AplusB(#1)(@TMP@a){@TMP@A}%
\nodexn{#3}{@TMP@a}%
\AplusB(#1)(@TMP@a){@TMP@B}%
\psline(@TMP@A)(@TMP@B)%
}%
\ignorespaces}%
\def\curvepnodes{\pst@object{curvepnodes}}
\def\curvepnodes@i#1#2#3#4{{%optional [plotpoints=xx]
  \pst@killglue
  \use@par
  \edef\my@tempA{#3}% x(t) y(t) expanded
  \expandafter\testAlg\my@tempA|\@nil %
  \pstVerb{% 
	tx@Dict begin % so we can use definitions from tx@Dict
	/t0 #1 def
	/t1 #2 def  
	 t1 t0 sub end \psk@plotpoints div /dt exch def }%
  \pst@cntc=\psk@plotpoints\relax%\psk@plotpoints=plotpoints-1
  \advance\pst@cntc by \@ne\relax %=plotpoints
  \ifx\my@tempD\@empty\pstVerb{tx@Dict begin /Func (#3) cvx def end }%add tx@Dict
  \else\pstVerb{tx@Dict begin /Func (#3 ) AlgParser cvx def end }%
  \fi%
    \multido{\i=0+1}{\pst@cntc}{%
      \pnode(! /t #1 dt \i\space mul add def Func ){#4\i}}% remove t before Func
    \expandafter\xdef \csname #4nodecount\endcsname {\psk@plotpoints}%
    \ifnum\Pst@Debug>0 \typeout{Created nodes #40 .. #4\psk@plotpoints}\fi%
}\ignorespaces}%
\def\fnpnode{\pst@object{fnpnode}}
\def\fnpnode@i#1#2#3{{%optional [algebraic]
  \pst@killglue
  \use@par
  \ifPst@algebraic\pnode(*#1 {#2}){#3}\else\pnode(! /x #1 def x #2){#3}\fi
}\ignorespaces}%
\def\fnpnodes{\pst@object{fnpnodes}}
\def\fnpnodes@i#1#2#3#4{{%optional [algebraic]
\pst@killglue
\use@par
\pst@dima=#1pt \pst@dimb=#2pt \advance\pst@dimb -\pst@dima%
\pst@cnta=\psk@plotpoints \relax %=plotpoints-1
\def\PST@root{#4}
\divide\pst@dimb by \pst@cnta%plotpoint-1 intervals
\pst@cntc=\pst@cnta %
\advance\pst@cntc by 1 \relax %=plotpoints
\ifPst@algebraic 
  \multido{\i=0+1}{\pst@cntc}{\pnode(*{\pst@number\pst@dima} {#3}){#4\i}%  hv 20130713
  \advance\pst@dima \pst@dimb}%
\else
    \multido{\i=0+1}{\pst@cntc}{\pnode(!/x \pst@number\pst@dima\space def x #3){#4\i}%
  \advance\pst@dima \pst@dimb}%
\fi%
  \expandafter\xdef \csname \PST@root nodecount\endcsname {\the\pst@cnta}%
  \ifnum\Pst@Debug>0 \typeout{Created nodes #40 .. #4\the\pst@cnta}\fi%
}\ignorespaces}%
\def\AtoB(#1)(#2)#3{\psLCNodeVar(#1)(#2)(-1,1){#3}}
\def\AplusB(#1)(#2)#3{\psLCNodeVar(#1)(#2)(1,1){#3}}
\def\midAB(#1)(#2)#3{\psLCNodeVar(#1)(#2)(.5,.5){#3}}
\def\psnline{\pst@object{psnline}}%line of nodes
\def\psnline@i{\pst@getarrows{\psnline@ii}}
\def\psnline@ii(#1,#2)#3{{%
\pst@killglue%
\use@par%
\pst@cnta=#2 \relax\advance\pst@cnta by 1
\edef\@tmp{}%
\multido{\i=#1+1}{\pst@cnta}{\xdef\@tmp{\@tmp(#3\i)}}%
\expandafter\psline\@tmp}%
\ignorespaces}%
\def\psnpolygon{\pst@object{psnpolygon}}%polygon of nodes
\def\psnpolygon@i{\pst@getarrows{\psnpolygon@ii}}
\def\psnpolygon@ii(#1,#2)#3{{%
\pst@killglue%
\use@par%
\pst@cnta=#2 \relax\advance\pst@cnta by 1
\edef\@tmp{}%
\multido{\i=#1+1}{\pst@cnta}{\xdef\@tmp{\@tmp(#3\i)}}%
\expandafter\pspolygon\@tmp}%
\ignorespaces}%
\def\psncurve{\pst@object{psncurve}}%line of nodes
\def\psncurve@i{\pst@getarrows{\psncurve@ii}}
\def\psncurve@ii(#1,#2)#3{{%
\pst@killglue%
\use@par%
\pst@cnta=#2 \relax\advance\pst@cnta by 1
\edef\@tmp{}%
\multido{\i=#1+1}{\pst@cnta}{\xdef\@tmp{\@tmp(#3\i)}}%
\expandafter\pscurve\@tmp}%
\ignorespaces}%
\def\psnccurve{\pst@object{psnccurve}}%line of nodes
\def\psnccurve@i{\pst@getarrows{\psnccurve@ii}}
\def\psnccurve@ii(#1,#2)#3{{%
\pst@killglue%
\use@par%
\pst@cnta=#2 \relax\advance\pst@cnta by 1
\xdef\@tmp{}%
\multido{\i=#1+1}{\pst@cnta}{\xdef\@tmp{\@tmp(#3\i)}}%
\expandafter\psccurve\@tmp}%
\ignorespaces}%
\def\shownode(#1){%display node user coords in console
  \pst@killglue%
  \pstVerb{% 
    gsave tx@Dict begin %
    tx@NodeDict /N@#1 known { % known node
      /tmpar [(Node #1: ) <28> () (, ) () <29>] def %
      /str 12 string def 
      STV CP T \psGetNodeCenter{#1}\space 
      tmpar 2 #1.x str cvs put 
      /str 12 string def 
      tmpar 4 #1.y str cvs put 
      tmpar concatstringarray = }%
    {% not known
      (Node #1: (NOT KNOWN)) = %
    } ifelse %
    end grestore }%
  \ignorespaces}%
\def\pnodes@ii#1{\getnodelist{#1}{}}
\def\getnodelist#1#2{%
\pst@args=0 \relax%
\def\PST@root{#1}%
\def\pst@next{#2}% command to perform after reading list
\getnext@Node}%
\def\getnext@Node{\@ifnextchar({\getnext@Node@i}%
  {\advance\pst@args by \m@ne \expandafter\xdef \csname \PST@root nodecount\endcsname {\the\pst@args}
  \ifnum\Pst@Debug>0 \typeout{Created nodes \PST@root0 .. \PST@root\the\pst@args}\fi% 
  \pst@next}%
}%
\def\getnext@Node@i(#1){%
\pnode(#1){\PST@root\the\pst@args}%
\advance\pst@args by \@ne\relax%
\getnext@Node}%
\def\psLCNodeVar(#1)(#2)(#3)#4{%
\pst@getcoor{#1}\my@tempA%
\pst@getcoor{#2}\my@tempB%
\pnode(#3){tmpLCn@de}%
\pnode(!%
  \my@tempA /YA exch \pst@number\psyunit div def
  /XA exch \pst@number\psxunit div def
  \my@tempB /YB exch \pst@number\psyunit div def
  /XB exch \pst@number\psxunit div def %stack now empty
  \psGetNodeCenter{tmpLCn@de}\space
  XA tmpLCn@de.x mul XB tmpLCn@de.y mul add
  YA tmpLCn@de.x mul YB tmpLCn@de.y mul add){tmpLCn@deA}%
\pnode(tmpLCn@deA){#4}%
}%
\def\psRelNodeVar{\pst@object{psRelNodeVar}}
\def\psRelNodeVar@i(#1)(#2)(#3)#4{{% A - B - factor;angle - node name
  \use@par
  \pst@getcoor{#1}\my@tempA%
  \pst@getcoor{#2}\my@tempB%
   \pnode(#3){tmpn@de}%
\pnode(!
  /unit \pst@number\psyunit \pst@number\psxunit div def % yunit/xunit
    \my@tempA /YA exch \pst@number\psyunit div def
    /XA exch \pst@number\psxunit div def
    \my@tempB /YB exch \pst@number\psyunit div YA sub 
    \ifPst@trueAngle\space unit mul \fi\space def
    /XB exch \pst@number\psxunit div XA sub def
    \psGetNodeCenter{tmpn@de}
    XB tmpn@de.x mul YB tmpn@de.y mul sub
    YB tmpn@de.x mul XB tmpn@de.y mul add
    \ifPst@trueAngle\space unit div \fi\space 
   YA add exch XA add exch %x, y coords on stack
    ){#4}%
}}
\def\psRelLineVar{\pst@object{psRelLineVar}}
\def\psRelLineVar@i{\@ifnextchar({\psRelLineVar@iii}{\psRelLineVar@ii}}
\def\psRelLineVar@ii#1{%
  \addto@par{arrows=#1}%
  \psRelLineVar@iii}
\def\psRelLineVar@iii(#1)(#2)(#3)#4{{%
  \pst@killglue
  \use@par
  \psRelNodeVar(#1)(#2)(#3){#4}%
  \psline(#1)(#4)%
}\ignorespaces}
\def\rhombus#1(#2)(#3)#4#5{% \rhombus{m}(B)(D){A}{C} 
\AtoB(#2)(#3){node@P}% P=BD
\pnode(! %compute angle and scale in PS
/tmp \psGetNodeCenter{node@P} node@P.x node@P.y 
Pyth 2 div def %tmp=half-length of BD
/ang tmp #1\space div Acos def %ang=angle from BD to BC & BA
#1\space tmp 2 mul div %scale factor s=m/BD
dup ang cos mul exch ang sin mul ){node@A1}% s cos(ang), s sin(ang)
\pnode(! \psGetNodeCenter{node@A1} node@A1.x node@A1.y neg ){node@A2}%reflect in x axis
\psRelNodeVar(#2)(#3)(node@A1){#4}%
\psRelNodeVar(#2)(#3)(node@A2){#5}%
}%
\def\psrline{\pst@object{psrline}}% relative lines
\def\psrline@i{\@ifnextchar({\psrline@iii}{\psrline@ii}}%
\def\psrline@ii#1{%
\addto@par{arrows=#1}%
\psrline@iii}%
\def\psrline@iii{%
\getnodelist{@tmpnode}{\psrline@iv}%
}%
\def\psrline@iv{%
   \ifnum\pst@args<0\else%do nothing
      \pnode(@tmpnode0){@tmpnodeB0}%
      \multido{\iA=1+1,\iB=0+1}{\pst@args}{%
      \AplusB(@tmpnodeB\iB)(@tmpnode\iA){@tmpnodeB\iA}}%
      \psrline@v%
   \fi%
}%
\def\psrline@v{{%finish up
  \pst@killglue%
  \use@par%
  \xdef\tmp{(@tmpnodeB0)}%
  \multido{\i=1+1}{\pst@args}%
{\xdef\tmp{\tmp(@tmpnodeB\i)}}%
\expandafter\psline\tmp%
}\ignorespaces}%
\def\polyIntersections#1#2(#3)(#4){%
\def\nodenameA{#1}\def\nodenameB{#2}%
\pnode(#3){P@A}\pnode(#4){P@B}%
\@ifnextchar({\polyIntersections@next}{\polyIntersections@ii}%
}%
\def\polyIntersections@ii#1#2{%
\def\root@node{#1}\num@pts=#2 \relax%
\polyIntersections@iii}% 
\def\polyIntersections@next{%read as many points as exist
\def\root@node{P@}\getnodelist{P@}{\num@pts=\pst@args \relax\polyIntersections@iii}%
}%
\def\polyIntersections@iii{%nodes are now XXX0....XXXn, n=num@pts
\pst@cnta=\num@pts \relax\advance\pst@cnta by 1 \relax%
\pstVerb{%
 /xarray \the\pst@cnta\space array def
 /yarray \the\pst@cnta\space array def  tx@Dict begin }%
\multido{\i=0+1}{\the\pst@cnta}{\pstVerb{ \psGetNodeCenter{\root@node\i} xarray \i\space \root@node\i.x put yarray \i\space \root@node\i.y put }}%
\pstVerb{ /tposmin 100 def /tnegmax -100 def %/argposmin 0 def /argposmax 0 def 
\psGetNodeCenter{P@B} \psGetNodeCenter{P@A} 
/dx P@B.x P@A.x sub def 
/dy P@B.y P@A.y sub def 
/lenAB dx dy Pyth def
/oldx xarray 0 get def /oldy yarray 0 get def 
1 1 \the\num@pts\space {/k exch def /newx xarray k get def /newy yarray k get def 
/ddx newx oldx sub def /ddy newy oldy sub def 
/det ddy dx mul ddx dy mul sub def
det abs lenAB ddx ddy Pyth mul .001 mul gt 
{/ac oldx P@A.x sub def /bd oldy P@A.y sub def 
 /tt  ac ddy mul bd ddx mul sub det div def %solve for t value at intersection
 /ss ac  dy mul bd dx mul sub det div def % solve for s value at intersection
ss 0 ge 
   {ss 1 le 
        {tt 0 lt {tt tnegmax gt {/tnegmax tt def} if } {tt tposmin lt {/tposmin tt def} if } ifelse }
    if } % ss 1 le
if }%ss 0 ge
 if %det>
 /oldx newx def /oldy newy def} for end }%
\pnode(! \psGetNodeCenter{P@A} \psGetNodeCenter{P@B} P@B.x P@A.x sub  tposmin mul P@A.x add  P@B.y P@A.y sub tposmin  mul P@A.y add ){\nodenameA}%
\pnode(! \psGetNodeCenter{P@A} \psGetNodeCenter{P@B} P@B.x P@A.x sub tnegmax mul P@A.x add P@B.y P@A.y sub tnegmax mul P@A.y add){\nodenameB}%
}%
\def\actualscale#1 #2 scale{% extract x-scale from, eg,  {2. 2. scale}
#1}
\def\psGetCenter#1{ tx@NodeDict begin /N@#1 load GetCenter end }% x y on stack in system coor
\def\ArrowNotch{\pst@object{ArrowNotch}}
\def\ArrowNotch@i#1#2#3#4{{%
\pst@killglue%
\use@par%
\def\inc{-1}%
\ifx#3<\def\inc{1}\fi% -1 means notch to left of arrowhead
\pstVerb{ 
    1 \psk@arrowinset\space sub \psk@arrowlength\space \psk@arrowsize\space  
    \pst@number\pslinewidth \space mul add  mul mul 
    \expandafter\actualscale\psk@arrowscale \space  mul 
    /hh exch def /hh1 hh .05 sub def }% PS variable hh contains dist from tip to notch of arrow, in pts
\def\root@node{#1}\num@pts=\csname\root@node nodecount\endcsname %
\pst@cntb=\num@pts \advance\pst@cntb by \@ne%actual node count
\pst@cnta=\num@pts \advance\pst@cnta by \thr@@%size of PS array
\pst@cntc=#2 \relax% index of center of circle
\ifnum\pst@cntc>\num@pts \pnode(0,0){#4}\else
\pstVerb{%
/PythSq { dup mul exch dup mul add } def
/PtSub {					%  xA yA xB yB
  3 -1 roll 		% xA xB yB yA
  sub neg		% xA xB yA-yB
  3 1 roll 		% yB-yA xA xB
  sub			% yB-yA xA-xB
  exch                     % xB-xA yA-yB
} def
  /xarray \the\pst@cnta\space array def
  /yarray \the\pst@cnta\space array def  
  tx@Dict begin }% end pstVerb
\multido{\i=0+1,\ib=1+1}{\the\pst@cntb}{\pnode(! \psGetCenter{\root@node\i}\space  % center on stack in system coords, not user coords
yarray \ib\space 3 -1 roll put xarray \ib\space 3 -1 roll put 0 0 ){@tmp}}% end multido
\pnode(! xarray 1 get dup yarray 1 get dup 3 1 roll % x1 y1 x1 y1
xarray 2 get yarray 2 get PtSub  % x1 y1 x1-x2 y1-y2
2 copy Pyth hh div 2 div dup % x1 y1 x1-x2 y1-y2 d d ,d->d/2*hh
3 1 roll % x1 y1 x1-x2 d y1-y2 d
div 3 1 roll div %x1 y1  (y1-y2)/d (x1-x2)/d
3 1 roll %x1  (x1-x2)/d y1  (y1-y2)/d
add 3 1 roll add %  y1-(y2-y1)/d x1-(x2-x1)/d
 xarray 0 3 -1 roll put yarray 0 3 -1 roll put %stack empty
 xarray length 2 sub /topnum exch def 
 xarray topnum get dup yarray topnum get dup 3 1 roll %xn yn xn yn
topnum 1 sub /topnum exch def xarray topnum get yarray topnum get % xn yn xn yn x(n-1) y(n-1)
3 -1 roll sub  neg 3 1 roll sub exch % xn yn (xn-x(n-1)) (yn-y(n-1))
2 copy Pyth hh div 2 div dup % xn yn (xn-x(n-1)) (yn-y(n-1)) d d (d->d/(2*h))
3 1 roll div 3 1 roll div %xn yn (xn-x(n-1))/d (yn-y(n-1))/d
3 -1 roll add 3 1 roll % y(n+1) x(n+1)
topnum 2 add /topnum exch def xarray topnum 3 -1 roll put yarray topnum 3 -1 roll put % empty
 /oldcindex \the\pst@cntc\space 1 add def %position in array
 xarray oldcindex get /xc exch def yarray oldcindex get /yc exch def
/inc \inc\space def %+1 for left facing arrow, else -1 
/cindex oldcindex def 
{cindex inc add /cindex exch def xarray cindex get xc sub yarray cindex get yc sub Pyth dup hh1 gt 
{ exit } if } loop % exit from loop with cindex the first index of an external point--dist on stack
 hh1 .1 add lt { xarray cindex get yarray cindex get } %else within segment
{ xarray cindex inc sub get dup yarray cindex inc sub get dup 4 -1 roll exch 
xarray cindex get yarray cindex get PtSub /dy1 exch def /dx1 exch def dx1 dy1 PythSq /Aterm exch def 
 2 copy xc yc PtSub % x(n-inc) y(n-inc) (x(n-inc)-xc) (y(n-inc)-yc)
 2 copy 2 copy 3 -1 roll mul 3 1 roll mul add hh dup mul sub % x(n-inc) y(n-inc) (y(n-inc)-yc) (x(n-inc)-xc) (x(n-inc)-xc)^2+(y(n-inc)-yc)^2-hh^2
 Aterm div /Cterm exch def  % x(n-inc) y(n-inc)  (y(n-inc)-yc) (x(n-inc)-xc) , Cterm=((x(n-inc)-xc)^2+(y(n-inc)-yc)^2-hh^2)/(Aterm) (<0)
 dx1 dy1 %  x(n-inc) y(n-inc) (y(n-inc)-yc)  (x(n-inc)-xc)  dx1 dy1
 4 1 roll mul 3 1 roll mul add Aterm div /Bterm exch def %   x(n-inc) y(n-inc) , Bterm=( (x(n-inc)-xc)*dx1+(y(n-inc)-yc)*dy1)/Aterm
 Bterm abs neg dup dup mul Cterm sub sqrt add dup /tval exch def
 dup dx1 dy1 4 1 roll mul 3 1 roll mul  % x(n-inc) y(n-inc) t*dx1 t*dy1
 PtSub } ifelse % x y screen coords of arrow notch now on stack---convert to user x y
 \pst@number\psyunit div exch \pst@number\psxunit div exch  %use coords now on stack
){#4}\fi%
\pstVerb { end } %tx@Dict
}\ignorespaces}%
\def\saveDataAsNodes#1#2{%  Filename NodePrefix
  \psLoopIndex=0\relax
  \typeout{Open file #1}%
  \openin7=#1
  \loop
    \read7 to \@Data
    \ifeof7\else
      \ifx\@Data\@empty
      \else
        \pnode(!\@Data){#2\the\psLoopIndex}%
        \typeout{#2\the\psLoopIndex -> \@Data}%
	\advance\psLoopIndex by 1
        \let\@oldData\@Data
      \fi
  \repeat
  \closein7
  \advance\psLoopIndex by -1
  \pnode(!\@oldData){#2Last}%  
}
\catcode`\@=\TheAtCode\relax
 
\csname PSTcoilsLoaded\endcsname
\let\PSTcoilsLoaded 
\ifx\PSTricksLoaded   \else\input pstricks.tex\fi
\ifx\PSTnodeLoaded    \else\fi
\ifx\PSTXKeyLoaded    \else\input pst-xkey \fi
\def\fileversion{1.07}
\def\filedate{2015/05/13}
\edef\TheAtCode{\the\catcode`\@}
\catcode`\@=11
\pst@addfams{pst-coil}
\pstheader{pst-coil.pro}
\edef\pst@theheaders{\pst@theheaders,pst-coil.pro}
\def\pst@CoilDict{tx@CoilDict begin }
\def\tx@CoilLoop  {\pst@CoilDict CoilLoop   end }
\def\tx@Coil      {\pst@CoilDict Coil       end }
\def\tx@AltCoil   {\pst@CoilDict AltCoil    end }
\def\tx@ZigZag    {\pst@CoilDict ZigZag     end }
\def\tx@ZigZagCirc{\pst@CoilDict ZigZagCirc end }
\def\tx@Sin       {\pst@CoilDict Sin        end }
\define@key[psset]{pst-coil}{coilwidth}[1cm]{\pst@getlength{#1}\psk@coilwidth}
\define@key[psset]{pst-coil}{coilheight}[1]{\pst@checknum{#1}\pscoilheight}
\define@key[psset]{pst-coil}{coilarmA}[0.5cm]{\pst@getlength{#1}\psk@coilarmA}
\define@key[psset]{pst-coil}{coilarmB}[0.5cm]{\pst@getlength{#1}\psk@coilarmB}
\define@key[psset]{pst-coil}{coilarm}[0.5cm]{%
  \pst@getlength{#1}\psk@coilarmA%
  \let\psk@coilarmB\psk@coilarmA}
\define@key[psset]{pst-coil}{coilaspect}[45]{\pst@getangle{#1}\psk@coilaspect}
\define@key[psset]{pst-coil}{coilinc}[10]{\pst@getangle{#1}\psk@coilinc}
\psset[pst-coil]{coilaspect=45,coilarm=.5cm,coilheight=1,coilwidth=1cm,coilinc=10}
\def\pscoil{\def\pst@par{}\pst@object{pscoil}}
\def\pscoil@i{\pst@getarrows\pscoil@ii}
\def\pscoil@ii(#1){\@ifnextchar({\pscoil@iii{1}(#1)}{\pscoil@iii{\z@}(0,0)(#1)}}
\def\pscoil@iii#1(#2)(#3){%
  \begin@OpenObj
  \pst@getcoor{#2}\pst@tempa
  \pst@getcoor{#3}\pst@tempb
  \pst@optcp{#1}\pst@tempa
  \addto@pscode{%
    \pst@tempa \pst@tempb
    \psk@coilwidth \pscoilheight
    \psk@coilarmA \psk@coilarmB
    \psk@coilaspect \psk@coilinc
    \tx@Coil }%
    \showpointsfalse
  \end@OpenObj}
\def\psCoil{\def\pst@par{}\pst@object{psCoil}}
\def\psCoil@i#1#2{%
  \begin@AltOpenObj
  \showpointsfalse
  \pst@getangle{#1}\pst@tempa
  \pst@getangle{#2}\pst@tempb
  \addto@pscode{%
    \pst@tempa
    \pst@tempb
    \psk@coilwidth
    \pscoilheight
    \psk@coilaspect
    \psk@coilinc
    \tx@AltCoil  
    \@nameuse{psls@\pslinestyle} }%
  \end@OpenObj}
\define@key[psset]{pst-coil}{bow}[0]{%
  \pst@getlength{#1}\psk@bow% 
  \pst@dimm=\psk@bow pt%
  \pst@absdim{\pst@dimm}{\pst@dimn}%
  \ifdim\pst@dimn<1pt \def\psk@bow{0}\fi}%
\psset[pst-coil]{bow=0}
\def\pszigzag{\def\pst@par{}\pst@object{pszigzag}}
\def\pszigzag@i{\pst@getarrows\pszigzag@ii}
\def\pszigzag@ii(#1){\@ifnextchar({\pszigzag@iii{1}(#1)}{\pszigzag@iii{\z@}(0,0)(#1)}}
\def\pszigzag@iii#1(#2)(#3){%
  \addbefore@par{bow=0}%
  \begin@OpenObj%
  \pst@getcoor{#2}\pst@tempA%
  \pst@getcoor{#3}\pst@tempB%
  \pst@optcp{#1}\pst@tempA%
  \addto@pscode{%
    \pst@tempA
    \pst@tempB
    \pscoilheight
    \psk@coilwidth
    \psk@coilarmA
    \psk@coilarmB 
    \ifdim\psk@bow pt=\z@ \tx@ZigZag \else \psk@bow\space \tx@ZigZagCirc \fi
    \psline@iii
    \tx@Line }%
  \end@OpenObj}
\def\nccoil{\pst@object{nccoil}}
\def\nccoil@i{\check@arrow{\nccoil@ii}}
\def\nccoil@ii#1#2{\nc@object{Open}{#1}{#2}{.5}{%
  \tx@NCCoor
  tx@Dict begin
  4 2 roll
  \psk@coilwidth \pscoilheight
  \psk@coilarmA \psk@coilarmB
  \psk@coilaspect \psk@coilinc
  \tx@Coil 
  end }}
\def\pccoil{\def\pst@par{}\pst@object{pccoil}}
\def\pccoil@i{\pc@object\nccoil@ii}
\def\nczigzag{\pst@object{nczigzag}}
\def\nczigzag@i{\check@arrow{\nczigzag@ii}}
\def\nczigzag@ii#1#2{\nc@object{Open}{#1}{#2}{.5}{%
  \tx@NCCoor
  tx@Dict begin
  4 2 roll
  \pscoilheight
  \psk@coilwidth
  \psk@coilarmA
  \psk@coilarmB
  \ifdim\psk@bow pt=\z@\tx@ZigZag\else\psk@bow\space\tx@ZigZagCirc\fi 
  \psline@iii
  \tx@Line
  end }}
\def\pczigzag{\def\pst@par{}\pst@object{pczigzag}}
\def\pczigzag@i{\pc@object\nczigzag@ii}
\def\pst@checkUnit#1#2{\expandafter\pst@checkUnit@i#1!!#2}
\def\pst@checkUnit@i{\@ifnextchar*%
  {\def\pst@roundValue{0 }\pst@checkUnit@ii}%
  {\def\pst@roundValue{-1 }\pst@checkUnit@iii**}}
\def\pst@checkUnit@ii*{\@ifnextchar*%
  {\def\pst@roundValue{1 }\pst@checkUnit@iii*}%
  {\pst@checkUnit@iii**}}
\def\pst@checkUnit@iii**#1!!#2{%
  \edef\ps@next{#1}%
  \ifx\ps@next\@empty\let\pst@num\z@%
  \else\expandafter\pst@@checknum\ps@next..\@nil%
  \fi%
  \ifnum\pst@num=\z@\pst@getlength{#1}{#2}\def\pst@relativePeriod{false }%
  \else%
    \def\pst@relativePeriod{true }%
    \edef#2{\ifnum\pst@num=\tw@-\fi\the\pst@cntg.%
    \expandafter\@gobble\the\pst@cnth\space}%
  \fi}
\define@key[psset]{pst-coil}{periods}[1]{\pst@checkUnit{#1}{\psk@periods}}
\define@key[psset]{pst-coil}{amplitude}[1]{\def\psk@amplitude{#1 }}
\define@key[psset]{pst-coil}{ppoints}[360]{\def\psk@ppoints{#1 }}
\define@key[psset]{pst-coil}{function}[sin]{\def\psk@function{#1 }}
\psset[pst-coil]{periods=1,amplitude=1,ppoints=360,function=sin}
\def\pssin{\pst@object{pssin}}
\def\pssin@i{\pst@getarrows\pssin@ii}
\def\pssin@ii(#1){\@ifnextchar({\pssin@iii{1}(#1)}{\pssin@iii{\z@}(0,0)(#1)}}
\def\pssin@iii#1(#2)(#3){%
  \begin@OpenObj
  \pst@getcoor{#2}\pst@tempa
  \pst@getcoor{#3}\pst@tempb
  \pst@optcp{#1}\pst@tempa
  \addto@pscode{%
    \pst@tempa \pst@tempb
    \psk@periods 
    \pst@relativePeriod 
    \pst@roundValue
    \psk@amplitude \pst@number\psyunit mul
    \psk@coilarmA \psk@coilarmB 
    \psk@ppoints
    { \psk@function }
    \tx@Sin
  }%
  \showpointsfalse%
  \end@OpenObj}
\def\ncsin{\pst@object{ncsin}}
\def\ncsin@i{\check@arrow{\ncsin@ii}}
\def\ncsin@ii#1#2{\nc@object{Open}{#1}{#2}{.5}{%
  \tx@NCCoor
  tx@Dict begin
  4 2 roll
  \psk@periods 
  \pst@relativePeriod 
  \pst@roundValue
  \psk@amplitude \pst@number\psyunit mul
  \psk@coilarmA \psk@coilarmB 
  \psk@ppoints
  { \psk@function }
  \tx@Sin 
  end }}
\def\pcsin{\def\pst@par{}\pst@object{pcsin}}
\def\pcsin@i{\pc@object\ncsin@ii}
\catcode`\@=\TheAtCode\relax

 \psset{unit=1.0\unitlength}
 \SpecialCoor

\newcommand{\schemaPOG}[3]{
 \setlength{\unitlength}{#2}
 \psset{unit=1.0\unitlength}
 \SpecialCoor
 \begin{pspicture}#1
 \put(0,0){$#3$}
 \end{pspicture}
}

\newcommand{\punto}{\pscircle*{0.16}}

\newcommand{\tratteggio}{\psline[linestyle=dashed,linewidth=0.2pt](0,-0.75)(0,10.75)}

\newcommand{\NOtratteggio}[1]{\rput(#1,0){\psline[linestyle=dashed,linecolor=white,linewidth=0.4pt](0,-0.75)(0,10.75)}}

\newcommand{\addspace}[1]{\begin{pspicture}(#1,0)
\end{pspicture}}

\newcommand{\piu}{\begin{pspicture}(0,0)
 \put(0,0){\pscircle[linewidth=0.4pt,fillstyle=none]{0.5}}
 \psline(0.5; 45)(0.5;225)
 \psline(0.5;315)(0.5;135)
\end{pspicture}}

\newcommand{\dst}{\begin{pspicture}(0,0)
 \put(0.293,0){\pscircle*{0.207}}
\end{pspicture}}

\newcommand{\snt}{\begin{pspicture}(0,0)
 \put(-0.293,0){\pscircle*{0.207}}
\end{pspicture}}

\newcommand{\giu}{\begin{pspicture}(0,0)
 \put(0      ,-0.293 ){\pscircle*{0.207}}
\end{pspicture}}

\newcommand{\su}{\begin{pspicture}(0,0)
 \put(0      ,0.293 ){\pscircle*{0.207}}
\end{pspicture}}

\newcommand{\POGcontdst}{\begin{pspicture}(0,10)
   \put(0.5,5){\makebox(0,0)[l]{{\Large$\cdots$}}}
\end{pspicture}}

\newcommand{\POGcontsnt}{\begin{pspicture}(0,10)
   \put(-0.5,5){\makebox(0,0)[r]{{\Large$\cdots$}}}
\end{pspicture}}

\newcommand{\addgiuvar}[3]{\begin{pspicture}(0,10)
   \put(#1,-2){\psline{->}(0,1.5)}
   \put(#1,-2){\punto}
   \put(#1,-2){\put(0.5,0){\makebox(0,0)[l]{$#2$}}}
   \put(#1,0){#3}
\end{pspicture}}

\newcommand{\addsuvar}[3]{\begin{pspicture}(0,10)
   \put(#1,12){\psline{->}(0,-1.5)}
   \put(#1,12){\punto}
   \put(#1,12){\put(0.5,0){\makebox(0,0)[l]{$#2$}}}
   \put(#1,10){#3}
\end{pspicture}}

\newcommand{\escigiuvar}[2]{\begin{pspicture}(0,10)
   \put(#1,0){\psline{->}(0,-2)}
   \put(#1,-2){\punto}
   \put(#1,-2){\put(0.5,0){\makebox(0,0)[l]{$#2$}}}
\end{pspicture}}

\newcommand{\escisuvar}[2]{\begin{pspicture}(0,10)
   \put(#1,10){\psline{->}(0,2)}
   \put(#1,12){\punto}
   \put(#1,12){\put(0.5,0){\makebox(0,0)[l]{$#2$}}}
\end{pspicture}}

\newcommand{\sumdstsnt}[1]{\begin{pspicture}(0,1)
   \put(4,0){\piu}
   \put(4,0){#1}
   \put(0,0){\psline{->}(3.5,0)}
   \put(8,0){\psline{->}(-3.5,0)}
\end{pspicture}}

\newcommand{\sumdstdst}[1]{\begin{pspicture}(0,1)
   \put(4,0){\piu}
   \put(4,0){#1}
   \put(0,0){\psline{->}(3.5,0)}
   \put(4.5,0){\psline{->}(3.5,0)}
\end{pspicture}}

\newcommand{\sumsntsnt}[1]{\begin{pspicture}(0,1)
   \put(4,0){\piu}
   \put(4,0){#1}
   \put(3.5,0){\psline{->}(-3.5,0)}
   \put(8,0){\psline{->}(-3.5,0)}
\end{pspicture}}

\newcommand{\sumsntdst}[1]{\begin{pspicture}(0,1)
   \put(4,0){\piu}
   \put(4,0){#1}
   \put(3.5,0){\psline{->}(-3.5,0)}
   \put(4.5,0){\psline{->}(3.5,0)}
\end{pspicture}}

\newcommand{\sumdstsntvi}[1]{\begin{pspicture}(0,1)
   \put(3,0){\piu}
   \put(3,0){#1}
   \put(0,0){\psline{->}(2.5,0)}
   \put(6,0){\psline{->}(-2.5,0)}
\end{pspicture}}

\newcommand{\sumdstdstvi}[1]{\begin{pspicture}(0,1)
   \put(3,0){\piu}
   \put(3,0){#1}
   \put(0,0){\psline{->}(2.5,0)}
   \put(3.5,0){\psline{->}(2.5,0)}
\end{pspicture}}

\newcommand{\sumsntsntvi}[1]{\begin{pspicture}(0,1)
   \put(3,0){\piu}
   \put(3,0){#1}
   \put(2.5,0){\psline{->}(-2.5,0)}
   \put(6,0){\psline{->}(-2.5,0)}
\end{pspicture}}

\newcommand{\sumsntdstvi}[1]{\begin{pspicture}(0,1)
   \put(3,0){\piu}
   \put(3,0){#1}
   \put(2.5,0){\psline{->}(-2.5,0)}
   \put(3.5,0){\psline{->}(2.5,0)}
\end{pspicture}}

\newcommand{\sumdstsntiv}[1]{\begin{pspicture}(0,1)
   \put(2,0){\piu}
   \put(2,0){#1}
   \put(0,0){\psline{->}(1.5,0)}
   \put(4,0){\psline{->}(-1.5,0)}
\end{pspicture}}

\newcommand{\sumdstdstiv}[1]{\begin{pspicture}(0,1)
   \put(2,0){\piu}
   \put(2,0){#1}
   \put(0,0){\psline{->}(1.5,0)}
   \put(2.5,0){\psline{->}(1.5,0)}
\end{pspicture}}

\newcommand{\sumsntsntiv}[1]{\begin{pspicture}(0,1)
   \put(2,0){\piu}
   \put(2,0){#1}
   \put(1.5,0){\psline{->}(-1.5,0)}
   \put(4,0){\psline{->}(-1.5,0)}
\end{pspicture}}

\newcommand{\sumsntdstiv}[1]{\begin{pspicture}(0,1)
   \put(2,0){\piu}
   \put(2,0){#1}
   \put(1.5,0){\psline{->}(-1.5,0)}
   \put(2.5,0){\psline{->}(1.5,0)}
\end{pspicture}}

\newcommand{\dotdstsnt}{\begin{pspicture}(0,1)
   \put(4,0){\punto}
   \put(0,0){\psline{->}(4,0)}
   \put(8,0){\psline{->}(-4,0)}
\end{pspicture}}

\newcommand{\dotdstdst}{\begin{pspicture}(0,1)
   \put(4,0){\punto}
   \put(0,0){\psline{->}(4,0)}
   \put(4,0){\psline{->}(4,0)}
\end{pspicture}}

\newcommand{\dotsntsnt}{\begin{pspicture}(0,1)
   \put(4,0){\punto}
   \put(4,0){\psline{->}(-4,0)}
   \put(8,0){\psline{->}(-4,0)}
\end{pspicture}}

\newcommand{\dotsntdst}{\begin{pspicture}(0,1)
   \put(4,0){\punto}
   \put(4,0){\psline{->}(-4,0)}
   \put(4,0){\psline{->}(4,0)}
\end{pspicture}}

\newcommand{\dotdstsntvi}{\begin{pspicture}(0,1)
   \put(3,0){\punto}
   \put(0,0){\psline{->}(3,0)}
   \put(6,0){\psline{->}(-3,0)}
\end{pspicture}}

\newcommand{\dotdstdstvi}{\begin{pspicture}(0,1)
   \put(3,0){\punto}
   \put(0,0){\psline{->}(3,0)}
   \put(3,0){\psline{->}(3,0)}
\end{pspicture}}

\newcommand{\dotsntsntvi}{\begin{pspicture}(0,1)
   \put(3,0){\punto}
   \put(3,0){\psline{->}(-3,0)}
   \put(6,0){\psline{->}(-3,0)}
\end{pspicture}}

\newcommand{\dotsntdstvi}{\begin{pspicture}(0,1)
   \put(3,0){\punto}
   \put(3,0){\psline{->}(-3,0)}
   \put(3,0){\psline{->}(3,0)}
\end{pspicture}}

\newcommand{\dotdstsntiv}{\begin{pspicture}(0,1)
   \put(2,0){\punto}
   \put(0,0){\psline{->}(2,0)}
   \put(4,0){\psline{->}(-2,0)}
\end{pspicture}}

\newcommand{\dotdstdstiv}{\begin{pspicture}(0,1)
   \put(2,0){\punto}
   \put(0,0){\psline{->}(2,0)}
   \put(2,0){\psline{->}(2,0)}
\end{pspicture}}

\newcommand{\dotsntsntiv}{\begin{pspicture}(0,1)
   \put(2,0){\punto}
   \put(2,0){\psline{->}(-2,0)}
   \put(4,0){\psline{->}(-2,0)}
\end{pspicture}}

\newcommand{\dotsntdstiv}{\begin{pspicture}(0,1)
   \put(2,0){\punto}
   \put(2,0){\psline{->}(-2,0)}
   \put(2,0){\psline{->}(2,0)}
\end{pspicture}}

\newcommand{\midgiuaddbase}[2]{\begin{pspicture}(0,1)
   \put( 0,0.5){#1}
   \put( 0.25,1.5){\makebox(0,0)[lb]{$#2$}}
\end{pspicture}}

\newcommand{\midsuaddbase}[2]{\begin{pspicture}(0,1)
   \put( 0,-0.5){#1}
   \put( 0.25,8.5){\makebox(0,0)[lt]{$#2$}}
\end{pspicture}}

\newcommand{\midgiuaddbasetop}[2]{\begin{pspicture}(0,1)
   \put( 0,0.5){#1}
   \put( 0.25,0.75){\makebox(0,0)[lb]{$#2$}}
\end{pspicture}}

\newcommand{\midsuaddbasetop}[2]{\begin{pspicture}(0,1)
   \put( 0,-0.5){#1}
   \put( 0.25,9.25){\makebox(0,0)[lt]{$#2$}}
\end{pspicture}}

\newcommand{\midgiu}[2]{\begin{pspicture}(0,1)
   \put(-2,3){\psframe(4,4)\rput(2,2){$#1$}}
   \put(0,3){\psline{->}(0,-3)}
   \put(0,9.5){\psline{->}(0,-2.5)}
   \put(0,-0.5){\makebox(0,0)[t]{$#2$}}
\end{pspicture}}

\newcommand{\midgiuadd}[2]{\midgiuaddbase{\midgiu{#1}{}}{#2}}

\newcommand{\midsu}[2]{\begin{pspicture}(0,1)
   \put(-2,3){\psframe(4,4)\rput(2,2){$#1$}}
   \put(0,7){\psline{->}(0,3)}
   \put(0,0.5){\psline{->}(0,2.5)}
   \put(0,10.5){\makebox(0,0)[b]{$#2$}}
\end{pspicture}}

\newcommand{\midsuadd}[2]{\midsuaddbase{\midsu{#1}{}}{#2}}

\newcommand{\midgiularge}[2]{\begin{pspicture}(0,1)
   \put(-3,3){\psframe(6,4)\rput(3,2){$#1$}}
   \put(0,3){\psline{->}(0,-3)}
   \put(0,9.5){\psline{->}(0,-2.5)}
   \put(0,-0.5){\makebox(0,0)[t]{$#2$}}
\end{pspicture}}

\newcommand{\midgiulargeadd}[2]{\midgiuaddbase{\midgiularge{#1}{}}{#2}}

\newcommand{\midsularge}[2]{\begin{pspicture}(0,1)
   \put(-3,3){\psframe(6,4)\rput(3,2){$#1$}}
   \put(0,7){\psline{->}(0,3)}
   \put(0,0.5){\psline{->}(0,2.5)}
   \put(0,10.5){\makebox(0,0)[b]{$#2$}}
\end{pspicture}}

\newcommand{\midsulargeadd}[2]{\midsuaddbase{\midsularge{#1}{}}{#2}}

\newcommand{\midgiubig}[2]{\begin{pspicture}(0,1)
   \put( 0, 9.5){\psline{->}(0,-1.5)}
   \put(-3,   2){\psframe(6,6)\rput(3,3){$#1$}}
   \put( 0,   2){\psline{->}(0,-2)}
   \put( 0,-0.5){\makebox(0,0)[t]{$#2$}}
\end{pspicture}}

\newcommand{\midgiubigadd}[2]{\midgiuaddbase{\midgiubig{#1}{}}{#2}}

\newcommand{\midsubig}[2]{\begin{pspicture}(0,1)
   \put(0,0.5){\psline{->}(0,1.5)}
   \put(-3,2){\psframe(6,6)\rput(3,3){$#1$}}
   \put(0,8){\psline{->}(0,2)}
   \put(0,10.5){\makebox(0,0)[b]{$#2$}}
\end{pspicture}}

\newcommand{\midsubigadd}[2]{\midsuaddbase{\midsubig{#1}{}}{#2}}

\newcommand{\midgiuvi}[2]{\begin{pspicture}(0,1)
   \put(   0, 9.5){\psline{->}(0,-3)}
   \put(-1.5, 3.5){\psframe(3,3)\rput(1.5,1.5){$#1$}}
   \put(   0, 3.5){\psline{->}(0,-3.5)}
   \put(   0,-0.5){\makebox(0,0)[t]{$#2$}}
\end{pspicture}}

\newcommand{\midgiuviadd}[2]{\midgiuaddbase{\midgiuvi{#1}{}}{#2}}

\newcommand{\midsuvi}[2]{\begin{pspicture}(0,1)
   \put(   0, 0.5){\psline{->}(0,3)}
   \put(-1.5, 3.5){\psframe(3,3)\rput(1.5,1.5){$#1$}}
   \put(   0, 6.5){\psline{->}(0,3.5)}
   \put(   0,10.5){\makebox(0,0)[b]{$#2$}}
\end{pspicture}}

\newcommand{\midsuviadd}[2]{\midsuaddbase{\midsuvi{#1}{}}{#2}}

\newcommand{\midgiumint}[2]{\begin{pspicture}(0,1)
   \put(-2,4.5){\psframe(4,4)\rput(2,2){$#1$}}
   \put(-1.25,1){\psframe(2.5,2.5)\rput(1.25,1.25){{\bf $\frac{1}{s}$}}}
   \put(0,1){\psline{->}(0,-1)}
   \put(0,4.5){\psline{->}(0,-1)}
   \put(0,9.5){\psline{->}(0,-1)}
   \put(0,-0.5){\makebox(0,0)[t]{$#2$}}
\end{pspicture}}

\newcommand{\midgiumintlarge}[2]{\begin{pspicture}(0,1)
   \put(-3,4.5){\psframe(6,4)\rput(3,2){$#1$}}
   \put(-1.25,1){\psframe(2.5,2.5)\rput(1.25,1.25){{\bf $\frac{1}{s}$}}}
   \put(0,1){\psline{->}(0,-1)}
   \put(0,4.5){\psline{->}(0,-1)}
   \put(0,9.5){\psline{->}(0,-1)}
   \put(0,-0.5){\makebox(0,0)[t]{$#2$}}
\end{pspicture}}

\newcommand{\midgiumintadd}[2]{\midgiuaddbasetop{\midgiumint{#1}{}}{#2}}
\newcommand{\midgiumintaddlarge}[2]{\midgiuaddbasetop{\midgiumintlarge{#1}{}}{#2}}

\newcommand{\midsuintm}[2]{\begin{pspicture}(0,1)
   \put(    0,0.5){\psline{->}(0,1)}
   \put(-1.25,1.5){\psframe(2.5,2.5)\rput(1.25,1.25){{\bf $\frac{1}{s}$}}}
   \put(    0,  4){\psline{->}(0,1)}
   \put(   -2,  5){\psframe(4,4)\rput(2,2){$#1$}}
   \put(    0,  9){\psline{->}(0,1)}
   \put(0,10.5){\makebox(0,0)[b]{$#2$}}
\end{pspicture}}

\newcommand{\midsuintmlarge}[2]{\begin{pspicture}(0,1)
   \put(    0,0.5){\psline{->}(0,1)}
   \put(-1.25,1.5){\psframe(2.5,2.5)\rput(1.25,1.25){{\bf $\frac{1}{s}$}}}
   \put(    0,  4){\psline{->}(0,1)}
   \put(   -3,  5){\psframe(6,4)\rput(3,2){$#1$}}
   \put(    0,  9){\psline{->}(0,1)}
   \put(0,10.5){\makebox(0,0)[b]{$#2$}}
\end{pspicture}}

\newcommand{\midsuintmadd}[2]{\midsuaddbasetop{\midsuintm{#1}{}}{#2}}
\newcommand{\midsuintmaddlarge}[2]{\midsuaddbasetop{\midsuintmlarge{#1}{}}{#2}}

\newcommand{\midgiumder}[2]{\begin{pspicture}(0,1)
   \put(-2,4.5){\psframe(4,4)\rput(2,2){$#1$}}
   \put(-1.25,1){\psframe(2.5,2.5)\rput(1.25,1.25){{\bf $s$}}}
   \put(0,1){\psline{->}(0,-1)}
   \put(0,4.5){\psline{->}(0,-1)}
   \put(0,9.5){\psline{->}(0,-1)}
   \put(0,-0.5){\makebox(0,0)[t]{$#2$}}
\end{pspicture}}

\newcommand{\midgiumderlarge}[2]{\begin{pspicture}(0,1)
   \put(-3,4.5){\psframe(6,4)\rput(3,2){$#1$}}
   \put(-1.25,1){\psframe(2.5,2.5)\rput(1.25,1.25){{\bf $s$}}}
   \put(0,1){\psline{->}(0,-1)}
   \put(0,4.5){\psline{->}(0,-1)}
   \put(0,9.5){\psline{->}(0,-1)}
   \put(0,-0.5){\makebox(0,0)[t]{$#2$}}
\end{pspicture}}

\newcommand{\midgiumderadd}[2]{\midgiuaddbasetop{\midgiumder{#1}{}}{#2}}
\newcommand{\midgiumderaddlarge}[2]{\midgiuaddbasetop{\midgiumderlarge{#1}{}}{#2}}

\newcommand{\midsuderm}[2]{\begin{pspicture}(0,1)
   \put(    0,0.5){\psline{->}(0,1)}
   \put(-1.25,1.5){\psframe(2.5,2.5)\rput(1.25,1.25){{\bf $s$}}}
   \put(    0,  4){\psline{->}(0,1)}
   \put(   -2,  5){\psframe(4,4)\rput(2,2){$#1$}}
   \put(    0,  9){\psline{->}(0,1)}
   \put(0,10.5){\makebox(0,0)[b]{$#2$}}
\end{pspicture}}

\newcommand{\midsudermlarge}[2]{\begin{pspicture}(0,1)
   \put(    0,0.5){\psline{->}(0,1)}
   \put(-1.25,1.5){\psframe(2.5,2.5)\rput(1.25,1.25){{\bf $s$}}}
   \put(    0,  4){\psline{->}(0,1)}
   \put(   -3,  5){\psframe(6,4)\rput(3,2){$#1$}}
   \put(    0,  9){\psline{->}(0,1)}
   \put(0,10.5){\makebox(0,0)[b]{$#2$}}
\end{pspicture}}

\newcommand{\midsudermadd}[2]{\midsuaddbasetop{\midsuderm{#1}{}}{#2}}
\newcommand{\midsudermaddlarge}[2]{\midsuaddbasetop{\midsudermlarge{#1}{}}{#2}}

\newcommand{\midgiumxm}[4]{\begin{pspicture}(0,1)
   \put(    0,9.5){\psline{->}(0,-1)}
   \put(-1.25,6.0){\psframe(2.5,2.5)\rput(1.25,1.25){{\bf $#3$}}}
   \put(    0,  6){\psline{->}(0,-1)}
   \put(  0.5,5.5){\makebox(0,0)[l]{$#4$}}
   \put(   -2,  1){\psframe(4,4)\rput(2,2){$#1$}}
   \put(    0,  1){\psline{->}(0,-1)}
   \put(0,-0.5){\makebox(0,0)[t]{$#2$}}
\end{pspicture}}

\newcommand{\midgiumxmlarge}[4]{\begin{pspicture}(0,1)
   \put(    0,9.5){\psline{->}(0,-1)}
   \put(-1.25,6.0){\psframe(2.5,2.5)\rput(1.25,1.25){{\bf $#3$}}}
   \put(    0,  6){\psline{->}(0,-1)}
   \put(  0.5,5.5){\makebox(0,0)[l]{$#4$}}
   \put(   -3,  1){\psframe(6,4)\rput(3,2){$#1$}}
   \put(    0,  1){\psline{->}(0,-1)}
   \put(0,-0.5){\makebox(0,0)[t]{$#2$}}
\end{pspicture}}

\newcommand{\midgiumxmadd}[4]{\midgiuaddbasetop{\midgiumxm{#1}{}{#3}{#4}}{#2}}
\newcommand{\midgiumxmaddlarge}[4]{\midgiuaddbasetop{\midgiumxmlarge{#1}{}{#3}{#4}}{#2}}

\newcommand{\midgiuivxxx}[4]{\begin{pspicture}(0,1)
   \put(    0,9.5){\psline{->}(0,-1)}
   \put(-1.5,6.5){\psframe(3,2)\rput(1.5,1){{\bf $#1$}}}
   \put(    0,6.5){\psline{->}(0,-1)}
   \put(-1.25,4.0){\psframe(2.5,1.5)\rput(1.25,0.75){{\bf $#2$}}}
   \put(    0,4.0){\psline{->}(0,-1)}
   \put(-1.5,1.0){\psframe(3,2)\rput(1.5,1){{\bf $#3$}}}
   \put(    0,1.0){\psline{->}(0,-1)}
   \put(0,-0.5){\makebox(0,0)[t]{$#4$}}
\end{pspicture}}

\newcommand{\midgiuivxxxr}[4]{\begin{pspicture}(0,1)
   \put(    0,9.5){\psline{->}(0,-0.75)}
   \put(-1.5,7.25){\psframe(3,1.5)\rput(1.5,0.75){{\bf $#1$}}}
   \put(    0,7.25){\psline{->}(0,-0.75)}
   \put(0,6){\piu{}}
   \put(    0,5.5){\psline{->}(0,-0.75)}
   \put(-1.25,3.25){\psframe(2.5,1.5)\rput(1.25,0.75){{\bf $#2$}}}
    \put(    0,3.25){\psline{->}(0,-0.75)}
    \put(-1.5,1.0){\psframe(3,1.5)\rput(1.5,0.75){{\bf $#3$}}}
   \put(    0,1.0){\psline{->}(0,-0.888)}
  \put(0,-0.5){\makebox(0,0)[t]{$#4$}}
\end{pspicture}}

\newcommand{\midgiuivxxxadd}[4]{\midgiuaddbasetop{\midgiuivxxx{#1}{#2}{#3}{}}{#4}}

\newcommand{\midgiuxxx}[4]{\begin{pspicture}(0,1)
   \put(    0,9.5){\psline{->}(0,-1)}
   \put(-2.5,6.5){\psframe(5,2)\rput(2.5,1){{\bf $#1$}}}
   \put(    0,6.5){\psline{->}(0,-1)}
   \put(-1.5,4.0){\psframe(3,1.5)\rput(1.5,0.75){{\bf $#2$}}}
   \put(    0,4.0){\psline{->}(0,-1)}
   \put(-2.5,1.0){\psframe(5,2)\rput(2.5,1){{\bf $#3$}}}
   \put(    0,1.0){\psline{->}(0,-1)}
   \put(0,-0.5){\makebox(0,0)[t]{$#4$}}
\end{pspicture}}

\newcommand{\midgiuxxxadd}[4]{\midgiuaddbasetop{\midgiuxxx{#1}{#2}{#3}{}}{#4}}

\newcommand{\midgiuxxxlarge}[4]{\begin{pspicture}(0,1)
   \put(    0,9.5){\psline{->}(0,-1)}
   \put(-3.5,6.5){\psframe(7,2)\rput(3.5,1){{\bf $#1$}}}
   \put(    0,6.5){\psline{->}(0,-1)}
   \put(-1.5,4.0){\psframe(3,1.5)\rput(1.5,0.75){{\bf $#2$}}}
   \put(    0,4.0){\psline{->}(0,-1)}
   \put(-3.5,1.0){\psframe(7,2)\rput(3.5,1){{\bf $#3$}}}
   \put(    0,1.0){\psline{->}(0,-1)}
   \put(0,-0.5){\makebox(0,0)[t]{$#4$}}
\end{pspicture}}

\newcommand{\midgiuxxxlargeadd}[4]{\midgiuaddbasetop{\midgiuxxxlarge{#1}{#2}{#3}{}}{#4}}

\newcommand{\midsuivxxx}[4]{\begin{pspicture}(0,1)
   \put(    0,0.5){\psline{->}(0,1)}
   \put(-1.5,7.0){\psframe(3,2)\rput(1.5,1){{\bf $#3$}}}
   \put(    0,3.5){\psline{->}(0,1)}
   \put(-1.25,4.5){\psframe(2.5,1.5)\rput(1.25,0.75){{\bf $#2$}}}
   \put(    0,6.0){\psline{->}(0,1)}
   \put(-1.5,1.5){\psframe(3,2)\rput(1.5,1){{\bf $#1$}}}
   \put(    0,9.0){\psline{->}(0,1)}
   \put(0,10.5){\makebox(0,0)[b]{$#4$}}
\end{pspicture}}

\newcommand{\midsuivxxxadd}[4]{\midsuaddbasetop{\midsuivxxx{#1}{#2}{#3}{}}{#4}}

\newcommand{\midsuxxx}[4]{\begin{pspicture}(0,1)
   \put(    0,0.5){\psline{->}(0,1)}
   \put(-2.5,7.0){\psframe(5,2)\rput(2.5,1){{\bf $#3$}}}
   \put(    0,3.5){\psline{->}(0,1)}
   \put(-1.5,4.5){\psframe(3,1.5)\rput(1.5,0.75){{\bf $#2$}}}
   \put(    0,6.0){\psline{->}(0,1)}
   \put(-2.5,1.5){\psframe(5,2)\rput(2.5,1){{\bf $#1$}}}
   \put(    0,9.0){\psline{->}(0,1)}
   \put(0,10.5){\makebox(0,0)[b]{$#4$}}
\end{pspicture}}

\newcommand{\midsuxxxadd}[4]{\midsuaddbasetop{\midsuxxx{#1}{#2}{#3}{}}{#4}}

\newcommand{\midsuxxxlarge}[4]{\begin{pspicture}(0,1)
   \put(    0,0.5){\psline{->}(0,1)}
   \put(-3.5,7.0){\psframe(7,2)\rput(3.5,1){{\bf $#3$}}}
   \put(    0,3.5){\psline{->}(0,1)}
   \put(-1.5,4.5){\psframe(3,1.5)\rput(1.5,0.75){{\bf $#2$}}}
   \put(    0,6.0){\psline{->}(0,1)}
   \put(-3.5,1.5){\psframe(7,2)\rput(3.5,1){{\bf $#1$}}}
   \put(    0,9.0){\psline{->}(0,1)}
   \put(0,10.5){\makebox(0,0)[b]{$#4$}}
\end{pspicture}}

\newcommand{\midsuxxxlargeadd}[4]{\midsuaddbasetop{\midsuxxxlarge{#1}{#2}{#3}{}}{#4}}

\newcommand{\midgiuivmxm}[4]{\begin{pspicture}(0,1)
   \put(    0,9.5){\psline{->}(0,-1.5)}
   \put(-1.25,5.5){\psframe(2.5,2.5)\rput(1.25,1.25){{\bf $#3$}}}
   \put(    0,5.5){\psline{->}(0,-1.5)}
   \put(  0.5,4.75){\makebox(0,0)[l]{$#4$}}
   \put(-1.25,1.5){\psframe(2.5,2.5)\rput(1.25,1.25){$#1$}}
   \put(    0,1.5){\psline{->}(0,-1.5)}
   \put(0,-0.5){\makebox(0,0)[t]{$#2$}}
\end{pspicture}}

\newcommand{\midgiuivmxmadd}[4]{\midgiuaddbasetop{\midgiuivmxm{#1}{}{#3}{#4}}{#2}}

\newcommand{\midsumxm}[4]{\begin{pspicture}(0,1)
   \put(    0,0.5){\psline{->}(0,1)}
   \put(-1.25,1.5){\psframe(2.5,2.5)\rput(1.25,1.25){{\bf $#3$}}}
   \put(    0,  4){\psline{->}(0,1)}
   \put(  0.5,4.5){\makebox(0,0)[l]{$#4$}}
   \put(   -2,  5){\psframe(4,4)\rput(2,2){$#1$}}
   \put(    0,  9){\psline{->}(0,1)}
   \put(0,10.5){\makebox(0,0)[b]{$#2$}}
\end{pspicture}}

\newcommand{\midsumxmlarge}[4]{\begin{pspicture}(0,1)
   \put(    0,0.5){\psline{->}(0,1)}
   \put(-1.25,1.5){\psframe(2.5,2.5)\rput(1.25,1.25){{\bf $#3$}}}
   \put(    0,  4){\psline{->}(0,1)}
   \put(  0.5,4.5){\makebox(0,0)[l]{$#4$}}
   \put(   -3,  5){\psframe(6,4)\rput(3,2){$#1$}}
   \put(    0,  9){\psline{->}(0,1)}
   \put(0,10.5){\makebox(0,0)[b]{$#2$}}
\end{pspicture}}

\newcommand{\midsumxmadd}[4]{\midsuaddbasetop{\midsumxm{#1}{}{#3}{#4}}{#2}}
\newcommand{\midsumxmaddlarge}[4]{\midsuaddbasetop{\midsumxmlarge{#1}{}{#3}{#4}}{#2}}

\newcommand{\midsuivmxm}[4]{\begin{pspicture}(0,1)
   \put(    0,0.5){\psline{->}(0,1.5)}
   \put(-1.25,  2){\psframe(2.5,2.5)\rput(1.25,1.25){{\bf $#3$}}}
   \put(    0,4.5){\psline{->}(0,1.5)}
   \put(  0.5,5.25){\makebox(0,0)[l]{$#4$}}
   \put(-1.25,  6){\psframe(2.5,2.5)\rput(1.25,1.25){$#1$}}
   \put(    0,8.5){\psline{->}(0,1.5)}
   \put(0,10.5){\makebox(0,0)[b]{$#2$}}
\end{pspicture}}

\newcommand{\midsuivmxmadd}[4]{\midsuaddbasetop{\midsuivmxm{#1}{}{#3}{#4}}{#2}}

\newcommand{\midgiuintm}[2]{\begin{pspicture}(0,1)
   \put(    0,9.5){\psline{->}(0,-1)}
   \put(-1.25,  6){\psframe(2.5,2.5)\rput(1.25,1.25){{\bf $\frac{1}{s}$}}}
   \put(    0,  6){\psline{->}(0,-1)}
   \put(   -2,  1){\psframe(4,4)\rput(2,2){$#1$}}
   \put(    0,  1){\psline{->}(0,-1)}
   \put(0,-0.5  ){\makebox(0,0)[t]{$#2$}}
\end{pspicture}}

\newcommand{\midgiuintmlarge}[2]{\begin{pspicture}(0,1)
   \put(    0,9.5){\psline{->}(0,-1)}
   \put(-1.25,  6){\psframe(2.5,2.5)\rput(1.25,1.25){{\bf $\frac{1}{s}$}}}
   \put(    0,  6){\psline{->}(0,-1)}
   \put(   -3,  1){\psframe(6,4)\rput(3,2){$#1$}}
   \put(    0,  1){\psline{->}(0,-1)}
   \put(0,-0.5  ){\makebox(0,0)[t]{$#2$}}
\end{pspicture}}

\newcommand{\midgiuintmadd}[2]{\midgiuaddbasetop{\midgiuintm{#1}{}}{#2}}
\newcommand{\midgiuintmaddlarge}[2]{\midgiuaddbasetop{\midgiuintmlarge{#1}{}}{#2}}

\newcommand{\midsumint}[2]{\begin{pspicture}(0,1)
   \put(-2,1.5){\psframe(4,4)\rput(2,2){$#1$}}
   \put(-1.25,6.5){\psframe(2.5,2.5)\rput(1.25,1.25){{\bf $\frac{1}{s}$}}}
   \put(0,0.5){\psline{->}(0,1)}
   \put(0,5.5){\psline{->}(0,1)}
   \put(0,9){\psline{->}(0,1)}
   \put(0,10.5){\makebox(0,0)[b]{$#2$}}
\end{pspicture}}

\newcommand{\midsumintlarge}[2]{\begin{pspicture}(0,1)
   \put(-3,1.5){\psframe(6,4)\rput(3,2){$#1$}}
   \put(-1.25,6.5){\psframe(2.5,2.5)\rput(1.25,1.25){{\bf $\frac{1}{s}$}}}
   \put(0,0.5){\psline{->}(0,1)}
   \put(0,5.5){\psline{->}(0,1)}
   \put(0,9){\psline{->}(0,1)}
   \put(0,10.5){\makebox(0,0)[b]{$#2$}}
\end{pspicture}}

\newcommand{\midsumintadd}[2]{\midsuaddbasetop{\midsumint{#1}{}}{#2}}
\newcommand{\midsumintaddlarge}[2]{\midsuaddbasetop{\midsumintlarge{#1}{}}{#2}}

\newcommand{\midgiuderm}[2]{\begin{pspicture}(0,1)
   \put(    0,9.5){\psline{->}(0,-1)}
   \put(-1.25,  6){\psframe(2.5,2.5)\rput(1.25,1.25){{\bf $s$}}}
   \put(    0,  6){\psline{->}(0,-1)}
   \put(   -2,  1){\psframe(4,4)\rput(2,2){$#1$}}
   \put(    0,  1){\psline{->}(0,-1)}
   \put(0,-0.5  ){\makebox(0,0)[t]{$#2$}}
\end{pspicture}}

\newcommand{\midgiudermlarge}[2]{\begin{pspicture}(0,1)
   \put(    0,9.5){\psline{->}(0,-1)}
   \put(-1.25,  6){\psframe(2.5,2.5)\rput(1.25,1.25){{\bf $s$}}}
   \put(    0,  6){\psline{->}(0,-1)}
   \put(   -3,  1){\psframe(6,4)\rput(3,2){$#1$}}
   \put(    0,  1){\psline{->}(0,-1)}
   \put(0,-0.5  ){\makebox(0,0)[t]{$#2$}}
\end{pspicture}}

\newcommand{\midgiudermadd}[2]{\midgiuaddbasetop{\midgiuderm{#1}{}}{#2}}
\newcommand{\midgiudermaddlarge}[2]{\midgiuaddbasetop{\midgiudermlarge{#1}{}}{#2}}

\newcommand{\midsumder}[2]{\begin{pspicture}(0,1)
   \put(-2,1.5){\psframe(4,4)\rput(2,2){$#1$}}
   \put(-1.25,6.5){\psframe(2.5,2.5)\rput(1.25,1.25){{\bf $s$}}}
   \put(0,0.5){\psline{->}(0,1)}
   \put(0,5.5){\psline{->}(0,1)}
   \put(0,9){\psline{->}(0,1)}
   \put(0,10.5){\makebox(0,0)[b]{$#2$}}
\end{pspicture}}

\newcommand{\midsumderlarge}[2]{\begin{pspicture}(0,1)
   \put(-3,1.5){\psframe(6,4)\rput(3,2){$#1$}}
   \put(-1.25,6.5){\psframe(2.5,2.5)\rput(1.25,1.25){{\bf $s$}}}
   \put(0,0.5){\psline{->}(0,1)}
   \put(0,5.5){\psline{->}(0,1)}
   \put(0,9){\psline{->}(0,1)}
   \put(0,10.5){\makebox(0,0)[b]{$#2$}}
\end{pspicture}}

\newcommand{\midsumderadd}[2]{\midsuaddbasetop{\midsumder{#1}{}}{#2}}
\newcommand{\midsumderaddlarge}[2]{\midsuaddbasetop{\midsumderlarge{#1}{}}{#2}}

\newcommand{\midgiuiv}[2]{\begin{pspicture}(0,1)
   \put(-1.5,3.5){\psframe(3,3)\rput(1.5,1.5){$#1$}}
   \put(0,9.5){\psline{->}(0,-3)}
   \put(0,3.5){\psline{->}(0,-3.5)}
   \put(0,-0.5){\makebox(0,0)[t]{$#2$}}
\end{pspicture}}

\newcommand{\midgiuivadd}[2]{\midgiuaddbase{\midgiuiv{#1}{}}{#2}}

\newcommand{\midsuiv}[2]{\begin{pspicture}(0,1)
   \put(-1.5,3.5){\psframe(3,3)\rput(1.5,1.5){$#1$}}
   \put(0,6.5){\psline{->}(0,3.5)}
   \put(0,0.5){\psline{->}(0,3)}
   \put(0,10.5){\makebox(0,0)[b]{$#2$}}
\end{pspicture}}

\newcommand{\midsuivadd}[2]{\midsuaddbase{\midsuiv{#1}{}}{#2}}

\newcommand{\midgiuintmtre}[2]{\begin{pspicture}(0,1)
   \put(    0,9.5){\psline{->}(0,-1.5)}
   \put(-1.25,5.5){\psframe(2.5,2.5)\rput(1.25,1.25){{\bf $\frac{1}{s}$}}}
   \put(    0,5.5){\psline{->}(0,-1.5)}
   \put(-1.25,1.5){\psframe(2.5,2.5)\rput(1.25,1.25){$#1$}}
   \put(    0,1.5){\psline{->}(0,-1.5)}
   \put(0,-0.5  ){\makebox(0,0)[t]{$#2$}}
\end{pspicture}}

\newcommand{\midgiuintmtreadd}[2]{\midgiuaddbasetop{\midgiuintmtre{#1}{}}{#2}}

\newcommand{\midsuintmtre}[2]{\begin{pspicture}(0,1)
   \put(    0,0.5){\psline{->}(0,1.5)}
   \put(-1.25,2){\psframe(2.5,2.5)\rput(1.25,1.25){{\bf $\frac{1}{s}$}}}
   \put(    0,4.5){\psline{->}(0,1.5)}
   \put(-1.25,  6){\psframe(2.5,2.5)\rput(1.25,1.25){$#1$}}
   \put(    0,8.5){\psline{->}(0,1.5)}
   \put(0,10.5){\makebox(0,0)[b]{$#2$}}
\end{pspicture}}

\newcommand{\midsuintmtreadd}[2]{\midsuaddbasetop{\midsuintmtre{#1}{}}{#2}}

\newcommand{\midgiuminttre}[2]{\begin{pspicture}(0,1)
   \put(    0,9.5){\psline{->}(0,-1.5)}
   \put(-1.25,5.5){\psframe(2.5,2.5)\rput(1.25,1.25){$#1$}}
   \put(    0,5.5){\psline{->}(0,-1.5)}
   \put(-1.25,1.5){\psframe(2.5,2.5)\rput(1.25,1.25){{\bf $\frac{1}{s}$}}}
   \put(    0,1.5){\psline{->}(0,-1.5)}
   \put(0,-0.5  ){\makebox(0,0)[t]{$#2$}}
\end{pspicture}}

\newcommand{\midgiuminttreadd}[2]{\midgiuaddbasetop{\midgiuminttre{#1}{}}{#2}}

\newcommand{\midsuminttre}[2]{\begin{pspicture}(0,1)
   \put(0,0.5){\psline{->}(0,1.5)}
   \put(-1.25,   2){\psframe(2.5,2.5)\rput(1.25,1.25){{$#1$}}}
   \put(    0, 4.5){\psline{->}(0,1.5)}
   \put(-1.25,   6){\psframe(2.5,2.5)\rput(1.25,1.25){{\bf $\frac{1}{s}$}}}
   \put(    0, 8.5){\psline{->}(0,1.5)}
   \put(    0,10.5){\makebox(0,0)[b]{$#2$}}
\end{pspicture}}

\newcommand{\midsuminttreadd}[2]{\midsuaddbasetop{\midsuminttre{#1}{}}{#2}}

\newcommand{\midgiudermtre}[2]{\begin{pspicture}(0,1)
   \put(    0,9.5){\psline{->}(0,-1.5)}
   \put(-1.25,5.5){\psframe(2.5,2.5)\rput(1.25,1.25){{\bf $s$}}}
   \put(    0,5.5){\psline{->}(0,-1.5)}
   \put(-1.25,1.5){\psframe(2.5,2.5)\rput(1.25,1.25){$#1$}}
   \put(    0,1.5){\psline{->}(0,-1.5)}
   \put(0,-0.5  ){\makebox(0,0)[t]{$#2$}}
\end{pspicture}}

\newcommand{\midgiudermtreadd}[2]{\midgiuaddbasetop{\midgiudermtre{#1}{}}{#2}}

\newcommand{\midsudermtre}[2]{\begin{pspicture}(0,1)
   \put(    0,0.5){\psline{->}(0,1.5)}
   \put(-1.25,2){\psframe(2.5,2.5)\rput(1.25,1.25){{\bf $s$}}}
   \put(    0,4.5){\psline{->}(0,1.5)}
   \put(-1.25,  6){\psframe(2.5,2.5)\rput(1.25,1.25){$#1$}}
   \put(    0,8.5){\psline{->}(0,1.5)}
   \put(0,10.5){\makebox(0,0)[b]{$#2$}}
\end{pspicture}}

\newcommand{\midsudermtreadd}[2]{\midsuaddbasetop{\midsudermtre{#1}{}}{#2}}

\newcommand{\midgiumdertre}[2]{\begin{pspicture}(0,1)
   \put(    0,9.5){\psline{->}(0,-1.5)}
   \put(-1.25,5.5){\psframe(2.5,2.5)\rput(1.25,1.25){$#1$}}
   \put(    0,5.5){\psline{->}(0,-1.5)}
   \put(-1.25,1.5){\psframe(2.5,2.5)\rput(1.25,1.25){{\bf $s$}}}
   \put(    0,1.5){\psline{->}(0,-1.5)}
   \put(0,-0.5  ){\makebox(0,0)[t]{$#2$}}
\end{pspicture}}

\newcommand{\midgiumdertreadd}[2]{\midgiuaddbasetop{\midgiumdertre{#1}{}}{#2}}

\newcommand{\midsumdertre}[2]{\begin{pspicture}(0,1)
   \put(0,0.5){\psline{->}(0,1.5)}
   \put(-1.25,   2){\psframe(2.5,2.5)\rput(1.25,1.25){{$#1$}}}
   \put(    0, 4.5){\psline{->}(0,1.5)}
   \put(-1.25,   6){\psframe(2.5,2.5)\rput(1.25,1.25){{\bf $s$}}}
   \put(    0, 8.5){\psline{->}(0,1.5)}
   \put(    0,10.5){\makebox(0,0)[b]{$#2$}}
\end{pspicture}}

\newcommand{\midsumdertreadd}[2]{\midsuaddbasetop{\midsumdertre{#1}{}}{#2}}

\newcommand{\midsupiu}[4]{\begin{pspicture}(0,1)
   \put(-2,2.5){\psframe(4,4)\rput(2,2){$#1$}}
   \put(0,6.5){\psline{->}(0,1.5)}
   \put(0,8.5){\piu}
   \put(0,9){\psline{->}(0,1)}
   \put(1.5,8.5){\psline{->}(-1,0)}
   \put(1.5,8.5){\punto}
   \put(0,0.5){\psline{->}(0,2)}
   \put(0,10.5){\makebox(0,0)[b]{$#2$}}
   \put(0,8.5){\makebox(0,0){$#4$}}
   \put(1.75,8.5){\makebox(0,0)[l]{$#3$}}
\end{pspicture}}

\newcommand{\midsupiutre}[4]{\begin{pspicture}(0,1)
   \put(-1.5,2.5){\psframe(3,3)\rput(1.5,1.5){$#1$}}
   \put(0,5.5){\psline{->}(0,2)}
   \put(0,8){\piu}
   \put(0,8.5){\psline{->}(0,1.5)}
   \put(1.5,8){\psline{->}(-1,0)}
   \put(1.5,8){\punto}
   \put(0,0.5){\psline{->}(0,2)}
   \put(0,10.5){\makebox(0,0)[b]{$#2$}}
   \put(0,8){\makebox(0,0){$#4$}}
   \put(1.75,7.75){\makebox(0,0)[tr]{$#3$}}
\end{pspicture}}

\newcommand{\midgiuintmvuoto}[2]{\begin{pspicture}(0,1)
   \put(    0,9.5){\psline{->}(0,-1)}
   \put(    0,8  ){\makebox(0,0)[t]{$#1$}}
   \put(    0,1.5){\makebox(0,0)[b]{$#2$}}
   \put(    0,  1){\punto}
   \put(    0,  1){\psline{->}(0,-1)}
\end{pspicture}}

\newcommand{\midsuintmvuoto}[2]{\begin{pspicture}(0,1)
   \put(    0,0.5){\psline{->}(0,1)}
   \put(    0,  2 ){\makebox(0,0)[b]{$#1$}}
   \put(    0,8.5 ){\makebox(0,0)[t]{$#2$}}
   \put(    0,  9){\punto}
   \put(    0,  9){\psline{->}(0,1)}
\end{pspicture}}

\newcommand{\blosubase}[2]{\begin{pspicture}(8,10)
   \put(0, 0){\tratteggio}
   \put(0,10){\dotsntdst}
   \put(4, 0){#1}
   \put(0, 0){\sumdstsnt{#2}}
\end{pspicture}}

\newcommand{\blovisubase}[2]{\begin{pspicture}(6,10)
   \put(0, 0){\tratteggio}
   \put(0,10){\dotsntdstvi}
   \put(3, 0){#1}
   \put(0, 0){\sumdstsntvi{#2}}
\end{pspicture}}

\newcommand{\bloivsubase}[2]{\begin{pspicture}(4,10)
   \put(0, 0){\tratteggio}
   \put(0,10){\dotsntdstiv}
   \put(2, 0){#1}
   \put(0, 0){\sumdstsntiv{#2}}
\end{pspicture}}

\newcommand{\blosuorbase}[2]{\begin{pspicture}(8,10)
   \put(0, 0){\tratteggio}
   \put(0,10){\sumdstdst{#2}}
   \put(4, 0){#1}
   \put(0, 0){\dotsntsnt}
\end{pspicture}}

\newcommand{\blovisuorbase}[2]{\begin{pspicture}(6,10)
   \put(0, 0){\tratteggio}
   \put(0,10){\sumdstdstvi{#2}}
   \put(3, 0){#1}
   \put(0, 0){\dotsntsntvi}
\end{pspicture}}

\newcommand{\bloivsuorbase}[2]{\begin{pspicture}(4,10)
   \put(0, 0){\tratteggio}
   \put(0,10){\sumdstdstiv{#2}}
   \put(2, 0){#1}
   \put(0, 0){\dotsntsntiv}
\end{pspicture}}

\newcommand{\blosuaorbase}[2]{\begin{pspicture}(8,10)
   \put(0, 0){\tratteggio}
   \put(0,10){\sumsntsnt{#2}}
   \put(4, 0){#1}
   \put(0, 0){\dotdstdst}
\end{pspicture}}

\newcommand{\blovisuaorbase}[2]{\begin{pspicture}(6,10)
   \put(0, 0){\tratteggio}
   \put(0,10){\sumsntsntvi{#2}}
   \put(3, 0){#1}
   \put(0, 0){\dotdstdstvi}
\end{pspicture}}

\newcommand{\bloivsuaorbase}[2]{\begin{pspicture}(4,10)
   \put(0, 0){\tratteggio}
   \put(0,10){\sumsntsntiv{#2}}
   \put(2, 0){#1}
   \put(0, 0){\dotdstdstiv}
\end{pspicture}}

\newcommand{\blosudstbase}[2]{\begin{pspicture}(8,10)
   \put(0, 0){\tratteggio}
   \put(0,10){\sumdstdst{#2}}
   \put(4, 0){#1}
   \put(0, 0){\dotdstdst}
\end{pspicture}}

\newcommand{\bloivsudstbase}[2]{\begin{pspicture}(4,10)
   \put(0, 0){\tratteggio}
   \put(0,10){\sumdstdstiv{#2}}
   \put(2, 0){#1}
   \put(0, 0){\dotdstdstiv}
\end{pspicture}}

\newcommand{\blovisudstbase}[2]{\begin{pspicture}(6,10)
   \put(0, 0){\tratteggio}
   \put(0,10){\sumdstdstvi{#2}}
   \put(3, 0){#1}
   \put(0, 0){\dotdstdstvi}
\end{pspicture}}

\newcommand{\blosusntbase}[2]{\begin{pspicture}(8,10)
   \put(0, 0){\tratteggio}
   \put(0,10){\sumsntsnt{#2}}
   \put(4, 0){#1}
   \put(0, 0){\dotsntsnt}
\end{pspicture}}

\newcommand{\bloivsusntbase}[2]{\begin{pspicture}(4,10)
   \put(0, 0){\tratteggio}
   \put(0,10){\sumsntsntiv{#2}}
   \put(2, 0){#1}
   \put(0, 0){\dotsntsntiv}
\end{pspicture}}

\newcommand{\blovisusntbase}[2]{\begin{pspicture}(6,10)
   \put(0, 0){\tratteggio}
   \put(0,10){\sumsntsntvi{#2}}
   \put(3, 0){#1}
   \put(0, 0){\dotsntsntvi}
\end{pspicture}}

\newcommand{\blogiubase}[2]{\begin{pspicture}(8,10)
   \put(0, 0){\tratteggio}
   \put(0,10){\sumdstsnt{#2}}
   \put(4, 0){#1}
   \put(0, 0){\dotsntdst}
\end{pspicture}}

\newcommand{\blovigiubase}[2]{\begin{pspicture}(6,10)
   \put(0, 0){\tratteggio}
   \put(0,10){\sumdstsntvi{#2}}
   \put(3, 0){#1}
   \put(0, 0){\dotsntdstvi}
\end{pspicture}}

\newcommand{\bloivgiubase}[2]{\begin{pspicture}(4,10)
   \put(0, 0){\tratteggio}
   \put(0,10){\sumdstsntiv{#2}}
   \put(2, 0){#1}
   \put(0, 0){\dotsntdstiv}
\end{pspicture}}

\newcommand{\bloivgiubaser}[2]{\begin{pspicture}(4,10)
   \put(0, 0){\tratteggio}
   \put(0,10){\sumdstsntiv{#2}}
   \put(2, 0){#1}
   \psline{->}(4,0)(0,0)
   \put(2, 0){\punto}
\end{pspicture}
}

\newcommand{\blogiuorbase}[2]{\begin{pspicture}(8,10)
   \put(0, 0){\tratteggio}
   \put(0,10){\dotdstdst}
   \put(4, 0){#1}
   \put(0, 0){\sumsntsnt{#2}}
\end{pspicture}}

\newcommand{\blovigiuorbase}[2]{\begin{pspicture}(6,10)
   \put(0, 0){\tratteggio}
   \put(0,10){\dotdstdstvi}
   \put(3, 0){#1}
   \put(0, 0){\sumsntsntvi{#2}}
\end{pspicture}}

\newcommand{\bloivgiuorbase}[2]{\begin{pspicture}(4,10)
   \put(0, 0){\tratteggio}
   \put(0,10){\dotdstdstiv}
   \put(2, 0){#1}
   \put(0, 0){\sumsntsntiv{#2}}
\end{pspicture}}

\newcommand{\blogiuaorbase}[2]{\begin{pspicture}(8,10)
   \put(0, 0){\tratteggio}
   \put(0,10){\dotsntsnt}
   \put(4, 0){#1}
   \put(0, 0){\sumdstdst{#2}}
\end{pspicture}}

\newcommand{\blovigiuaorbase}[2]{\begin{pspicture}(6,10)
   \put(0, 0){\tratteggio}
   \put(0,10){\dotsntsntvi}
   \put(3, 0){#1}
   \put(0, 0){\sumdstdstvi{#2}}
\end{pspicture}}

\newcommand{\bloivgiuaorbase}[2]{\begin{pspicture}(4,10)
   \put(0, 0){\tratteggio}
   \put(0,10){\dotsntsntiv}
   \put(2, 0){#1}
   \put(0, 0){\sumdstdstiv{#2}}
\end{pspicture}}

\newcommand{\blogiudstbase}[2]{\begin{pspicture}(8,10)
   \put(0, 0){\tratteggio}
   \put(0,10){\dotdstdst}
   \put(4, 0){#1}
   \put(0, 0){\sumdstdst{#2}}
\end{pspicture}}

\newcommand{\bloivgiudstbase}[2]{\begin{pspicture}(4,10)
   \put(0, 0){\tratteggio}
   \put(0,10){\dotdstdstiv}
   \put(2, 0){#1}
   \put(0, 0){\sumdstdstiv{#2}}
\end{pspicture}}

\newcommand{\blovigiudstbase}[2]{\begin{pspicture}(6,10)
   \put(0, 0){\tratteggio}
   \put(0,10){\dotdstdstvi}
   \put(3, 0){#1}
   \put(0, 0){\sumdstdstvi{#2}}
\end{pspicture}}

\newcommand{\blogiusntbase}[2]{\begin{pspicture}(8,10)
   \put(0, 0){\tratteggio}
   \put(0,10){\dotsntsnt}
   \put(4, 0){#1}
   \put(0, 0){\sumsntsnt{#2}}
\end{pspicture}}

\newcommand{\bloivgiusntbase}[2]{\begin{pspicture}(4,10)
   \put(0, 0){\tratteggio}
   \put(0,10){\dotsntsntiv}
   \put(2, 0){#1}
   \put(0, 0){\sumsntsntiv{#2}}
\end{pspicture}}

\newcommand{\blovigiusntbase}[2]{\begin{pspicture}(6,10)
   \put(0, 0){\tratteggio}
   \put(0,10){\dotsntsntvi}
   \put(3, 0){#1}
   \put(0, 0){\sumsntsntvi{#2}}
\end{pspicture}}

\newcommand{\bloin}[2]{\begin{pspicture}(0,0)
 \put(-0.5,10){\makebox(0,0)[r]{$#1$}}
 \put(-0.5, 0){\makebox(0,0)[r]{$#2$}}
 \put(0, 0){\punto}
 \put(0,10){\punto}
 \end{pspicture}}

\newcommand{\bloout}[2]{\begin{pspicture}(0,10)
   \put(0.5,10){\makebox(0,0)[l]{$#1$}}
   \put(0.5, 0){\makebox(0,0)[l]{$#2$}}
   \put(0, 0){\punto}
   \put(0,10){\punto}
   \put(0,0){\tratteggio}
\end{pspicture}}

\newcommand{\bloivinsu}[3]{\begin{pspicture}(4,0)
 \put(0,0){\tratteggio}
 \put(0,0){\dotsntdstiv}
 \put(0,10){\sumdstsntiv{#3}}
 \put(2,-2){\punto}
 \put(2,10.5){\psline{->}(0,1.5)}
 \put(2.25,12){\makebox(0,0)[l]{$#2$}}
 \put(2,-2){\psline{->}(0,2)}
 \put(2.25,-2){\makebox(0,0)[l]{$#1$}}
 \end{pspicture}}

\newcommand{\bloivingiu}[3]{\begin{pspicture}(4,0)
 \put(0,0){\tratteggio}
 \put(0,0){\sumdstsntiv{#3}}
 \put(0,10){\dotsntdstiv}
 \put(2,12){\punto}
 \put(2,-0.5){\psline{->}(0,-1.5)}
 \put(2.25,-2){\makebox(0,0)[l]{$#2$}}
 \put(2,12){\psline{->}(0,-2)}
 \put(2.25,12){\makebox(0,0)[l]{$#1$}}
 \end{pspicture}}

\newcommand{\bloivinsuor}[3]{\begin{pspicture}(4,0)
 \put(0,0){\tratteggio}
 \put(0,0){\sumsntsntiv{#3}}
 \put(0,10){\dotdstdstiv}
 \put(2,10){\punto}
 \put(2,-2){\punto}
 \put(2.25,-2){\makebox(0,0)[l]{$#1$}}
 \put(2.25,12){\makebox(0,0)[l]{$#2$}}
 \put(2,-2){\psline{->}(0,1.5)}
 \put(2,10){\psline{->}(0,2)}
 \end{pspicture}}

\newcommand{\bloivinsuaor}[3]{\begin{pspicture}(4,0)
 \put(0,0){\tratteggio}
 \put(0,0){\sumdstdstiv{#3}}
 \put(0,10){\dotsntsntiv}
 \put(2,10){\punto}
 \put(2,-2){\punto}
 \put(2,-2){\psline{->}(0,1.5)}
 \put(2.25,-2){\makebox(0,0)[l]{$#1$}}
 \put(2,10){\psline{->}(0,2)}
 \put(2.25,12){\makebox(0,0)[l]{$#2$}}
 \end{pspicture}}

\newcommand{\bloivingiuaor}[3]{\begin{pspicture}(4,0)
 \put(0,0){\tratteggio}
 \put(0,0){\dotdstdstiv}
 \put(0,10){\sumsntsntiv{#3}}
 \put(2,12){\punto}
 \put(2,12){\psline{->}(0,-1.5)}
 \put(2.25,12){\makebox(0,0)[l]{$#1$}}
 \put(2,0){\psline{->}(0,-2)}
 \put(2.25,-2){\makebox(0,0)[l]{$#2$}}
 \end{pspicture}}

\newcommand{\bloivingiuor}[3]{\begin{pspicture}(4,0)
 \put(0,0){\tratteggio}
 \put(0,0){\dotsntsntiv}
 \put(0,10){\sumdstdstiv{#3}}
 \put(2,12){\punto}
 \put(2,12){\psline{->}(0,-1.5)}
 \put(2.25,12){\makebox(0,0)[l]{$#1$}}
 \put(2,0){\psline{->}(0,-2)}
 \put(2.25,-2){\makebox(0,0)[l]{$#2$}}
 \end{pspicture}}

\newcommand{\blosu}[3]{\blosubase{\midsu{#1}{#2}}{#3}}
\newcommand{\blovisu}[3]{\blovisubase{\midsu{#1}{#2}}{#3}}
\newcommand{\bloivsu}[3]{\bloivsubase{\midsuiv{#1}{#2}}{#3}}
\newcommand{\blosuor}[3]{\blosuorbase{\midsuadd{#1}{#2}}{#3}}
\newcommand{\blovisuor}[3]{\blovisuorbase{\midsuadd{#1}{#2}}{#3}}
\newcommand{\bloivsuor}[3]{\bloivsuorbase{\midsuivadd{#1}{#2}}{#3}}
\newcommand{\blosuaor}[3]{\blosuaorbase{\midsuadd{#1}{#2}}{#3}}
\newcommand{\blovisuaor}[3]{\blovisuaorbase{\midsuadd{#1}{#2}}{#3}}
\newcommand{\bloivsuaor}[3]{\bloivsuaorbase{\midsuivadd{#1}{#2}}{#3}}
\newcommand{\blosupiu}[5]{\blosubase{\midsupiu{#1}{#2}{#4}{#5}}{#3}}
\newcommand{\blovisupiu}[5]{\blovisubase{\midsupiu{#1}{#2}{#4}{#5}}{#3}}
\newcommand{\bloivsupiu}[5]{\bloivsubase{\midsupiutre{#1}{#2}{#4}{#5}}{#3}}
\newcommand{\blosularge}[3]{\blosubase{\midsularge{#1}{#2}}{#3}}
\newcommand{\blosulargeor}[3]{\blosuorbase{\midsulargeadd{#1}{#2}}{#3}}
\newcommand{\blosulargeaor}[3]{\blosuaorbase{\midsulargeadd{#1}{#2}}{#3}}
\newcommand{\blosubig}[3]{\blosubase{\midsubig{#1}{#2}}{#3}}
\newcommand{\blosubigor}[3]{\blosuorbase{\midsubigadd{#1}{#2}}{#3}}
\newcommand{\blosubigaor}[3]{\blosuaorbase{\midsubigadd{#1}{#2}}{#3}}
\newcommand{\blosuintmlarge}[3]{\blosubase{\midsuintmlarge{#1}{#2}}{#3}}
\newcommand{\blosuintm}[3]{\blosubase{\midsuintm{#1}{#2}}{#3}}
\newcommand{\blovisuintm}[3]{\blovisubase{\midsuintm{#1}{#2}}{#3}}
\newcommand{\bloivsuintm}[3]{\bloivsubase{\midsuintmtre{#1}{#2}}{#3}}
\newcommand{\blosuintmorlarge}[3]{\blosuorbase{\midsuintmaddlarge{#1}{#2}}{#3}}
\newcommand{\blosuintmor}[3]{\blosuorbase{\midsuintmadd{#1}{#2}}{#3}}
\newcommand{\blovisuintmor}[3]{\blovisuorbase{\midsuintmadd{#1}{#2}}{#3}}
\newcommand{\bloivsuintmor}[3]{\bloivsuorbase{\midsuintmtreadd{#1}{#2}}{#3}}
\newcommand{\blosuintmaorlarge}[3]{\blosuaorbase{\midsuintmaddlarge{#1}{#2}}{#3}}
\newcommand{\blosuintmaor}[3]{\blosuaorbase{\midsuintmadd{#1}{#2}}{#3}}
\newcommand{\blovisuintmaor}[3]{\blovisuaorbase{\midsuintmadd{#1}{#2}}{#3}}
\newcommand{\bloivsuintmaor}[3]{\bloivsuaorbase{\midsuintmtreadd{#1}{#2}}{#3}}
\newcommand{\blosumintlarge}[3]{\blosubase{\midsumintlarge{#1}{#2}}{#3}}
\newcommand{\blosumint}[3]{\blosubase{\midsumint{#1}{#2}}{#3}}
\newcommand{\blovisumint}[3]{\blovisubase{\midsumint{#1}{#2}}{#3}}
\newcommand{\bloivsumint}[3]{\bloivsubase{\midsuminttre{#1}{#2}}{#3}}
\newcommand{\blosumintorlarge}[3]{\blosuorbase{\midsumintaddlarge{#1}{#2}}{#3}}
\newcommand{\blosumintor}[3]{\blosuorbase{\midsumintadd{#1}{#2}}{#3}}
\newcommand{\blovisumintor}[3]{\blovisuorbase{\midsumintadd{#1}{#2}}{#3}}
\newcommand{\bloivsumintor}[3]{\bloivsuorbase{\midsuminttreadd{#1}{#2}}{#3}}
\newcommand{\blosumintaorlarge}[3]{\blosuaorbase{\midsumintaddlarge{#1}{#2}}{#3}}
\newcommand{\blosumintaor}[3]{\blosuaorbase{\midsumintadd{#1}{#2}}{#3}}
\newcommand{\blovisumintaor}[3]{\blovisuaorbase{\midsumintadd{#1}{#2}}{#3}}
\newcommand{\bloivsumintaor}[3]{\bloivsuaorbase{\midsuminttreadd{#1}{#2}}{#3}}
\newcommand{\blosudermlarge}[3]{\blosubase{\midsudermlarge{#1}{#2}}{#3}}
\newcommand{\blosuderm}[3]{\blosubase{\midsuderm{#1}{#2}}{#3}}
\newcommand{\blovisuderm}[3]{\blovisubase{\midsuderm{#1}{#2}}{#3}}
\newcommand{\bloivsuderm}[3]{\bloivsubase{\midsudermtre{#1}{#2}}{#3}}
\newcommand{\blosudermorlarge}[3]{\blosuorbase{\midsudermaddlarge{#1}{#2}}{#3}}
\newcommand{\blosudermor}[3]{\blosuorbase{\midsudermadd{#1}{#2}}{#3}}
\newcommand{\blovisudermor}[3]{\blovisuorbase{\midsudermadd{#1}{#2}}{#3}}
\newcommand{\bloivsudermor}[3]{\bloivsuorbase{\midsudermtreadd{#1}{#2}}{#3}}
\newcommand{\blosudermaorlarge}[3]{\blosuaorbase{\midsudermaddlarge{#1}{#2}}{#3}}
\newcommand{\blosudermaor}[3]{\blosuaorbase{\midsudermadd{#1}{#2}}{#3}}
\newcommand{\blovisudermaor}[3]{\blovisuaorbase{\midsudermadd{#1}{#2}}{#3}}
\newcommand{\bloivsudermaor}[3]{\bloivsuaorbase{\midsudermtreadd{#1}{#2}}{#3}}
\newcommand{\blosumderlarge}[3]{\blosubase{\midsumderlarge{#1}{#2}}{#3}}
\newcommand{\blosumder}[3]{\blosubase{\midsumder{#1}{#2}}{#3}}
\newcommand{\blovisumder}[3]{\blovisubase{\midsumder{#1}{#2}}{#3}}
\newcommand{\bloivsumder}[3]{\bloivsubase{\midsumdertre{#1}{#2}}{#3}}
\newcommand{\blosumderorlarge}[3]{\blosuorbase{\midsumderaddlarge{#1}{#2}}{#3}}
\newcommand{\blosumderor}[3]{\blosuorbase{\midsumderadd{#1}{#2}}{#3}}
\newcommand{\blovisumderor}[3]{\blovisuorbase{\midsumderadd{#1}{#2}}{#3}}
\newcommand{\bloivsumderor}[3]{\bloivsuorbase{\midsumdertreadd{#1}{#2}}{#3}}
\newcommand{\blosumderaorlarge}[3]{\blosuaorbase{\midsumderaddlarge{#1}{#2}}{#3}}
\newcommand{\blosumderaor}[3]{\blosuaorbase{\midsumderadd{#1}{#2}}{#3}}
\newcommand{\blovisumderaor}[3]{\blovisuaorbase{\midsumderadd{#1}{#2}}{#3}}
\newcommand{\bloivsumderaor}[3]{\bloivsuaorbase{\midsumdertreadd{#1}{#2}}{#3}}
\newcommand{\blosumxmlarge}[5]{\blosubase{\midsumxmlarge{#1}{#2}{#4}{#5}}{#3}}
\newcommand{\blosumxm}[5]{\blosubase{\midsumxm{#1}{#2}{#4}{#5}}{#3}}
\newcommand{\blovisumxm}[5]{\blovisubase{\midsumxm{#1}{#2}{#4}{#5}}{#3}}
\newcommand{\bloivsumxm}[5]{\bloivsubase{\midsuivmxm{#1}{#2}{#4}{#5}}{#3}}
\newcommand{\blosumxmorlarge}[5]{\blosuorbase{\midsumxmaddlarge{#1}{#2}{#4}{#5}}{#3}}
\newcommand{\blosumxmor}[5]{\blosuorbase{\midsumxmadd{#1}{#2}{#4}{#5}}{#3}}
\newcommand{\blovisumxmor}[5]{\blovisuorbase{\midsumxmadd{#1}{#2}{#4}{#5}}{#3}}
\newcommand{\bloivsumxmor}[5]{\bloivsuorbase{\midsuivmxmadd{#1}{#2}{#4}{#5}}{#3}}
\newcommand{\blosumxmaorlarge}[5]{\blosuaorbase{\midsumxmaddlarge{#1}{#2}{#4}{#5}}{#3}}
\newcommand{\blosumxmaor}[5]{\blosuaorbase{\midsumxmadd{#1}{#2}{#4}{#5}}{#3}}
\newcommand{\blovisumxmaor}[5]{\blovisuaorbase{\midsumxmadd{#1}{#2}{#4}{#5}}{#3}}
\newcommand{\bloivsumxmaor}[5]{\bloivsuaorbase{\midsuivmxmadd{#1}{#2}{#4}{#5}}{#3}}
\newcommand{\blosuxxxlarge}[5]{\blosubase{\midsuxxxlarge{#1}{#2}{#3}{#4}}{#5}}
\newcommand{\blosuxxx}[5]{\blosubase{\midsuxxx{#1}{#2}{#3}{#4}}{#5}}
\newcommand{\blovisuxxx}[5]{\blovisubase{\midsuivxxx{#1}{#2}{#3}{#4}}{#5}}
\newcommand{\bloivsuxxx}[5]{\bloivsubase{\midsuivxxx{#1}{#2}{#3}{#4}}{#5}}
\newcommand{\blosuxxxlargeor}[5]{\blosuorbase{\midsuxxxlargeadd{#1}{#2}{#3}{#4}}{#5}}
\newcommand{\blosuxxxor}[5]{\blosuorbase{\midsuxxxadd{#1}{#2}{#3}{#4}}{#5}}
\newcommand{\blovisuxxxor}[5]{\blovisuorbase{\midsuivxxxadd{#1}{#2}{#3}{#4}}{#5}}
\newcommand{\bloivsuxxxor}[5]{\bloivsuorbase{\midsuivxxxadd{#1}{#2}{#3}{#4}}{#5}}
\newcommand{\blosuxxxlargeaor}[5]{\blosuaorbase{\midsuxxxlargeadd{#1}{#2}{#3}{#4}}{#5}}
\newcommand{\blosuxxxaor}[5]{\blosuaorbase{\midsuxxxadd{#1}{#2}{#3}{#4}}{#5}}
\newcommand{\blovisuxxxaor}[5]{\blovisuaorbase{\midsuivxxxadd{#1}{#2}{#3}{#4}}{#5}}
\newcommand{\bloivsuxxxaor}[5]{\bloivsuaorbase{\midsuivxxxadd{#1}{#2}{#3}{#4}}{#5}}
\newcommand{\blosuintmvuoto}[3]{\blosubase{\midsuintmvuoto{#1}{#2}}{#3}}
\newcommand{\blovisuintmvuoto}[3]{\blovisubase{\midsuintmvuoto{#1}{#2}}{#3}}
\newcommand{\bloivsuintmvuoto}[3]{\bloivsubase{\midsuintmvuoto{#1}{#2}}{#3}}
\newcommand{\blosudst}[3]{\blosudstbase{\midsuadd{#1}{#2}}{#3}}
\newcommand{\bloivsudst}[3]{\bloivsudstbase{\midsuivadd{#1}{#2}}{#3}}
\newcommand{\blovisudst}[3]{\blovisudstbase{\midsuadd{#1}{#2}}{#3}}
\newcommand{\blosusnt}[3]{\blosusntbase{\midsuadd{#1}{#2}}{#3}}
\newcommand{\bloivsusnt}[3]{\bloivsusntbase{\midsuivadd{#1}{#2}}{#3}}
\newcommand{\blovisusnt}[3]{\blovisusntbase{\midsuadd{#1}{#2}}{#3}}
\newcommand{\blosumxmdst}[5]{\blosudstbase{\midsumxmadd{#1}{#2}{#4}{#5}}{#3}}
\newcommand{\bloivsumxmdst}[5]{\bloivsudstbase{\midsuivmxmadd{#1}{#2}{#4}{#5}}{#3}}
\newcommand{\blovisumxmdst}[5]{\blovisudstbase{\midsumxmadd{#1}{#2}{#4}{#5}}{#3}}
\newcommand{\blosumxmsnt}[5]{\blosusntbase{\midsumxmadd{#1}{#2}{#4}{#5}}{#3}}
\newcommand{\bloivsumxmsnt}[5]{\bloivsusntbase{\midsuivmxmadd{#1}{#2}{#4}{#5}}{#3}}
\newcommand{\blovisumxmsnt}[5]{\blovisusntbase{\midsumxmadd{#1}{#2}{#4}{#5}}{#3}}

\newcommand{\blogiu}[3]{\blogiubase{\midgiu{#1}{#2}}{#3}}
\newcommand{\blovigiu}[3]{\blovigiubase{\midgiu{#1}{#2}}{#3}}
\newcommand{\bloivgiu}[3]{\bloivgiubase{\midgiuiv{#1}{#2}}{#3}}
\newcommand{\blogiuor}[3]{\blogiuorbase{\midgiuadd{#1}{#2}}{#3}}
\newcommand{\blovigiuor}[3]{\blovigiuorbase{\midgiuadd{#1}{#2}}{#3}}
\newcommand{\bloivgiuor}[3]{\bloivgiuorbase{\midgiuivadd{#1}{#2}}{#3}}
\newcommand{\blogiuaor}[3]{\blogiuaorbase{\midgiuadd{#1}{#2}}{#3}}
\newcommand{\blovigiuaor}[3]{\blovigiuaorbase{\midgiuadd{#1}{#2}}{#3}}
\newcommand{\bloivgiuaor}[3]{\bloivgiuaorbase{\midgiuivadd{#1}{#2}}{#3}}
\newcommand{\blogiularge}[3]{\blogiubase{\midgiularge{#1}{#2}}{#3}}
\newcommand{\blogiulargeor}[3]{\blogiuorbase{\midgiulargeadd{#1}{#2}}{#3}}
\newcommand{\blogiulargeaor}[3]{\blogiuaorbase{\midgiulargeadd{#1}{#2}}{#3}}
\newcommand{\blogiubig}[3]{\blogiubase{\midgiubig{#1}{#2}}{#3}}
\newcommand{\blogiubigor}[3]{\blogiuorbase{\midgiubigadd{#1}{#2}}{#3}}
\newcommand{\blogiubigaor}[3]{\blogiuaorbase{\midgiubigadd{#1}{#2}}{#3}}
\newcommand{\blogiuintmlarge}[3]{\blogiubase{\midgiuintmlarge{#1}{#2}}{#3}}
\newcommand{\blogiuintm}[3]{\blogiubase{\midgiuintm{#1}{#2}}{#3}}
\newcommand{\blovigiuintm}[3]{\blovigiubase{\midgiuintm{#1}{#2}}{#3}}
\newcommand{\bloivgiuintm}[3]{\bloivgiubase{\midgiuintmtre{#1}{#2}}{#3}}
\newcommand{\blogiuintmorlarge}[3]{\blogiuorbase{\midgiuintmaddlarge{#1}{#2}}{#3}}
\newcommand{\blogiuintmor}[3]{\blogiuorbase{\midgiuintmadd{#1}{#2}}{#3}}
\newcommand{\blovigiuintmor}[3]{\blovigiuorbase{\midgiuintmadd{#1}{#2}}{#3}}
\newcommand{\bloivgiuintmor}[3]{\bloivgiuorbase{\midgiuintmtreadd{#1}{#2}}{#3}}
\newcommand{\blogiuintmaorlarge}[3]{\blogiuaorbase{\midgiuintmaddlarge{#1}{#2}}{#3}}
\newcommand{\blogiuintmaor}[3]{\blogiuaorbase{\midgiuintmadd{#1}{#2}}{#3}}
\newcommand{\blovigiuintmaor}[3]{\blovigiuaorbase{\midgiuintmadd{#1}{#2}}{#3}}
\newcommand{\bloivgiuintmaor}[3]{\bloivgiuaorbase{\midgiuintmtreadd{#1}{#2}}{#3}}
\newcommand{\blogiumintlarge}[3]{\blogiubase{\midgiumintlarge{#1}{#2}}{#3}}
\newcommand{\blogiumint}[3]{\blogiubase{\midgiumint{#1}{#2}}{#3}}
\newcommand{\blovigiumint}[3]{\blovigiubase{\midgiumint{#1}{#2}}{#3}}
\newcommand{\bloivgiumint}[3]{\bloivgiubase{\midgiuminttre{#1}{#2}}{#3}}
\newcommand{\blogiumintorlarge}[3]{\blogiuorbase{\midgiumintaddlarge{#1}{#2}}{#3}}
\newcommand{\blogiumintor}[3]{\blogiuorbase{\midgiumintadd{#1}{#2}}{#3}}
\newcommand{\blovigiumintor}[3]{\blovigiuorbase{\midgiumintadd{#1}{#2}}{#3}}
\newcommand{\bloivgiumintor}[3]{\bloivgiuorbase{\midgiuminttreadd{#1}{#2}}{#3}}
\newcommand{\blogiumintaorlarge}[3]{\blogiuaorbase{\midgiumintaddlarge{#1}{#2}}{#3}}
\newcommand{\blogiumintaor}[3]{\blogiuaorbase{\midgiumintadd{#1}{#2}}{#3}}
\newcommand{\blovigiumintaor}[3]{\blovigiuaorbase{\midgiumintadd{#1}{#2}}{#3}}
\newcommand{\bloivgiumintaor}[3]{\bloivgiuaorbase{\midgiuminttreadd{#1}{#2}}{#3}}
\newcommand{\blogiudermlarge}[3]{\blogiubase{\midgiudermlarge{#1}{#2}}{#3}}
\newcommand{\blogiuderm}[3]{\blogiubase{\midgiuderm{#1}{#2}}{#3}}
\newcommand{\blovigiuderm}[3]{\blovigiubase{\midgiuderm{#1}{#2}}{#3}}
\newcommand{\bloivgiuderm}[3]{\bloivgiubase{\midgiudermtre{#1}{#2}}{#3}}
\newcommand{\blogiudermorlarge}[3]{\blogiuorbase{\midgiudermaddlarge{#1}{#2}}{#3}}
\newcommand{\blogiudermor}[3]{\blogiuorbase{\midgiudermadd{#1}{#2}}{#3}}
\newcommand{\blovigiudermor}[3]{\blovigiuorbase{\midgiudermadd{#1}{#2}}{#3}}
\newcommand{\bloivgiudermor}[3]{\bloivgiuorbase{\midgiudermtreadd{#1}{#2}}{#3}}
\newcommand{\blogiudermaorlarge}[3]{\blogiuaorbase{\midgiudermaddlarge{#1}{#2}}{#3}}
\newcommand{\blogiudermaor}[3]{\blogiuaorbase{\midgiudermadd{#1}{#2}}{#3}}
\newcommand{\blovigiudermaor}[3]{\blovigiuaorbase{\midgiudermadd{#1}{#2}}{#3}}
\newcommand{\bloivgiudermaor}[3]{\bloivgiuaorbase{\midgiudermtreadd{#1}{#2}}{#3}}
\newcommand{\blogiumderlarge}[3]{\blogiubase{\midgiumderlarge{#1}{#2}}{#3}}
\newcommand{\blogiumder}[3]{\blogiubase{\midgiumder{#1}{#2}}{#3}}
\newcommand{\blovigiumder}[3]{\blovigiubase{\midgiumder{#1}{#2}}{#3}}
\newcommand{\bloivgiumder}[3]{\bloivgiubase{\midgiumdertre{#1}{#2}}{#3}}
\newcommand{\blogiumderorlarge}[3]{\blogiuorbase{\midgiumderaddlarge{#1}{#2}}{#3}}
\newcommand{\blogiumderor}[3]{\blogiuorbase{\midgiumderadd{#1}{#2}}{#3}}
\newcommand{\blovigiumderor}[3]{\blovigiuorbase{\midgiumderadd{#1}{#2}}{#3}}
\newcommand{\bloivgiumderor}[3]{\bloivgiuorbase{\midgiumdertreadd{#1}{#2}}{#3}}
\newcommand{\blogiumderaorlarge}[3]{\blogiuaorbase{\midgiumderaddlarge{#1}{#2}}{#3}}
\newcommand{\blogiumderaor}[3]{\blogiuaorbase{\midgiumderadd{#1}{#2}}{#3}}
\newcommand{\blovigiumderaor}[3]{\blovigiuaorbase{\midgiumderadd{#1}{#2}}{#3}}
\newcommand{\bloivgiumderaor}[3]{\bloivgiuaorbase{\midgiumdertreadd{#1}{#2}}{#3}}
\newcommand{\blogiumxmlarge}[5]{\blogiubase{\midgiumxmlarge{#1}{#2}{#4}{#5}}{#3}}
\newcommand{\blogiumxm}[5]{\blogiubase{\midgiumxm{#1}{#2}{#4}{#5}}{#3}}
\newcommand{\blovigiumxm}[5]{\blovigiubase{\midgiumxm{#1}{#2}{#4}{#5}}{#3}}
\newcommand{\bloivgiumxm}[5]{\bloivgiubase{\midgiuivmxm{#1}{#2}{#4}{#5}}{#3}}
\newcommand{\blogiumxmorlarge}[5]{\blogiuorbase{\midgiumxmaddlarge{#1}{#2}{#4}{#5}}{#3}}
\newcommand{\blogiumxmor}[5]{\blogiuorbase{\midgiumxmadd{#1}{#2}{#4}{#5}}{#3}}
\newcommand{\blovigiumxmor}[5]{\blovigiuorbase{\midgiumxmadd{#1}{#2}{#4}{#5}}{#3}}
\newcommand{\bloivgiumxmor}[5]{\bloivgiuorbase{\midgiuivmxmadd{#1}{#2}{#4}{#5}}{#3}} %%
\newcommand{\blogiumxmaorlarge}[5]{\blogiuaorbase{\midgiumxmaddlarge{#1}{#2}{#4}{#5}}{#3}}
\newcommand{\blogiumxmaor}[5]{\blogiuaorbase{\midgiumxmadd{#1}{#2}{#4}{#5}}{#3}}
\newcommand{\blovigiumxmaor}[5]{\blovigiuaorbase{\midgiumxmadd{#1}{#2}{#4}{#5}}{#3}}
\newcommand{\bloivgiumxmaor}[5]{\bloivgiuaorbase{\midgiuivmxmadd{#1}{#2}{#4}{#5}}{#3}}
\newcommand{\blogiuxxxlarge}[5]{\blogiubase{\midgiuxxxlarge{#1}{#2}{#3}{#4}}{#5}}
\newcommand{\blogiuxxx}[5]{\blogiubase{\midgiuxxx{#1}{#2}{#3}{#4}}{#5}}
\newcommand{\blovigiuxxx}[5]{\blovigiubase{\midgiuxxx{#1}{#2}{#3}{#4}}{#5}}
\newcommand{\bloivgiuxxx}[5]{\bloivgiubase{\midgiuivxxx{#1}{#2}{#3}{#4}}{#5}}
\newcommand{\bloivgiuxxxr}[5]{\bloivgiubaser{\midgiuivxxxr{#1}{#2}{#3}{#4}}{#5}}
\newcommand{\blogiuxxxlargeor}[5]{\blogiuorbase{\midgiuxxxlargeadd{#1}{#2}{#3}{#4}}{#5}}
\newcommand{\blogiuxxxor}[5]{\blogiuorbase{\midgiuxxxadd{#1}{#2}{#3}{#4}}{#5}}
\newcommand{\blovigiuxxxor}[5]{\blovigiuorbase{\midgiuxxxadd{#1}{#2}{#3}{#4}}{#5}}
\newcommand{\bloivgiuxxxor}[5]{\bloivgiuorbase{\midgiuivxxxadd{#1}{#2}{#3}{#4}}{#5}} %%
\newcommand{\blogiuxxxlargeaor}[5]{\blogiuaorbase{\midgiuxxxlargeadd{#1}{#2}{#3}{#4}}{#5}}
\newcommand{\blogiuxxxaor}[5]{\blogiuaorbase{\midgiuxxxadd{#1}{#2}{#3}{#4}}{#5}}
\newcommand{\blovigiuxxxaor}[5]{\blovigiuaorbase{\midgiuxxxadd{#1}{#2}{#3}{#4}}{#5}}
\newcommand{\bloivgiuxxxaor}[5]{\bloivgiuaorbase{\midgiuivxxxadd{#1}{#2}{#3}{#4}}{#5}}
\newcommand{\blogiuintmvuoto}[3]{\blogiubase{\midgiuintmvuoto{#1}{#2}}{#3}}
\newcommand{\blovigiuintmvuoto}[3]{\blovigiubase{\midgiuintmvuoto{#1}{#2}}{#3}}
\newcommand{\bloivgiuintmvuoto}[3]{\bloivgiubase{\midgiuintmvuoto{#1}{#2}}{#3}}
\newcommand{\blogiudst}[3]{\blogiudstbase{\midgiuadd{#1}{#2}}{#3}}
\newcommand{\bloivgiudst}[3]{\bloivgiudstbase{\midgiuivadd{#1}{#2}}{#3}}
\newcommand{\blovigiudst}[3]{\blovigiudstbase{\midgiuadd{#1}{#2}}{#3}}
\newcommand{\blogiusnt}[3]{\blogiusntbase{\midgiuadd{#1}{#2}}{#3}}
\newcommand{\bloivgiusnt}[3]{\bloivgiusntbase{\midgiuivadd{#1}{#2}}{#3}}
\newcommand{\blovigiusnt}[3]{\blovigiusntbase{\midgiuadd{#1}{#2}}{#3}}
\newcommand{\blogiumxmdst}[5]{\blogiudstbase{\midgiumxmadd{#1}{#2}{#4}{#5}}{#3}}
\newcommand{\bloivgiumxmdst}[5]{\bloivgiudstbase{\midgiuivmxmadd{#1}{#2}{#4}{#5}}{#3}}
\newcommand{\blovigiumxmdst}[5]{\blovigiudstbase{\midgiumxmadd{#1}{#2}{#4}{#5}}{#3}}
\newcommand{\blogiumxmsnt}[5]{\blogiusntbase{\midgiumxmadd{#1}{#2}{#4}{#5}}{#3}}
\newcommand{\bloivgiumxmsnt}[5]{\bloivgiusntbase{\midgiuivmxmadd{#1}{#2}{#4}{#5}}{#3}}
\newcommand{\blovigiumxmsnt}[5]{\blovigiusntbase{\midgiumxmadd{#1}{#2}{#4}{#5}}{#3}}

\newcommand{\kxdst}[4]{\begin{pspicture}(0,1)
   \put(0,0){\psline{->}(1,0)}
   \psset{unit=0.5\unitlength}
   \rput[lb](2,-#1){\psset{unit=2.0\unitlength}
   \framebox(#1,#1){$#2$}}
   \psset{unit=2.0\unitlength}
   \put(#1,0){
   \put(1,0){\psline{->}(1,0)}
   \put(2.25, 0.5){\makebox(0,0)[bl]{$#3$}}
   \put(2.25,-0.5){\makebox(0,0)[tl]{$#4$}}}
\end{pspicture}}

\newcommand{\kxsnt}[4]{\begin{pspicture}(0,1)
   \put(1,0){\psline{->}(-1,0)}
   \put(-0.25, 0.5){\makebox(0,0)[br]{$#3$}}
   \put(-0.25,-0.5){\makebox(0,0)[tr]{$#4$}}
   \put(#1,0){
   \put(2,0){\psline{->}(-1,0)}
   \psset{unit=0.5\unitlength}
   \rput[rb](2,-#1){\psset{unit=2.0\unitlength}
   \framebox(#1,#1){$#2$}}
   \psset{unit=2.0\unitlength}}
\end{pspicture}}

\newcommand{\kxydst}[5]{\begin{pspicture}(0,1)
   \put(0,0){\psline{->}(1,0)}
   \psset{unit=0.5\unitlength}
   \rput[lb](2,-#2){\psset{unit=2.0\unitlength}
   \framebox(#1,#2){$#3$}}
   \psset{unit=2.0\unitlength}
   \put(#1,0){
   \put(1,0){\psline{->}(1,0)}
   \put(2.25, 0.5){\makebox(0,0)[bl]{$#4$}}
   \put(2.25,-0.5){\makebox(0,0)[tl]{$#5$}}}
\end{pspicture}}

\newcommand{\kxysnt}[5]{\begin{pspicture}(0,1)
   \put(1,0){\psline{->}(-1,0)}
   \put(-0.25, 0.5){\makebox(0,0)[br]{$#4$}}
   \put(-0.25,-0.5){\makebox(0,0)[tr]{$#5$}}
   \put(#1,0){
   \put(2,0){\psline{->}(-1,0)}
   \psset{unit=0.5\unitlength}
   \rput[rb](2,-#2){\psset{unit=2.0\unitlength}
   \framebox(#1,#2){$#3$}}
   \psset{unit=2.0\unitlength}}
\end{pspicture}}

\newcommand{\blokxbase}[3]{\begin{pspicture}(#1,10)
   \put(0, 0){\tratteggio}
   \put(0, 0){#2}
   \put(0,10){#3}
\end{pspicture}
\begin{pspicture}(2,10)\put(0, 0){}\end{pspicture}
}
\newcommand{\blokxor}[5]{\blokxbase{#1}{\kxsnt{#1}{#3}{}{#5}}{\kxdst{#1}{#2}{#4}{}}}
\newcommand{\blokxaor}[5]{\blokxbase{#1}{\kxdst{#1}{#3}{}{#5}}{\kxsnt{#1}{#2}{#4}{}}}
\newcommand{\blokxyor}[6]{\blokxbase{#1}{\kxysnt{#1}{#2}{#4}{}{#6}}{\kxydst{#1}{#2}{#3}{#5}{}}}
\newcommand{\blokxyaor}[6]{\blokxbase{#1}{\kxydst{#1}{#2}{#4}{}{#6}}{\kxysnt{#1}{#2}{#3}{#5}{}}}
\newcommand{\blokxydst}[6]{\blokxbase{#1}{\kxydst{#1}{#2}{#4}{}{#6}}{\kxydst{#1}{#2}{#3}{#5}{}}}
\newcommand{\blokxysnt}[6]{\blokxbase{#1}{\kxysnt{#1}{#2}{#4}{}{#6}}{\kxysnt{#1}{#2}{#3}{#5}{}}}

\newcommand{\kduedst}[3]{\kxdst{2}{#1}{#2}{#3}}
\newcommand{\kduesnt}[3]{\kxsnt{2}{#1}{#2}{#3}}
\newcommand{\ktredst}[3]{\kxdst{4}{#1}{#2}{#3}}
\newcommand{\ktresnt}[3]{\kxsnt{4}{#1}{#2}{#3}}
\newcommand{\kquadst}[3]{\kxydst{4}{3}{#1}{#2}{#3}}
\newcommand{\kquasnt}[3]{\kxysnt{4}{3}{#1}{#2}{#3}}
\newcommand{\kcindst}[3]{\kxydst{8}{4}{#1}{#2}{#3}}
\newcommand{\kcinsnt}[3]{\kxysnt{8}{4}{#1}{#2}{#3}}
\newcommand{\blokbase}[2]{\blokxbase{2}{#1}{#2}}
\newcommand{\bloko}[4]{\blokxbase{2}{\kduesnt{#2}{}{#4}}{\kduedst{#1}{#3}{}}}
\newcommand{\blokao}[4]{\blokxbase{2}{\kduedst{#2}{}{#4}}{\kduesnt{#1}{#3}{}}}
\newcommand{\bloktreo}[4]{\blokxbase{4}{\ktresnt{#2}{}{#4}}{\ktredst{#1}{#3}{}}}
\newcommand{\bloktreao}[4]{\blokxbase{4}{\ktredst{#2}{}{#4}}{\ktresnt{#1}{#3}{}}}
\newcommand{\blokquao}[4]{\blokxbase{4}{\kquasnt{#2}{}{#4}}{\kquadst{#1}{#3}{}}}
\newcommand{\blokquaao}[4]{\blokxbase{4}{\kquadst{#2}{}{#4}}{\kquasnt{#1}{#3}{}}}
\newcommand{\blokcino}[4]{\blokxbase{8}{\kcinsnt{#2}{}{#4}}{\kcindst{#1}{#3}{}}}
\newcommand{\blokcinao}[4]{\blokxbase{8}{\kcindst{#2}{}{#4}}{\kcinsnt{#1}{#3}{}}}

\newcommand{\blosuatt}[5]{\begin{pspicture}(8,10)
   \put(2.5,10.5){\makebox(0,0)[b]{$#4$}}
   \put(2,7.5){\makebox(0,0)[r]{$#3$}}
   \put(0,0){\tratteggio}
   \put(0,0){\psline{->}(2,0)}
   \put(8,0){\psline{->}(-5,0)}
   \put(2.5,0){\circle{1}}
   \put(2.5,0){\piu}
   \put(2.5,0){#5}
   \put(2.5,0.5){\psline{->}(0,3)}
   \put(2.5, 2){\punto}
   \put(2.5, 2){\psline(3.5,0)}
   \put(  6, 2){\psline{->}(0,2)}
   \put(  5, 4){\psframe(2,2)\rput(1,1){$#2$}}
   \put(  1, 3.5){\psframe(3,3)\rput(1.5,1.5){$#1$}}
   \put(2.5, 6.5){\psline{->}(0,1.5)}
   \put(  6, 6){\psline(0,2.5)}
   \put(  6, 8.5){\psline{->}(-3,0)}
   \put(2.5,8.5){\circle{1}}
   \put(2.5,8.5){\piu}
   \put(2.5,  9){\psline{->}(0,1)}
   \put(2.5, 10){\psline{->}(-2.5,0)}
   \put(2.5, 10){\psline{->}(5.5,0)}
   \end{pspicture}}

\newcommand{\blogiuattraor}[5]{\begin{pspicture}(8,10)
   \put(2,2){\makebox(0,0)[r]{$#3$}}
   \put(0,0){\tratteggio}
   \put(0,0){\psline{->}(2,0)}
   \put(3,0){\psline{->}(5,0)}
   \put(2.5,0){\circle{1}}
   \put(2.5,0){\piu}
   \put(2.5,0){#5}
   \put(2.5,3.5){\psline{->}(0,-3)}
   \put(2.5, 2){\punto}
   \put(2.5, 2){\psline(3.5,0)}
   \put(  6, 2){\psline{->}(0,2)}
   \put(  5, 4){\psframe(2,2)\rput(1,1){$#2$}}
   \put(  1, 3.5){\psframe(3,3)\rput(1.5,1.5){$#1$}}
   \put(2.5, 8){\psline{->}(0,-1.5)}
   \put(  6, 6){\psline(0,2.5)}
   \put(  6, 8.5){\psline{->}(-3,0)}
   \put(2.5,8.5){\circle{1}}
   \put(2.5,8.5){\piu}
   \put(2.5,8.5){#4}
   \put(2.5, 10){\psline{->}(0,-1)}
   \put(2.5, 10){\psline{->}(-2.5,0)}
   \put(8, 10){\psline{->}(-5.5,0)}
   \end{pspicture}}

\newcommand{\blogiuattror}[5]{\begin{pspicture}(8,10)
   \put(2,2){\makebox(0,0)[r]{$#3$}}
   \put(0,0){\tratteggio}
   \put(2,0){\psline{->}(-2,0)}
   \put(8,0){\psline{->}(-5,0)}
   \put(2.5,0){\circle{1}}
   \put(2.5,0){\piu}
   \put(2.5,0){#5}
   \put(2.5,3.5){\psline{->}(0,-3)}
   \put(2.5, 2){\punto}
   \put(2.5, 2){\psline(3.5,0)}
   \put(  6, 2){\psline{->}(0,2)}
   \put(  5, 4){\psframe(2,2)\rput(1,1){$#2$}}
   \put(  1, 3.5){\psframe(3,3)\rput(1.5,1.5){$#1$}}
   \put(2.5, 8){\psline{->}(0,-1.5)}
   \put(  6, 6){\psline(0,2.5)}
   \put(  6, 8.5){\psline{->}(-3,0)}
   \put(2.5,8.5){\circle{1}}
   \put(2.5,8.5){\piu}
   \put(2.5,8.5){#4}
   \put(2.5, 10){\psline{->}(0,-1)}
   \put(0, 10){\psline{->}(2.5,0)}
   \put(2.5, 10){\psline{->}(5.5,0)}
   \end{pspicture}}

\newcommand{\blogiuatt}[5]{\begin{pspicture}(8,10)
   \put(2.5,-0.5){\makebox(0,0)[t]{$#4$}}
   \put(2,2.5){\makebox(0,0)[r]{$#3$}}
   \put(0,0){\tratteggio}
   \put(0,10){\psline{->}(2,0)}
   \put(8,10){\psline{->}(-5,0)}
   \put(2.5,10){\circle{1}}
   \put(2.5,10){\piu}
   \put(2.5,10){#5}
   \put(2.5,9.5){\psline{->}(0,-3)}
   \put(2.5, 8){\punto}
   \put(2.5, 8){\psline(3.5,0)}
   \put(  6, 8){\psline{->}(0,-2)}
   \put(  5, 4){\psframe(2,2)\rput(1,1){$#2$}}
   \put(  1, 3.5){\psframe(3,3)\rput(1.5,1.5){$#1$}}
   \put(2.5, 3.5){\psline{->}(0,-1.5)}
   \put(  6, 1.5){\psline(0,2.5)}
   \put(  6, 1.5){\psline{->}(-3,0)}
   \put(2.5,1.5){\circle{1}}
   \put(2.5,1.5){\piu}
   \put(2.5,  1){\psline{->}(0,-1)}
   \put(2.5,  0){\psline{->}(-2.5,0)}
   \put(2.5,  0){\psline{->}(5.5,0)}
   \end{pspicture}}

\newcommand{\blosuattaor}[5]{\begin{pspicture}(8,10)
   \put(2.5,-0.5){\makebox(0,0)[t]{$#4$}}
   \put(2,2.5){\makebox(0,0)[r]{$#3$}}
   \put(0,0){\tratteggio}
   \put(2,10){\psline{->}(-2,0)}
   \put(8,10){\psline{->}(-5,0)}
   \put(2.5,10){\circle{1}}
   \put(2.5,10){\piu}
   \put(2.5,10){#5}
   \put(2.5,6.5){\psline{->}(0,3)}
   \put(2.5, 8){\punto}
   \put(2.5, 8){\psline(3.5,0)}
   \put(  6, 8){\psline{->}(0,-2)}
   \put(  5, 4){\psframe(2,2)\rput(1,1){$#2$}}
   \put(  1, 3.5){\psframe(3,3)\rput(1.5,1.5){$#1$}}
   \put(2.5, 2.0){\psline{->}(0,1.5)}
   \put(  6, 1.5){\psline(0,2.5)}
   \put(  6, 1.5){\psline{->}(-3,0)}
   \put(2.5,1.5){\circle{1}}
   \put(2.5,1.5){\piu}
   \put(2.5,  0){\psline{->}(0,1)}
   \put(0,  0){\psline{->}(2.5,0)}
   \put(2.5,  0){\psline{->}(5.5,0)}
   \end{pspicture}}

\newcommand{\blosuattor}[5]{\begin{pspicture}(8,10)
   \put(2.5,-0.5){\makebox(0,0)[t]{$#4$}}
   \put(2,2.5){\makebox(0,0)[r]{$#3$}}
   \put(0,0){\tratteggio}
   \put(0,10){\psline{->}(2,0)}
   \put(3,10){\psline{->}(5,0)}
   \put(2.5,10){\circle{1}}
   \put(2.5,10){\piu}
   \put(2.5,10){#5}
   \put(2.5,6.5){\psline{->}(0,3)}
   \put(2.5, 8){\punto}
   \put(2.5, 8){\psline(3.5,0)}
   \put(  6, 8){\psline{->}(0,-2)}
   \put(  5, 4){\psframe(2,2)\rput(1,1){$#2$}}
   \put(  1, 3.5){\psframe(3,3)\rput(1.5,1.5){$#1$}}
   \put(2.5, 2.0){\psline{->}(0,1.5)}
   \put(  6, 1.5){\psline(0,2.5)}
   \put(  6, 1.5){\psline{->}(-3,0)}
   \put(2.5,1.5){\circle{1}}
   \put(2.5,1.5){\piu}
   \put(2.5,  0){\psline{->}(0,1)}
   \put(2.5,  0){\psline{->}(-2.5,0)}
   \put(8,  0){\psline{->}(-5.5,0)}
   \end{pspicture}}

\newcommand{\blosufdk}[9]{\begin{pspicture}(8,10)
   \put(2.5,10.5){\makebox(0,0)[b]{$#7$}}
   \put(2.0,8.5){\makebox(0,0)[r]{$#4$}}
   \put(6.25,8){\makebox(0,0)[l]{$#5$}}
   \put(6.25,3.5){\makebox(0,0)[l]{$#6$}}
   \put(0,0){\tratteggio}
   \put(0,0){\psline{->}(2,0)}
   \put(8,0){\psline{->}(-5,0)}
   \put(2.5,0){\piu}
   \put(2.5,0){#8}
   \put(2.5,0.5){\psline{->}(0,1)}
   \put(2.5,2){\piu}
   \put(2.5,2){#9}
   \put(2.5,2.5){\psline{->}(0,1.5)}
   \put(2.0,3.5){\makebox(0,0)[rt]{$#3$}}
   \put(6, 2){\psline{->}(-3,0)}
   \put(  6, 2){\psline(0,2.5)}
   \put(  5, 4.5){\psframe(2,2)\rput(1,1){$#2$}}
   \put(  1, 4.0){\psframe(3,3)\rput(1.5,1.5){$#1$}}
   \put(2.5, 7.0){\psline{->}(0,1.5)}
   \put(  6, 8.5){\psline{->}(0,-2)}
   \put(  6, 8.5){\psline(-3.5,0)}
   \put(2.5, 8.5){\punto}
   \put(2.5, 8.5){\psline{->}(0,1.5)}
   \put(2.5, 10){\psline{->}(-2.5,0)}
   \put(2.5, 10){\psline{->}(5.5,0)}
   \end{pspicture}}

\newcommand{\blosuaorpar}[9]{\begin{pspicture}(8,10)
   \put(2.5,-0.5){\makebox(0,0)[t]{$#7$}}
   \put(2.0,1.5){\makebox(0,0)[r]{$#4$}}
   \put(6.25,2){\makebox(0,0)[l]{$#5$}}
   \put(6.25,6.5){\makebox(0,0)[l]{$#6$}}
   \put(0,0){\tratteggio}
   \put(2,10){\psline{->}(-2,0)}
   \put(8,10){\psline{->}(-5,0)}
   \put(2.5,10){\piu}
   \put(2.5,10){#8}
   \put(2.5,8.5){\psline{->}(0,1)}
   \put(2.5,8){\piu}
   \put(2.5,8){#9}
   \put(2.5,6.0){\psline{->}(0,1.5)}
   \put(2.0,6.5){\makebox(0,0)[rb]{$#3$}}
   \put(6, 8){\psline{->}(-3,0)}
   \put(  6, 8){\psline(0,-2.5)}
   \put(  5, 3.5){\psframe(2,2)\rput(1,1){$#2$}}
   \put(  1, 3.0){\psframe(3,3)\rput(1.5,1.5){$#1$}}
   \put(2.5, 1.5){\psline{->}(0,1.5)}
   \put(  6, 1.5){\psline{->}(0,2)}
   \put(  6, 1.5){\psline(-3.5,0)}
   \put(2.5, 1.5){\punto}
   \put(2.5, 0.0){\psline{->}(0,1.5)}
   \put(0, 0){\psline{->}(2.5,0)}
   \put(2.5, 0){\psline{->}(5.5,0)}
   \end{pspicture}}

\newcommand{\blosuorpar}[9]{\begin{pspicture}(8,10)
   \put(2.5,-0.5){\makebox(0,0)[t]{$#7$}}
   \put(2.0,1.5){\makebox(0,0)[r]{$#4$}}
   \put(6.25,2){\makebox(0,0)[l]{$#5$}}
   \put(6.25,6.5){\makebox(0,0)[l]{$#6$}}
   \put(0,0){\tratteggio}
   \put(0,10){\psline{->}(2,0)}
   \put(3,10){\psline{->}(5,0)}
   \put(2.5,10){\piu}
   \put(2.5,10){#8}
   \put(2.5,8.5){\psline{->}(0,1)}
   \put(2.5,8){\piu}
   \put(2.5,8){#9}
   \put(2.5,6.0){\psline{->}(0,1.5)}
   \put(2.0,6.5){\makebox(0,0)[rb]{$#3$}}
   \put(6, 8){\psline{->}(-3,0)}
   \put(  6, 8){\psline(0,-2.5)}
   \put(  5, 3.5){\psframe(2,2)\rput(1,1){$#2$}}
   \put(  1, 3.0){\psframe(3,3)\rput(1.5,1.5){$#1$}}
   \put(2.5, 1.5){\psline{->}(0,1.5)}
   \put(  6, 1.5){\psline{->}(0,2)}
   \put(  6, 1.5){\psline(-3.5,0)}
   \put(2.5, 1.5){\punto}
   \put(2.5, 0.0){\psline{->}(0,1.5)}
   \put(2.5, 0){\psline{->}(-2.5,0)}
   \put(8.0, 0){\psline{->}(-5.5,0)}
   \end{pspicture}}

\newcommand{\blogiufdk}[9]{\begin{pspicture}(8,10)
   \put(2.5,-0.5){\makebox(0,0)[t]{$#7$}}
   \put(2.0,1.5){\makebox(0,0)[r]{$#4$}}
   \put(6.25,2){\makebox(0,0)[l]{$#5$}}
   \put(6.25,6.5){\makebox(0,0)[l]{$#6$}}
   \put(0,0){\tratteggio}
   \put(0,10){\psline{->}(2,0)}
   \put(8,10){\psline{->}(-5,0)}
   \put(2.5,10){\piu}
   \put(2.5,10){#8}
   \put(2.5,9.5){\psline{->}(0,-1)}
   \put(2.5,8){\piu}
   \put(2.5,8){#9}
   \put(2.5,7.5){\psline{->}(0,-1.5)}
   \put(2.0,6.5){\makebox(0,0)[rb]{$#3$}}
   \put(6, 8){\psline{->}(-3,0)}
   \put(  6, 8){\psline(0,-2.5)}
   \put(  5, 3.5){\psframe(2,2)\rput(1,1){$#2$}}
   \put(  1, 3.0){\psframe(3,3)\rput(1.5,1.5){$#1$}}
   \put(2.5, 3.0){\psline{->}(0,-1.5)}
   \put(  6, 1.5){\psline{->}(0,2)}
   \put(  6, 1.5){\psline(-3.5,0)}
   \put(2.5, 1.5){\punto}
   \put(2.5, 1.5){\psline{->}(0,-1.5)}
   \put(2.5, 0){\psline{->}(-2.5,0)}
   \put(2.5, 0){\psline{->}(5.5,0)}
   \end{pspicture}}

\newcommand{\blogiuaorpar}[9]{\begin{pspicture}(8,10)
   \put(2.5,10.5){\makebox(0,0)[b]{$#7$}}
   \put(2.0,8.5){\makebox(0,0)[r]{$#4$}}
   \put(6.25,8){\makebox(0,0)[l]{$#5$}}
   \put(6.25,3.5){\makebox(0,0)[l]{$#6$}}
   \put(0,0){\tratteggio}
   \put(0,0){\psline{->}(2,0)}
   \put(3,0){\psline{->}(5,0)}
   \put(2.5,0){\piu}
   \put(2.5,0){#8}
   \put(2.5,1.5){\psline{->}(0,-1)}
   \put(2.5,2){\piu}
   \put(2.5,2){#9}
   \put(2.5,4.0){\psline{->}(0,-1.5)}
   \put(2.0,3.5){\makebox(0,0)[rt]{$#3$}}
   \put(6, 2){\psline{->}(-3,0)}
   \put(  6, 2){\psline(0,2.5)}
   \put(  5, 4.5){\psframe(2,2)\rput(1,1){$#2$}}
   \put(  1, 4.0){\psframe(3,3)\rput(1.5,1.5){$#1$}}
   \put(2.5, 8.5){\psline{->}(0,-1.5)}
   \put(  6, 8.5){\psline{->}(0,-2)}
   \put(  6, 8.5){\psline(-3.5,0)}
   \put(2.5, 8.5){\punto}
   \put(2.5, 10){\psline{->}(0,-1.5)}
   \put(2.5, 10){\psline{->}(-2.5,0)}
   \put(8, 10){\psline{->}(-5.5,0)}
   \end{pspicture}}

\newcommand{\blogiuorpar}[9]{\begin{pspicture}(8,10)
   \put(2.5,10.5){\makebox(0,0)[b]{$#7$}}
   \put(2.0,8.5){\makebox(0,0)[r]{$#4$}}
   \put(6.25,8){\makebox(0,0)[l]{$#5$}}
   \put(6.25,3.5){\makebox(0,0)[l]{$#6$}}
   \put(0,0){\tratteggio}
   \put(2,0){\psline{->}(-2,0)}
   \put(8,0){\psline{->}(-5,0)}
   \put(2.5,0){\piu}
   \put(2.5,0){#8}
   \put(2.5,1.5){\psline{->}(0,-1)}
   \put(2.5,2){\piu}
   \put(2.5,2){#9}
   \put(2.5,4.0){\psline{->}(0,-1.5)}
   \put(2.0,3.5){\makebox(0,0)[rt]{$#3$}}
   \put(6, 2){\psline{->}(-3,0)}
   \put(  6, 2){\psline(0,2.5)}
   \put(  5, 4.5){\psframe(2,2)\rput(1,1){$#2$}}
   \put(  1, 4.0){\psframe(3,3)\rput(1.5,1.5){$#1$}}
   \put(2.5, 8.5){\psline{->}(0,-1.5)}
   \put(  6, 8.5){\psline{->}(0,-2)}
   \put(  6, 8.5){\psline(-3.5,0)}
   \put(2.5, 8.5){\punto}
   \put(2.5, 10){\psline{->}(0,-1.5)}
   \put(0, 10){\psline{->}(2.5,0)}
   \put(2.5, 10){\psline{->}(5.5,0)}
   \end{pspicture}}

\newcommand{\blogiupar}[9]{\begin{pspicture}(8,10)
   \put(2.5,-0.5){\makebox(0,0)[t]{$#7$}}
   \put(2,2.5){\makebox(0,0)[r]{$#4$}}
   \put(2,8.0){\makebox(0,0)[r]{$#3$}}
   \put(6.25,2.5){\makebox(0,0)[l]{$#6$}}
   \put(6.25,7.5){\makebox(0,0)[l]{$#5$}}
   \put(0,0){\tratteggio}
   \put(0,10){\psline{->}(2,0)}
   \put(8,10){\psline{->}(-5,0)}
   \put(2.5,10){\circle{1}}
   \put(2.5,10){\piu}
   \put(2.5,10){#8}
   \put(2.5,9.5){\psline{->}(0,-1.5)}
   \put(2.5, 8){\punto}
   \put(2.5, 8){\psline{->}(0,-1.5)}
   \put(2.5, 8){\psline(3.5,0)}
   \put(  6, 8){\psline{->}(0,-2)}
   \put(  5, 4){\psframe(2,2)\rput(1,1){$#2$}}
   \put(  1, 3.5){\psframe(3,3)\rput(1.5,1.5){$#1$}}
   \put(2.5, 3.5){\psline{->}(0,-1.5)}
   \put(  6, 1.5){\psline(0,2.5)}
   \put(  6, 1.5){\psline{->}(-3,0)}
   \put(2.5,1.5){\piu}
   \put(2.5,1.5){#9}
   \put(2.5,  1){\psline{->}(0,-1)}
   \put(2.5,  0){\psline{->}(-2.5,0)}
   \put(2.5,  0){\psline{->}(5.5,0)}
   \end{pspicture}}

\newcommand{\blogiuorfdk}[9]{\begin{pspicture}(8,10)
   \put(2.5,10.5){\makebox(0,0)[b]{$#7$}}
   \put(2,7.5){\makebox(0,0)[r]{$#4$}}
   \put(2,2.0){\makebox(0,0)[r]{$#3$}}
   \put(6.25,2.5){\makebox(0,0)[l]{$#5$}}
   \put(6.25,7.5){\makebox(0,0)[l]{$#6$}}
   \put(0,0){\tratteggio}
   \put(2,0){\psline{->}(-2,0)}
   \put(8,0){\psline{->}(-5,0)}
   \put(2.5,0){\circle{1}}
   \put(2.5,2.0){\psline{->}(0,-1.5)}
   \put(2.5,0){\piu}
   \put(2.5,0){#8}
   \put(2.5,3.5){\psline{->}(0,-1.5)}
   \put(2.5, 2){\punto}
   \put(2.5, 2){\psline(3.5,0)}
   \put(  6, 2){\psline{->}(0,2)}
   \put(  5, 4){\psframe(2,2)\rput(1,1){$#2$}}
   \put(  1, 3.5){\psframe(3,3)\rput(1.5,1.5){$#1$}}
   \put(2.5, 8.0){\psline{->}(0,-1.5)}
   \put(  6, 6){\psline(0,2.5)}
   \put(  6, 8.5){\psline{->}(-3,0)}
   \put(2.5,8.5){\piu}
   \put(2.5,8.5){#9}
   \put(2.5, 10){\psline{->}(0,-1)}
   \put(0, 10){\psline{->}(2.5,0)}
   \put(2.5, 10){\psline{->}(5.5,0)}
   \end{pspicture}}

\newcommand{\blogiuaorfdk}[9]{\begin{pspicture}(8,10)
   \put(2.5,10.5){\makebox(0,0)[b]{$#7$}}
   \put(2,7.5){\makebox(0,0)[r]{$#4$}}
   \put(2,2.0){\makebox(0,0)[r]{$#3$}}
   \put(6.25,2.5){\makebox(0,0)[l]{$#5$}}
   \put(6.25,7.5){\makebox(0,0)[l]{$#6$}}
   \put(0,0){\tratteggio}
   \put(0,0){\psline{->}(2,0)}
   \put(3,0){\psline{->}(5,0)}
   \put(2.5,0){\circle{1}}
   \put(2.5,2.0){\psline{->}(0,-1.5)}
   \put(2.5,0){\piu}
   \put(2.5,0){#8}
   \put(2.5,3.5){\psline{->}(0,-1.5)}
   \put(2.5, 2){\punto}
   \put(2.5, 2){\psline(3.5,0)}
   \put(  6, 2){\psline{->}(0,2)}
   \put(  5, 4){\psframe(2,2)\rput(1,1){$#2$}}
   \put(  1, 3.5){\psframe(3,3)\rput(1.5,1.5){$#1$}}
   \put(2.5, 8.0){\psline{->}(0,-1.5)}
   \put(  6, 6){\psline(0,2.5)}
   \put(  6, 8.5){\psline{->}(-3,0)}
   \put(2.5,8.5){\piu}
   \put(2.5,8.5){#9}
   \put(2.5, 10){\psline{->}(0,-1)}
   \put(2.5, 10){\psline{->}(-2.5,0)}
   \put(8, 10){\psline{->}(-5.5,0)}
   \end{pspicture}}

\newcommand{\blosupar}[9]{\begin{pspicture}(8,10)
   \put(2.5,10.5){\makebox(0,0)[b]{$#7$}}
   \put(2,7.5){\makebox(0,0)[r]{$#4$}}
   \put(2,2.0){\makebox(0,0)[r]{$#3$}}
   \put(6.25,2.5){\makebox(0,0)[l]{$#5$}}
   \put(6.25,7.5){\makebox(0,0)[l]{$#6$}}
   \put(0,0){\tratteggio}
   \put(0,0){\psline{->}(2,0)}
   \put(8,0){\psline{->}(-5,0)}
   \put(2.5,2.0){\psline{->}(0,1.5)}
   \put(2.5,0){\piu}
   \put(2.5,0){#8}
   \put(2.5,0.5){\psline{->}(0,1.5)}
   \put(2.5, 2){\punto}
   \put(2.5, 2){\psline(3.5,0)}
   \put(  6, 2){\psline{->}(0,2)}
   \put(  5, 4){\psframe(2,2)\rput(1,1){$#2$}}
   \put(  1, 3.5){\psframe(3,3)\rput(1.5,1.5){$#1$}}
   \put(2.5, 6.5){\psline{->}(0,1.5)}
   \put(  6, 6){\psline(0,2.5)}
   \put(  6, 8.5){\psline{->}(-3,0)}
   \put(2.5,8.5){\circle{1}}
   \put(2.5,8.5){\piu}
   \put(2.5,8.5){#9}
   \put(2.5,  9){\psline{->}(0,1)}
   \put(2.5, 10){\psline{->}(-2.5,0)}
   \put(2.5, 10){\psline{->}(5.5,0)}
   \end{pspicture}}

\newcommand{\blosuorfdk}[9]{\begin{pspicture}(8,10)
   \put(2.5,-0.5){\makebox(0,0)[t]{$#7$}}
   \put(2,2.5){\makebox(0,0)[r]{$#4$}}
   \put(2,8.0){\makebox(0,0)[r]{$#3$}}
   \put(6.25,2.5){\makebox(0,0)[l]{$#6$}}
   \put(6.25,7.5){\makebox(0,0)[l]{$#5$}}
   \put(0,0){\tratteggio}
   \put(0,10){\psline{->}(2,0)}
   \put(3,10){\psline{->}(5,0)}
   \put(2.5,10){\piu}
   \put(2.5,10){#8}
   \put(2.5,8.0){\psline{->}(0,1.5)}
   \put(2.5,6.5){\psline{->}(0,1.5)}
   \put(2.5, 8){\punto}
   \put(2.5, 8){\psline(3.5,0)}
   \put(  6, 8){\psline{->}(0,-2)}
   \put(  5, 4){\psframe(2,2)\rput(1,1){$#2$}}
   \put(  1, 3.5){\psframe(3,3)\rput(1.5,1.5){$#1$}}
   \put(2.5, 2.0){\psline{->}(0,1.5)}
   \put(  6, 1.5){\psline(0,2.5)}
   \put(  6, 1.5){\psline{->}(-3,0)}
   \put(2.5,1.5){\piu}
   \put(2.5,1.5){#9}
   \put(2.5,  0){\psline{->}(0,1)}
   \put(2.5,  0){\psline{->}(-2.5,0)}
   \put(8,  0){\psline{->}(-5.5,0)}
   \end{pspicture}}

\newcommand{\blosuaorfdk}[9]{\begin{pspicture}(8,10)
   \put(2.5,-0.5){\makebox(0,0)[t]{$#7$}}
   \put(2,2.5){\makebox(0,0)[r]{$#4$}}
   \put(2,8.0){\makebox(0,0)[r]{$#3$}}
   \put(6.25,2.5){\makebox(0,0)[l]{$#6$}}
   \put(6.25,7.5){\makebox(0,0)[l]{$#5$}}
   \put(0,0){\tratteggio}
   \put(2,10){\psline{->}(-2,0)}
   \put(8,10){\psline{->}(-5,0)}
   \put(2.5,10){\piu}
   \put(2.5,10){#8}
   \put(2.5,8.0){\psline{->}(0,1.5)}
   \put(2.5,6.5){\psline{->}(0,1.5)}
   \put(2.5, 8){\punto}
   \put(2.5, 8){\psline(3.5,0)}
   \put(  6, 8){\psline{->}(0,-2)}
   \put(  5, 4){\psframe(2,2)\rput(1,1){$#2$}}
   \put(  1, 3.5){\psframe(3,3)\rput(1.5,1.5){$#1$}}
   \put(2.5, 2.0){\psline{->}(0,1.5)}
   \put(  6, 1.5){\psline(0,2.5)}
   \put(  6, 1.5){\psline{->}(-3,0)}
   \put(2.5,1.5){\piu}
   \put(2.5,1.5){#9}
   \put(2.5,  0){\psline{->}(0,1)}
   \put(0,  0){\psline{->}(2.5,0)}
   \put(2.5,  0){\psline{->}(5.5,0)}
   \end{pspicture}}

\newcommand{\blovigiuattror}[5]{\begin{pspicture}(6,10)
   \put(1.25,2){\makebox(0,0)[r]{$#3$}}
   \put(0,0){\tratteggio}
   \put(1,0){\psline{->}(-1,0)}
   \put(6,0){\psline{->}(-4,0)}
   \put(1.5,0){\circle{1}}
   \put(1.5,0){\piu}
   \put(1.5,0){#5}
   \put(1.5,4){\psline{->}(0,-3.5)}
   \put(1.5, 2){\punto}
   \put(1.5, 2){\psline(3,0)}
   \put(4.5, 2){\psline{->}(0,2)}
   \put(3.5, 4){\psframe(2,2)\rput(1,1){$#2$}}
   \put(0.5, 4){\psframe(2,2)\rput(1,1){$#1$}}
   \put(1.5, 7.5){\psline{->}(0,-1.5)}
   \put(4.5, 6){\psline(0,2)}
   \put(4.5, 8){\psline{->}(-2.5,0)}
   \put(1.5,8){\circle{1}}
   \put(1.5,8){\piu}
   \put(1.5,8){#4}
   \put(1.5, 10){\psline{->}(0,-1.5)}
   \put(0, 10){\psline{->}(1.5,0)}
   \put(1.5, 10){\psline{->}(4.5,0)}
   \end{pspicture}}

\newcommand{\blovigiuattr}[5]{\begin{pspicture}(6,10)
   \put(1.5,-0.5){\makebox(0,0)[t]{$#3$}}
   \put(0,0){\tratteggio}
   \put(1.5,0){\psline{->}(-1.5,0)}
   \put(1.5,0){\psline{->}(4.5,0)}
   \put(1.5,4){\psline{->}(0,-4)}
   \put(1.5, 2){\punto}
   \put(1.5, 2){\psline(3,0)}
   \put(4.5, 2){\psline{->}(0,2)}
   \put(3.5, 4){\psframe(2,2)\rput(1,1){$#2$}}
   \put(0.5, 4){\psframe(2,2)\rput(1,1){$#1$}}
   \put(1.5, 7.5){\psline{->}(0,-1.5)}
   \put(4.5, 6){\psline(0,2)}
   \put(4.5, 8){\psline{->}(-2.5,0)}
   \put(1.5,8){\circle{1}}
   \put(1.5,8){\piu}
   \put(1.5,8){#4}
   \put(1.5, 9.5){\psline{->}(0,-1)}
   \put(1.5,10){\circle{1}}
   \put(1.5,10){\piu}
   \put(1.5,10){#5}
   \put(0, 10){\psline{->}(1,0)}
   \put(6, 10){\psline{->}(-4,0)}
   \end{pspicture}}

\newcommand{\blosuintpium}[6]{\blosubase{\midsuintpium{#1}{#2}{#4}{#5}{#6}}{#3}}
\newcommand{\blovisuintpium}[6]{\blovisubase{\midsuintpium{#1}{#2}{#4}{#5}{#6}}{#3}}

\newcommand{\blosuattdue}[6]{\begin{pspicture}(8,10)
   \put(2.5,10.5){\makebox(0,0)[b]{$#5$}}
   \put(3.5,9){\makebox(0,0)[bl]{$#4$}}
   \put(0,0){\tratteggio}
   \put(0,0){\psline{->}(2,0)}
   \put(8,0){\psline{->}(-5,0)}
   \put(2.5,0){\circle{1}}
   \put(2.5,0){\piu}
   \put(2.5,0){#6}
   \put(2.5,0.5){\psline{->}(0,1.5)}
   \put(2.5, 1){\punto}
   \put(2.5, 1){\psline(3.5,0)}
   \put(  6, 1){\psline{->}(0,3)}
   \put( 4.5, 4){\psframe(3,2)\rput(1.5,1){$#3$}}
   \put(  1.5, 2){\psframe(2,2)\rput(1,1){$\frac{1}{s}$}}
   \put(2.5, 4){\psline{->}(0,1.5)}
   \put(2.25 ,4.75){\makebox(0,0)[r]{$#1$}}
   \put( 1.5, 5.5){\psframe(2,2)\rput(1,1){$#2$}}
   \put(2.5, 7.5){\psline{->}(0,0.75)}
   \put(  6, 6){\psline(0,2.75)}
   \put(  6, 8.75){\psline{->}(-3,0)}
   \put(2.5,8.75){\circle{1}}
   \put(2.5,8.75){\piu}
   \put(2.5,  9.25){\psline{->}(0,0.75)}
   \put(2.5, 10){\psline{->}(-2.5,0)}
   \put(2.5, 10){\psline{->}(5.5,0)}
   \end{pspicture}}

\newcommand{\blosuatttre}[6]{\begin{pspicture}(8,10)
   \put(2.5,10.5){\makebox(0,0)[b]{$#5$}}
   \put(3.5,9){\makebox(0,0)[bl]{$#4$}}
   \put(0,0){\tratteggio}
   \put(0,0){\psline{->}(2,0)}
   \put(8,0){\psline{->}(-5,0)}
   \put(2.5,0){\circle{1}}
   \put(2.5,0){\piu}
   \put(2.5,0){#6}
   \put(2.5,0.5){\psline{->}(0,1.5)}
   \put(2.5, 1){\punto}
   \put(2.5, 1){\psline(3.5,0)}
   \put(  6, 1){\psline{->}(0,2.5)}
   \put( 4.5,3.5){\psframe(3,2)\rput(1.5,1){$#3$}}
   \put(  1.5, 2){\psframe(2,2)\rput(1,1){$\frac{1}{s}$}}
   \put(2.5, 4){\psline{->}(0,1.5)}
   \put(2.25 ,4.75){\makebox(0,0)[r]{$#1$}}
   \put(1, 5.5){\psframe(3,2)\rput(1.5,1){$#2$}}
   \put(2.5, 7.5){\psline{->}(0,0.75)}
   \put(  6, 5.5){\psline(0,3.25)}
   \put(  6, 8.75){\psline{->}(-3,0)}
   \put(2.5,8.75){\circle{1}}
   \put(2.5,8.75){\piu}
   \put(2.5,  9.25){\psline{->}(0,0.75)}
   \put(2.5, 10){\psline{->}(-2.5,0)}
   \put(2.5, 10){\psline{->}(5.5,0)}
   \end{pspicture}}

\newcommand{\blogiuatttre}[6]{\begin{pspicture}(8,10)
   \put(2.5,-0.5){\makebox(0,0)[t]{$#5$}}
   \put(3.5,1){\makebox(0,0)[tl]{$#4$}}
   \put(0,0){\tratteggio}
   \put(0,10){\psline{->}(2,0)}
   \put(8,10){\psline{->}(-5,0)}
   \put(2.5,10){\circle{1}}
   \put(2.5,10){\piu}
   \put(2.5,10){#6}
   \put(2.5,9.5){\psline{->}(0,-1.5)}
   \put(2.5, 9){\punto}
   \put(2.5, 9){\psline(3.5,0)}
   \put(  6, 9){\psline{->}(0,-2.5)}
   \put( 4.5,4.5){\psframe(3,2)\rput(1.5,1){$#3$}}
   \put(  1.5, 6){\psframe(2,2)\rput(1,1){$\frac{1}{s}$}}
   \put(2.5, 6){\psline{->}(0,-1.5)}
   \put(2.25 ,5.25){\makebox(0,0)[r]{$#1$}}
   \put(1, 2.5){\psframe(3,2)\rput(1.5,1){$#2$}}
   \put(2.5, 2.5){\psline{->}(0,-0.75)}
   \put(  6, 4.5){\psline(0,-3.25)}
   \put(  6, 1.25){\psline{->}(-3,0)}
   \put(2.5,1.25){\circle{1}}
   \put(2.5,1.25){\piu}
   \put(2.5,  0.75){\psline{->}(0,-0.75)}
   \put(2.5,  0){\psline{->}(-2.5,0)}
   \put(2.5,  0){\psline{->}(5.5,0)}
   \end{pspicture}}

\newcommand{\tratteggiogiu}[1]{\begin{pspicture}(0,10)
   \put(0,-1.125){\psline[linestyle=dashed,linewidth=0.2pt](0,0)(0,-#1)}
   \end{pspicture}}

\newcommand{\scrivigiu}[3]{\begin{pspicture}(0,10)
   \put(0,0){\tratteggiogiu{#2}}
   \put(0,-1){
   \put(0,-#2){\begin{minipage}[t]{#1\unitlength} \centering #3 \end{minipage}}
   }
   \end{pspicture}}

\newcommand{\tratteggiosu}[1]{\begin{pspicture}(0,10)
   \put(0, 11.125){\psline[linestyle=dashed,linewidth=0.2pt](0,0)(0,#1)}
   \end{pspicture}}

\newcommand{\scrivisu}[3]{\begin{pspicture}(0,10)
   \put(0,0){\tratteggiosu{#2}}
   \put(0,10){\put(0,#2){\begin{minipage}[b]{#1\unitlength} \centering #3 \end{minipage}}}
   \end{pspicture}}

\newcommand{\nrsezsu}[2]{\begin{pspicture}(0,10)
   \put(0,#2){\small\cput[linewidth=0.3pt,framesep=1pt](0,11.75){\tiny #1}}
  \end{pspicture}}

\newcommand{\nrsezxsu}[3]{\begin{pspicture}(0,10)
   \put(0,#3){\pscircle[linewidth=0.3pt](0,11.75){#1}\rput(0,11.75){\tiny #2}}
  \end{pspicture}}

\newcommand{\nrsezgiu}[2]{\begin{pspicture}(0,10)
   \put(0,#2){\small\cput[linewidth=0.3pt,framesep=1pt](0,-2){\tiny #1}}
  \end{pspicture}}

\newcommand{\nrsezxgiu}[3]{\begin{pspicture}(0,10)
   \put(0,#3){\pscircle[linewidth=0.3pt](0,-2){#1}\rput(0,-2){\tiny #2}}
  \end{pspicture}}

\newcommand{\blogiogiu}[3]{\begin{pspicture}(8,10)
   \put(0,0){\tratteggio}
   \put(2,1){\psframe(4,4)\rput(2,2){}}
   \put(3,6){\psframe(2,2)\rput(1,1){{\bf $1/s$}}}
   \put(2.5,3){\psline(3,0)}
   \put(4,1.5){\psline(0,3)}
   \put(4.75,3){\psline(1,1.5)}
   \put(3.25,3){\psline(-1,-1.5)}
   \put(4,10){\piu\dst}
   \put(0,10){\psline{->}(3.5,0)}
   \put(8,10){\psline{->}(-3.5,0)}
   \put(4,0){\psline{->}(4,0)}
   \put(4,0){\psline{->}(-4,0)}
   \put(4,0){\punto}
   \put(4,1){\psline{->}(0,-1)}
   \put(4,6){\psline{->}(0,-1)}
   \put(4,9.5){\psline{->}(0,-1.5)}
   \put(3,-2){\makebox(2,2){{\bf $#3$}}}
   \put(4,1.5){\makebox(1.5,1.5){{\bf $#1$}}}
   \put(4,3.5){\makebox(1.5,1.5){{\bf $#2$}}}
\end{pspicture}}

\newcommand{\flowrightgiu}[1]{\begin{pspicture}(0,10)
   \psline[linewidth=0.4pt]{->}(-0.5,-1)(0.5,-1)
   \put(0,-1.5){\makebox(0,0)[t]{$#1$}}
\end{pspicture}}

\newcommand{\flowleftgiu}[1]{\begin{pspicture}(0,10)
   \Psline[linewidth=0.4pt]{<-}(-0.5,-1)(0.5,-1)
   \put(0,-1.5){\makebox(0,0)[t]{$#1$}}
\end{pspicture}}

\newcommand{\flowrightsu}[1]{\begin{pspicture}(0,10)
   \psline[linewidth=0.4pt]{->}(-0.5,11.5)(0.5,11.5)
   \put(0,11.75){\makebox(0,0)[b]{$#1$}}
\end{pspicture}}

\newcommand{\flowleftsu}[1]{\begin{pspicture}(0,10)
   \psline[linewidth=0.4pt]{<-}(-0.5,11.5)(0.5,11.5)
   \put(0,11.75){\makebox(0,0)[b]{$#1$}}
\end{pspicture}}

\newcommand{\flowright}[1]{\begin{pspicture}(0,12)
   \put(0,-1){\makebox(0,0)[t]{$\Rightarrow$}}
\end{pspicture}}

\newcommand{\flowleft}[1]{\begin{pspicture}(0,12)
   \put(0,-1){\makebox(0,0)[t]{$\leftarrow$}}
\end{pspicture}}

\newcommand{\downdot}[1]{\begin{pspicture}(0,12)
   \put(0,-1){\makebox(0,0)[t]{$#1$}}
\end{pspicture}}

\newcommand{\updot}[1]{\begin{pspicture}(0,12)
   \put(0,11){\makebox(0,0)[b]{$#1$}}
\end{pspicture}}

\newcommand{\attdstsnt}[1]{\begin{pspicture}(0,1)
   \put(2.5,-1.5){\psframe(3,3)\rput(1.5,1.5){#1}}
   \put(0,0){\psline{->}(2.5,0)}
   \put(8,0){\psline{->}(-2.5,0)}
\end{pspicture}}

\newcommand{\midgiuvipiu}[2]{\begin{pspicture}(0,1)
   \put(   0, 10){\psline{->}(0,-3.5)}
   \put(-1.5,3.5){\psframe(3,3)\rput(1.5,1.5){$#1$}}
   \put(   0, 3.5){\psline{->}(0,-3)}
   \put(   0,1.5){\makebox(0,0)[l]{$#2$}}
\end{pspicture}}

\newcommand{\midsuvipiu}[2]{\begin{pspicture}(0,1)
   \put(   0,  0){\psline{->}(0,3.5)}
   \put(-1.5,3.5){\psframe(3,3)\rput(1.5,1.5){$#1$}}
   \put(   0,6.5){\psline{->}(0,3)}
   \put(   0,8.5){\makebox(0,0)[l]{$#2$}}
\end{pspicture}}

\newcommand{\midsutrelarge}[2]{\begin{pspicture}(0,1)
   \put(-2,3.5){\psframe(4,3)\rput(2,1.5){$#1$}}
   \put(0,6.5){\psline{->}(0,3.5)}
   \put(0,  0){\psline{->}(0,3.5)}
   \put(0,10.5){\makebox(0,0)[b]{$#2$}}
\end{pspicture}}

\newcommand{\midsuintpium}[5]{\begin{pspicture}(0,1)
   \put(    0,0.5){\psline{->}(0,0.5)}
   \put(-1.25,  1){\psframe(2.5,2.5)\rput(1.25,1.25){{\bf $\frac{1}{s}$}}}
   \put(    0,3.5){\psline{->}(0,1)}
   \put(    0,  5){\piu}
   \put(    0,5.5){\psline{->}(0,0.5)}
   \put(    2, 5){\psline{->}(-1.5,0)}
   \put(    2, 5){\punto}
   \put(   -1.5,  6){\psframe(3,3)\rput(1.5,1.5){$#1$}}
   \put(    0,  9){\psline{->}(0,1)}
   \put(0,10.5){\makebox(0,0)[b]{$#2$}}
   \put(  2.25,5){\makebox(0,0)[l]{$#3$}}
   \put(    0, 5){\makebox(0,0)[c]{$#4$}}
   \put(-0.25, 3.75){\makebox(0,0)[rb]{$#5$}}
\end{pspicture}}

\newcommand{\bloseiretaor}[2]{\begin{pspicture}(8,10)
   \put(0,0){\tratteggio}
   \put(3,10){\psline{->}(-3,0)}
   \put(3, 0){\midsutre{#1}{#2}}
   \put(0, 0){\psline{->}(3,0)}
\end{pspicture}}

\newcommand{\bloseiretaorlarge}[2]{\begin{pspicture}(8,10)
   \put(0,0){\tratteggio}
   \put(3,10){\psline{->}(-3,0)}
   \put(3, 0){\midsutrelarge{#1}{#2}}
   \put(0, 0){\psline{->}(3,0)}
\end{pspicture}}

\newcommand{\bloseiretor}[2]{\begin{pspicture}(8,10)
   \put(0,0){\bloseiintgiu{#1}{#2}}
   \put(3,0){\psline{->}(-3,0)}
\end{pspicture}}

\newcommand{\bloseigiuao}[2]{\begin{pspicture}(6,10)
   \put(6,10){\psline{->}(-3,0)}
   \put(3,0){\midgiutre{#1}{#2}}
   \put(3,0){\psline{->}(3,0)}
\end{pspicture}}

\newcommand{\bloseiintgiu}[2]{\begin{pspicture}(8,10)
   \put(0,0){\tratteggio}
   \put(0,10){\psline{->}(3,0)}
   \put(3,0){\midgiutre{#1}{#2}}
\end{pspicture}}

\newcommand{\blogiuretor}[3]{\begin{pspicture}(8,10)
   \put(0,0){\tratteggio}
   \put(4.5,8.0){\makebox(0,0)[lb]{$#2$}}
   \put(0,10){\psline(4,0)}
   \put(4,10){\psline{->}(0,-3)}
   \put(2,3){\psframe(4,4)\rput(2,2){$#1$}}
   \put(4,3){\psline(0,-3)}
   \put(4,0){\psline{->}(-4,0)}
   \put(4.5,2.0){\makebox(0,0)[lt]{$#3$}}
\end{pspicture}}

\newcommand{\blogiuretao}[3]{\begin{pspicture}(0,10)
   \put(8,0){\tratteggio}
   \put(3.5,8.0){\makebox(0,0)[rb]{$#2$}}
   \put(8,10){\psline(-4,0)}
   \put(4,10){\psline{->}(0,-3)}
   \put(2,3){\psframe(4,4)\rput(2,2){$#1$}}
   \put(4,3){\psline(0,-3)}
   \put(4,0){\psline{->}(4,0)}
   \put(3.5,2.0){\makebox(0,0)[rt]{$#3$}}
\end{pspicture}}

\newcommand{\blosuretor}[3]{\begin{pspicture}(8,10)
   \put(8,0){\tratteggio}
   \put(3.5,8.0){\makebox(0,0)[rb]{$#2$}}
   \put(4,10){\psline{->}(4,0)}
   \put(4,7){\psline(0,3)}
   \put(2,3){\psframe(4,4)\rput(2,2){$#1$}}
   \put(4,0){\psline{->}(0,3)}
   \put(8,0){\psline(-4,0)}
   \put(3.5,2.0){\makebox(0,0)[rt]{$#3$}}
\end{pspicture}}

\newcommand{\blosuretao}[3]{\begin{pspicture}(8,10)
   \put(0,0){\tratteggio}
   \put(4.5,8.0){\makebox(0,0)[lb]{$#2$}}
   \put(4,10){\psline{->}(-4,0)}
   \put(4,7){\psline(0,3)}
   \put(2,3){\psframe(4,4)\rput(2,2){$#1$}}
   \put(4,0){\psline{->}(0,3)}
   \put(0,0){\psline(4,0)}
   \put(4.5,2.0){\makebox(0,0)[lt]{$#3$}}
\end{pspicture}}

\newcommand{\blovigiuretintmor}[2]{\begin{pspicture}(6,10)
   \put(0,0){\tratteggio}
   \put(    0,10){\psline(3,0)}
   \put(3,0){
   \put(    0,10){\psline{->}(0,-1.5)}
   \put(-1.25,  6){\psframe(2.5,2.5)\rput(1.25,1.25){{\bf $\frac{1}{s}$}}}
   \put(    0,  6){\psline{->}(0,-1)}
   \put(   -2,  1){\psframe(4,4)\rput(2,2){$#1$}}
   \put(    0,  1){\psline(0,-1)}
   \put(0,-0.5  ){\makebox(0,0)[t]{$#2$}}
   }
   \put(    3,0){\psline{->}(-3,0)}
\end{pspicture}}

\newcommand{\blovigiuretor}[3]{\begin{pspicture}(8,10)
   \put(0,0){\tratteggio}
   \put(3.5,8.0){\makebox(0,0)[lb]{$#2$}}
   \put(0,10){\psline(3,0)}
   \put(3,10){\psline{->}(0,-3)}
   \put(1,3){\psframe(4,4)\rput(2,2){$#1$}}
   \put(3,3){\psline(0,-3)}
   \put(3,0){\psline{->}(-3,0)}
   \put(3.5,2.0){\makebox(0,0)[lt]{$#3$}}
\end{pspicture}}

\newcommand{\blovigiuretao}[3]{\begin{pspicture}(6,10)
   \put(6,0){\tratteggio}
   \put(2.5,8.0){\makebox(0,0)[rb]{$#2$}}
   \put(6,10){\psline(-3,0)}
   \put(3,10){\psline{->}(0,-3)}
   \put(1,3){\psframe(4,4)\rput(2,2){$#1$}}
   \put(3,3){\psline(0,-3)}
   \put(3,0){\psline{->}(3,0)}
   \put(2.5,2.0){\makebox(0,0)[rt]{$#3$}}
\end{pspicture}}

\newcommand{\blovisuretor}[3]{\begin{pspicture}(6,10)
   \put(6,0){\tratteggio}
   \put(2.5,8.0){\makebox(0,0)[rb]{$#2$}}
   \put(3,10){\psline{->}(3,0)}
   \put(3,7){\psline(0,3)}
   \put(1,3){\psframe(4,4)\rput(2,2){$#1$}}
   \put(3,0){\psline{->}(0,3)}
   \put(6,0){\psline(-3,0)}
   \put(2.5,2.0){\makebox(0,0)[rt]{$#3$}}
\end{pspicture}}

\newcommand{\blovisuretao}[3]{\begin{pspicture}(6,10)
   \put(0,0){\tratteggio}
   \put(3.5,8.0){\makebox(0,0)[lb]{$#2$}}
   \put(3,10){\psline{->}(-3,0)}
   \put(3,7){\psline(0,3)}
   \put(1,3){\psframe(4,4)\rput(2,2){$#1$}}
   \put(3,0){\psline{->}(0,3)}
   \put(0,0){\psline(3,0)}
   \put(3.5,2.0){\makebox(0,0)[lt]{$#3$}}
\end{pspicture}}

\newcommand{\bloivgiuretor}[3]{\begin{pspicture}(4,10)
   \put(0,0){\tratteggio}
   \put(2.5,8.0){\makebox(0,0)[lb]{$#2$}}
   \put(0,10){\psline(2,0)}
   \put(2,10){\psline{->}(0,-3.5)}
   \put(0.5,3.5){\psframe(3,3)\rput(1.5,1.5){$#1$}}
   \put(2,3.5){\psline(0,-3.5)}
   \put(2,0){\psline{->}(-2,0)}
   \put(2.5,2.0){\makebox(0,0)[lt]{$#3$}}
\end{pspicture}}

\newcommand{\bloivgiuretao}[3]{\begin{pspicture}(4,10)
   \put(4,0){\tratteggio}
   \put(1.5,8.0){\makebox(0,0)[rb]{$#2$}}
   \put(4,10){\psline(-2,0)}
   \put(2,10){\psline{->}(0,-3.5)}
   \put(0.5,3.5){\psframe(3,3)\rput(1.5,1.5){$#1$}}
   \put(2,3.5){\psline(0,-3.5)}
   \put(2,0){\psline{->}(2,0)}
   \put(1.5,2.0){\makebox(0,0)[rt]{$#3$}}
\end{pspicture}}

\newcommand{\bloivsuretor}[3]{\begin{pspicture}(4,10)
   \put(4,0){\tratteggio}
   \put(1.5,8.0){\makebox(0,0)[rb]{$#2$}}
   \put(2,10){\psline{->}(2,0)}
   \put(2,6.5){\psline(0,3.5)}
   \put(0.5,3.5){\psframe(3,3)\rput(1.5,1.5){$#1$}}
   \put(2,0){\psline{->}(0,3.5)}
   \put(4,0){\psline(-2,0)}
   \put(1.5,2.0){\makebox(0,0)[rt]{$#3$}}
\end{pspicture}}

\newcommand{\bloivsuretao}[3]{\begin{pspicture}(4,10)
   \put(0,0){\tratteggio}
   \put(2.5,8.0){\makebox(0,0)[lb]{$#2$}}
   \put(2,10){\psline{->}(-2,0)}
   \put(2,6.5){\psline(0,3.5)}
   \put(0.5,3.5){\psframe(3,3)\rput(1.5,1.5){$#1$}}
   \put(2,0){\psline{->}(0,3.5)}
   \put(0,0){\psline(2,0)}
   \put(2.5,2.0){\makebox(0,0)[lt]{$#3$}}
\end{pspicture}}

\newcommand{\addgiu}[3]{\begin{pspicture}(4,10)
   \put(0,0){\tratteggio}
   \put(0,0){\psline{->}(1.5,0)}
   \put(2,0){\piu}
   \put(2,0){\makebox(0,0){$#2$}}
   \put(2.5,-2){\makebox(0,0)[l]{$#3$}}
   \put(4,0){\psline{->}(-1.5,0)}
   \put(2,-0.5){\psline{->}(0,-1)}
   \put(2,10){\psline{->}(-2,0)}
   \put(2,10){\psline{->}(2,0)}
   \put(2,10){\punto}
   \put(2,12){\psline{->}(0,-2)}
   \put(2,12){\punto}
   \put(2.5,12){\makebox(0,0)[l]{$#1$}}
\end{pspicture}}

\newcommand{\addsu}[3]{\begin{pspicture}(4,10)
   \put(0,0){\tratteggio}
   \put(0,10){\psline{->}(1.5,0)}
   \put(2,10){\piu}
   \put(2,10){\makebox(0,0){$#2$}}
   \put(2.5,12){\makebox(0,0)[l]{$#3$}}
   \put(4,10){\psline{->}(-1.5,0)}
   \put(2,10.5){\psline{->}(0,1)}
   \put(2,0){\psline{->}(-2,0)}
   \put(2,0){\psline{->}(2,0)}
   \put(2,0){\punto}
   \put(2,-2){\psline{->}(0,2)}
   \put(2,-2){\punto}
   \put(2.5,-2){\makebox(0,0)[l]{$#1$}}
\end{pspicture}}

\newcommand{\blopiugiuor}[5]{\begin{pspicture}(8,10)
   \put(0, 0){\tratteggio}
   \put(0,10){\dotdstdst}
   \put(4,10){\psline{->}(0,-1.5)}
   \put(4,8){\piu}
   \put(6,8){\psline{->}(-1.5,0)}
   \put(6,8){\punto}
   \put(6.25,8){\makebox(0,0)[l]{$#4$}}
   \put(4,8){\makebox(0,0){$#5$}}
   \put(4,7.5){\psline{->}(0,-1)}
   \put(4,2.5){\psline{->}(0,-2)}
   \put(1.5,2.5){\psframe(5,4)\rput(2.5,2){$#1$}}
   \put(4.5,1.5){\makebox(0,0)[l]{$#2$}}
   \put(0, 0){\sumsntsnt{#3}}
\end{pspicture}}

\newcommand{\blopiugiuorlarge}[5]{\begin{pspicture}(8,10)
   \put(0, 0){\tratteggio}
   \put(0,10){\dotdstdst}
   \put(4,10){\psline{->}(0,-1.5)}
   \put(4,8){\piu}
   \put(6,8){\psline{->}(-1.5,0)}
   \put(6,8){\punto}
   \put(6.25,8){\makebox(0,0)[l]{$#4$}}
   \put(4,8){\makebox(0,0){$#5$}}
   \put(4,7.5){\psline{->}(0,-1)}
   \put(4,2.5){\psline{->}(0,-2)}
   \put(0.5,2.5){\psframe(7,4)\rput(3.5,2){$#1$}}
   \put(4.5,1.5){\makebox(0,0)[l]{$#2$}}
   \put(0, 0){\sumsntsnt{#3}}
\end{pspicture}}

\newcommand{\blopiugiuaor}[5]{\begin{pspicture}(8,10)
   \put(0, 0){\tratteggio}
   \put(0,10){\dotsntsnt}
   \put(4,10){\psline{->}(0,-1.5)}
   \put(4,8){\piu}
   \put(6,8){\psline{->}(-1.5,0)}
   \put(6,8){\punto}
   \put(6.25,8){\makebox(0,0)[l]{$#4$}}
   \put(4,8){\makebox(0,0){$#5$}}
   \put(4,7.5){\psline{->}(0,-1)}
   \put(4,2.5){\psline{->}(0,-2)}
   \put(1.5,2.5){\psframe(5,4)\rput(2.5,2){$#1$}}
   \put(4.5,1.5){\makebox(0,0)[l]{$#2$}}
   \put(0, 0){\sumdstdst{#3}}
\end{pspicture}}

\newcommand{\blopiugiuaorlarge}[5]{\begin{pspicture}(8,10)
   \put(0, 0){\tratteggio}
   \put(0,10){\dotsntsnt}
   \put(4,10){\psline{->}(0,-1.5)}
   \put(4,8){\piu}
   \put(6,8){\psline{->}(-1.5,0)}
   \put(6,8){\punto}
   \put(6.25,8){\makebox(0,0)[l]{$#4$}}
   \put(4,8){\makebox(0,0){$#5$}}
   \put(4,7.5){\psline{->}(0,-1)}
   \put(4,2.5){\psline{->}(0,-2)}
   \put(0.5,2.5){\psframe(7,4)\rput(3.5,2){$#1$}}
   \put(4.5,1.5){\makebox(0,0)[l]{$#2$}}
   \put(0, 0){\sumdstdst{#3}}
\end{pspicture}}

\newcommand{\midgiutre}[2]{\begin{pspicture}(0,1)
   \put(-1.5,3.5){\psframe(3,3)\rput(1.5,1.5){$#1$}}
   \put(0,10){\psline{->}(0,-3.5)}
   \put(0,3.5){\psline{->}(0,-3.5)}
   \put(0,-0.5){\makebox(0,0)[t]{$#2$}}
\end{pspicture}}

\newcommand{\midgiutreadd}[2]{\midgiuaddbase{\midgiutre{#1}{}}{#2}}

\newcommand{\midsutre}[2]{\begin{pspicture}(0,1)
   \put(-1.5,3.5){\psframe(3,3)\rput(1.5,1.5){$#1$}}
   \put(0,6.5){\psline{->}(0,3.5)}
   \put(0,  0){\psline{->}(0,3.5)}
   \put(0,10.5){\makebox(0,0)[b]{$#2$}}
\end{pspicture}}

\newcommand{\midsutreadd}[2]{\midsuaddbase{\midsutre{#1}{}}{#2}}

\newcommand{\blokcindst}[1]{\begin{pspicture}(0,1)
   \put(1,-2){\psframe(8,4)\rput(4,2){$#1$}}
   \put(0,0){\psline{->}(1,0)}
   \put(9,0){\psline{->}(1,0)}
\end{pspicture}}

\newcommand{\blokcinsnt}[1]{\begin{pspicture}(0,1)
   \put(1,-2){\psframe(8,4)\rput(4,2){$#1$}}
   \put(1,0){\psline{->}(-1,0)}
   \put(10,0){\psline{->}(-1,0)}
\end{pspicture}}

\newcommand{\blokaocin}[4]{\begin{pspicture}(8,10)
   \put(0, 0){\tratteggio}
   \put(0, 0){\blokcindst{#2}}
   \put(0,10){\blokcinsnt{#1}}
   \put(0,10.5){\makebox(0,0)[br]{$#3$}}
   \put(10,-0.5){\makebox(0,0)[tl]{$#4$}}
\end{pspicture}}

\newcommand{\blokocin}[4]{\begin{pspicture}(6,10)
   \put(0,0){\tratteggio}
   \put(0,10){\blokcindst{#1}}
   \put(0, 0){\blokcinsnt{#2}}
   \put(10,10.5){\makebox(0,0)[bl]{$#3$}}
   \put(0,-0.5){\makebox(0,0)[tr]{$#4$}}
\end{pspicture}}

\newcommand{\piudst}{\begin{pspicture}(1,1)
 \put(0,0){\piu}
 \put(0,0){\dst}
\end{pspicture}}

\newcommand{\midgiupiu}[2]{\begin{pspicture}(0,1)
   \put( 0, 10){\psline{->}(0,-3)}
   \put(-2,  3){\psframe(4,4)\rput(2,2){$#1$}}
   \put( 0,  3){\psline{->}(0,-2.5)}
   \put( 0,1.5){\makebox(0,0)[l]{$#2$}}
\end{pspicture}}

\newcommand{\addsuaoriv}[2]{\begin{pspicture}(4,10)
   \put(0,0){\tratteggio}
   \put(1.5,10){\psline{->}(-1.5,0)}
   \put(2,10){\piu}
   \put(2,11.5){\punto}
   \put(2,10){\makebox(0,0){$#2$}}
   \put(2.5,12){\makebox(0,0)[l]{$#1$}}
   \put(4,10){\psline{->}(-1.5,0)}
   \put(2,11.5){\psline{->}(0,-1)}
   \put(0,0){\psline{->}(4,0)}
\end{pspicture}}

\newcommand{\blosuintmp}[5]{\blosubase{\midsuintmp{#1}{#2}{#3}{#4}}{#5}}

\newcommand{\midsuintmp}[4]{\begin{pspicture}(0,1)
   \put(0,0.5){\psline{->}(0,1)}
   \put(-1,1.5){\psframe(2,2)\rput(1,1){{\bf $\frac{1}{s}$}}}
   \put(0,3.5){\psline{->}(0,1)}
   \put(0,5){\piu}
   \put(0,5){\makebox(0,0){$#4$}}
   \put(2,5){\psline{->}(-1.5,0)}
   \put(2,5){\punto}
   \put(2.25,5){\makebox(0,0)[l]{$#2$}}
   \put(0,10.5){\makebox(0,0)[b]{$#3$}}
   \put(0,5.5){\psline{->}(0,0.75)}
   \put(0,9.25){\psline{->}(0,0.75)}
   \put(-2,6.25){\psframe(4,3)\rput(2,1.5){$#1$}}
\end{pspicture}}

\newcommand{\blokodueuno}[8]{\begin{pspicture}(10,10)
   \put(0,0){\tratteggio}
   \put(0,2.5){\tratteggio}
    \rput(0,10){\punto}
   \psline[linewidth=0.4pt]{->}(1.5,0)(0,0)
   \psline[linewidth=0.4pt]{->}(1.5,2)(0,2)
   \psframe[linewidth=0.4pt](1.5,-1)(6.5,3)
   \put(3.95,1){\makebox(0,0){$#7$}}
   \put(-0.75,3){\makebox(0,0){$#3$}}
   \put(-0.75,-1){\makebox(0,0){$#4$}}
   \psline[linewidth=0.4pt]{->}(8,1)(6.5,1)
   \put(8.65,1.85){\makebox(0,0){$#6$}}
   \rput(8,1){\punto}
   \psline[linewidth=0.4pt]{->}(0,10)(1.5,10)
   \psline[linewidth=0.4pt]{->}(0,12)(1.5,12)
   \psframe[linewidth=0.4pt](1.5,9)(6.6,13)
   \put(3.9,10.8){\makebox(0,0){$#8$}}
   \put(-0.75,9){\makebox(0,0){$#2$}}
   \rput(0,12){\punto}
   \put(-0.75,12.5){\makebox(0,0){$#1$}}
   \psline[linewidth=0.4pt]{->}(6.5,11)(8,11)
   \put(8.65,11.95){\makebox(0,0){$#5$}}
\end{pspicture}}

\newcommand{\POGgiularge}[4]{\begin{pspicture}(8,10)
   \put(0, 0){\blogiularge{#1}{#3}{#4}}
   \rput[lb](3.25,8.25){$#2$}
\end{pspicture}}

\newcommand{\POGvigiuaor}[4]{\begin{pspicture}(6,10)
   \put(0, 0){\blovigiuaor{#1}{#3}{#4}}
   \rput[lb](3.25,8.25){$#2$}
\end{pspicture}}

\newcommand{\POGvigiuor}[4]{\begin{pspicture}(6,10)
   \put(0, 0){\blovigiuor{#1}{#3}{#4}}
   \rput[lb](3.25,8.25){$#2$}
\end{pspicture}}

\newcommand{\POGvisu}[4]{\begin{pspicture}(6,10)
   \put(0, 0){\blovisu{#1}{}{#4}}
   \rput[lb](3.25,8.25){$#2$}
   \rput[lb](3.25,1.25){$#3$}
\end{pspicture}}

\newcommand{\POGvisuaor}[4]{\begin{pspicture}(6,10)
   \put(0, 0){\blovisuaor{#1}{#3}{#4}}
   \rput[lb](3.25,0.75){$#2$}
\end{pspicture}}

\newcommand{\POGvisuor}[4]{\begin{pspicture}(6,10)
   \put(0, 0){\blovisuor{#1}{#3}{#4}}
   \rput[lb](3.25,0.75){$#2$}
\end{pspicture}}

\newcommand{\POGvigiu}[4]{\begin{pspicture}(6,10)
   \put(0, 0){\blovigiu{#1}{}{#4}}
   \rput[lb](3.25,8.00){$#2$}
   \rput[lb](3.25,1.00){$#3$}
\end{pspicture}}

\newcommand{\POGsugiu}[2]{\begin{pspicture}(0,10)
   \rput[b](0,10.95){$#1$}
   \rput[t](0,-0.95){$#2$}
\end{pspicture}}

\newcommand{\POGsuatt}[9]{\begin{pspicture}(8,10)
   \put(0, 0){\blosuatt{#1}{#2}{#4}{#8}{#9}}
   \rput[rt](2.25,3.25){$#3$}
   \rput[l](6.25,2.75){$#5$}
   \rput[l](6.25,7.25){$#6$}
   \rput[lb](2.75,0.75){$#7$}
\end{pspicture}}

\newcommand{\POGgiuattor}[8]{\begin{pspicture}(8,10)
   \put(0, 0){\blogiuattror{#1}{#2}{#4}{#7}{#8}}
   \rput[rb](2.25,6.75){$#3$}
   \rput[l](6.25,7.25){$#5$}
   \rput[l](6.25,2.75){$#6$}
\end{pspicture}}

\newcommand{\POGgiuattaor}[8]{\begin{pspicture}(8,10)
   \put(0, 0){\blogiuattraor{#1}{#2}{#4}{#7}{#8}}
   \rput[rb](2.25,6.75){$#3$}
   \rput[l](6.25,7.25){$#5$}
   \rput[l](6.25,2.75){$#6$}
\end{pspicture}}

\newcommand{\POGgiuatt}[9]{\begin{pspicture}(8,10)
   \put(0, 0){\blogiuatt{#1}{#2}{#4}{#8}{#9}}
   \rput[rb](2.25,6.75){$#3$}
   \rput[l](6.25,7.25){$#5$}
   \rput[l](6.25,2.75){$#6$}
   \rput[lt](2.75,9.25){$#7$}
\end{pspicture}}

\newcommand{\POGgiupar}[9]{\blogiupar{#1}{#2}{#3}{#4}{#5}{#6}{#7}{#8}{#9}}
\newcommand{\POGgiuorpar}[9]{\blogiuorpar{#1}{#2}{#3}{#4}{#5}{#6}{#7}{#8}{#9}}
\newcommand{\POGgiuaorpar}[9]{\blogiuaorpar{#1}{#2}{#3}{#4}{#5}{#6}{#7}{#8}{#9}}
\newcommand{\POGgiufdk}[9]{\blogiufdk{#1}{#2}{#3}{#4}{#5}{#6}{#7}{#8}{#9}}
\newcommand{\POGgiuorfdk}[9]{\blogiuorfdk{#1}{#2}{#3}{#4}{#5}{#6}{#7}{#8}{#9}}
\newcommand{\POGgiuaorfdk}[9]{\blogiuaorfdk{#1}{#2}{#3}{#4}{#5}{#6}{#7}{#8}{#9}}
\newcommand{\POGsupar}[9]{\blosupar{#1}{#2}{#3}{#4}{#5}{#6}{#7}{#8}{#9}}
\newcommand{\POGsuorpar}[9]{\blosuorpar{#1}{#2}{#3}{#4}{#5}{#6}{#7}{#8}{#9}}
\newcommand{\POGsuaorpar}[9]{\blosuaorpar{#1}{#2}{#3}{#4}{#5}{#6}{#7}{#8}{#9}}
\newcommand{\POGsufdk}[9]{\blosufdk{#1}{#2}{#3}{#4}{#5}{#6}{#7}{#8}{#9}}
\newcommand{\POGsuorfdk}[9]{\blosuorfdk{#1}{#2}{#3}{#4}{#5}{#6}{#7}{#8}{#9}}
\newcommand{\POGsuaorfdk}[9]{\blosuaorfdk{#1}{#2}{#3}{#4}{#5}{#6}{#7}{#8}{#9}}

\newcommand{\bloivparingiubase}[9]{\begin{pspicture}(4,10)
   \put(0,0){\tratteggio}
   \put(0,10){\psline{#1}(1.5,0)}
   \put(1.5,10){\punto}
   \put(1.5,10){\psline{#2}(2,8)}
   \put(1.5,10){\psline{#3}(2,-8)}
   \psline(3.5,-8)(4,-8)
   \psline(3.5,8)(4,8)
   \psline(3.5,2)(4,2)
   \psline(3.5,18)(4,18)
   \put(1.5,0){
   \put(0,0){\pscircle[linewidth=0.4pt]{0.5}}
   \psline{#5}(0.5;58)(2,8)
   \psline{#6}(0.5;-58)(2,-8)
   \psline(0.5; 0)(0;0)
   \psline(0.5; 120)(0;0)
   \psline(0.5; 240)(0;0)
   \psline{#4}(-1.5,0)(-0.5,0)
    #9
    }
    \put(4,8){$#7$}
    \put(4,-8){$#8$}
\end{pspicture}}

\newcommand{\bloivparingiulab}[4]{\begin{pspicture}(0,10)
   \put(1.25,10.25){\makebox(0,0)[rb]{$#1$}}
   \put(0.75,0.25){\makebox(0,0)[rb]{$#2$}}
   \put(2.25,3){\makebox(0,0)[rb]{$#3$}}
   \put(2.25,-3){\makebox(0,0)[rt]{$#4$}}
\end{pspicture}}

\newcommand{\bloivparingiu}[7]{\bloivparingiulab{#4}{#5}{#6}{#7}\bloivparingiubase{->}{->}{->}{<-}{<-}{<-}{#1}{#2}{#3}}
\newcommand{\bloivparingiuuno}[7]{\bloivparingiulab{#4}{#5}{#6}{#7}\bloivparingiubase{<-}{<-}{->}{->}{->}{<-}{#1}{#2}{#3}}
\newcommand{\bloivparingiudue}[7]{\bloivparingiulab{#4}{#5}{#6}{#7}\bloivparingiubase{<-}{->}{<-}{->}{<-}{->}{#1}{#2}{#3}}

\newcommand{\bloivparinsubase}[9]{\begin{pspicture}(4,10)
   \put(0,0){\tratteggio}
   \put(0,0){\psline{#1}(1.5,0)}
   \put(1.5,0){\punto}
   \put(1.5,0){\psline{#2}(2,8)}
   \put(1.5,0){\psline{#3}(2,-8)}
   \psline(3.5,-8)(4,-8)
   \psline(3.5,8)(4,8)
   \psline(3.5,2)(4,2)
   \psline(3.5,18)(4,18)
   \put(1.5,10){
   \put(0,0){\pscircle[linewidth=0.4pt]{0.5}}
   \psline{#5}(0.5;58)(2,8)
   \psline{#6}(0.5;-58)(2,-8)
   \psline(0.5; 0)(0;0)
   \psline(0.5; 120)(0;0)
   \psline(0.5; 240)(0;0)
   \psline{#4}(-1.5,0)(-0.5,0)
    #9
    }
    \put(4,8){$#7$}
    \put(4,-8){$#8$}
\end{pspicture}}

\newcommand{\bloivparinsulab}[4]{\begin{pspicture}(0,10)
   \put(0.75,10.25){\makebox(0,0)[rb]{$#1$}}
   \put(1.25,0.25){\makebox(0,0)[rb]{$#2$}}
   \put(2.25,13){\makebox(0,0)[rb]{$#3$}}
   \put(2.25,7){\makebox(0,0)[rt]{$#4$}}
\end{pspicture}}

\newcommand{\bloivparinsu}[7]{\bloivparinsulab{#4}{#5}{#6}{#7}\bloivparinsubase{->}{->}{->}{<-}{<-}{<-}{#1}{#2}{#3}}
\newcommand{\bloivparinsuuno}[7]{\bloivparinsulab{#4}{#5}{#6}{#7}\bloivparinsubase{<-}{<-}{->}{->}{->}{<-}{#1}{#2}{#3}}
\newcommand{\bloivparinsudue}[7]{\bloivparinsulab{#4}{#5}{#6}{#7}\bloivparinsubase{<-}{->}{<-}{->}{<-}{->}{#1}{#2}{#3}}

\newcommand{\bloivparoutgiubase}[7]{\begin{pspicture}(4,10)
   \put(0,8){\tratteggio}
   \put(0,-8){\tratteggio}
   \put(2.5,10){\psline{#2}(1.5,0)}
   \put(2.5,10){\punto}
   \put(2.5,10){\psline{#1}(-2,8)}
   \put(2.5,10){\psline{#3}(-2,-8)}
   \put(2.5,0){\psline{#4}(-2,8)(0.5;120)}
   \put(2.5,0){\psline{#6}(-2,-8)(0.5;-120)}
   \psline(0.5,-8)(0,-8)
   \psline(0.5,8)(0,8)
   \psline(0.5,2)(0,2)
   \psline(0.5,18)(0,18)
   \put(2.5,0){
   \put(0,0){\pscircle[linewidth=0.4pt]{0.5}}
   \psline(0.5;180)(0;0)
   \psline(0.5; 60)(0;0)
   \psline(0.5;-60)(0;0)
   \psline{#5}(0.5,0)(1.5,0)
    #7
    }
\end{pspicture}}

\newcommand{\bloivparoutgiulab}[4]{\begin{pspicture}(0,10)
   \put(2.75,10.25){\makebox(0,0)[lb]{$#1$}}
   \put(3.25,0.25){\makebox(0,0)[lb]{$#2$}}
   \put(1.75,3){\makebox(0,0)[lb]{$#3$}}
   \put(1.75,-3){\makebox(0,0)[lt]{$#4$}}
\end{pspicture}}

\newcommand{\bloivparoutgiu}[5]{\bloivparoutgiulab{#2}{#3}{#4}{#5}\bloivparoutgiubase{->}{<-}{->}{->}{->}{->}{#1}}
\newcommand{\bloivparoutgiuuno}[5]{\bloivparoutgiulab{#2}{#3}{#4}{#5}\bloivparoutgiubase{<-}{->}{->}{<-}{<-}{->}{#1}}
\newcommand{\bloivparoutgiudue}[5]{\bloivparoutgiulab{#2}{#3}{#4}{#5}\bloivparoutgiubase{->}{->}{<-}{->}{<-}{<-}{#1}}

\newcommand{\bloivparoutsubase}[7]{\begin{pspicture}(4,10)
   \put(0,8){\tratteggio}
   \put(0,-8){\tratteggio}
   \put(2.5,0){\psline{#2}(1.5,0)}
   \put(2.5,0){\punto}
   \put(2.5,0){\psline{#1}(-2,8)}
   \put(2.5,0){\psline{#3}(-2,-8)}
   \put(2.5,10){\psline{#4}(-2,8)(0.5;120)}
   \put(2.5,10){\psline{#6}(-2,-8)(0.5;-120)}
   \psline(0.5,-8)(0,-8)
   \psline(0.5,8)(0,8)
   \psline(0.5,2)(0,2)
   \psline(0.5,18)(0,18)
   \put(2.5,10){
   \put(0,0){\pscircle[linewidth=0.4pt]{0.5}}
   \psline(0.5;180)(0;0)
   \psline(0.5; 60)(0;0)
   \psline(0.5;-60)(0;0)
   \psline{#5}(0.5,0)(1.5,0)
    #7
    }
\end{pspicture}}

\newcommand{\bloivparoutsulab}[4]{\begin{pspicture}(0,10)
   \put(3.25,10.25){\makebox(0,0)[lb]{$#1$}}
   \put(2.75,0.25){\makebox(0,0)[lb]{$#2$}}
   \put(1.75,13){\makebox(0,0)[lb]{$#3$}}
   \put(1.75,7){\makebox(0,0)[lt]{$#4$}}
\end{pspicture}}

\newcommand{\bloivparoutsu}[5]{\bloivparoutsulab{#2}{#3}{#4}{#5}\bloivparoutsubase{->}{<-}{->}{->}{->}{->}{#1}}
\newcommand{\bloivparoutsuuno}[5]{\bloivparoutsulab{#2}{#3}{#4}{#5}\bloivparoutsubase{<-}{->}{->}{<-}{<-}{->}{#1}}
\newcommand{\bloivparoutsudue}[5]{\bloivparoutsulab{#2}{#3}{#4}{#5}\bloivparoutsubase{->}{->}{<-}{->}{<-}{<-}{#1}}

\newcommand{\thdot}[1]{\begin{pspicture}(0,0)
 \pscircle*(0.25;#1){0.25}
\end{pspicture}}

\newcommand{\linex}[1]{\begin{pspicture}(#1,10)
 \psline(#1,0)
 \psline(0,10)(#1,10)
\end{pspicture}}

\newcommand{\bloivquaingiubase}[9]{\begin{pspicture}(4,10)
   \put(0,0){\tratteggio}
   \put(0,10){\psline{#1}(3.0,0)}
   \put(3.0,10){\punto}
   \psline{#2}(3,10)(4,10)
   \put(3,0){\psline{#3}(0,10)(0,0.3)(0.3;80)(0.3;60)(0.3;40)(0.3;20)(0.3;0)(0.3;-20)(0.3;-40)(0.3;-60)(0.3;-80)(0,-0.3)(0,-6)(1,-6)}
   \put(1.5,0){\piu}
   \put(1.5,0){#9}
   \psline{#5}(2,0)(4,0)
   \psline{#6}(1.5,-0.5)(1.5,-16)(4,-16)
   \psline{#4}(0,0)(1,0)
   \put(4,0){$#7$}
   \put(4,-16){$#8$}
\end{pspicture}}

\newcommand{\bloivquaingiulab}[4]{\begin{pspicture}(0,10)
   \put(2.0,10.25){\makebox(0,0)[b]{$#1$}}
   \put(0.75,0.25){\makebox(0,0)[rb]{$#2$}}
   \put(2.5,5){\makebox(0,0)[r]{$#3$}}
   \put(2.0,-11){\makebox(0,0)[l]{$#4$}}
\end{pspicture}}

\newcommand{\bloivquaingiu}[7]{\bloivquaingiulab{#4}{#5}{#6}{#7}\bloivquaingiubase{->}{->}{->}{<-}{<-}{<-}{#1}{#2}{#3}}
\newcommand{\bloivquaingiuuno}[7]{\bloivquaingiulab{#4}{#5}{#6}{#7}\bloivquaingiubase{<-}{<-}{->}{->}{->}{<-}{#1}{#2}{#3}}
\newcommand{\bloivquaingiudue}[7]{\bloivquaingiulab{#4}{#5}{#6}{#7}\bloivquaingiubase{<-}{->}{<-}{->}{<-}{->}{#1}{#2}{#3}}

\newcommand{\bloivquainsubase}[9]{\begin{pspicture}(4,10)
   \put(0,0){\tratteggio}
   \put(0,0){\psline{#1}(1,0)}
   \put(1,0){\punto}
   \psline{#2}(1,0)(4,0)
   \psline{#3}(1,0)(1,-16)(4,-16)
   \psline{#4}(0,10)(2,10)
   \put(2.5,10){\piu}
   \put(2.5,10){#9}
   \psline{#5}(3,10)(4,10)
   \put(2.5,0){\psline{#6}(0,9.5)(0,0.3)(0.3;80)(0.3;60)(0.3;40)(0.3;20)(0.3;0)(0.3;-20)(0.3;-40)(0.3;-60)(0.3;-80)(0,-0.3)(0,-6)(1.5,-6)}
   \put(4,0){$#7$}
   \put(4,-16){$#8$}
\end{pspicture}}

\newcommand{\bloivquainsulab}[4]{\begin{pspicture}(0,10)
   \put(1.0,10.25){\makebox(0,0)[b]{$#1$}}
   \put(1,0.5){\makebox(0,0)[b]{$#2$}}
   \put(2,5){\makebox(0,0)[r]{$#3$}}
   \put(1.5,-11){\makebox(0,0)[l]{$#4$}}
\end{pspicture}}

\newcommand{\bloivquainsu}[7]{\bloivquainsulab{#4}{#5}{#6}{#7}\bloivquainsubase{->}{->}{->}{<-}{<-}{<-}{#1}{#2}{#3}}
\newcommand{\bloivquainsuuno}[7]{\bloivquainsulab{#4}{#5}{#6}{#7}\bloivquainsubase{<-}{<-}{->}{->}{->}{<-}{#1}{#2}{#3}}
\newcommand{\bloivquainsudue}[7]{\bloivquainsulab{#4}{#5}{#6}{#7}\bloivquainsubase{<-}{->}{<-}{->}{<-}{->}{#1}{#2}{#3}}

\newcommand{\bloivquaoutgiubase}[7]{\begin{pspicture}(4,10)
   \put(0,0){\tratteggio}
   \put(0,10){\psline{#1}(1,0)}
   \put(1,10){\punto}
   \psline{#2}(1,10)(4,10)
   \put(1,0){\psline{#3}(0,10)(0,0.3)(0.3;80)(0.3;60)(0.3;40)(0.3;20)(0.3;0)(0.3;-20)(0.3;-40)(0.3;-60)(0.3;-80)(0,-0.3)(0,-6)(-1,-6)}
   \put(2.5,0){\piu}
   \put(2.5,0){#7}
   \psline{#5}(3,0)(4,0)
   \psline{#6}(2.5,-0.5)(2.5,-16)(0,-16)
   \psline{#4}(0,0)(2,0)
\end{pspicture}}

\newcommand{\bloivquaoutgiulab}[4]{\begin{pspicture}(0,10)
   \put(1,10.25){\makebox(0,0)[b]{$#1$}}
   \put(3.25,0.25){\makebox(0,0)[lb]{$#2$}}
   \put(1.5,5){\makebox(0,0)[l]{$#3$}}
   \put(2,-11){\makebox(0,0)[r]{$#4$}}
\end{pspicture}}

\newcommand{\bloivquaoutgiu}[5]{\bloivquaoutgiulab{#2}{#3}{#4}{#5}\bloivquaoutgiubase{<-}{<-}{->}{->}{->}{<-}{#1}}
\newcommand{\bloivquaoutgiuuno}[5]{\bloivquaoutgiulab{#2}{#3}{#4}{#5}\bloivquaoutgiubase{->}{->}{->}{<-}{<-}{<-}{#1}}
\newcommand{\bloivquaoutgiudue}[5]{\bloivquaoutgiulab{#2}{#3}{#4}{#5}\bloivquaoutgiubase{<-}{->}{<-}{->}{<-}{->}{#1}}

\newcommand{\bloivquaoutsubase}[7]{\begin{pspicture}(4,10)
   \put(0,0){\tratteggio}
   \put(0,0){\psline{#1}(3,0)}
   \put(3,0){\punto}
   \psline{#2}(3,0)(4,0)
   \psline{#3}(3,0)(3,-16)(0,-16)
   \put(1.5,10){\piu}
   \put(1.5,10){#7}
   \psline{#5}(2,10)(4,10)
   \put(1.5,0){\psline{#6}(0,9.5)(0,0.3)(0.3;80)(0.3;60)(0.3;40)(0.3;20)(0.3;0)(0.3;-20)(0.3;-40)(0.3;-60)(0.3;-80)(0,-0.3)(0,-6)(-1.5,-6)}
   \psline{#4}(0,10)(1,10)
\end{pspicture}}

\newcommand{\bloivquaoutsulab}[4]{\begin{pspicture}(0,10)
   \put(3,10.25){\makebox(0,0)[b]{$#1$}}
   \put(3,0.25){\makebox(0,0)[b]{$#2$}}
   \put(2,5){\makebox(0,0)[l]{$#3$}}
   \put(2.5,-11){\makebox(0,0)[r]{$#4$}}
\end{pspicture}}

\newcommand{\bloivquaoutsu}[5]{\bloivquaoutsulab{#2}{#3}{#4}{#5}\bloivquaoutsubase{<-}{<-}{->}{->}{->}{<-}{#1}}
\newcommand{\bloivquaoutsuuno}[5]{\bloivquaoutsulab{#2}{#3}{#4}{#5}\bloivquaoutsubase{->}{->}{->}{<-}{<-}{<-}{#1}}
\newcommand{\bloivquaoutsudue}[5]{\bloivquaoutsulab{#2}{#3}{#4}{#5}\bloivquaoutsubase{<-}{->}{<-}{->}{<-}{->}{#1}}

  \newcommand{\indxver}[4]{
  \rput(0,#2){
  \rput(#1,0){\psline(-0.25,0)(0.25,0)}
  \rput(-0.5,0){\rput[r](#1,0){\tiny #3}}
  \rput(0.5,0){\rput[l](#1,0){\tiny #4}}  }
  }
  \newcommand{\indxorzsu}[3]{
  \psline(#1,9.75)(#1,10.25)
  \rput[b](#1,10.5){\tiny #2}
  \rput[t](#1,9.5){\tiny #3}
  }
  \newcommand{\indxorzgiu}[3]{
  \psline(#1,-0.25)(#1,0.25)
  \rput[b](#1,0.5){\tiny #2}
  \rput[t](#1,-0.5){\tiny #3}
  }
  \newcommand{\loopxysnt}[5]{%
  \rput(1,4.5){\psline[#5]{#3}(-1,#2)(-#1,#2)}
  \rput(1,4.0){\psarc[#5](-#1,#2){0.5}{90}{180}}
  \rput(0.5,4){\psline[#5](-#1,#2)(-#1,1)}
  \rput(0.5,6){\psline[#5](-#1,-1)(-#1,-#2)}
  \rput(1,6.0){\psarc[#5](-#1,-#2){0.5}{180}{270}}
  \rput(1,5.5){\psline[#5]{#4}(-1,-#2)(-#1,-#2)}
  }
  \newcommand{\loopxydst}[5]{
  \rput(-1,4.5){\psline[#5]{#3}(1,#2)(#1,#2)}
  \rput(-1,4.0){\psarc[#5](#1,#2){0.5}{0}{90}}
  \rput(-0.5,4){\psline[#5](#1,#2)(#1,1)}
  \rput(-0.5,6){\psline[#5](#1,-1)(#1,-#2)}
  \rput(-1,6.0){\psarc[#5](#1,-#2){0.5}{270}{0}}
  \rput(-1,5.5){\psline[#5]{#4}(1,-#2)(#1,-#2)}
  }
  \newcommand{\loopxsnt}[4]{\loopxysnt{#1}{5}{#2}{#3}{#4}}

  \newcommand{\loopxdst}[4]{\loopxydst{#1}{5}{#2}{#3}{#4}}
  \newcommand{\loopxsntor}[2]{\loopxsnt{#1}{<-}{-}{#2}}
  \newcommand{\loopxsntaor}[2]{\loopxsnt{#1}{-}{<-}{#2}}
  \newcommand{\loopsntor}[1]{\loopxsnt{4}{<-}{-}{#1}}
  \newcommand{\loopsntaor}[1]{\loopxsnt{4}{-}{<-}{#1}}
  \newcommand{\loopvisntor}[1]{\loopxsnt{3}{<-}{-}{#1}}
  \newcommand{\loopvisntaor}[1]{\loopxsnt{3}{-}{<-}{#1}}
  \newcommand{\loopivsntor}[1]{\loopxsnt{2}{<-}{-}{#1}}
  \newcommand{\loopivsntaor}[1]{\loopxsnt{2}{-}{<-}{#1}}
  \newcommand{\loopxysntor}[3]{\loopxysnt{#1}{#2}{<-}{-}{#3}}
  \newcommand{\loopxysntaor}[3]{\loopxysnt{#1}{#2}{-}{<-}{#3}}
  \newcommand{\loopxdstor}[2]{\loopxdst{#1}{-}{<-}{#2}}
  \newcommand{\loopxdstaor}[2]{\loopxdst{#1}{<-}{-}{#2}}
  \newcommand{\loopdstor}[1]{\loopxdst{4}{-}{<-}{#1}}
  \newcommand{\loopdstaor}[1]{\loopxdst{4}{<-}{-}{#1}}
  \newcommand{\loopvidstor}[1]{\loopxdst{3}{-}{<-}{#1}}
  \newcommand{\loopvidstaor}[1]{\loopxdst{3}{<-}{-}{#1}}
  \newcommand{\loopivdstor}[1]{\loopxdst{2}{-}{<-}{#1}}
  \newcommand{\loopivdstaor}[1]{\loopxdst{2}{<-}{-}{#1}}
  \newcommand{\loopxydstor}[3]{\loopxydst{#1}{#2}{-}{<-}{#3}}
  \newcommand{\loopxydstaor}[3]{\loopxydst{#1}{#2}{<-}{-}{#3}}
  \newcommand{\loopor}[1]{\loopxdstor{4}{#1}\loopxsntor{4}{#1}  }
  \newcommand{\loopaor}[1]{\loopxdstaor{4}{#1} \loopxsntaor{4}{#1}  }
  \newcommand{\loopvior}[1]{\loopxdstor{3}{#1}\loopxsntor{3}{#1}  }
  \newcommand{\loopviaor}[1]{\loopxdstaor{3}{#1} \loopxsntaor{3}{#1}  }
  \newcommand{\loopivor}[1]{\loopxdstor{2}{#1}\loopxsntor{2}{#1}  }
  \newcommand{\loopivaor}[1]{\loopxdstaor{2}{#1} \loopxsntaor{2}{#1}  }
  \newcommand{\loopxor}[2]{\loopxdstor{#1}{#2}\loopxsntor{#1}{#2}  }
  \newcommand{\loopxaor}[2]{\loopxdstaor{#1}{#2} \loopxsntaor{#1}{#2}  }
  \newcommand{\loopxyor}[3]{\loopxydstor{#1}{#2}{#3}\loopxysntor{#1}{#2}{#3}  }
  \newcommand{\loopxyaor}[3]{\loopxydstaor{#1}{#2}{#3}\loopxysntaor{#1}{#2}{#3}  }

\newcommand{\addsusolo}[2]{\addsuvar{4}{#1}{#2}}
\newcommand{\tras}{^{\mbox{\tiny T}}}
\newcommand{\mtras}{^{\mbox{\tiny -T}}}
\newcommand{\dtras}{^{\mbox{\tiny T}\dag}}
\newcommand{\muno}{^{\mbox{\tiny -1}}}
\newcommand{\tracon}{^{*}}
\newcommand{\mtracon}{^{-*}}
\newcommand{\ds}{\displaystyle}
\newcommand{\ts}{\textstyle}
\newcommand{\scr}{\scriptstyle}
\newcommand{\sscr}{\scriptscriptstyle}
\newcommand{\bbar}[1]{\bar{\bar{#1}}}

 \newcommand{\evidenzia}[1]{\psframebox[framesep=5pt]{#1}}
 \newcommand{\BoxedEPSF}{\epsfbox} %% Primo capitolo
 \newcommand{\bo}{\mbox{\bf o}}
 \newcommand{\und}{\underline}
 \newcommand{\Span}{\mbox{\rm span}}
 \newcommand{\spann}{\mbox{\rm span}}
 \newcommand{\Agg}{\mbox{\rm agg}}
 \newcommand{\adj}{\mbox{\rm adj}}
 \newcommand{\Imm}{\mbox{\rm Im}}
 \newcommand{\Ker}{\mbox{\rm ker}}
 \newcommand{\rango}{\mbox{\rm rango}}
 \newcommand{\bM}{{\bf M}}
 \newcommand{\mat}[2]{\left[\begin{array}{#1} #2 \end{array}\right]}
 \newcommand{\mdet}[2]{\left|\begin{array}{#1} #2 \end{array}\right|}
 \newcommand{\cofbin}[2]{\left(\!\begin{array}{c}#1\\ #2 \end{array}\!\right)}

\def\a{{\bf a}} \def\b{{\bf b}} \def\c{{\bf c}} \def\d{{\bf d}}
\def\e{{\bf e}} \def\f{{\bf f}} \def\g{{\bf g}} \def\h{{\bf h}}
\def\ib{{\bf i}} \def\j{{\bf j}} \def\k{{\bf k}} \def\l{{\bf l}}
\def\m{{\bf m}} \def\n{{\bf n}} \def\o{{\bf o}} \def\p{{\bf p}}
\def\q{{\bf q}} \def\r{{\bf r}} \def\s{{\bf s}} \def\t{{\bf t}}
\def\u{{\bf u}} \def\v{{\bf v}} \def\w{{\bf w}} \def\x{{\bf x}}
\def\y{{\bf y}} \def\z{{\bf z}}

\def\A{{\bf A}} \def\B{{\bf B}} \def\C{{\bf C}} \def\D{{\bf D}}
\def\E{{\bf E}} \def\F{{\bf F}} \def\G{{\bf G}} \def\H{{\bf H}}
\def\I{{\bf I}} \def\J{{\bf J}} \def\K{{\bf K}} \def\L{{\bf L}}
\def\M{{\bf M}} \def\N{{\bf N}} \def\O{{\bf O}} \def\P{{\bf P}}
\def\Q{{\bf Q}} \def\R{{\bf R}} \def\S{{\bf S}} \def\T{{\bf T}}
\def\U{{\bf U}} \def\V{{\bf V}} \def\W{{\bf W}} \def\X{{\bf X}}
\def\Y{{\bf Y}} \def\Z{{\bf Z}}

\def\cA{{\cal A}} \def\cB{{\cal B}} \def\cC{{\cal C}} \def\cD{{\cal D}}
\def\cE{{\cal E}} \def\cF{{\cal F}} \def\cG{{\cal G}} \def\cH{{\cal H}}
\def\cI{{\cal I}} \def\cJ{{\cal J}} \def\cK{{\cal K}} \def\cL{{\cal L}}
\def\cM{{\cal M}} \def\cN{{\cal N}} \def\cO{{\cal O}} \def\cP{{\cal P}}
\def\cQ{{\cal Q}} \def\cR{{\cal R}} \def\cS{{\cal S}} \def\cT{{\cal T}}
\def\cU{{\cal U}} \def\cV{{\cal V}} \def\cW{{\cal W}} \def\cX{{\cal X}}
\def\cY{{\cal Y}} \def\cZ{{\cal Z}}

\def\ta{\tilde{\a}} \def\tb{\tilde{\b}} \def\tc{\tilde{\c}}
\def\td{\tilde{\d}} \def\te{\tilde{\e}} \def\tf{\tilde{\f}}
\def\tg{\tilde{\g}} \def\th{\tilde{\h}} \def\tib{\tilde{\i}}
\def\tj{\tilde{\j}} \def\tk{\tilde{\k}} \def\tl{\tilde{\l}}
\def\tm{\tilde{\m}} \def\tn{\tilde{\n}} \def\to{\tilde{\o}}
\def\tp{\tilde{\p}} \def\tq{\tilde{\q}} \def\tr{\tilde{\r}}
\def\tts{\tilde{\s}} \def\ttt{\tilde{\t}} \def\tu{\tilde{\u}}
\def\tv{\tilde{\v}} \def\tw{\tilde{\w}} \def\tx{\tilde{\x}}
\def\ty{\tilde{\y}} \def\tz{\tilde{\z}}

\def\tA{\tilde{\A}} \def\tB{\tilde{\B}} \def\tC{\tilde{\C}}
\def\tD{\tilde{\D}} \def\tE{\tilde{\E}} \def\tF{\tilde{\F}}
\def\tG{\tilde{\G}} \def\tH{\tilde{\H}} \def\tI{\tilde{\I}}
\def\tJ{\tilde{\J}} \def\tK{\tilde{\K}} \def\tL{\tilde{\L}}
\def\tM{\tilde{\M}} \def\tN{\tilde{\N}} \def\tO{\tilde{\O}}
\def\tP{\tilde{\P}} \def\tQ{\tilde{\Q}} \def\tR{\tilde{\R}}
\def\tS{\tilde{\S}} \def\tT{\tilde{\T}} \def\tU{\tilde{\U}}
\def\tV{\tilde{\V}} \def\tW{\tilde{\W}} \def\tX{\tilde{\X}}
\def\tY{\tilde{\Y}} \def\tZ{\tilde{\Z}}

\def\ooa{\overline{\a}} \def\oob{\overline{\b}}
\def\ooc{\overline{\c}} \def\ood{\overline{\d}}
\def\ooe{\overline{\e}} \def\oof{\overline{\f}}
\def\oog{\overline{\g}} \def\ooh{\overline{\h}}
\def\ooib{\overline{\ib}} \def\ooj{\overline{\j}}
\def\ook{\overline{\k}} \def\ool{\overline{\l}}
\def\oom{\overline{\m}} \def\oon{\overline{\n}}
\def\ooo{\overline{\o}} \def\oop{\overline{\p}}
\def\ooq{\overline{\q}} \def\oor{\overline{\r}}
\def\oos{\overline{\s}} \def\oot{\overline{\t}}
\def\oou{\overline{\u}} \def\oov{\overline{\v}}
\def\oow{\overline{\w}} \def\oox{\overline{\x}}
\def\ooy{\overline{\y}} \def\ooz{\overline{\z}}

\def\ooA{\overline{\A}} \def\ooB{\overline{\B}}
\def\ooC{\overline{\C}} \def\ooD{\overline{\D}}
\def\ooE{\overline{\E}} \def\ooF{\overline{\F}}
\def\ooG{\overline{\G}} \def\ooH{\overline{\H}}
\def\ooI{\overline{\I}} \def\ooJ{\overline{\J}}
\def\ooK{\overline{\K}} \def\ooL{\overline{\L}}
\def\ooM{\overline{\M}} \def\ooN{\overline{\N}}
\def\ooO{\overline{\O}} \def\ooP{\overline{\P}}
\def\ooQ{\overline{\Q}} \def\ooR{\overline{\R}}
\def\ooS{\overline{\S}} \def\ooT{\overline{\T}}
\def\ooU{\overline{\U}} \def\ooV{\overline{\V}}
\def\ooW{\overline{\W}} \def\ooX{\overline{\X}}
\def\ooY{\overline{\Y}} \def\ooZ{\overline{\Z}}

\def\va{\vec{\a}} \def\vb{\vec{\b}} \def\vc{\vec{\c}}
\def\vd{\vec{\d}} \def\ve{\vec{\e}} \def\vf{\vec{\f}}
\def\vg{\vec{\g}} \def\vh{\vec{\h}} \def\vib{\vec{\i}}
\def\vj{\vec{\j}} \def\vk{\vec{\k}} \def\vl{\vec{\l}}
\def\vm{\vec{\m}} \def\vn{\vec{\n}} \def\vo{\vec{\o}}
\def\vp{\vec{\p}} \def\vq{\vec{\q}} \def\vr{\vec{\r}}
\def\vts{\vec{\s}} \def\vtt{\vec{\t}} \def\vu{\vec{\u}}
\def\vv{\vec{\v}} \def\vw{\vec{\w}} \def\vx{\vec{\x}}
\def\vy{\vec{\y}} \def\vz{\vec{\z}}

\def\vA{\vec{\A}} \def\vB{\vec{\B}} \def\vC{\vec{\C}}
\def\vD{\vec{\D}} \def\vE{\vec{\E}} \def\vF{\vec{\F}}
\def\vG{\vec{\G}} \def\vH{\vec{\H}} \def\vI{\vec{\I}}
\def\vJ{\vec{\J}} \def\vK{\vec{\K}} \def\vL{\vec{\L}}
\def\vM{\vec{\M}} \def\vN{\vec{\N}} \def\vO{\vec{\O}}
\def\vP{\vec{\P}} \def\vQ{\vec{\Q}} \def\vR{\vec{\R}}
\def\vS{\vec{\S}} \def\vT{\vec{\T}} \def\vU{\vec{\U}}
\def\vV{\vec{\V}} \def\vW{\vec{\W}} \def\vX{\vec{\X}}
\def\vY{\vec{\Y}} \def\vZ{\vec{\Z}}

\def\ha{\hat{\a}} \def\hb{\hat{\b}} \def\hc{\hat{\c}}
\def\hd{\hat{\d}} \def\he{\hat{\e}} \def\hf{\hat{\f}}
\def\hg{\hat{\g}} \def\hh{\hat{\h}} \def\hib{\hat{\i}}
\def\hj{\hat{\j}} \def\hk{\hat{\k}} \def\hl{\hat{\l}}
\def\hm{\hat{\m}} \def\hn{\hat{\n}} \def\ho{\hat{\o}}
\def\hp{\hat{\p}} \def\hq{\hat{\q}} \def\hr{\hat{\r}}
\def\hts{\hat{\s}} \def\htt{\hat{\t}} \def\hu{\hat{\u}}
\def\hv{\hat{\v}} \def\hw{\hat{\w}} \def\hx{\hat{\x}}
\def\hy{\hat{\y}} \def\hz{\hat{\z}}

\def\hA{\hat{\A}} \def\hB{\hat{\B}} \def\hC{\hat{\C}}
\def\hD{\hat{\D}} \def\hE{\hat{\E}} \def\hF{\hat{\F}}
\def\hG{\hat{\G}} \def\hH{\hat{\H}} \def\hI{\hat{\I}}
\def\hJ{\hat{\J}} \def\hK{\hat{\K}} \def\hL{\hat{\L}}
\def\hM{\hat{\M}} \def\hN{\hat{\N}} \def\hO{\hat{\O}}
\def\hP{\hat{\P}} \def\hQ{\hat{\Q}} \def\hR{\hat{\R}}
\def\hS{\hat{\S}} \def\hT{\hat{\T}} \def\hU{\hat{\U}}
\def\hV{\hat{\V}} \def\hW{\hat{\W}} \def\hX{\hat{\X}}
\def\hY{\hat{\Y}} \def\hZ{\hat{\Z}}

 \newcommand{\btheta}{\boldsymbol{\theta}}
 \newcommand{\btau}{\boldsymbol{\tau}}

 \newcommand{\Evidenzia}[1]{{\magenta #1}}
 \newcommand{\ShortFig}{}

\definecolor{mygreen}{RGB}{28,172,0}

  \newcommand{\setsmall}{
 \lstset{language=Matlab,
        deletekeywords={zeros,length,max,min,interp2},
        basicstyle=\tiny,
        numberstyle=\tiny,
        commentstyle=\color{mygreen},
        keywordstyle=\color{blue},
        stepnumber=0.1,
        firstnumber = 1,
        numberfirstline=false,
        stepnumber=1,
        frame=single,
        rulecolor=\color{blue},
 }
 }

 \newtheorem{Theo}{\bf Theorem}

  \newtheorem{Remark}{\bf Remark}

\newcommand{\ENG}{1=1}                              %

\newcommand{\ItaEng}[2]{\ifnum\ENG\relax#2\else#1\fi}

\newcommand{\alza}[1]{\rule{0mm}{#1}}

  \newcommand{\ii}{\`i }
 \newcommand{\ea}{\'e }
 \newcommand{\ee}{\`e }
 \newcommand{\oo}{\'o }
 \newcommand{\uu}{\'u }
 \newcommand{\ac}{\'a }
 \newcommand{\ad}{\'a}

 \newcommand{\eps}{.}

\definecolor{mio_red}{rgb}{0.0, 0.0, 0.0}
\definecolor{mioo_red}{rgb}{0.0, 0.0, 0.0}

\newcommand{\bloivretaor}[2]{\begin{picture}(4,10)
   \put(0,0){\tratteggio}
   \put(2,10){\psline(-2,0)}
   \put(2, 0){
   \put(-1.5,3){
  \psframe(3,3)
  \rput(1.5,1.5){$#1$}
   }
   \put(0,6){\psline{->}(0,4.0)}
   \put(0,  0){\psline{->}(0,3)}
   \put(0,10.5){\makebox(0,0)[b]{$#2$}}}
   \put(0, 0){\psline(2,0)}
\end{picture}}

\journal{European Journal of Control}

\begin{document}

\begin{frontmatter}

\title{Efficient and Robust Modeling of Nonlinear Mechanical Systems}

\author{Davide Tebaldi\orcidlink{0000-0003-1432-0489}\corref{mycorrespondingauthor}}
\cortext[mycorrespondingauthor]{Corresponding author}
\ead{davide.tebaldi@unimore.it}

\author{Roberto Zanasi\orcidlink{0000-0001-5507-825X}}
\ead{roberto.zanasi@unimore.it}

\affiliation{organization={University of Modena and Reggio Emilia},
            addressline={Via Pietro Vivarelli 10 - int. 1},
            city={Modena},
            postcode={41125},
            country={Italy}}

\begin{abstract}
 The development of efficient and robust dynamic models is fundamental in the field of systems and control engineering. In this paper, a new formulation for the dynamic model of nonlinear mechanical systems, that can be applied to different automotive and robotic case studies, is proposed, together with a modeling procedure allowing to automatically obtain the model formulation.
Compared with the Euler-Lagrange formulation,
the proposed model is shown to give superior performances in terms of robustness against measurement noise for systems exhibiting dependence on some external variables, as well as in terms of execution time when computing the inverse dynamics of the system.
\end{abstract}

\begin{keyword}
Dynamic Modeling, Model Formulation, Physical Systems, Mechanical Systems.
\end{keyword}

\end{frontmatter}

%%%%%%%%%%%%%%%%%%%%%%%%%%%%%%%%%%%%%%%%%%%%%%%%%%%%%%%%%%%%%%%%%%%%%%%%
%%%%%%%%%%%%%%%%%%%%%%%%%%%%%%%%%%%%%%%%%%%%%%%%%%%%%%%%%%%%%%%%%%%%%%%%

\section{Introduction}
 The development of approaches to model physical systems in different energetic domains has been the subject of studies
 by many renowned scientists throughout history \cite{maxwell1861}. Indeed, the dynamic modeling of physical systems is fundamental
 to develop effective control strategies \cite{MatinLCSS2024}.

The choice of the adopted modeling approach historically depends on the systems energetic domains. When dealing with mechanical systems, the Lagrange equations \cite{lagrange1811mecanique} represent one of the most widespread approaches to derive the dynamic model of the system. Alternatives to the Lagrangian approach involve the initial definition of an augmented dynamic model of the system~\cite{TEBALDI2023105420}, followed by a model-order reduction assuming rigid connections. As far as electromechanical systems such as electric machines are concerned, they are typically modeled by directly writing the model equations in the desired reference frame \cite{Zhang2020stfr,C2024100931}.

In the context of automotive systems, the detailed modeling of the physical systems involved in a large variety of vehicles, such as planetary gear sets \cite{LhommeTVT2017}, permanent magnet synchronous motors \cite{TILLI201916} or Full Toroidal Variators (FTVs) \cite{FuchsSAE2002}, represents an important aspect for developing effective
 energy management strategies \cite{ZhangEnergies2020}.
This aspect is strictly related to the calculation of the system inverse dynamics, which allows to compute the optimal actuator torques for energy recovery or trajectory tracking in articulated vehicles. The inverse dynamics problem is also relevant in other areas of engineering, such as robotics for example,
 where the inverse dynamics control~\cite{SicilianoBook} represents the basis for
more articulated control approaches such as impedance control \cite{tsumugiwa2007stability} or admittance control \cite{SIRINTUNA2024104725}.

The modeling problem becomes more complex when multi-physics systems are considered \cite{CSL_New_last}, that are systems composed of physical elements belonging to different energetic domains.
In this context, approaches such as Bond Graph (BG)~\cite{6144754,GONZALEZ201825}, Power-Oriented Graphs (POG)~\cite{tebaldi2025pog,TEBALDI2023105420}, and Energetic Macroscopic Representation (EMR)~\cite{Li24EMR,Bouscayrol_2000} have been developed, based on the concept of having a unified framework for modeling physical systems independently of their energetic domains~\cite{kron1939tensor}. 

In this paper, we focus on the dynamic modeling of nonlinear mechanical systems whose energy $E_n$ can be expressed in the quadratic form
$E_n(\x,\dot\x)=\frac{1}{2}\dot\x\tras\M(\x)\dot\x$,
where $\M(\x)$ is a full rank symmetric matrix, while the terms $\x=[x_1,\,x_2,\,\ldots,x_N]$ and $\dot{\x}=[\dot{x}_1,\,\dot{x}_2,\,\ldots,\dot{x}_N]$ are the vector of the generalized coordinates and its time derivative, respectively.
In this context, the typically used description of the system dynamic model is the so-called Euler-Lagrange form~\cite{SicilianoBook}: 

 \begin{equation} \label{generic_8}
\M(\x)\,\ddot{\x}+\N(\x,\dot{\x})\,\dot\x = \boldsymbol{\tau}
\hspace{6mm}
\leftrightarrow
\hspace{6mm}
\M\,\ddot{\x}+\N\,\dot\x = \boldsymbol{\tau},
\end{equation}
where $\boldsymbol{\tau}$ is the generalized wrench vector. Hereinafter, the explicit dependencies ``$(\x)$'' and ``$(\x,\dot\x)$'' for matrices $\M$ and $\N$, respectively, will omitted for the sake of brevity.

In this paper, a new alternative formulation of the model \eqref{generic_8} is proposed, which is called {\it factorized} form and is based on the calculation of the characterizing factorization matrix $\T$. Together with this new formulation, we also provide the steps of a modeling procedure allowing to systematically obtained the model equations, as well as the factorization matrix $\T$. The derivative-causality and integral-causality implementations of the factorized model are derived through rigorous theorems supported by analytical proofs. 

Furthermore, the advantages of the new proposed factorized model are shown
in terms of:

\begin{enumerate}[a)]
\item better simulation performances of the factorized model for time-varying physical systems, 
such as the FTV in Sec.~\ref{FTV_section}, compared to the Euler-Lagrange formulation. For the FTV case study, it is shown that the simulation results given by the factorized model in the noisy scenario remain the same as the noiseless case.
\item better simulation performances of the factorized model in terms of execution time when computing the inverse dynamics of a robotic manipulator
compared to the Euler-Lagrange formulation.
\end{enumerate}

The remainder of this paper is structured as follows. The proposed factorized model is given in Sec.~\ref{POG_models_sect}, together with the proposed modeling procedure. The different case studies, showing the good performances of the proposed factorized model,
are addressed in Sec.~\ref{CCR_section}, Sec.~\ref{FTV_section} and Sec.~\ref{Chap_2_DOF}, respectively.
Finally, the appendix part is dedicated to analytical proof, and the conclusions are given in Sec.~\ref{conclusions_sect}.

 \section{Proposed Model Formulation and Modeling Procedure}\label{POG_models_sect}

This section deals with the introduction of the proposed factorized model in Sec.~\ref{fact_model_section}, and with the introduction of the proposed modeling procedure allowing to compute the factorization matrix in Sec.~\ref{prop_appr_section}.

\subsection{Factorized Model Formulation}\label{fact_model_section}

\vspace{1mm}
 \begin{Theo}\label{Prop_1}
 The Euler-Lagrange model \eqref{generic_8} can be   expressed in the following alternative {\it factorized} form:
 \begin{equation} \label{Factorized:_form}
   \T\tras\frac{d}{dt}\!\left(\T\,\dot\x\right)  =\boldsymbol{\tau},
\end{equation}
where the factorization matrix $\T$, which can also be rectangular with more rows than columns, is such that the matrices $\M$ and $\N$ in \eqref{generic_8} can be expressed as follows:
\begin{equation} \label{basic_relations}
\begin{array}{c}
 \M= \T\tras\T,
 \hspace{16.8mm}
 \N= \T\tras\dot\T,
\end{array}
\end{equation}
where $\dot\T$ is the time derivative of matrix $\T$.
 \hfill $\square$\vspace{1mm}
 \end{Theo}

\noindent {\em Proof.}  The left side of model \eqref{Factorized:_form} can be written as:
\begin{equation}\label{proof1}
 \begin{array}{@{}r@{}c@{}l@{}}
\T\tras \frac{d}{dt}( \T\dot\x)
=  \underbrace{\T\tras \T}_{\M} \ddot\x \!+\! \underbrace{\T\tras \dot\T}_{\N} \dot\x\!,
 \end{array}
\end{equation}
resulting in the Euler-Largange form \eqref{generic_8}.

The factorization matrix $\T$ can be obtained using the modeling procedure described in the next Sec.~\ref{prop_appr_section}.

Systems~\eqref{generic_8} and \eqref{Factorized:_form} are written using a derivative causality. While, for what concerns systems \eqref{generic_8}, the integral causality implementation can be obtained through direct inversion, the following theorem applies to the factorized system \eqref{Factorized:_form}.

\begin{Theo}\label{Prop_2}
The integral causality implementation of the factorized form \eqref{Factorized:_form} cannot be obtained through direct model inversion because,  when matrix $\T$ is rectangular, the dimension of vector $\T\dot\x$ in \eqref{Factorized:_form} is larger than that of vector $\x$.  Instead, the integral causality implementation of the factorized form \eqref{Factorized:_form} must be computed as follows:
  \begin{equation}\label{fact_integral}
\dot\x=
\ds \T^+ \int_0^t
\left((\T^+)\tras\boldsymbol{\tau}
+ \P \dot\T \dot\x
\right) dt,
  \end{equation}
where $\T^+ = (\T\tras\T)\muno\T\tras$ is the pseudo-inverse of matrix $\T$, and  $\P=\I - \T\T^+$  is an orthogonal  projection matrix on subspace {\em Ker}$(\T\tras)$.
   \hfill  $\square$\vspace{1mm}
\end{Theo}
  The proof of Theorem~\ref{Prop_2} is reported in App.~A. 
\vspace{1mm}

The Power-Oriented Graphs (POG)~\cite{tebaldi2025pog} representation of the integral implementations of the two models~\eqref{generic_8} and~\eqref{Factorized:_form}, which are suitable for direct implementation in the Simulink environment, are shown in Fig.~\ref{time_varying_dynamic_iii}. Specifically, Fig.~\ref{time_varying_dynamic_iii}(a) shows the power-oriented graph representation of the Euler-Lagrange model \eqref{generic_8} in integral form, while Fig.~\ref{time_varying_dynamic_iii}(b) shows the power-oriented graph representation of the factorized model \eqref{Factorized:_form} in integral form.

\begin{figure}[tb]
\centering
 \includegraphics[clip,width=0.68\columnwidth]{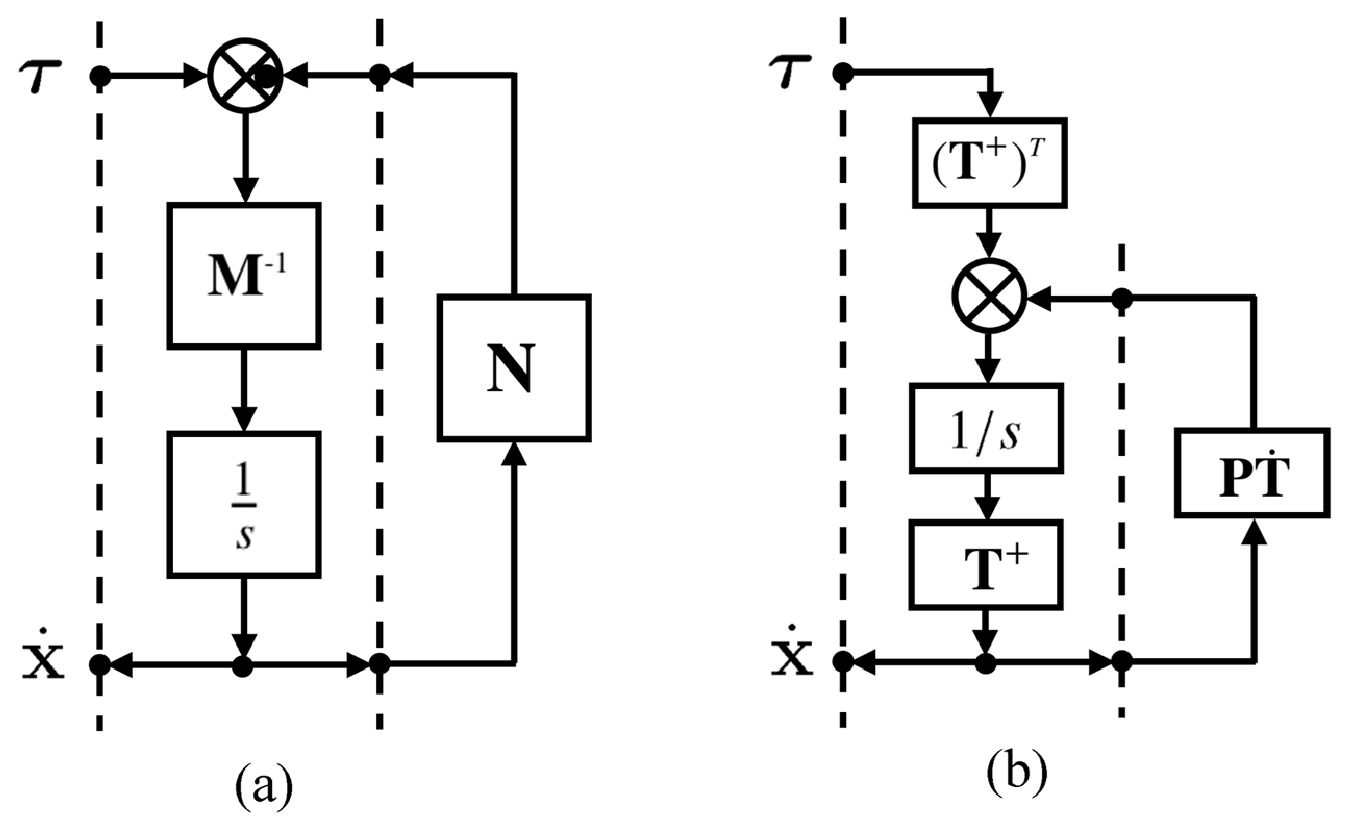}
    \vspace{-3.2mm}
  \caption{Power-oriented graph representation using integral causality of: (a) the Euler-Lagrange model \eqref{generic_8}, and (b) the factorized model \eqref{Factorized:_form}.}\label{time_varying_dynamic_iii}
    \vspace{-2mm}
\end{figure}

\vspace{2mm}

\subsection{Modeling Procedure}\label{prop_appr_section}
The dynamic equations of the considered class of mechanical systems can be systematically computed as follows:

\noindent 1) Let $n$ denote the number of coordinates variables $x_i$ present in the system, for $i \in [1,\,2,\,\ldots,\,n]$, and let $\x=[x_1,\; x_2,\; \ldots,\; x_n]\tras$  denote the  coordinates vector.

\noindent 2) Let $m$ denote the number of masses  $m_j$ in the system, for $j\in[1,\,2,\,\ldots,\,m]$, and let $r_j$ denote the degrees of freedom  in the space 
(translational and rotational) of mass $m_j$.

\noindent 3) Write the position $P_{j,k}(\x)$ of each degree of freedom $k$ of each  mass $m_j$, for  $k\in[1,\,2,\,\ldots,\,r_j]$ and $j\in[1,\,2,\,\ldots,\,m]$. Let $\P_j(\x)=[P_{j,1},\; P_{j,2},\; \ldots,\; P_{j,r_j}]\tras$ denote the position vector of each mass $m_j$, and let $\P_0(\x)=[\P_{j},\; \P_{j},\; \ldots,\; \P_{j}]\tras$ denote the overall mass position vector of the system.
At this point, other approaches in the literature would proceed in the direct writing of the differential equations characterizing the considered system \cite{karnopp2012modeling}. In this paper, we introduce the factorized model formulation \eqref{Factorized:_form} instead, and we propose an automatic way of computing the factorization matrix $\T$ characterizing it, through the subsequent steps.

\noindent 4) Let $N_{j,k}$ denote the inertial  coefficient (mass $m_{j,k}$ or inertial coefficient $J_{j,k}$) of each degree of freedom $k$  of each mass $m_j$. Let $\N_j=[N_{j,1},\; N_{j,2},\; \ldots,\; N_{j,r_j}]$ denote the inertial vector of mass $m_j$, and let $\N_0=\mbox{diag}([\N_{j},\; \N_{j},\; \ldots,\; \N_{j}])$ denote a diagonal matrix of all the inertial  coefficients of the system.

\noindent 5) The matrix $\T$ of the factorized model \eqref{Factorized:_form} can be computed as follows:
\begin{equation}
\label{T_def}
\T= \sqrt{\N_0}\, \H_0,
\hspace{5mm}
\mbox{where}
\hspace{5mm}
\H_0 = \frac{\partial \P_0(\x)}{\partial \x}.
\end{equation}
The system matrices $\M$ and $\N$ of the Euler-Lagrange model \eqref{generic_8}
can be obtained using~\eqref{basic_relations}.

%%%%%%%%%%%%%%%%%%%%%%%%%%%%%%%%%%%%%%%%%%%%%%%%%%%%%%%%%%%%%%%%%%%%%%%%%
\section{Dynamic Model of a crank-connecting rod system}\label{CCR_section}
%%%%%%%%%%%%%%%%%%%%%%%%%%%%%%%%%%%%%%%%%%%%%%%%%%%%%%%%%%%%%%%%%%%%%%%%%

Reference is made to the crank-connecting rod system shown in
Fig.~\ref{crank_connecting_rod} as a first case study, composed of two masses.
\begin{figure}[tbp]
  \centering\footnotesize
 %%%%%%%%%%%%%%%%
 \setlength{\unitlength}{3.68mm}
 \psset{unit=\unitlength}
 \SpecialCoor
 \begin{pspicture}(-2,-1.5)(21.0,8.0)
 \psline[linewidth=0.3pt]{->}(0,0)(20,0)
 \psline[linestyle=dashed,linewidth=0.3pt](4,3)(19.5,3)
 \psline[linestyle=dashed,linewidth=0.3pt](8;60)(4,0)
 \pscircle[fillstyle=solid,fillcolor=gray](0,0){0.55}
 \psline[linewidth=1.0pt]{o-o}(0,0)(8;60)
 \rput(16,3){\psframe[fillstyle=solid,fillcolor=gray](-1,-0.5)(1,0.5)}
 \psline[linewidth=1.0pt]{o-o}(8;60)(15,3)
 \psarc[linewidth=0.6pt]{->}{2}{0}{60}\
 \rput[lb](2.1;30){$\theta,\,\tau_m$}
 \pcline[offset=10pt,linewidth=0.3pt]{|-|}(0,0)(8;60)
 \lput*{U}{$R$}
 \pcline[offset=10pt,linewidth=0.3pt]{|-|}(8;60)(15,3)
 \lput*{U}{$L$}
 \psline[linewidth=0.3pt,linestyle=dashed](15,-1.25)(15,3)
 \pcline[offset=-1,linewidth=0.3pt]{|->}(0,0)(15,0)
 \lput*{U}{$x$}
 \pcline[linewidth=0.3pt]{<->}(18,0)(18,3)
 \lput*{U}{$d$}
 \psline[linewidth=1.0pt]{->}(17,3)(20,3)
 \rput[b](18.5,3.25){$F$}
 \rput[rt](-0.35,-0.2){$J_m$}
 \rput[lt]( 0.5,-0.2){$d_m$}
 \rput[b](16.2,3.6){$m_p$}
 \rput[t](16.2,2.4){$b_p$}
\end{pspicture}
 \caption{Structure of the crank-connecting rod system.}\label{crank_connecting_rod}
\end{figure}
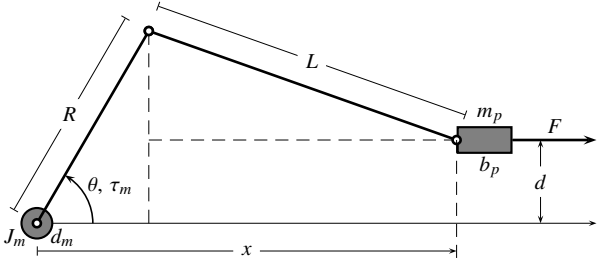
 Applying the procedure of Sec.~\ref{prop_appr_section}, 
in this case the system coordinate variable $x$ is the angular position $x=\theta$ shown in
Fig.~\ref{crank_connecting_rod}, and the generalized torque  $\tau$ is $\tau=-d_m-H^2(\theta)b_p+\tau_{m}+H(\theta)F$, where $H(\theta)$ is defined in \eqref{FTV_new_eq}.  Each mass is characterized by one degree of freedom, see step 2) of the modeling procedure.
The position vectors of the two masses at step 3) are the following:
\[
\P_1=\theta,
\hspace{10mm}
\P_2=R\cos(\theta)+\sqrt{L^2-(R \sin(\theta)-d)^2}.
\]
The vector $\P_0$ and the matrix $\N_0$
have the following structure:
\[
\P_0=\mat{@{\,}c@{\,}}{\P_1\\ \P_2},
\hspace{14mm}
\N_0=\mat{@{\,}c@{\;}c@{\,}}{J_m & 0 \\ 0 &  m_p},
\]
while the vectors $\H_0$ and the factorization matrix $\T$ are given by:
\begin{equation} \label{H0_T}
\H_0
\!=\!
\frac{\partial \P_0}{\partial \theta}
\!=\!
\mat{@{}c@{}}{1\\ H(\theta)}\!,
\hspace{10mm}
\T
\!=\!
\sqrt{\N_0}\, \H_0
\!=\!
\mat{@{}c@{}}{\sqrt{J_m}\\ \sqrt{m_p}H(\theta)}\!,\!
\end{equation}
where
\begin{equation}\label{FTV_new_eq}
H(\theta)
\!=\!
\frac{\partial \P_2}{\partial \theta}
\!=\! %\ts
-R\sin\theta +\frac{R\cos\theta(d-R\sin\theta)}{\sqrt{L^2-(R\sin\theta-d)^2}}.
\end{equation}
The matrix $\T$ in \eqref{H0_T} fully defines the proposed factorized model \eqref{Factorized:_form}. Additionally, using \eqref{basic_relations} and \eqref{H0_T}, the matrices $\M$ and $\N$ of the Euler-Lagrange model \eqref{generic_8} can be obtained  as follows:
 \[
 \begin{array}{l}
 \M(\theta)
 = \T\tras \T
 = M(\theta)
 = J_{m}+m_p\,H^2(\theta),
 \\[1mm]
   \N(\theta,\dot\theta)
   =\T\tras\dot\T
   = \frac{\dot{M}(\theta)}{2}
   =m_p\, H(\theta)\dot{H}(\theta).
 \end{array}
 \]
% %%%%%%%%%%%%%%%%%%%%%%%%%%%%%%%%%%%%%%%%%%%%

%%%%%%%%%%%%%%%%%%%%%%%%%%%%%%%%%%%%%%%%%%%%%%%%%%%%%%%%%%%%%%%%%%%%%%%%%%%%%%%%
\section{Dynamic Model of a Full Toroidal Variator}\label{FTV_section}

The Full Toroidal Variator (FTV) is a key element in automotive, that finds
applications such as Kinetic Energy Recovery System (KERS) and Infinitely Variable Transmission (IVT).
 A schematic of the considered FTV is shown in
Fig.~\ref{full_tor_variat_Grafico} \cite{TEBALDI2023105420},
where the speed-ratio
is a function of the rollers tilt angle $\theta_t=\theta_t(t)$.
%%%%%%%%%%%%%%%%%%%%%%%%%%%%%%%%%%%%%
\begin{figure}[t]
\centering \small \setlength{\unitlength}{5.96mm}
\psset{unit=\unitlength} \SpecialCoor
\begin{pspicture}(4,-1)(19,13)
\rput(-0.26,0){
\newrgbcolor{color_0}{0 0 0}
\newrgbcolor{color_1}{0.98        0.81        0.69}
\psline[fillstyle=solid,fillcolor=color_1,linewidth=0.9pt](6.25,4.5)(6.25,9.5)(7.75,9.5)(7.75,8.5)(7.6376,8.3916)(7.5353,8.2736)(7.4439,8.1469)(7.3642,8.0126)(7.2968,7.8717)(7.2422,7.7254)(7.2009,7.5748)(7.1731,7.4211)(7.1592,7.2655)(7.1592,7.1094)(7.1731,6.9538)(7.2009,6.8001)(7.2422,6.6495)(7.2968,6.5032)(7.3642,6.3623)(7.4439,6.228)(7.5353,6.1014)(7.6376,5.9833)(7.75,5.8749)(7.75,4.8749)(15.25,4.8749)(15.25,5.8749)(15.357,5.9777)(15.455,6.0891)(15.5431,6.2085)(15.6208,6.3349)(15.6875,6.4674)(15.7428,6.6051)(15.7862,6.747)(15.8174,6.892)(15.8362,7.0392)(15.8425,7.1875)(15.8425,7.1875)(15.8327,7.3726)(15.8033,7.5557)(15.7548,7.7347)(15.6875,7.9075)(15.6024,8.0722)(15.5003,8.227)(15.3824,8.3702)(15.25,8.5)(15.25,9.5)(16.75,9.5)(16.75,4.5)
\psline[fillstyle=solid,fillcolor=color_1,linewidth=0.9pt](6.25,4.5)(6.25,-0.5)(7.75,-0.5)(7.75,0.5)(7.6376,0.6084)(7.5353,0.7264)(7.4439,0.8531)(7.3642,0.9874)(7.2968,1.1283)(7.2422,1.2746)(7.2009,1.4252)(7.1731,1.5789)(7.1592,1.7345)(7.1592,1.8906)(7.1731,2.0462)(7.2009,2.1999)(7.2422,2.3505)(7.2968,2.4968)(7.3642,2.6377)(7.4439,2.772)(7.5353,2.8986)(7.6376,3.0167)(7.75,3.1251)(7.75,4.1251)(15.25,4.1251)(15.25,3.1251)(15.357,3.0223)(15.455,2.9109)(15.5431,2.7915)(15.6208,2.6651)(15.6875,2.5326)(15.7428,2.3949)(15.7862,2.253)(15.8174,2.108)(15.8362,1.9608)(15.8425,1.8125)(15.8425,1.8125)(15.8327,1.6274)(15.8033,1.4443)(15.7548,1.2653)(15.6875,1.0925)(15.6024,0.9278)(15.5003,0.773)(15.3824,0.6298)(15.25,0.5)(15.25,-0.5)(16.75,-0.5)(16.75,4.5)
\newrgbcolor{color_2}{0.55        0.55        0.55}
\psline[fillstyle=solid,fillcolor=color_2,linewidth=0.9pt](10,5)(10,5.9)(10.1184,6.0147)(10.2254,6.1401)(10.3202,6.275)(10.4018,6.4182)(10.4696,6.5684)(10.523,6.7244)(10.5615,6.8846)(10.5847,7.0478)(10.5925,7.2125)(10.5925,7.2125)(10.5847,7.3772)(10.5615,7.5404)(10.523,7.7006)(10.4696,7.8566)(10.4018,8.0068)(10.3202,8.15)(10.2254,8.2849)(10.1184,8.4103)(10,8.525)(10,9.525)(10.75,9.525)(10.75,10.275)(12.25,10.275)(12.25,9.525)(13,9.525)(13,8.525)(12.8876,8.4166)(12.7853,8.2986)(12.6939,8.1719)(12.6142,8.0376)(12.5468,7.8967)(12.4922,7.7504)(12.4509,7.5998)(12.4232,7.4461)(12.4093,7.2906)(12.4093,7.1344)(12.4232,6.9789)(12.4509,6.8252)(12.4922,6.6746)(12.5468,6.5283)(12.6142,6.3874)(12.6939,6.2531)(12.7853,6.1264)(12.8876,6.0084)(13,5.9)(13,5)(10,5)
\psline[fillstyle=solid,fillcolor=color_2,linewidth=0.9pt](10,4)(10,3.1)(10.1184,2.9853)(10.2254,2.8599)(10.3202,2.725)(10.4018,2.5818)(10.4696,2.4316)(10.523,2.2756)(10.5615,2.1154)(10.5847,1.9522)(10.5925,1.7875)(10.5925,1.7875)(10.5847,1.6228)(10.5615,1.4596)(10.523,1.2994)(10.4696,1.1434)(10.4018,0.9932)(10.3202,0.85)(10.2254,0.7151)(10.1184,0.5897)(10,0.475)(10,-0.525)(10.75,-0.525)(10.75,-1.275)(12.25,-1.275)(12.25,-0.525)(13,-0.525)(13,0.475)(12.8876,0.5834)(12.7853,0.7014)(12.6939,0.8281)(12.6142,0.9624)(12.5468,1.1033)(12.4922,1.2496)(12.4509,1.4002)(12.4232,1.5539)(12.4093,1.7094)(12.4093,1.8656)(12.4232,2.0211)(12.4509,2.1748)(12.4922,2.3254)(12.5468,2.4717)(12.6142,2.6126)(12.6939,2.7469)(12.7853,2.8736)(12.8876,2.9916)(13,3.1)(13,4)(10,4)
\newrgbcolor{color_3}{0.01     0.28           1}
\psline[fillstyle=solid,fillcolor=color_3,linewidth=0.9pt](7.408,6.9897)(7.3705,6.9724)(7.3364,6.9492)(7.3065,6.9207)(7.2817,6.8876)(7.2628,6.851)(7.2501,6.8117)(7.244,6.7708)(7.2448,6.7296)(7.2523,6.689)(7.2664,6.6502)(7.2868,6.6142)(7.3127,6.5821)(7.3436,6.5547)(7.3786,6.5328)(7.4167,6.5169)(7.4569,6.5075)(7.4981,6.5048)(7.5392,6.509)(7.579,6.5199)(10.279,7.5199)(10.3185,7.5383)(10.3542,7.5633)(10.385,7.5941)(10.41,7.6298)(10.4284,7.6693)(10.4397,7.7114)(10.4435,7.7548)(10.4435,7.7548)(10.4397,7.7982)(10.4284,7.8403)(10.41,7.8798)(10.385,7.9155)(10.3542,7.9463)(10.3185,7.9713)(10.279,7.9897)(10.2369,8.001)(10.1935,8.0048)(10.1501,8.001)(10.108,7.9897)(7.408,6.9897)
\psline[fillstyle=solid,fillcolor=color_3,linewidth=0.9pt](7.408,2.0103)(7.3705,2.0276)(7.3364,2.0508)(7.3065,2.0793)(7.2817,2.1124)(7.2628,2.149)(7.2501,2.1883)(7.244,2.2292)(7.2448,2.2704)(7.2523,2.311)(7.2664,2.3498)(7.2868,2.3858)(7.3127,2.4179)(7.3436,2.4453)(7.3786,2.4672)(7.4167,2.4831)(7.4569,2.4925)(7.4981,2.4952)(7.5392,2.491)(7.579,2.4801)(10.279,1.4801)(10.3185,1.4617)(10.3542,1.4367)(10.385,1.4059)(10.41,1.3702)(10.4284,1.3307)(10.4397,1.2886)(10.4435,1.2452)(10.4435,1.2452)(10.4397,1.2018)(10.4284,1.1597)(10.41,1.1202)(10.385,1.0845)(10.3542,1.0537)(10.3185,1.0287)(10.279,1.0103)(10.2369,0.999)(10.1935,0.9952)(10.1501,0.999)(10.108,1.0103)(7.408,2.0103)
\newrgbcolor{color_4}{1        0.44        0.37}
\psline[fillstyle=solid,fillcolor=color_4,linewidth=0.9pt](8,11.4)(8,11.7)(10.75,11.7)(10.75,12.4)(12.25,12.4)(12.25,11.4)
\psline[fillstyle=solid,fillcolor=color_4,linewidth=0.9pt](8,11.4)(8,11.1)(10.75,11.1)(10.75,10.4)(12.25,10.4)(12.25,11.4)
\newrgbcolor{color_5}{0.01        0.28           1}
\psline[fillstyle=solid,fillcolor=color_5,linewidth=0.9pt](15.392,6.4897)(15.4305,6.4791)(15.4702,6.4747)(15.5101,6.4768)(15.5492,6.4851)(15.5864,6.4996)(15.6209,6.5198)(15.6517,6.5453)(15.678,6.5753)(15.6993,6.6092)(15.7148,6.646)(15.7243,6.6848)(15.7275,6.7246)(15.7275,6.7246)(15.7223,6.7752)(15.7071,6.8236)(15.6823,6.868)(15.6491,6.9065)(15.6087,6.9374)(15.563,6.9595)(12.863,7.9595)(12.8232,7.9704)(12.7821,7.9746)(12.7409,7.9719)(12.7007,7.9625)(12.6626,7.9466)(12.6276,7.9247)(12.5967,7.8973)(12.5708,7.8652)(12.5504,7.8292)(12.5363,7.7904)(12.5288,7.7498)(12.528,7.7086)(12.5341,7.6677)(12.5468,7.6284)(12.5657,7.5918)(12.5905,7.5587)(12.6204,7.5302)(12.6545,7.507)(12.692,7.4897)(15.392,6.4897)
\newrgbcolor{color_6}{0.01        0.28           1}
\psline[fillstyle=solid,fillcolor=color_6,linewidth=0.9pt](15.517,2.0103)(15.5591,2.0216)(15.5986,2.04)(15.6343,2.065)(15.6651,2.0958)(15.6901,2.1315)(15.7085,2.171)(15.7198,2.2131)(15.7236,2.2565)(15.7236,2.2565)(15.719,2.3042)(15.7054,2.3502)(15.6833,2.3927)(15.6534,2.4302)(15.617,2.4613)(15.5753,2.4849)(15.5298,2.5001)(15.4823,2.5063)(15.4345,2.5034)(15.3881,2.4914)(12.6881,1.4914)(12.6546,1.4764)(12.6237,1.4566)(12.596,1.4325)(12.5721,1.4046)(12.5526,1.3735)(12.5379,1.3398)(12.5282,1.3044)(12.5239,1.268)(12.5249,1.2313)(12.5312,1.1951)(12.5428,1.1603)(12.5594,1.1275)(12.5806,1.0976)(12.606,1.071)(12.6349,1.0485)(12.6669,1.0304)(12.7012,1.0172)(12.737,1.0092)(12.7736,1.0065)(15.5286,2.0228)
\pcline[linecolor=color_1,linewidth=1.25pt]{->}(5.6,4.5)(4.4,4.5)
 \aput[1pt]{0}(0.5){$\vec\omega_d$} \bput[1pt]{0}(0.5){$\tau_d$}
\pcline[linecolor=color_2,linewidth=1.25pt]{->}(17.5,4.5)(18.7,4.5)
 \aput[1pt]{0}(0.5){$\vec\omega_b$} \bput[1pt]{0}(0.5){} % $\tau_b$
\pcline[linecolor=color_3,linewidth=1.25pt]{->}(8.1916,8.8626)(8.6084,7.7374)
 \aput[1pt]{0}(0.5){$\vec\omega_c$} \bput[1pt]{0}(0.5){} %$\tau_c$
\pcline[linecolor=color_4,linewidth=1.25pt]{->}(14.1,11.4)(12.9,11.4)
 \aput[1pt]{0}(0.5){$\vec\omega_a$} \bput[1pt]{0}(0.5){$\tau_a$}
\rput{0}(8.775,7.225){\psline[linewidth=0.4pt,linestyle=dashed](0,0)(1.2,0)}
\rput{0}(8.775,7.225){\psline[linewidth=0.4pt,linestyle=dashed](0,0)(1.2,0.4)}
\rput{0}(8.775,7.275){\psarc(0,0){1.25}{-2}{17}}
\rput{0}(10.16,6.896){\small$\theta_t$}
\pcline[linewidth=0.4pt]{|->}(8.775,4.5)(8.775,7.275)
\aput[1pt]{0}(0.5){$r_{b1}$}
\psline[linewidth=0.4pt,linestyle=dashed]{-}(8.775,7.275)(8.775,7.25)
\pcline[linewidth=0.4pt]{|->}(11,4.5)(11,7.81)
\aput[1pt]{0}(0.305){$\scr r_{b}(\!\theta_t\!)$}
\psline[linewidth=0.4pt,linestyle=dashed]{-}(11,7.81)(10.4,7.81)
\pcline[linewidth=0.4pt]{|->}(6.5,4.5)(6.5,6.65)
\bput[1pt]{0}(0.5){$\scr r_d(\!\theta_t\!)$}
\psline[linewidth=0.4pt,linestyle=dashed]{-}(6.5,6.65)(7.25,6.65)
\pcline[linewidth=0.4pt]{|->}(11.5,4.5)(11.5,10.3)
\bput[1pt]{0}(0.305){$r_{b2}$}
\psline[linewidth=0.4pt,linestyle=dashed]{-}(11.5,10.3)(11.5,10.3)
\pcline[linewidth=0.4pt]{|->}(11.5,11.4)(11.5,10.4)
\aput[1pt]{0}(0.5){$r_a$}
\psline[linewidth=0.4pt,linestyle=dashed]{-}(11.5,10.4)(11.5,10.4)
\pcline[linewidth=0.4pt]{|->}(12.775,8.425)(14.325,7.8)
\aput[1pt]{0}(0.5){$r_c$}
\psline[linewidth=0.4pt,linestyle=dashed]{-}(14.325,7.8)(14.1,7.25)
\pcline[linecolor=red,linewidth=1.25pt]{->}(10.5524,8.4782)(9.8476,8.2218)
\bput[1pt]{0}(0.5){$F_{bc}$}
\pcline[linecolor=red,linewidth=1.25pt]{-}(10.3816,8.0379)(10.5184,7.6621)
\pcline[linecolor=red,linewidth=1.25pt]{->}(6.6596,7.0418)(7.3644,7.2982)
\aput[1pt]{0}(0.5){$F_{ac}$}
\pcline[linecolor=red,linewidth=1.25pt]{-}(7.1936,6.8579)(7.3304,6.4821)
\pcline[linecolor=red,linewidth=1.25pt]{->}(10.45,10.715)(10.45,9.965)
\bput[1pt]{0}(0.5){$F_{ab}$}
\pcline[linecolor=red,linewidth=1.25pt]{-}(12.275,10.34)(10.725,10.34)
}
 \rput[l](14.1,11.4){$\theta$}
 \rput[l](11,11.8){$\scr J_a$}
 \rput[t](11.25,3.7){$\scr J_b$}
 \rput[lt](8.6,6.8){$\scr J_c$}
 \rput[l](6.5,3.7){$\scr J_d$}
%\psgrid[gridwidth=0.15pt,subgridwidth=0.1pt,subgriddiv=2,gridlabels=5pt](4,-1)(19,13)
\end{pspicture}
\caption{Schematic representation of the Full Toroidal Variator.}
\vspace{-2.2mm}\label{full_tor_variat_Grafico}
\end{figure}
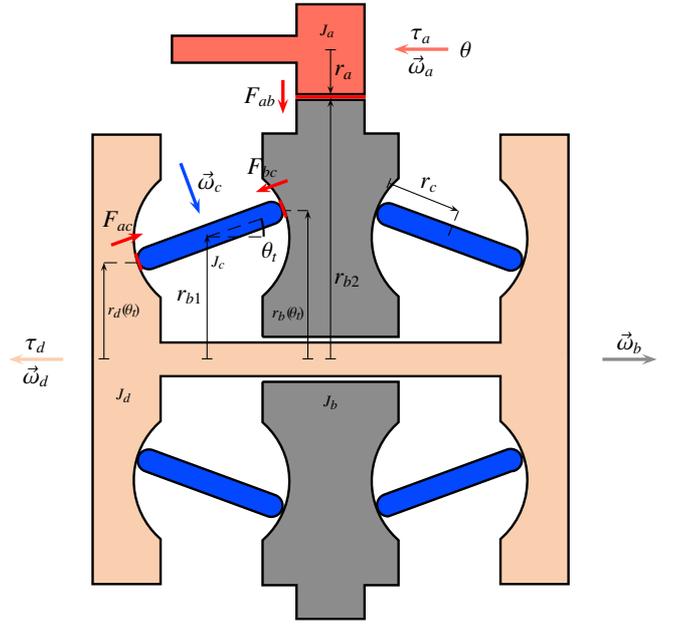
%%%%%%%%%%%%%%%%%%%%%%%%%%%%%%%%%%%%%
In this case, the system is composed by four rotating  masses, and  the generalized torque $\tau$ is $\tau=
-(b_a+R^2_1 b_b +  R^2_2(\theta_t) b_c + R^2_3(\theta_t) b_d )\omega_a+\tau_a
$, where $R_1$, $R_2(\theta_t)$ and $R_3(\theta_t)$ are defined in \eqref{FTV_new_eq_1}. The configuration variable $x$ at step 1) of the procedure proposed in Sec.~\ref{prop_appr_section} is the angular position $x=\theta_a=\theta$ of 
inertia $J_a$. Each mass is characterized by one degree of freedom, see step 2) of the modeling procedure.
  The position vector $\P_0$ of the system at step 3)  is given by:
\[
\P_0
\!=\!
\mat{@{}c@{}}{\P_a\\[1mm] \P_b\\[1mm] \P_c\\[1mm] \P_d\\ }
\!=\!
\mat{@{}c@{}}{\theta_a\\[1mm] \theta_b\\[1mm] \theta_c\\[1mm] \theta_d\\ }
\!=\!
\mat{@{}c@{}}{
\theta
\\[1mm]
\frac{r_a}{r_{b2}}\theta_a
\\[1mm]
\frac{r_b(\!\theta_t\!)}{r_{c}}\theta_b
\\[1mm]
\frac{r_{c}}{r_d(\!\theta_t\!)}\theta_c
}
\!=\!
\mat{@{}c@{}}{
1\\[1mm]
R_1\\[1mm]
R_2(\theta_t)\\[1mm]
R_3(\theta_t)
} \theta
=\R(\theta_t) \theta,
\]
 where the gear ratios $R_{1}$, $R_{2}(\theta_t)$ and $R_{3}(\theta_t)$ are defined as:
\begin{equation}\label{FTV_new_eq_1}
 R_1 \!=\!\frac{r_a}{r_{b2}},
  \hspace{4mm}
 R_2(\theta_t)
 \!=\!\frac{r_a\,r_{b}(\theta_t)}{r_{b2}\,r_c},
  \hspace{4mm}
 R_3(\theta_t)
 \!=\!\frac{r_a\,r_{b}(\theta_t)}{r_{b2}\,r_d(\theta_t)},
\end{equation}
 and the radii $r_{b}(\theta_t)$ and $r_d(\theta_t)$ are given by:
 \[
 r_{b}(\theta_t)\!=\!(r_{b1}\!+\!r_c\sin\theta_t),
 \hspace{16mm}
 r_d(\theta_t)\!=\!(r_{b1}\!-\!r_c\sin\theta_t).
 \]
The matrix $\N_0$ is 
% given by 
$\N_0=\mbox{diag}([J_a,\,J_b,\,J_c,\,J_d])$, while the
vectors $\H_0$ and the factorization matrix $\T$ have the following structure:
\begin{equation}\label{FTV_new_eq_2}
\H_0
\!=\!
\frac{\partial \P_0}{\partial \theta}
\!=\!
\R(\theta_t),
\hspace{15mm}
\T
\!=\!
\sqrt{\N_0}\, \H_0
\!=\!
\sqrt{\N_0} \R(\theta_t).
\end{equation}
 The matrix $\T$ in \eqref{FTV_new_eq_2} 
 fully defines the proposed factorized model \eqref{Factorized:_form}. Using \eqref{basic_relations} and \eqref{FTV_new_eq_2}, the matrices $\M$ and $\N$ of the Euler-Lagrange model \eqref{generic_8} can also be obtained as follows:
\begin{equation}\label{FTV_new_eq_2_new}
 \begin{array}{@{}l}
 \M(\theta_t)
 \!=\! \T\tras \T
 \!=\! M(\theta_t)
 \!=\! J_a
  \!+\! J_b R_1^2
  \!+\! J_c R_2^2(\theta_t)
  \!+\! J_d R_3^2(\theta_t),
 \\[2mm]
   \N(\theta_t,\dot{\theta}_t)
   \!=\!\T\tras\dot\T
   \!=\!   \frac{\dot{M}(\theta_t)}{2}
   \!=\!
   r_{b}(\theta_t)r_c\dot\theta_t\cos\theta_t R_1^2\left(
  \frac{J_c}{r_c^2}
  \!+\!
  \frac{2 r_{b1}J_d}{r_d^3(\theta_t)}\right)\!.
 \end{array}
\end{equation}

 \begin{Remark}
Matrix $\N(\theta_t,\dot{\theta}_t)$ in \eqref{FTV_new_eq_2_new}, present in the Euler-Lagrange model, 
is function of both $\theta_t$ and $\dot{\theta}_t$, where $\dot{\theta}_t$ is not an internal system variable and therefore must be computed numerically from the available measurements of the external variable $\theta_t$ representing the rollers tilt angle. This makes the Euler-Lagrange model sensitive to measurement noise, as further discussed in the results presented in Sec.~\ref{KERS_simul}.
The aforementioned drawback is not present in the proposed factorized model  \eqref{Factorized:_form},
because it does not depend on variable $\dot\theta_t$.
 \end{Remark}

\subsection{Simulation of the Full Toroidal Variator system}\label{KERS_simul}
\black

\begin{table}[t]
    \centering
    \caption{Parameters of the Full Toroidal Variator.}
    \label{tab:param_KERS}
\vspace{1.1mm}
    \begin{tabular}{c}
        \hline\\[-2mm]
        $J_a\!=\! 0.026,\; J_b\!=\! 7.56\cdot10^{-4},\; J_c\!=\!0.0018,\;J_d\!=\!0.179\;$ \mbox{ kg m$^2$} \\[1mm]
        \hline \\[-2mm]
        $r_a= 17.2,\; r_b= 51.5,\; r_c=27.5,\;r_d=34.3\;$ \mbox{mm} \\[1mm]
        \hline
    \end{tabular}
    \vspace{2mm}
\end{table}

The results given by the simulation of the FTV using the Euler-Lagrange model  \eqref{generic_8} and the factorized model  \eqref{Factorized:_form} have been compared in the following two conditions: 1) noiseless scenario; 2) noisy scenario.
The frictionless case and the parameters in Table \ref{tab:param_KERS} have been considered. The input torques have been assumed equal to zero, the initial condition is $\omega_{a0}=2000$~rpm, and the tilt angle is $\theta_t=1.22 \sin{(3\,t)}$. In the noisy case, a disturbance $d=12.2 \cdot 10^{-3} \sin{(3 \cdot 10^3\,t)}$ has been superimposed to the tilt angle $\theta_t$. The results are shown in Fig.~\ref{fig:KERS_results_figure}.
%%%%%%%%%%%%%%%%%%%%%%%%%%%%%%%%%%%%%%%%%%%%
\begin{figure}[t]
\centering
\begin{minipage}{0.48\columnwidth}
    \centering
    \includegraphics[width=1\linewidth]{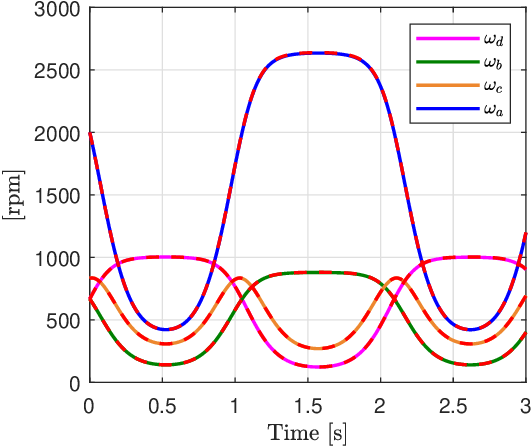}
\end{minipage}
%%%%%%%%%%%%%%%%%%%%
%%%%%%%%%%%%%%%%%%%%
\begin{minipage}{0.48\columnwidth}
    \centering
    \includegraphics[clip,width=1\linewidth]{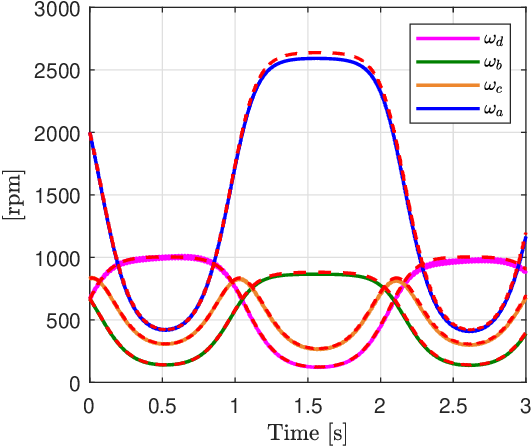}
\end{minipage}
%%%%%%%%%%%%%%%%%%%%
\psset{unit=\unitlength} \SpecialCoor
       \rput(-62.25,16.2){\footnotesize (a)}
   \rput(-18.88,16.2){\footnotesize (b)}
%%%%%%%%%%%%%%%%%%%%
\vspace{-0.4mm}
\caption{Full Toroidal Variator: Simulation results in the noiseless scenario (a) and noisy scenario (b). The colored results are given by the Euler-Lagrange model, while the red dashed results are given by the proposed factorized model.}\label{fig:KERS_results_figure}
\vspace{-2mm}
 \end{figure}
%%%%%%%%%%%%%%%%%%%%%%%%%%%%%%%%%%%%%%%%%%%%
It can be noticed that both the Euler-Lagrange model  and the factorized model  give the same results in the noiseless scenario. 
Conversely, in the noisy scenario, the proposed factorized model yields the same results as in the noiseless scenario, thanks to its independence from $\dot{\theta}_t$, whereas the performance of the Euler–Lagrange model degrades.
In fact, the maximum and mean percentage differences $\mbox{max}(|\Delta_{\omega}|)$ and $\mbox{mean}(|\Delta_{\omega}|)$ of the angular velocities given by the Euler-Lagrange model  with respect to the correct ones (given by the proposed factorized model) are shown in Table~\ref{tab:metrics_comparison_KERS}, showing a
non-negligible degradation of the Euler-Lagrange model results even with the considered small disturbance $d$.

\begin{table}[t]
    \centering
    \caption{Mean and maximum percentage differences
    of the angular velocities given by the Euler-Lagrange model, with respect to those given by the proposed factorized model,
    in the noisy scenario.}
    \vspace{1.1mm}
    \label{tab:metrics_comparison_KERS}
    \begin{tabular}{c@{\;\;\;\;\;\;\,}c@{\;\;\;\;\;\;\,}c@{\;\;\;\;\;\;\,}c@{\;\;\;\;\;\;\,}c}
        \hline\\[-4mm]
                              &    $\omega_a$ & $\omega_b$ & $\omega_c$ & $\omega_d$ \\[1mm]
                              %\hline \\[-2.8mm]
$\mbox{mean}(|\Delta_{\omega}|)$ [\%] &    1.73      &  1.73      &    1.74    &  1.81     \\[1mm]
$\mbox{max}(|\Delta_{\omega}|)$ [\%] &   2.35       &   2.35   &    3.60   &   5.02     \\[1mm]
        \hline
    \end{tabular}
    \vspace{2mm}
\end{table}

\section{Dynamic Model of a Planar robot with 2 DoF}\label{Chap_2_DOF}

Reference is made to the planar robot with 2 DoF shown in Fig.~\ref{planar_robot_with_2_DOF}, expressed
in the fixed frame
 $\Sigma_0$. Applying the modeling approach proposed in Sec.~\ref{prop_appr_section},
 %
%%%%%%%%%%%%%%%%
\begin{figure}[t]
 \begin{center}
 \setlength{\unitlength}{4.68mm}
 \psset{unit=\unitlength}
 \SpecialCoor
 \newrgbcolor{darkgreen}{0 0.5 0}
 \begin{pspicture}(-2,-1)(15.0,12.0)
 %
 %%%%%%%%%%%%%%%%%%%%%%%%%
 %%%%%%   Link 0
 %%%%%%%%%%%%%%%%%%%%%%%%%
 \rput[rb](0,0){
 \psline[linewidth=0.3pt]{|->}(0,0)(15,0)
 \psline[linewidth=0.3pt]{|->}(0,0)(0,12)
 \rput[r](-0.25,0.25){\normalsize $\Sigma_0$}
 \rput[r](-0.25,11.0){\normalsize $y$}
 \rput[t](14.0,-0.25){\normalsize $x$}
 }
 %%%%%%%%%%%%%%%%%%%%%%%%%
 %%%%%%   Link 1
 %%%%%%%%%%%%%%%%%%%%%%%%%
 \rput{30}(0,0){\psset{linecolor=blue}
 \psarc[linewidth=0.8pt,,linestyle=dashed]{->}{3}{-30}{0}     %  Joint 1
 \rput(3.1;-20){\rput[lb]{-30}(0,0){\small$\theta_1$}}
 \psline[linewidth=1.0pt]{o-}(0,0)(8;0)                    % point b1
 \pcline[linewidth=0.2pt,offset=20pt]{|-|}(0,0)(8;0)                    % point b1
 \lput*{:U}{$L_1$}
 \psline[linewidth=0.3pt]{|->}(0,0)(12,0)                   % Freme  1
 \pcline[linewidth=0.2pt,offset=10pt]{|-|}(0,0)(4.5;0)
 \lput*{:U}{$r_1$}
 \pscircle[fillstyle=solid,fillcolor=gray](4.5;0){0.42}    % mass N1
 %
 %%%%%%%%%%%%%%%%%%%%%%%%%
 %%%%%%   Link 2
 %%%%%%%%%%%%%%%%%%%%%%%%%
 \rput{30}(8,0){\psset{linecolor=red}
% \psellipse(4,0.125)(5.5,2)
 \psarc[linewidth=0.8pt,linestyle=dashed]{->}{3}{-30}{0}                 % Joint 2
 \rput(3.1;-20){\rput[lb]{-30}(0,0){\small$\theta_2$}}
 \psline[linewidth=1.0pt]{o->}(0,0)(8;0)                % point b2
 \psline[linewidth=1.0pt]{o}(8;0)(8;0)                % point b2
 \pcline[linewidth=0.2pt,offset=20pt]{|-|}(0,0)(8;0)                    % point b1
 \lput*{:U}{$L_2$}
 \psline[linewidth=1.0pt]{->}(5.5;0)
 \pcline[linewidth=0.2pt,offset=10pt]{|-|}(0,0)(5;0)
 \lput*{:U}{$r_2$}
 \pscircle[fillstyle=solid,fillcolor=gray](5.0;0){0.45}    % mass N2
 }  %%%% End Link 2
 }  %%%% End Link 1
 \rput[t](4.3,1.85){\footnotesize $m_1,I_{z1}$}
 \rput[t](10.3,7.85){\footnotesize $m_2,I_{z2}$}
 \rput[b](4.0,2.75){\footnotesize $\P_1$}
 \rput[b](9.35,8.9){\footnotesize $\P_2$}
 \rput[b](11,11.15){\footnotesize $P_e$}
\end{pspicture}
 \end{center}
 \vspace{-5.68mm}
  \caption{Planar robot with $2$-DoF expressed in the fixed frame $\Sigma_0$.}\label{planar_robot_with_2_DOF}
\end{figure}
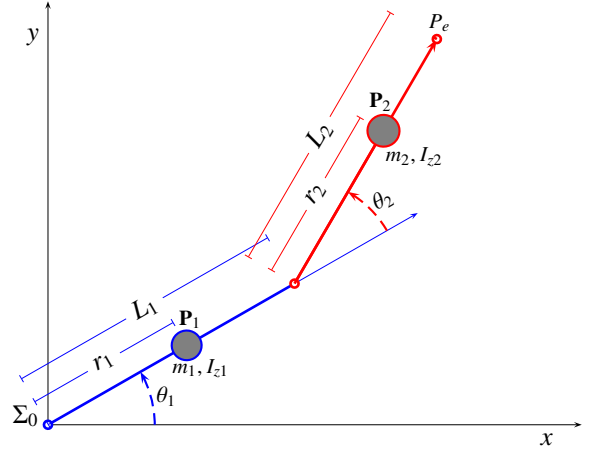
the planar positions $\P_1$ and $\P_2$ of the two masses $m_1$ and $m_2$ are the following:
 %%%%
 \begin{equation} \label{P1_P2}
 \P_1 = \theta_1,
\hspace{10mm}
 \P_2\!=\!\mat{@{}c@{}}{L_1\! \cos(\theta_1)\!+\!r_2\!  \cos(\theta_1\!+\!\theta_2)\\
 L_1\! \sin(\theta_1)\!+\!r_2\!  \sin(\theta_1\!+\!\theta_2)\\ \theta_1\!+\!\theta_2}.\!
 \end{equation}
% } 
The vectors $\x$, $\P_0$ and the matrix $\H_0(\x)$ are given by:
\[
\x
\!=\!
\mat{@{}c@{}}{x_1\\ x_2}
\!=\!
\mat{@{}c@{}}{\theta_1\\ \theta_2},
\hspace{4mm}
\P_0
\!=\!
\mat{@{}c@{}}{\P_1\\ \P_2},
\hspace{4mm}
\H_0(\x)
\!=\!
\frac{\partial \P_0}{\partial \x}
\!=\!
\mat{@{}c@{\,}c@{}}{\ds
\frac{\partial \P_0}{\partial \theta_1}
&\ds
\frac{\partial \P_0}{\partial \theta_2}
}.
\]
The diagonal matrix $\N_0$ has the following structure:
\[
\N_0=\mbox{diag}\{[m_1 r_1^2 + I_{z1},\;\;m_2,\;\;m_2,\;\;I_{z2}]\},
\]
while the matrix $\T$ characterizing the proposed factorized model  \eqref{Factorized:_form} can be computed as
$
\T
\!=\!
\sqrt{\N_0}\, \H_0.
$
The matrices of the Euler-Lagrange model  \eqref{generic_8} can be computed using \eqref{basic_relations}. Specifically, matrix $\M$ has the following structure:
\begin{equation} \label{ooL_theta}
\M(\x)
=
\T\tras \T
=
\left\lbrack \begin{array}{cc}
\alpha +2\beta \cos \theta_2
& \delta +\beta \cos \theta_2 \\\delta +\beta \cos \theta_2  & \delta \end{array}\right\rbrack,
\end{equation}
where the parameters $\alpha$, $\beta$ and $\delta$ are defined as follows:
\[
\begin{array}{c}
\alpha =I_{\mathrm{z1}}  +m_1 r_1^2 +m_2 L_1^2 +I_{\mathrm{z2}} +m_2 r_2^2,
\\[2mm]
\beta =m_2 L_1 r_2,
\hspace{10mm}
\delta =I_{\mathrm{z2}} +m_2 r_2^2.
\end{array}
\]
Finally, matrix 
$\N$ has the following structure:
\[
\N(\x,\dot\x)
\!=\!
\T\tras\dot\T
=
\mat{@{}c@{\;\;}c@{}}{
-\beta\, \sin\theta_2 \,\dot\theta_2 & -\beta\, \sin\theta_2 (\dot\theta_1 \!+\! \dot\theta_2)\\
 \beta\,  \sin\theta_2 \,\dot\theta_1 &         0}.
\]

\subsection{Simulation of the Planar Robot With 2 DOF}\label{Planar_simul}
Many applications require robots to be controlled using inverse dynamics control~\cite{SicilianoBook}.
The latter approach necessarily requires to compute the inverse dynamics of the manipulator. Assuming friction and gravity terms to be already compensated for, the inverse dynamics of robotic systems is typically expressed using the Euler-Lagrange model  \eqref{generic_8}~\cite{SicilianoBook}:
\begin{equation}\label{robot_eq}
\btau=\B(\x)\ddot{\x}+\C(\x,\dot{\x})\dot{\x}
=
\M\ddot\x+\N\dot\x,
\end{equation}
where
$\dot{\x}$ is the state vector, $\M=\B(\x)$ is the inertia matrix, and  $\N=\C(\x,\dot{\x})$ is the Coriolis forces matrix. In this section, we compare the Euler-Lagrange model  \eqref{robot_eq}
and the proposed factorized model   \eqref{Factorized:_form} in computing the inverse dynamics vector $\btau=\mat{@{}c@{\;\;}c@{}}{\tau_1 & \tau_2}\tras$ of the considered 2-DoF planar robot in~\cite{PlanarSimul2025} when the desired joint space  trajectory $\x=\mat{@{}c@{\;\;}c@{}}{\theta_1 & \theta_2}\tras$ is the one shown in Fig.~\ref{fig:planar_robot_traj_figure}.
The two models have been implemented in the Simulink environment, using MATLAB R2023b, and simulated
 in automatic solver selection mode in order to let MATLAB choose the best solver for each simulation, by using the same solver details. The results in terms of torque vector $\btau=\mat{@{}c@{\;\;}c@{}}{\tau_1 & \tau_2}\tras$ are very close,
 as shown in Table~\ref{tab:metrics_comparison_tau12} by the maximum and mean percentage differences $\mbox{max}(|\Delta{\tau}_1|)$, $\mbox{mean}(|\Delta{\tau}_1|)$, $\mbox{max}(|\Delta{\tau}_2|)$ and $\mbox{mean}(|\Delta{\tau}_2|)$ between the torque profiles $\tau_1$ and $\tau_2$ generated by the two models. At the same time, Fig.~\ref{Figure_2_planar_robot} shows the lower execution time (computed for $100$ executions)  given by the factorized model  compared to the Euler-Lagrange model.
 The average reduction in the execution time
 given by the factorized model  is $17.08$~\% compared to the Euler-Lagrange model, showing the better performances of the factorized model  over the
Euler-Lagrange one.

 \begin{figure}[t]
\centering
\begin{minipage}{0.48\columnwidth}
    \centering
    \includegraphics[width=1\linewidth]{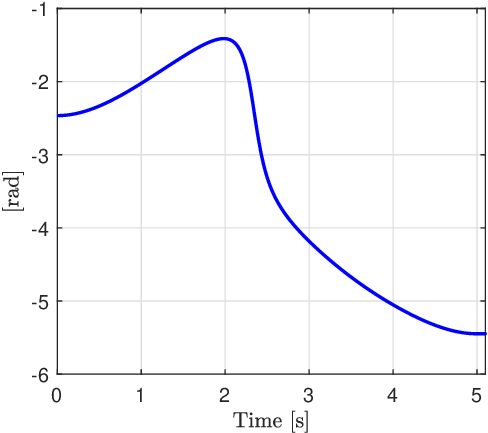}

\end{minipage}
%%%%%%%%%%%%%%%%%%%%
%%%%%%%%%%%%%%%%%%%%
\begin{minipage}{0.48\columnwidth}
    \centering
    \includegraphics[clip,width=1\linewidth]{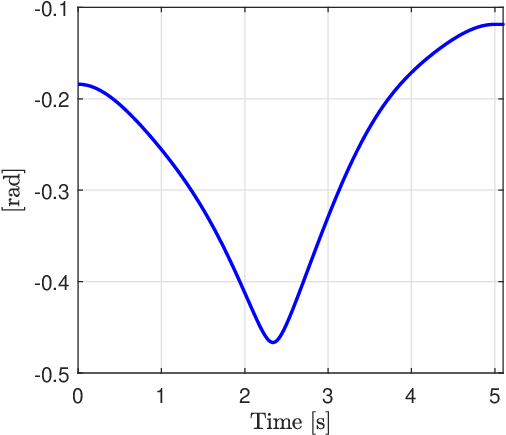}
\end{minipage}
%%%%%%%%%%%%%%%%%%%%
\psset{unit=\unitlength} \SpecialCoor
       \rput(-62.25,16.2){\footnotesize (a)}
   \rput(-18.88,16.2){\footnotesize (b)}
%%%%%%%%%%%%%%%%%%%%
\vspace{-0.8mm}
\caption{Considered trajectory for the 2-DoF planar robot: (a) $\theta_1$ and (b) $\theta_2$.}\label{fig:planar_robot_traj_figure}
\vspace{-2.08mm}
 \end{figure}

\begin{table}[t]
    \centering
    \caption{
    Mean and maximum percentage differences between the torque profiles $\tau_1$ and $\tau_2$ generated by the two model formulations: Euler-Lagrange model  and the proposed factorized model.
    }
\label{tab:metrics_comparison_tau12}
    \vspace{1.1mm}
    {\small
    \begin{tabular}{c@{\;\;\;}c@{\;\;\;}c@{\;\;\;}c}
        \hline\\[-2mm]
$\mbox{mean}(|\Delta{\tau}_1|)\!$ [\%] & $\mbox{max}(|\Delta{\tau}_1)|\!$ [\%] & $\mbox{mean}(|\Delta{\tau}_2|)\!$ [\%] & $\mbox{max}(|\Delta{\tau}_2|)\!$ [\%] \\[1mm]
 $0.071$      &  $0.016 $     &    $0.32$    &  $0.32$   \\[1mm]
        \hline
    \end{tabular}
    }
    \vspace{-2mm}
\end{table}

\begin{figure}[t]
\centering   
\includegraphics[clip,width=0.88\linewidth]{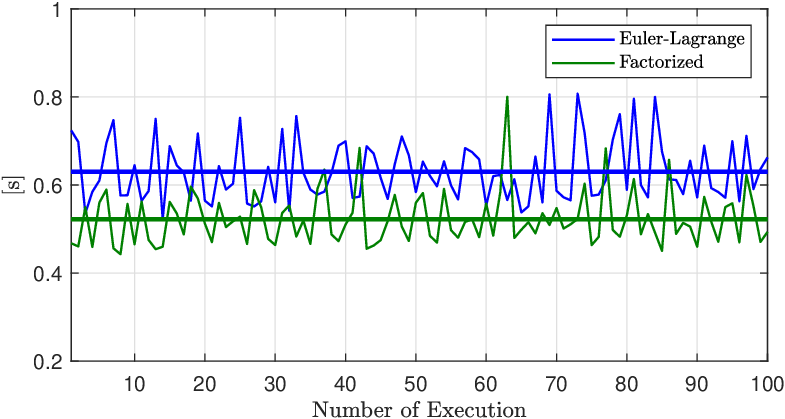} 
%%%%%%%%%%%%%%%%%%%%
%%%%%%%%%%%%%%%%%%%%
\vspace{-2mm}
\caption{Execution time of the Euler-Lagrange model  and of the factorized model  in computing the inverse dynamics over $100$ executions. The horizontal lines represent the average execution times. The average reduction in the execution time
 given by the factorized model  is $17.08$~\% compared to the Euler-Lagrange model.}\label{Figure_2_planar_robot}
\vspace{-3.6mm}
 \end{figure}
 %

%%%%%%%%%%%%%%%%%%%%%%%%%%%%%%%%%%%%%%%%%%%%%%%%%%%%%%%%%%%%%%%%%%%%%%%%%%%%%%%%%%%%%%%%%%%%%%%%%%%%%

\section{Conclusions}\label{conclusions_sect}
This paper has addressed the proposal of a new formulation for the dynamic model of nonlinear mechanical systems - called factorized model formulation - as well as the proposal of modeling approach enabling the user to automatically derive the system model.
The proposed factorized model  can be applied to different physical system case studies, and has been shown to give superior performance when compared with the Euler-Lagrange model  in terms of: robustness against measurement noise for physical systems where the system matrices exhibit a dependence on some external variables, and in terms of execution time when computing the inverse dynamics of the system.

\setcounter{secnumdepth}{0}
\section{Funding}
The work was partly supported by the University of Modena and Reggio Emilia
through the action FARD (Finanziamento Ateneo Ricerca Dipartimentale) 2023-2024, and funded under the National Recovery and Resilience Plan (NRRP), Mission 04 Component 2 Investment 1.5 - NextGenerationEU, Call for tender n. 3277 dated 30/12/2021 Award Number:  0001052 dated 23/06/2022.

\appendix
\renewcommand{\theequation}{A.\arabic{equation}}
\setcounter{equation}{0} \setcounter{secnumdepth}{0}
\section{Appendix A: Proof of Theorem~\ref{Prop_2}.}
\label{app2}

\noindent

Let $\T^+$ denote the Moore-Penrose pseudo-inverse of matrix $\T$. If matrix $\T$ in \eqref{basic_relations} is square, then $\T^+=\T\muno$. If $\T$ is rectangular, then it surely has more rows than columns since $\M$ is full rank, and thus
$\T^+=(\T\tras\T)\muno\T\tras$. 
Matrix $\T\T^+=\T(\T\tras\T)\muno\T\tras$ is an orthogonal projection matrix on subspace $\mbox{Im}(\T)$, and as such it is symmetric:   
\begin{equation}\label{symmetric}
\T\T^+
=\T(\T\tras\T)\muno\T\tras
=(\T(\T\tras\T)\muno\T\tras)\tras
=
(\T\T^+)\tras.
\end{equation}
Using \eqref{symmetric}, Eq.~\eqref{Factorized:_form} can be rewritten as follows:
\begin{equation}\label{fact_3}
\begin{array}{r@{\;}c@{\;}l}
(\T^+)\tras\T\tras \frac{d}{dt}(\T\dot\x)
&=&
 (\T^+)\tras \boldsymbol{\tau}
 \\[2mm]
(\T\T^+)\tras \frac{d}{dt}(\T\dot\x)
&=&
 (\T^+)\tras \boldsymbol{\tau}
 \\[2mm]
(\T\T^+)\frac{d}{dt}(\T\dot\x)
&=&
 (\T^+)\tras \boldsymbol{\tau}.
 \end{array}
\end{equation}
Let $\P\!=\!(\I-\T\T^+)$ denote the orthogonal projection matrix on subspace Ker$(\T\tras)$. Using $\P$, the following equalities can be obtained from \eqref{fact_3}:
\begin{equation}\label{fact_4}
\begin{array}{r@{\;}c@{\;}l}
(\T\T^+) \frac{d}{dt}(\T\dot\x)
&=&
 (\T^+)\tras \boldsymbol{\tau} \\[2mm]
(\I-\P) \frac{d}{dt}(\T\dot\x)
&=&
 (\T^+)\tras \boldsymbol{\tau}\\[2mm]
 \frac{d}{dt}(\T\dot\x)
&=& \P\frac{d}{dt}(\T\dot\x)+
 (\T^+)\tras \boldsymbol{\tau}\\[2mm]
 \frac{d}{dt}(\T\dot\x)
&=& \P(\dot{\T}\dot{\x}+\T\ddot{\x} )+
 (\T^+)\tras \boldsymbol{\tau}
 \\[2mm]
 \frac{d}{dt}(\T\dot\x)
&=& \P\dot{\T}\dot{\x}+
 (\T^+)\tras \boldsymbol{\tau},
 \end{array}
\end{equation}
where the latter equality holds because $\T\ddot{\x}\in \mbox{Im}(\T)$. Integrating both sides of the last equality in \eqref{fact_4} and left-multiplying by matrix $\T^+$ yield Eq.~\eqref{fact_integral}.

\bibliographystyle{elsarticle-num}
\bibliography{references}

\end{document}